\let\oldvec=\vec

\documentclass[smallcondensed]{svjour3}

\let\vec=\oldvec

\journalname{Journal of Automated Reasoning}

\usepackage{alltt}
\usepackage[isolatin]{inputenc}
\usepackage{amssymb}
\usepackage{tikz}
\usepackage{infrulex}
\usepackage{syntaxdefs}

\def\size#1{|#1|}
\def\id{{\it id}}
\def\sig{{\it sig}}
\def\chunk{\kappa}
\def\cst{{\it cst}}
\def\unop{{\it op}_1}
\def\binop{{\it op}_2}
\def\fn{F}
\def\fd{{\it Fd}}
\def\ef{{\it Fe}}
\def\prog{P}
\def\ofs{\delta}
\def\ev{\nu}
\def\event#1#2#3{#1(#2 \mapsto #3)}
\def\eval{\Rightarrow}
\def\evaltrace#1{\stackrel{#1}{\Rightarrow}}
\def\evalinf#1{\stackrel{#1}{\Rightarrow} \infty}
\def\twocases#1#2#3#4{\cases{#1 & #2 \cr #3 & #4 \cr}}
\def\threecases#1#2#3#4#5#6{\cases{#1 & #2 \cr #3 & #4 \cr #5 & #6 \cr}}
\def\fourcases#1#2#3#4#5#6#7#8{\cases{#1 & #2 \cr #3 & #4 \cr #5 & #6 \cr #7 & #8 \cr}}
\def\state{{\cal S}}
\def\callstate{{\cal C}}
\def\returnstate{{\cal R}}
\def\stackframe{{\cal F}}

\def\redtrace#1{\mathrel{\smash{\stackrel{#1}{\rightarrow}}\vphantom{\rightarrow}}}
\def\red{\redtrace{\epsilon}}
\def\redm#1{\redtrace{#1}^*}
\def\redp#1{\redtrace{#1}^+}
\def\redinf#1{\redtrace{#1} \infty}

\def\becomes{\leftarrow}
\def\sp{\sigma}
\def\trans#1#2{G \vdash #1 \red #2}
\def\transtrace#1#2#3{G \vdash #1 \redtrace{#2} #3}
\def\transm#1#2#3{G \vdash #1 \redm{#2} #3}
\def\transp#1#2#3{G \vdash #1 \redp{#2} #3}
\def\transinf#1#2{G \vdash #1 \redinf{#2}}
\def\Some#1{\lfloor #1 \rfloor}
\def\None{\emptyset}
\def\lbl{{\it lbl}}
\def\st{S}
\def\initialstate#1#2{{\tt initial}(#1,#2)}
\def\finalstate#1#2{{\tt final}(#1,#2)}
\def\lang#1{{\sf #1}}
\def\match{\sim}
\def\data{{\it data}}
\def\lbl{{\it lbl}}
\def\op{{\it op}}
\def\mode{{\it mode}}
\def\cond{{\it cond}}
\def\tr#1{[\![#1]\!]}
\def\Range{{\rm Rng}}
\def\ltrue{l_{\it true}}
\def\lfalse{l_{\it false}}
\def\lexit{l_e}
\def\lret{l_r}
\def\rret{r_r^?}
\def\lgoto{l_g}
\def\defequal{~~\stackrel{{\rm def}}{=}~~}
\def\loc{\ell}
\def\slot{s}
\def\mreg{r_m}
\def\ty{\tau}
\def\stk{\Sigma} %S

\def\retaddr{{\it ra}}

\def\parent#1#2{#2.{\tt top}.#1}
\def\ireg{r_i}
\def\freg{r_f}
\def\crbit{r_c}

\def\gpic{\begin{center}}
\def\endgpic{\end{center}}

\def\eg{e.g.\ }
\def\ie{i.e.\ }
\def\alt{\mid}
\def\th{^{\mbox{\scriptsize th}}}
\def\union{\cup}
\def\Dom{{\rm Dom}}
\def\fun{\rightarrow}
\def\citeurl#1{{\tt #1}}

\begin{document}

\title{A formally verified compiler back-end}
            
\author{Xavier Leroy}
\institute{
X. Leroy \at
INRIA Paris-Rocquencourt, B.P. 105, 78153 Le Chesnay, France
\\ \email{Xavier.Leroy@inria.fr}}

\date{Received: 21 July 2009 / Accepted: 22 October 2009}

\maketitle

\begin{abstract} 
This article describes the development and formal verification (proof
of semantic preservation) of a compiler back-end from Cminor (a simple
imperative intermediate language) to PowerPC assembly code, using the
Coq proof assistant both for programming the compiler and for proving
its soundness.  Such a verified compiler is useful in the context of
formal methods applied to the certification of critical software: the
verification of the compiler guarantees that the safety properties
proved on the source code hold for the executable compiled code as
well.
\end{abstract}


\section{Introduction} \label{s:intro}

Can you trust your compiler?  Compilers are generally assumed to be
semantically transparent: the compiled code should behave as
prescribed by the semantics of the source program.  Yet,
compilers---and especially optimizing compilers---are complex programs
that perform complicated symbolic transformations.  Despite intensive
testing, bugs in compilers do occur, causing the compiler to crash at
compile time or---much worse---to silently generate an incorrect
executable for a correct source program
\cite{McKeeman-98,Lindig-05,Regehr-08}.

For low-assurance software, validated only by testing, the impact of
compiler bugs is low: what is tested is the executable code produced
by the compiler; rigorous testing should expose compiler-introduced
errors along with errors already present in the source program.  Note,
however, that compiler-introduced bugs are notoriously difficult to
track down.  Moreover, test plans need to be made more complex if
optimizations are to be tested: for example, loop unrolling
introduces additional limit conditions that are not apparent in the
source loop.

The picture changes dramatically for safety-critical, high-assurance
software.  Here, validation by testing reaches its limits and needs to
be complemented or even replaced by the use of formal methods: model
checking, static analysis, program proof, etc..  Almost universally,
formal methods are applied to the source code of a program.  Bugs in
the compiler that is used to turn this formally verified source code into an
executable can potentially invalidate all the guarantees so painfully
obtained by the use of formal methods.  In a future where formal
methods are routinely applied to source programs, the compiler
could appear as a weak link in the chain that goes from
specifications to executables.

The safety-critical software industry is aware of these issues and
uses a variety of techniques to alleviate them: even more testing (of
the compiler and of the generated executable); turning compiler
optimizations off; and in extreme cases, conducting manual code
reviews of the generated assembly code.  These techniques do not fully
address the issue and are costly in terms of development time and
program performance.

An obviously better approach is to apply formal methods to the compiler
itself in order to gain assurance that it preserves the semantics of
the source programs.  Many different approaches have been proposed and
investigated, including on-paper and on-machine proofs of semantic
preservation, proof-carrying code, credible compilation, translation
validation, and type-preserving compilers.  (These approaches are
compared in section~\ref{s:framework}.)

For the last four years, we have been working on the development of a
{\em realistic}, {\em verified} compiler called Compcert.  By {\em
verified}, we mean a compiler that is accompanied by a machine-checked
proof that the generated code behaves exactly as prescribed by
the semantics of the source program (semantic preservation property).
By {\em realistic}, we mean a compiler that could realistically be
used in the context of production of critical software.  Namely, it
compiles a language commonly used for critical embedded software: not
Java, not ML, not assembly code, but a large subset of the C language.
It produces code for a processor commonly used in embedded systems, as
opposed \eg to a virtual machine: we chose the PowerPC
because it is popular in avionics.  Finally, the compiler must
generate code that is efficient enough and compact enough to fit the
requirements of critical embedded systems.  This implies a multi-pass
compiler that features good register allocation and some basic
optimizations.

This paper reports on the completion of a large part of this program:
the formal verification of a lightly-optimizing compiler back-end that
generates PowerPC assembly code from a simple imperative intermediate
language called Cminor.  This verification is mechanized using the Coq
proof assistant \cite{Coq,Bertot-Casteran-Coqart}.  Another part of
this program---the verification of a compiler front-end translating a
subset of C called Clight down to Cminor---has also been completed
and is described separately \cite{Blazy-Dargaye-Leroy-06,Blazy-Leroy-Clight-09}.

While there exists a considerable body of earlier work on
machine-checked correctness proofs of parts of compilers (see
section~\ref{s:related} for a review), our work is novel in two ways.
First, published work tends to focus on a few parts of a compiler,
such as optimizations and the underlying static analyses
\cite{Lerner-Chambers-05,Cachera-Jensen-05} or translation of a
high-level language to virtual machine code \cite{Klein-Nipkow-jinja}.
In contrast, our work emphasizes end-to-end verification of a complete
compilation chain from a structured imperative language down to
assembly code through 6 intermediate languages.  We found that many of
the non-optimizing translations performed, while often considered
obvious in compiler literature, are surprisingly tricky to 
prove correct formally.

Another novelty of this work is that most of the compiler is written
directly in the Coq specification language, in a purely functional
style.  The executable compiler is obtained by automatic extraction of
Caml code from this specification.  This approach is an attractive
alternative to writing the compiler in a conventional programming
language, then using a program logic to relate it with its
specifications.  This approach has never been applied before to a
program of the size and complexity of an optimizing compiler.

The complete source code of the Coq development, extensively
commented, is available on the Web \cite{Leroy-Compcert-Coq}.  We take
advantage of this availability to omit proofs and a number of
low-level details from this article, referring the interested reader
to the Coq development instead.  The purpose of this article is to
give a high-level presentation of a verified back-end, with just
enough details to enable readers to apply similar techniques in other
contexts.  The general perspective we adopt is to revisit classic
compiler technology from the viewpoint of the semanticist, in
particular by distinguishing clearly between the correctness-relevant
and the performance-relevant aspects of compilation algorithms, which
are inextricably mixed in compiler literature.

\begin{figure}
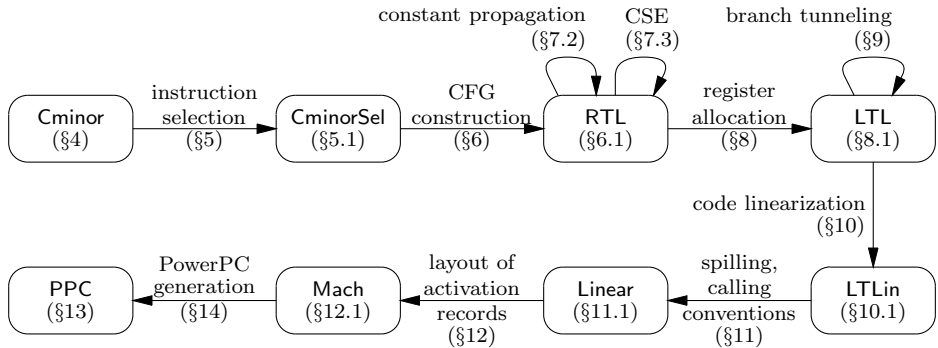


\begin{gpic}

\begin{pgfpicture}
  \pgfsetxvec{\pgfpoint{1.000in}{0in}}
  \pgfsetyvec{\pgfpoint{0in}{1.000in}}
  \begin{pgfscope}
    \pgfpathmoveto{\pgfpointxy{0.000}{0.112}}
    \pgfpatharc{180}{270}{0.100in}
    \pgfusepath{stroke}
  \end{pgfscope}
  \begin{pgfscope}
    \pgfpathmoveto{\pgfpointxy{0.100}{0.362}}
    \pgfpatharc{90}{179}{0.100in}
    \pgfusepath{stroke}
  \end{pgfscope}
  \begin{pgfscope}
    \pgfpathmoveto{\pgfpointxy{0.650}{0.262}}
    \pgfpatharc{0}{90}{0.100in}
    \pgfusepath{stroke}
  \end{pgfscope}
  \begin{pgfscope}
    \pgfpathmoveto{\pgfpointxy{0.550}{0.012}}
    \pgfpatharc{-90}{0}{0.100in}
    \pgfusepath{stroke}
  \end{pgfscope}
  \begin{pgfscope}
    \pgfpathmoveto{\pgfpointxy{0.000}{0.112}}
    \pgfpathlineto{\pgfpointxy{0.000}{0.262}}
    \pgfusepath{stroke}
  \end{pgfscope}
  \begin{pgfscope}
    \pgfpathmoveto{\pgfpointxy{0.100}{0.362}}
    \pgfpathlineto{\pgfpointxy{0.550}{0.362}}
    \pgfusepath{stroke}
  \end{pgfscope}
  \begin{pgfscope}
    \pgfpathmoveto{\pgfpointxy{0.650}{0.262}}
    \pgfpathlineto{\pgfpointxy{0.650}{0.112}}
    \pgfusepath{stroke}
  \end{pgfscope}
  \begin{pgfscope}
    \pgfpathmoveto{\pgfpointxy{0.550}{0.012}}
    \pgfpathlineto{\pgfpointxy{0.100}{0.012}}
    \pgfusepath{stroke}
  \end{pgfscope}
  \pgftext[at=\pgfpointadd{\pgfpointxy{0.325}{0.187}}{\pgfpoint{0pt}{0.5 \baselineskip}}]{\lang{Cminor}}
  \pgftext[at=\pgfpointadd{\pgfpointxy{0.325}{0.187}}{\pgfpoint{0pt}{-0.5 \baselineskip}}]{(\S \ref{s:cminor})}
  \begin{pgfscope}
    \pgfpathmoveto{\pgfpointxy{1.400}{0.112}}
    \pgfpatharc{180}{270}{0.100in}
    \pgfusepath{stroke}
  \end{pgfscope}
  \begin{pgfscope}
    \pgfpathmoveto{\pgfpointxy{1.500}{0.362}}
    \pgfpatharc{90}{179}{0.100in}
    \pgfusepath{stroke}
  \end{pgfscope}
  \begin{pgfscope}
    \pgfpathmoveto{\pgfpointxy{2.050}{0.262}}
    \pgfpatharc{0}{90}{0.100in}
    \pgfusepath{stroke}
  \end{pgfscope}
  \begin{pgfscope}
    \pgfpathmoveto{\pgfpointxy{1.950}{0.012}}
    \pgfpatharc{-90}{0}{0.100in}
    \pgfusepath{stroke}
  \end{pgfscope}
  \begin{pgfscope}
    \pgfpathmoveto{\pgfpointxy{1.400}{0.112}}
    \pgfpathlineto{\pgfpointxy{1.400}{0.262}}
    \pgfusepath{stroke}
  \end{pgfscope}
  \begin{pgfscope}
    \pgfpathmoveto{\pgfpointxy{1.500}{0.362}}
    \pgfpathlineto{\pgfpointxy{1.950}{0.362}}
    \pgfusepath{stroke}
  \end{pgfscope}
  \begin{pgfscope}
    \pgfpathmoveto{\pgfpointxy{2.050}{0.262}}
    \pgfpathlineto{\pgfpointxy{2.050}{0.112}}
    \pgfusepath{stroke}
  \end{pgfscope}
  \begin{pgfscope}
    \pgfpathmoveto{\pgfpointxy{1.950}{0.012}}
    \pgfpathlineto{\pgfpointxy{1.500}{0.012}}
    \pgfusepath{stroke}
  \end{pgfscope}
  \pgftext[at=\pgfpointadd{\pgfpointxy{1.725}{0.187}}{\pgfpoint{0pt}{0.5 \baselineskip}}]{\lang{CminorSel}}
  \pgftext[at=\pgfpointadd{\pgfpointxy{1.725}{0.187}}{\pgfpoint{0pt}{-0.5 \baselineskip}}]{(\S \ref{s:cminorsel})}
  \begin{pgfscope}
    \pgfpathmoveto{\pgfpointxy{2.800}{0.112}}
    \pgfpatharc{180}{270}{0.100in}
    \pgfusepath{stroke}
  \end{pgfscope}
  \begin{pgfscope}
    \pgfpathmoveto{\pgfpointxy{2.900}{0.362}}
    \pgfpatharc{90}{179}{0.100in}
    \pgfusepath{stroke}
  \end{pgfscope}
  \begin{pgfscope}
    \pgfpathmoveto{\pgfpointxy{3.450}{0.262}}
    \pgfpatharc{0}{90}{0.100in}
    \pgfusepath{stroke}
  \end{pgfscope}
  \begin{pgfscope}
    \pgfpathmoveto{\pgfpointxy{3.350}{0.012}}
    \pgfpatharc{-90}{0}{0.100in}
    \pgfusepath{stroke}
  \end{pgfscope}
  \begin{pgfscope}
    \pgfpathmoveto{\pgfpointxy{2.800}{0.112}}
    \pgfpathlineto{\pgfpointxy{2.800}{0.262}}
    \pgfusepath{stroke}
  \end{pgfscope}
  \begin{pgfscope}
    \pgfpathmoveto{\pgfpointxy{2.900}{0.362}}
    \pgfpathlineto{\pgfpointxy{3.350}{0.362}}
    \pgfusepath{stroke}
  \end{pgfscope}
  \begin{pgfscope}
    \pgfpathmoveto{\pgfpointxy{3.450}{0.262}}
    \pgfpathlineto{\pgfpointxy{3.450}{0.112}}
    \pgfusepath{stroke}
  \end{pgfscope}
  \begin{pgfscope}
    \pgfpathmoveto{\pgfpointxy{3.350}{0.012}}
    \pgfpathlineto{\pgfpointxy{2.900}{0.012}}
    \pgfusepath{stroke}
  \end{pgfscope}
  \pgftext[at=\pgfpointadd{\pgfpointxy{3.125}{0.187}}{\pgfpoint{0pt}{0.5 \baselineskip}}]{\lang{RTL}}
  \pgftext[at=\pgfpointadd{\pgfpointxy{3.125}{0.187}}{\pgfpoint{0pt}{-0.5 \baselineskip}}]{(\S \ref{s:rtl})}
  \begin{pgfscope}
    \pgfpathmoveto{\pgfpointxy{4.200}{0.112}}
    \pgfpatharc{180}{270}{0.100in}
    \pgfusepath{stroke}
  \end{pgfscope}
  \begin{pgfscope}
    \pgfpathmoveto{\pgfpointxy{4.300}{0.363}}
    \pgfpatharc{90}{180}{0.100in}
    \pgfusepath{stroke}
  \end{pgfscope}
  \begin{pgfscope}
    \pgfpathmoveto{\pgfpointxy{4.850}{0.262}}
    \pgfpatharc{0}{90}{0.100in}
    \pgfusepath{stroke}
  \end{pgfscope}
  \begin{pgfscope}
    \pgfpathmoveto{\pgfpointxy{4.750}{0.012}}
    \pgfpatharc{-90}{0}{0.100in}
    \pgfusepath{stroke}
  \end{pgfscope}
  \begin{pgfscope}
    \pgfpathmoveto{\pgfpointxy{4.200}{0.112}}
    \pgfpathlineto{\pgfpointxy{4.200}{0.262}}
    \pgfusepath{stroke}
  \end{pgfscope}
  \begin{pgfscope}
    \pgfpathmoveto{\pgfpointxy{4.300}{0.362}}
    \pgfpathlineto{\pgfpointxy{4.750}{0.362}}
    \pgfusepath{stroke}
  \end{pgfscope}
  \begin{pgfscope}
    \pgfpathmoveto{\pgfpointxy{4.850}{0.262}}
    \pgfpathlineto{\pgfpointxy{4.850}{0.112}}
    \pgfusepath{stroke}
  \end{pgfscope}
  \begin{pgfscope}
    \pgfpathmoveto{\pgfpointxy{4.750}{0.012}}
    \pgfpathlineto{\pgfpointxy{4.300}{0.012}}
    \pgfusepath{stroke}
  \end{pgfscope}
  \pgftext[at=\pgfpointadd{\pgfpointxy{4.525}{0.187}}{\pgfpoint{0pt}{0.5 \baselineskip}}]{\lang{LTL}}
  \pgftext[at=\pgfpointadd{\pgfpointxy{4.525}{0.187}}{\pgfpoint{0pt}{-0.5 \baselineskip}}]{(\S \ref{s:ltl})}
  \begin{pgfscope}
    \pgfpathmoveto{\pgfpointxy{4.200}{-0.788}}
    \pgfpatharc{180}{270}{0.100in}
    \pgfusepath{stroke}
  \end{pgfscope}
  \begin{pgfscope}
    \pgfpathmoveto{\pgfpointxy{4.300}{-0.537}}
    \pgfpatharc{90}{180}{0.100in}
    \pgfusepath{stroke}
  \end{pgfscope}
  \begin{pgfscope}
    \pgfpathmoveto{\pgfpointxy{4.850}{-0.638}}
    \pgfpatharc{0}{90}{0.100in}
    \pgfusepath{stroke}
  \end{pgfscope}
  \begin{pgfscope}
    \pgfpathmoveto{\pgfpointxy{4.750}{-0.888}}
    \pgfpatharc{-90}{0}{0.100in}
    \pgfusepath{stroke}
  \end{pgfscope}
  \begin{pgfscope}
    \pgfpathmoveto{\pgfpointxy{4.200}{-0.788}}
    \pgfpathlineto{\pgfpointxy{4.200}{-0.638}}
    \pgfusepath{stroke}
  \end{pgfscope}
  \begin{pgfscope}
    \pgfpathmoveto{\pgfpointxy{4.300}{-0.538}}
    \pgfpathlineto{\pgfpointxy{4.750}{-0.538}}
    \pgfusepath{stroke}
  \end{pgfscope}
  \begin{pgfscope}
    \pgfpathmoveto{\pgfpointxy{4.850}{-0.638}}
    \pgfpathlineto{\pgfpointxy{4.850}{-0.788}}
    \pgfusepath{stroke}
  \end{pgfscope}
  \begin{pgfscope}
    \pgfpathmoveto{\pgfpointxy{4.750}{-0.888}}
    \pgfpathlineto{\pgfpointxy{4.300}{-0.888}}
    \pgfusepath{stroke}
  \end{pgfscope}
  \pgftext[at=\pgfpointadd{\pgfpointxy{4.525}{-0.713}}{\pgfpoint{0pt}{0.5 \baselineskip}}]{\lang{LTLin}}
  \pgftext[at=\pgfpointadd{\pgfpointxy{4.525}{-0.713}}{\pgfpoint{0pt}{-0.5 \baselineskip}}]{(\S \ref{s:ltlin})}
  \begin{pgfscope}
    \pgfpathmoveto{\pgfpointxy{2.800}{-0.788}}
    \pgfpatharc{180}{270}{0.100in}
    \pgfusepath{stroke}
  \end{pgfscope}
  \begin{pgfscope}
    \pgfpathmoveto{\pgfpointxy{2.900}{-0.538}}
    \pgfpatharc{90}{179}{0.100in}
    \pgfusepath{stroke}
  \end{pgfscope}
  \begin{pgfscope}
    \pgfpathmoveto{\pgfpointxy{3.450}{-0.638}}
    \pgfpatharc{0}{90}{0.100in}
    \pgfusepath{stroke}
  \end{pgfscope}
  \begin{pgfscope}
    \pgfpathmoveto{\pgfpointxy{3.350}{-0.888}}
    \pgfpatharc{-90}{0}{0.100in}
    \pgfusepath{stroke}
  \end{pgfscope}
  \begin{pgfscope}
    \pgfpathmoveto{\pgfpointxy{2.800}{-0.788}}
    \pgfpathlineto{\pgfpointxy{2.800}{-0.638}}
    \pgfusepath{stroke}
  \end{pgfscope}
  \begin{pgfscope}
    \pgfpathmoveto{\pgfpointxy{2.900}{-0.538}}
    \pgfpathlineto{\pgfpointxy{3.350}{-0.538}}
    \pgfusepath{stroke}
  \end{pgfscope}
  \begin{pgfscope}
    \pgfpathmoveto{\pgfpointxy{3.450}{-0.638}}
    \pgfpathlineto{\pgfpointxy{3.450}{-0.788}}
    \pgfusepath{stroke}
  \end{pgfscope}
  \begin{pgfscope}
    \pgfpathmoveto{\pgfpointxy{3.350}{-0.888}}
    \pgfpathlineto{\pgfpointxy{2.900}{-0.888}}
    \pgfusepath{stroke}
  \end{pgfscope}
  \pgftext[at=\pgfpointadd{\pgfpointxy{3.125}{-0.713}}{\pgfpoint{0pt}{0.5 \baselineskip}}]{\lang{Linear}}
  \pgftext[at=\pgfpointadd{\pgfpointxy{3.125}{-0.713}}{\pgfpoint{0pt}{-0.5 \baselineskip}}]{(\S \ref{s:linear})}
  \begin{pgfscope}
    \pgfpathmoveto{\pgfpointxy{1.400}{-0.788}}
    \pgfpatharc{180}{270}{0.100in}
    \pgfusepath{stroke}
  \end{pgfscope}
  \begin{pgfscope}
    \pgfpathmoveto{\pgfpointxy{1.500}{-0.538}}
    \pgfpatharc{90}{179}{0.100in}
    \pgfusepath{stroke}
  \end{pgfscope}
  \begin{pgfscope}
    \pgfpathmoveto{\pgfpointxy{2.050}{-0.638}}
    \pgfpatharc{0}{90}{0.100in}
    \pgfusepath{stroke}
  \end{pgfscope}
  \begin{pgfscope}
    \pgfpathmoveto{\pgfpointxy{1.950}{-0.888}}
    \pgfpatharc{-90}{0}{0.100in}
    \pgfusepath{stroke}
  \end{pgfscope}
  \begin{pgfscope}
    \pgfpathmoveto{\pgfpointxy{1.400}{-0.788}}
    \pgfpathlineto{\pgfpointxy{1.400}{-0.638}}
    \pgfusepath{stroke}
  \end{pgfscope}
  \begin{pgfscope}
    \pgfpathmoveto{\pgfpointxy{1.500}{-0.538}}
    \pgfpathlineto{\pgfpointxy{1.950}{-0.538}}
    \pgfusepath{stroke}
  \end{pgfscope}
  \begin{pgfscope}
    \pgfpathmoveto{\pgfpointxy{2.050}{-0.638}}
    \pgfpathlineto{\pgfpointxy{2.050}{-0.788}}
    \pgfusepath{stroke}
  \end{pgfscope}
  \begin{pgfscope}
    \pgfpathmoveto{\pgfpointxy{1.950}{-0.888}}
    \pgfpathlineto{\pgfpointxy{1.500}{-0.888}}
    \pgfusepath{stroke}
  \end{pgfscope}
  \pgftext[at=\pgfpointadd{\pgfpointxy{1.725}{-0.713}}{\pgfpoint{0pt}{0.5 \baselineskip}}]{\lang{Mach}}
  \pgftext[at=\pgfpointadd{\pgfpointxy{1.725}{-0.713}}{\pgfpoint{0pt}{-0.5 \baselineskip}}]{(\S \ref{s:mach})}
  \begin{pgfscope}
    \pgfpathmoveto{\pgfpointxy{0.000}{-0.788}}
    \pgfpatharc{180}{270}{0.100in}
    \pgfusepath{stroke}
  \end{pgfscope}
  \begin{pgfscope}
    \pgfpathmoveto{\pgfpointxy{0.100}{-0.538}}
    \pgfpatharc{90}{180}{0.100in}
    \pgfusepath{stroke}
  \end{pgfscope}
  \begin{pgfscope}
    \pgfpathmoveto{\pgfpointxy{0.650}{-0.638}}
    \pgfpatharc{0}{90}{0.100in}
    \pgfusepath{stroke}
  \end{pgfscope}
  \begin{pgfscope}
    \pgfpathmoveto{\pgfpointxy{0.550}{-0.888}}
    \pgfpatharc{-90}{0}{0.100in}
    \pgfusepath{stroke}
  \end{pgfscope}
  \begin{pgfscope}
    \pgfpathmoveto{\pgfpointxy{0.000}{-0.788}}
    \pgfpathlineto{\pgfpointxy{0.000}{-0.638}}
    \pgfusepath{stroke}
  \end{pgfscope}
  \begin{pgfscope}
    \pgfpathmoveto{\pgfpointxy{0.100}{-0.538}}
    \pgfpathlineto{\pgfpointxy{0.550}{-0.538}}
    \pgfusepath{stroke}
  \end{pgfscope}
  \begin{pgfscope}
    \pgfpathmoveto{\pgfpointxy{0.650}{-0.638}}
    \pgfpathlineto{\pgfpointxy{0.650}{-0.788}}
    \pgfusepath{stroke}
  \end{pgfscope}
  \begin{pgfscope}
    \pgfpathmoveto{\pgfpointxy{0.550}{-0.888}}
    \pgfpathlineto{\pgfpointxy{0.100}{-0.888}}
    \pgfusepath{stroke}
  \end{pgfscope}
  \pgftext[at=\pgfpointadd{\pgfpointxy{0.325}{-0.713}}{\pgfpoint{0pt}{0.5 \baselineskip}}]{\lang{PPC}}
  \pgftext[at=\pgfpointadd{\pgfpointxy{0.325}{-0.713}}{\pgfpoint{0pt}{-0.5 \baselineskip}}]{(\S \ref{s:ppc})}
  \begin{pgfscope}
    \pgfpathmoveto{\pgfpointxy{1.300}{0.212}}
    \pgfpathlineto{\pgfpointxy{1.400}{0.187}}
    \pgfpathlineto{\pgfpointxy{1.300}{0.162}}
    \pgfpathlineto{\pgfpointxy{1.300}{0.212}}
    \pgfsetfillcolor{black}
    \pgfsetlinewidth{0.100pt}
    \pgfusepath{fill,stroke}
  \end{pgfscope}
  \begin{pgfscope}
    \pgfpathmoveto{\pgfpointxy{0.650}{0.187}}
    \pgfpathlineto{\pgfpointxy{1.300}{0.187}}
    \pgfusepath{stroke}
  \end{pgfscope}
  \pgftext[at=\pgfpointadd{\pgfpointxy{1.025}{0.187}}{\pgfpoint{0pt}{1.5 \baselineskip}}]{instruction}
  \pgftext[at=\pgfpointadd{\pgfpointxy{1.025}{0.187}}{\pgfpoint{0pt}{0.5 \baselineskip}}]{selection}
  \pgftext[at=\pgfpointadd{\pgfpointxy{1.025}{0.187}}{\pgfpoint{0pt}{-0.5 \baselineskip}}]{(\S \ref{s:selection})}
  \begin{pgfscope}
    \pgfpathmoveto{\pgfpointxy{2.700}{0.212}}
    \pgfpathlineto{\pgfpointxy{2.800}{0.187}}
    \pgfpathlineto{\pgfpointxy{2.700}{0.162}}
    \pgfpathlineto{\pgfpointxy{2.700}{0.212}}
    \pgfsetfillcolor{black}
    \pgfsetlinewidth{0.100pt}
    \pgfusepath{fill,stroke}
  \end{pgfscope}
  \begin{pgfscope}
    \pgfpathmoveto{\pgfpointxy{2.050}{0.187}}
    \pgfpathlineto{\pgfpointxy{2.700}{0.187}}
    \pgfusepath{stroke}
  \end{pgfscope}
  \pgftext[at=\pgfpointadd{\pgfpointxy{2.425}{0.187}}{\pgfpoint{0pt}{1.5 \baselineskip}}]{CFG}
  \pgftext[at=\pgfpointadd{\pgfpointxy{2.425}{0.187}}{\pgfpoint{0pt}{0.5 \baselineskip}}]{construction}
  \pgftext[at=\pgfpointadd{\pgfpointxy{2.425}{0.187}}{\pgfpoint{0pt}{-0.5 \baselineskip}}]{(\S \ref{s:rtlgen})}
  \begin{pgfscope}
    \pgfpathmoveto{\pgfpointxy{4.100}{0.212}}
    \pgfpathlineto{\pgfpointxy{4.200}{0.187}}
    \pgfpathlineto{\pgfpointxy{4.100}{0.162}}
    \pgfpathlineto{\pgfpointxy{4.100}{0.212}}
    \pgfsetfillcolor{black}
    \pgfsetlinewidth{0.100pt}
    \pgfusepath{fill,stroke}
  \end{pgfscope}
  \begin{pgfscope}
    \pgfpathmoveto{\pgfpointxy{3.450}{0.187}}
    \pgfpathlineto{\pgfpointxy{4.100}{0.187}}
    \pgfusepath{stroke}
  \end{pgfscope}
  \pgftext[at=\pgfpointadd{\pgfpointxy{3.825}{0.187}}{\pgfpoint{0pt}{1.5 \baselineskip}}]{register}
  \pgftext[at=\pgfpointadd{\pgfpointxy{3.825}{0.187}}{\pgfpoint{0pt}{0.5 \baselineskip}}]{allocation}
  \pgftext[at=\pgfpointadd{\pgfpointxy{3.825}{0.187}}{\pgfpoint{0pt}{-0.5 \baselineskip}}]{(\S \ref{s:regalloc})}
  \begin{pgfscope}
    \pgfpathmoveto{\pgfpointxy{4.550}{-0.438}}
    \pgfpathlineto{\pgfpointxy{4.525}{-0.538}}
    \pgfpathlineto{\pgfpointxy{4.500}{-0.438}}
    \pgfpathlineto{\pgfpointxy{4.550}{-0.438}}
    \pgfsetfillcolor{black}
    \pgfsetlinewidth{0.100pt}
    \pgfusepath{fill,stroke}
  \end{pgfscope}
  \begin{pgfscope}
    \pgfpathmoveto{\pgfpointxy{4.525}{0.012}}
    \pgfpathlineto{\pgfpointxy{4.525}{-0.438}}
    \pgfusepath{stroke}
  \end{pgfscope}
  \pgftext[at=\pgfpointadd{\pgfpointxy{4.525}{-0.263}}{\pgfpoint{0pt}{0.5 \baselineskip}},right]{code linearization~}
  \pgftext[at=\pgfpointadd{\pgfpointxy{4.525}{-0.263}}{\pgfpoint{0pt}{-0.5 \baselineskip}},right]{(\S \ref{s:linearize})~}
  \begin{pgfscope}
    \pgfpathmoveto{\pgfpointxy{3.550}{-0.738}}
    \pgfpathlineto{\pgfpointxy{3.450}{-0.713}}
    \pgfpathlineto{\pgfpointxy{3.550}{-0.688}}
    \pgfpathlineto{\pgfpointxy{3.550}{-0.738}}
    \pgfsetfillcolor{black}
    \pgfsetlinewidth{0.100pt}
    \pgfusepath{fill,stroke}
  \end{pgfscope}
  \begin{pgfscope}
    \pgfpathmoveto{\pgfpointxy{4.200}{-0.713}}
    \pgfpathlineto{\pgfpointxy{3.550}{-0.713}}
    \pgfusepath{stroke}
  \end{pgfscope}
  \pgftext[at=\pgfpointadd{\pgfpointxy{3.825}{-0.713}}{\pgfpoint{0pt}{1.5 \baselineskip}}]{spilling,}
  \pgftext[at=\pgfpointadd{\pgfpointxy{3.825}{-0.713}}{\pgfpoint{0pt}{0.5 \baselineskip}}]{calling}
  \pgftext[at=\pgfpointadd{\pgfpointxy{3.825}{-0.713}}{\pgfpoint{0pt}{-0.5 \baselineskip}}]{conventions}
  \pgftext[at=\pgfpointadd{\pgfpointxy{3.825}{-0.713}}{\pgfpoint{0pt}{-1.5 \baselineskip}}]{(\S \ref{s:reload})}
  \begin{pgfscope}
    \pgfpathmoveto{\pgfpointxy{2.150}{-0.738}}
    \pgfpathlineto{\pgfpointxy{2.050}{-0.713}}
    \pgfpathlineto{\pgfpointxy{2.150}{-0.688}}
    \pgfpathlineto{\pgfpointxy{2.150}{-0.738}}
    \pgfsetfillcolor{black}
    \pgfsetlinewidth{0.100pt}
    \pgfusepath{fill,stroke}
  \end{pgfscope}
  \begin{pgfscope}
    \pgfpathmoveto{\pgfpointxy{2.800}{-0.713}}
    \pgfpathlineto{\pgfpointxy{2.150}{-0.713}}
    \pgfusepath{stroke}
  \end{pgfscope}
  \pgftext[at=\pgfpointadd{\pgfpointxy{2.425}{-0.713}}{\pgfpoint{0pt}{1.5 \baselineskip}}]{layout of}
  \pgftext[at=\pgfpointadd{\pgfpointxy{2.425}{-0.713}}{\pgfpoint{0pt}{0.5 \baselineskip}}]{activation}
  \pgftext[at=\pgfpointadd{\pgfpointxy{2.425}{-0.713}}{\pgfpoint{0pt}{-0.5 \baselineskip}}]{records}
  \pgftext[at=\pgfpointadd{\pgfpointxy{2.425}{-0.713}}{\pgfpoint{0pt}{-1.5 \baselineskip}}]{(\S \ref{s:stacking})}
  \begin{pgfscope}
    \pgfpathmoveto{\pgfpointxy{0.750}{-0.738}}
    \pgfpathlineto{\pgfpointxy{0.650}{-0.713}}
    \pgfpathlineto{\pgfpointxy{0.750}{-0.688}}
    \pgfpathlineto{\pgfpointxy{0.750}{-0.738}}
    \pgfsetfillcolor{black}
    \pgfsetlinewidth{0.100pt}
    \pgfusepath{fill,stroke}
  \end{pgfscope}
  \begin{pgfscope}
    \pgfpathmoveto{\pgfpointxy{1.400}{-0.713}}
    \pgfpathlineto{\pgfpointxy{0.750}{-0.713}}
    \pgfusepath{stroke}
  \end{pgfscope}
  \pgftext[at=\pgfpointadd{\pgfpointxy{1.025}{-0.713}}{\pgfpoint{0pt}{1.5 \baselineskip}}]{PowerPC}
  \pgftext[at=\pgfpointadd{\pgfpointxy{1.025}{-0.713}}{\pgfpoint{0pt}{0.5 \baselineskip}}]{generation}
  \pgftext[at=\pgfpointadd{\pgfpointxy{1.025}{-0.713}}{\pgfpoint{0pt}{-0.5 \baselineskip}}]{(\S \ref{s:ppcgen})}
  \begin{pgfscope}
    \pgfpathmoveto{\pgfpointxy{3.100}{0.462}}
    \pgfpathlineto{\pgfpointxy{3.075}{0.362}}
    \pgfpathlineto{\pgfpointxy{3.050}{0.462}}
    \pgfpathlineto{\pgfpointxy{3.100}{0.462}}
    \pgfsetfillcolor{black}
    \pgfsetlinewidth{0.100pt}
    \pgfusepath{fill,stroke}
  \end{pgfscope}
  \begin{pgfscope}
    \pgfpathmoveto{\pgfpointxy{2.875}{0.362}}
    \pgfpathlineto{\pgfpointxy{2.825}{0.462}}
    \pgfpathcurveto{\pgfpointxy{2.792}{0.529}}{\pgfpointxy{2.825}{0.562}}{\pgfpointxy{2.925}{0.562}}
    \pgfpathcurveto{\pgfpointxy{3.025}{0.562}}{\pgfpointxy{3.075}{0.531}}{\pgfpointxy{3.075}{0.467}}
    \pgfpathlineto{\pgfpointxy{3.075}{0.372}}
    \pgfusepath{stroke}
  \end{pgfscope}
  \pgftext[at=\pgfpointadd{\pgfpointxy{3.025}{0.712}}{\pgfpoint{0pt}{0.5 \baselineskip}},right]{constant propagation}
  \pgftext[at=\pgfpointadd{\pgfpointxy{3.025}{0.712}}{\pgfpoint{0pt}{-0.5 \baselineskip}},right]{(\S \ref{s:constprop})}
  \begin{pgfscope}
    \pgfpathmoveto{\pgfpointxy{3.442}{0.441}}
    \pgfpathlineto{\pgfpointxy{3.375}{0.362}}
    \pgfpathlineto{\pgfpointxy{3.397}{0.463}}
    \pgfpathlineto{\pgfpointxy{3.442}{0.441}}
    \pgfsetfillcolor{black}
    \pgfsetlinewidth{0.100pt}
    \pgfusepath{fill,stroke}
  \end{pgfscope}
  \begin{pgfscope}
    \pgfpathmoveto{\pgfpointxy{3.379}{0.371}}
    \pgfpathlineto{\pgfpointxy{3.427}{0.467}}
    \pgfpathcurveto{\pgfpointxy{3.459}{0.531}}{\pgfpointxy{3.425}{0.562}}{\pgfpointxy{3.325}{0.562}}
    \pgfpathcurveto{\pgfpointxy{3.225}{0.562}}{\pgfpointxy{3.175}{0.529}}{\pgfpointxy{3.175}{0.462}}
    \pgfpathlineto{\pgfpointxy{3.175}{0.362}}
    \pgfusepath{stroke}
  \end{pgfscope}
  \pgftext[at=\pgfpointadd{\pgfpointxy{3.225}{0.712}}{\pgfpoint{0pt}{0.5 \baselineskip}},left]{CSE}
  \pgftext[at=\pgfpointadd{\pgfpointxy{3.225}{0.712}}{\pgfpoint{0pt}{-0.5 \baselineskip}},left]{(\S \ref{s:CSE})}
  \begin{pgfscope}
    \pgfpathmoveto{\pgfpointxy{4.692}{0.441}}
    \pgfpathlineto{\pgfpointxy{4.625}{0.362}}
    \pgfpathlineto{\pgfpointxy{4.647}{0.463}}
    \pgfpathlineto{\pgfpointxy{4.692}{0.441}}
    \pgfsetfillcolor{black}
    \pgfsetlinewidth{0.100pt}
    \pgfusepath{fill,stroke}
  \end{pgfscope}
  \begin{pgfscope}
    \pgfpathmoveto{\pgfpointxy{4.425}{0.362}}
    \pgfpathlineto{\pgfpointxy{4.375}{0.462}}
    \pgfpathcurveto{\pgfpointxy{4.342}{0.529}}{\pgfpointxy{4.392}{0.562}}{\pgfpointxy{4.525}{0.562}}
    \pgfpathcurveto{\pgfpointxy{4.658}{0.562}}{\pgfpointxy{4.709}{0.531}}{\pgfpointxy{4.677}{0.467}}
    \pgfpathlineto{\pgfpointxy{4.629}{0.371}}
    \pgfusepath{stroke}
  \end{pgfscope}
  \pgftext[at=\pgfpointadd{\pgfpointxy{4.625}{0.712}}{\pgfpoint{0pt}{0.5 \baselineskip}},right]{branch tunneling}
  \pgftext[at=\pgfpointadd{\pgfpointxy{4.625}{0.712}}{\pgfpoint{0pt}{-0.5 \baselineskip}},right]{(\S \ref{s:tunneling})}
\end{pgfpicture}
\ifx\Setlineno\undefined\else\Setlineno=59\fi
\end{gpic}

\caption{The passes and intermediate languages of Compcert.}
\label{f:compcert-passes}

\end{figure}

The remainder of this article is organized as follows.
Section~\ref{s:framework} formalizes various approaches to
establishing trust in the results of compilation.
Section~\ref{s:infrastructure} presents the main aspects of the
development that are shared between all passes of the compiler:
the value and memory models, labeled transition semantics, proofs by
simulation diagrams.  Sections~\ref{s:cminor} and~\ref{s:ppc} define
the semantics of our source language \lang{Cminor} and our target
language \lang{PPC}, respectively.  The bulk of this article 
(sections~\ref{s:selection} to~\ref{s:ppcgen}) is devoted to the 
description of the successive passes of the compiler, the intermediate
languages they operate on, and their soundness proofs.
(Figure~\ref{f:compcert-passes} summarizes the passes and the
intermediate languages.)
Experimental data on the Coq development and on the executable
compiler extracted from it are presented in sections~\ref{s:coq-devel}
and~\ref{s:experimental}.  Section~\ref{s:perspectives} discusses some
of the design choices and possible extensions.  Related work is
discussed in section~\ref{s:related}, followed by concluding remarks
in section~\ref{s:conclusions}.

\section{General framework} \label{s:framework}

\subsection{Notions of semantic preservation} \label{s:semantic-preservation}

Consider a source program $S$ and a compiled program $C$ produced by a
compiler.  Our aim is to prove that the semantics of $S$ was preserved
during compilation.  To make this notion of semantic preservation
precise, we assume given semantics for the source language $L_s$ and
the target language $L_t$.  These semantics associate one or several
observable behaviors $B$ to $S$ and $C$.  Typically, observable
behaviors include termination, divergence, and ``going wrong'' on
executing an undefined computation.  (In the remainder of this work,
behaviors also contain traces of input-output operations performed
during program execution.)  We write $S \Downarrow B$ to mean that
program $S$ executes with observable behavior $B$, and likewise for $C$.

The strongest notion of semantic preservation during compilation is
that the source program $S$ and the compiled code $C$ have exactly the
same sets of observable behaviors---a standard bisimulation property:

\begin{definition}[Bisimulation] \label{d:bisimulation}
$ \forall B, ~ S \Downarrow B \Longleftrightarrow C \Downarrow B$.
\end{definition}

Definition~\ref{d:bisimulation} is too strong to be usable as our
notion of semantic preservation.  If the source language is not
deterministic, compilers are allowed to select one of the possible
behaviors of the source program.  (For instance, C compilers choose
one particular evaluation order for expressions among the several
orders allowed by the C specifications.)  In this case, $C$ will have
fewer behaviors than $S$.  To account for this degree of freedom, we
can consider a backward simulation, or refinement, property:

\begin{definition}[Backward simulation] \label{d:backward-sim}
$ \forall B, ~ C \Downarrow B \Longrightarrow S \Downarrow B$.
\end{definition}

Definitions~\ref{d:bisimulation} and~\ref{d:backward-sim} imply that
if $S$ always goes wrong, so does $C$.
Several desirable optimizations violate this
requirement.  For instance, if $S$ contains an integer division whose
result is unused, and this division can cause $S$ to go wrong because
its second argument is zero, dead code elimination will result in a
compiled program $C$ that does not go wrong on this division.  To
leave more flexibility to the compiler, we can therefore restrict the
backward simulation property to {\em safe} source programs.  A program
$S$ is safe, written ${\it Safe}(S)$, if none of its possible behaviors is
in the set {\tt Wrong} of ``going wrong'' behaviors ($S \Downarrow B
\Longrightarrow B \notin {\tt Wrong}$).

\begin{definition}[Backward simulation for safe programs]
\label{d:safe-backward-sim}
If ${\it Safe}(S)$, then 
$ \forall B, ~ C \Downarrow B \Longrightarrow S \Downarrow B$.
\end{definition}

In other words, if $S$ cannot go wrong (a fact that can be established
by formal verification or static analysis of $S$), then neither does
$C$; moreover, all observable behaviors of $C$ are acceptable
behaviors of $S$. 

An alternative to backward simulation
(definitions~\ref{d:backward-sim} and~\ref{d:safe-backward-sim})
is forward simulation properties, showing that all possible behaviors
of the source program are also possible behaviors of the compiled
program:

\begin{definition}[Forward simulation] \label{d:forward-sim}
$ \forall B, ~ S \Downarrow B \Longrightarrow C \Downarrow B$.
\end{definition}

\begin{definition}[Forward simulation for safe programs]
\label{d:safe-forward-sim}
$ \forall B \notin {\tt Wrong}, ~ S \Downarrow B \Longrightarrow C \Downarrow B$.
\end{definition}

\begin{figure}
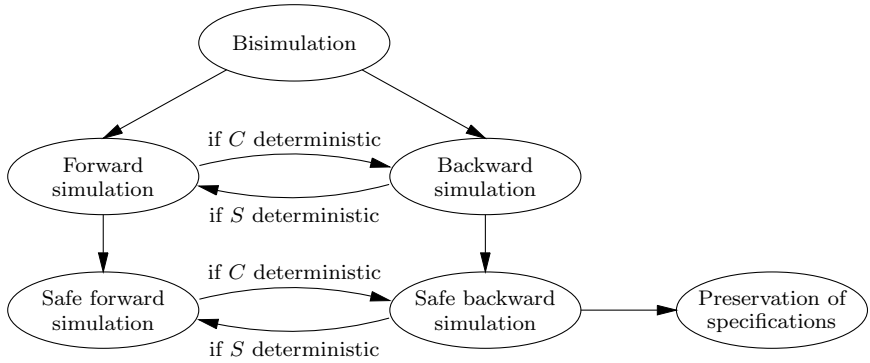


\begin{gpic}

\begin{pgfpicture}
  \pgfsetxvec{\pgfpoint{1.000in}{0in}}
  \pgfsetyvec{\pgfpoint{0in}{1.000in}}
  \begin{pgfscope}
    \pgfpathellipse{\pgfpointxy{1.500}{0.700}}{\pgfpointxy{0.500}{0}}{\pgfpointxy{0}{0.200}}
    \pgfusepath{stroke}
  \end{pgfscope}
  \pgftext[at=\pgfpointadd{\pgfpointxy{1.500}{0.700}}{\pgfpoint{0pt}{-0.0 \baselineskip}}]{Bisimulation}
  \begin{pgfscope}
    \pgfpathellipse{\pgfpointxy{2.500}{0.000}}{\pgfpointxy{0.500}{0}}{\pgfpointxy{0}{0.200}}
    \pgfusepath{stroke}
  \end{pgfscope}
  \pgftext[at=\pgfpointadd{\pgfpointxy{2.500}{0.000}}{\pgfpoint{0pt}{0.5 \baselineskip}}]{Backward}
  \pgftext[at=\pgfpointadd{\pgfpointxy{2.500}{0.000}}{\pgfpoint{0pt}{-0.5 \baselineskip}}]{simulation}
  \begin{pgfscope}
    \pgfpathellipse{\pgfpointxy{2.500}{-0.700}}{\pgfpointxy{0.500}{0}}{\pgfpointxy{0}{0.200}}
    \pgfusepath{stroke}
  \end{pgfscope}
  \pgftext[at=\pgfpointadd{\pgfpointxy{2.500}{-0.700}}{\pgfpoint{0pt}{0.5 \baselineskip}}]{Safe backward}
  \pgftext[at=\pgfpointadd{\pgfpointxy{2.500}{-0.700}}{\pgfpoint{0pt}{-0.5 \baselineskip}}]{simulation}
  \begin{pgfscope}
    \pgfpathellipse{\pgfpointxy{4.000}{-0.700}}{\pgfpointxy{0.500}{0}}{\pgfpointxy{0}{0.200}}
    \pgfusepath{stroke}
  \end{pgfscope}
  \pgftext[at=\pgfpointadd{\pgfpointxy{4.000}{-0.700}}{\pgfpoint{0pt}{0.5 \baselineskip}}]{Preservation of}
  \pgftext[at=\pgfpointadd{\pgfpointxy{4.000}{-0.700}}{\pgfpoint{0pt}{-0.5 \baselineskip}}]{specifications}
  \begin{pgfscope}
    \pgfpathellipse{\pgfpointxy{0.500}{0.000}}{\pgfpointxy{0.500}{0}}{\pgfpointxy{0}{0.200}}
    \pgfusepath{stroke}
  \end{pgfscope}
  \pgftext[at=\pgfpointadd{\pgfpointxy{0.500}{0.000}}{\pgfpoint{0pt}{0.5 \baselineskip}}]{Forward}
  \pgftext[at=\pgfpointadd{\pgfpointxy{0.500}{0.000}}{\pgfpoint{0pt}{-0.5 \baselineskip}}]{simulation}
  \begin{pgfscope}
    \pgfpathellipse{\pgfpointxy{0.500}{-0.700}}{\pgfpointxy{0.500}{0}}{\pgfpointxy{0}{0.200}}
    \pgfusepath{stroke}
  \end{pgfscope}
  \pgftext[at=\pgfpointadd{\pgfpointxy{0.500}{-0.700}}{\pgfpoint{0pt}{0.5 \baselineskip}}]{Safe forward}
  \pgftext[at=\pgfpointadd{\pgfpointxy{0.500}{-0.700}}{\pgfpoint{0pt}{-0.5 \baselineskip}}]{simulation}
  \begin{pgfscope}
    \pgfpathmoveto{\pgfpointxy{2.425}{0.270}}
    \pgfpathlineto{\pgfpointxy{2.500}{0.200}}
    \pgfpathlineto{\pgfpointxy{2.400}{0.227}}
    \pgfpathlineto{\pgfpointxy{2.425}{0.270}}
    \pgfsetfillcolor{black}
    \pgfsetlinewidth{0.100pt}
    \pgfusepath{fill,stroke}
  \end{pgfscope}
  \begin{pgfscope}
    \pgfpathmoveto{\pgfpointxy{1.854}{0.559}}
    \pgfpathlineto{\pgfpointxy{2.413}{0.249}}
    \pgfusepath{stroke}
  \end{pgfscope}
  \begin{pgfscope}
    \pgfpathmoveto{\pgfpointxy{0.600}{0.227}}
    \pgfpathlineto{\pgfpointxy{0.500}{0.200}}
    \pgfpathlineto{\pgfpointxy{0.575}{0.270}}
    \pgfpathlineto{\pgfpointxy{0.600}{0.227}}
    \pgfsetfillcolor{black}
    \pgfsetlinewidth{0.100pt}
    \pgfusepath{fill,stroke}
  \end{pgfscope}
  \begin{pgfscope}
    \pgfpathmoveto{\pgfpointxy{1.146}{0.559}}
    \pgfpathlineto{\pgfpointxy{0.587}{0.249}}
    \pgfusepath{stroke}
  \end{pgfscope}
  \begin{pgfscope}
    \pgfpathmoveto{\pgfpointxy{2.525}{-0.400}}
    \pgfpathlineto{\pgfpointxy{2.500}{-0.500}}
    \pgfpathlineto{\pgfpointxy{2.475}{-0.400}}
    \pgfpathlineto{\pgfpointxy{2.525}{-0.400}}
    \pgfsetfillcolor{black}
    \pgfsetlinewidth{0.100pt}
    \pgfusepath{fill,stroke}
  \end{pgfscope}
  \begin{pgfscope}
    \pgfpathmoveto{\pgfpointxy{2.500}{-0.200}}
    \pgfpathlineto{\pgfpointxy{2.500}{-0.400}}
    \pgfusepath{stroke}
  \end{pgfscope}
  \begin{pgfscope}
    \pgfpathmoveto{\pgfpointxy{3.400}{-0.675}}
    \pgfpathlineto{\pgfpointxy{3.500}{-0.700}}
    \pgfpathlineto{\pgfpointxy{3.400}{-0.725}}
    \pgfpathlineto{\pgfpointxy{3.400}{-0.675}}
    \pgfsetfillcolor{black}
    \pgfsetlinewidth{0.100pt}
    \pgfusepath{fill,stroke}
  \end{pgfscope}
  \begin{pgfscope}
    \pgfpathmoveto{\pgfpointxy{3.000}{-0.700}}
    \pgfpathlineto{\pgfpointxy{3.400}{-0.700}}
    \pgfusepath{stroke}
  \end{pgfscope}
  \begin{pgfscope}
    \pgfpathmoveto{\pgfpointxy{0.525}{-0.400}}
    \pgfpathlineto{\pgfpointxy{0.500}{-0.500}}
    \pgfpathlineto{\pgfpointxy{0.475}{-0.400}}
    \pgfpathlineto{\pgfpointxy{0.525}{-0.400}}
    \pgfsetfillcolor{black}
    \pgfsetlinewidth{0.100pt}
    \pgfusepath{fill,stroke}
  \end{pgfscope}
  \begin{pgfscope}
    \pgfpathmoveto{\pgfpointxy{0.500}{-0.200}}
    \pgfpathlineto{\pgfpointxy{0.500}{-0.400}}
    \pgfusepath{stroke}
  \end{pgfscope}
  \begin{pgfscope}
    \pgfpathmoveto{\pgfpointxy{1.908}{0.097}}
    \pgfpathlineto{\pgfpointxy{2.000}{0.050}}
    \pgfpathlineto{\pgfpointxy{1.897}{0.048}}
    \pgfpathlineto{\pgfpointxy{1.908}{0.097}}
    \pgfsetfillcolor{black}
    \pgfsetlinewidth{0.100pt}
    \pgfusepath{fill,stroke}
  \end{pgfscope}
  \begin{pgfscope}
    \pgfpathmoveto{\pgfpointxy{1.903}{0.073}}
    \pgfpatharc{78}{104}{2.000in}
    \pgfusepath{stroke}
  \end{pgfscope}
  \begin{pgfscope}
    \pgfpathmoveto{\pgfpointxy{1.092}{-0.097}}
    \pgfpathlineto{\pgfpointxy{1.000}{-0.050}}
    \pgfpathlineto{\pgfpointxy{1.103}{-0.048}}
    \pgfpathlineto{\pgfpointxy{1.092}{-0.097}}
    \pgfsetfillcolor{black}
    \pgfsetlinewidth{0.100pt}
    \pgfusepath{fill,stroke}
  \end{pgfscope}
  \begin{pgfscope}
    \pgfpathmoveto{\pgfpointxy{1.097}{-0.073}}
    \pgfpatharc{-101}{-75}{2.000in}
    \pgfusepath{stroke}
  \end{pgfscope}
  \pgftext[at=\pgfpointadd{\pgfpointxy{1.500}{0.200}}{\pgfpoint{0pt}{-0.0 \baselineskip}}]{if $C$ deterministic}
  \pgftext[at=\pgfpointadd{\pgfpointxy{1.500}{-0.200}}{\pgfpoint{0pt}{-0.0 \baselineskip}}]{if $S$ deterministic}
  \begin{pgfscope}
    \pgfpathmoveto{\pgfpointxy{1.908}{-0.603}}
    \pgfpathlineto{\pgfpointxy{2.000}{-0.650}}
    \pgfpathlineto{\pgfpointxy{1.897}{-0.652}}
    \pgfpathlineto{\pgfpointxy{1.908}{-0.603}}
    \pgfsetfillcolor{black}
    \pgfsetlinewidth{0.100pt}
    \pgfusepath{fill,stroke}
  \end{pgfscope}
  \begin{pgfscope}
    \pgfpathmoveto{\pgfpointxy{1.903}{-0.627}}
    \pgfpatharc{78}{104}{2.000in}
    \pgfusepath{stroke}
  \end{pgfscope}
  \begin{pgfscope}
    \pgfpathmoveto{\pgfpointxy{1.092}{-0.797}}
    \pgfpathlineto{\pgfpointxy{1.000}{-0.750}}
    \pgfpathlineto{\pgfpointxy{1.103}{-0.748}}
    \pgfpathlineto{\pgfpointxy{1.092}{-0.797}}
    \pgfsetfillcolor{black}
    \pgfsetlinewidth{0.100pt}
    \pgfusepath{fill,stroke}
  \end{pgfscope}
  \begin{pgfscope}
    \pgfpathmoveto{\pgfpointxy{1.097}{-0.773}}
    \pgfpatharc{-101}{-75}{2.000in}
    \pgfusepath{stroke}
  \end{pgfscope}
  \pgftext[at=\pgfpointadd{\pgfpointxy{1.500}{-0.500}}{\pgfpoint{0pt}{-0.0 \baselineskip}}]{if $C$ deterministic}
  \pgftext[at=\pgfpointadd{\pgfpointxy{1.500}{-0.900}}{\pgfpoint{0pt}{-0.0 \baselineskip}}]{if $S$ deterministic}
\end{pgfpicture}
\ifx\Setlineno\undefined\else\Setlineno=37\fi
\end{gpic}

\caption{Various semantic preservation properties and their
  relationships.  An arrow from $A$ to $B$ means that $A$ logically
  implies $B$.}
\label{f:sem-pres}

\end{figure}

In general, forward simulations are easier to prove than backward
simulations (by structural induction on an execution of $S$), but less
informative: even if forward simulation holds, the compiled code $C$
could have additional, undesirable behaviors beyond those of $S$.
However, this cannot happen if $C$ is deterministic, that is, if it
admits only one observable behavior ($C \Downarrow B_1 \wedge C
\Downarrow B_2 \Longrightarrow B_1 = B_2$).  This is the case if the
target language $L_t$ has no internal non-determinism (programs
change their behaviors only in response to different inputs but not
because of internal choices) and the execution environment is
deterministic (inputs given to programs are uniquely determined by
their previous outputs).\footnote{Section~\ref{s:ppc-determinism}
formalizes this notion of deterministic execution environment by, in
effect, restricting the set of behaviors $B$ to those generated by a
transition function that responds to the outputs of the program.}
In this case, it is easy to show that 
``forward simulation'' implies ``backward simulation'', and
``forward simulation for safe programs'' implies ``backward
simulation for safe programs''.  The reverse implications hold if the source
program is deterministic.  Figure~\ref{f:sem-pres} summarizes the logical
implications between the various notions of semantic preservation.

From a formal methods perspective, what we are really interested in is
whether the compiled code satisfies the functional specifications of
the application.  Assume that such a specification is given as a
predicate ${\it Spec}(B)$ of the observable behavior.  Further assume
that the specification rules out ``going wrong'' behaviors:
${\it Spec}(B) \Longrightarrow B \notin {\tt Wrong}$.  
We say that $C$
satisfies the specification, and write $C \models {\it Spec}$,
if all behaviors of $C$ satisfy $\it Spec$
($\forall B,~C \Downarrow B \Longrightarrow {\it Spec}(B)$).  The expected
soundness property of the compiler is that it preserves the fact
that the source code $S$ satisfies the specification, a fact that has
been established separately by formal verification of $S$.

\begin{definition}[Preservation of a specification] \label{d:pres-spec}
$ S \models {\it Spec} ~\Longrightarrow~ C \models {\it Spec}$.
\end{definition}

It is easy to show that ``backward simulation for safe programs''
implies ``preservation of a specification'' for all specifications {\it
Spec}.  In general, the latter property is weaker than the
former property.  For instance, if the specification of the
application is ``print a prime number'', and $S$ prints 7, and $C$
prints 11, the specification is preserved but backward simulation does
not hold.  Therefore, definition~\ref{d:pres-spec} leaves more liberty
for compiler optimizations that do not preserve semantics in general,
but are correct for specific programs.  However, it has the marked
disadvantage of depending on the specifications of the application, so
that changes in the latter can require the proof of preservation to be
redone.

A special case of preservation of a specification, of considerable
historical importance, is the preservation of type and memory safety,
which we can summarize as ``if $S$ does not go wrong, neither does
$C$'':

\begin{definition}[Preservation of safety] \label{d:pres-safety}
$ {\it Safe}(S) ~\Longrightarrow~ {\it Safe}(C)$.
\end{definition}

Combined with a separate check that $S$ is well-typed in a sound type
system, this property implies that $C$ executes without memory
violations.  Type-preserving compilation
\cite{Morrisett-TAL-99,Morrisett-STAL-02,Chen-PLDI-08}
obtains this guarantee by different means: under the assumption that
$S$ is well typed, $C$ is proved to be well-typed in a sound type
system, ensuring that it cannot go wrong.  Having proved
a semantic preservation property such as
definitions~\ref{d:safe-backward-sim} or~\ref{d:pres-spec}
provides the same guarantee without having to
equip the target and intermediate languages with sound type systems
and to prove type preservation for the compiler.

In summary, the approach we follow in this work is to prove a ``forward
simulation for safe programs'' property (sections~\ref{s:selection}
to~\ref{s:ppcgen}), and combine it with a separate proof of
determinism for the target language (section~\ref{s:ppc-determinism}),
the latter proof being particularly easy since the target is a
single-threaded assembly language.  Combining these two proofs, we obtain
that all specifications are preserved, in the sense of
definition~\ref{d:pres-spec}, which is the result that matters for
users of the compiler who practice formal verification at the source level.

\subsection{Verified compilers, validated compilers, and certifying compilers}

We now discuss several approaches to establishing that a compiler
preserves semantics of the compiled programs, in the sense of
section~\ref{s:semantic-preservation}.  In the following,
we write $S \approx C$, where $S$ is a source program and $C$ is
compiled code, to denote one of the semantic preservation
properties~\ref{d:bisimulation} to~\ref{d:pres-safety}
of section~\ref{s:semantic-preservation}.

\subsubsection{Verified compilers}

We model the compiler as a total function {\it Comp} from source
programs to either compiled code (written ${\it Comp}(S) = {\tt OK}(C)$)
or a compile-time error (written ${\it Comp}(S) = {\tt Error}$).
Compile-time errors correspond to cases where the compiler is unable
to produce code, for instance if the source program is incorrect (syntax
error, type error, etc.), but also if it exceeds the capacities of the
compiler (see section~\ref{s:stacking} for an example).

\begin{definition}[Verified compiler] \label{def:verified}
A compiler {\it Comp} is said to be verified if it is accompanied
with a formal proof of the following property:
$$
  \forall S,C,~~{\it Comp}(S) = {\tt OK}(C) \Longrightarrow S \approx C
  \eqno(i)
$$
\end{definition}

In other words, a verified compiler either reports an error or
produces code that satisfies the desired semantic preservation
property.  Notice that a compiler that always fails (${\it Comp}(S) =
{\tt Error}$ for all $S$) is indeed verified, although useless.  Whether
the compiler succeeds to compile the source programs of interest is
not a soundness issue, but a quality of implementation issue, which
is addressed by non-formal methods such as testing.  The important
feature, from a formal methods standpoint, is that the compiler never
silently produces incorrect code.

Verifying a compiler in the sense of definition~\ref{def:verified}
amounts to applying program proof technology to the compiler sources,
using one of the properties defined in section~\ref{s:framework}
as the high-level specification of the compiler.

\subsubsection{Translation validation with verified validators}
\label{s:translation-validation}

In the translation validation approach
\cite{Pnueli-SS-98,Necula-00}
the compiler does not need to be verified.  Instead, the compiler is
complemented by a {\em validator}: a boolean-valued function 
${\it Validate}(S, C)$ that verifies the property $S \approx C$
{\em a posteriori}.  If ${\it Comp}(S) = {\tt OK}(C)$ and ${\it Validate}(S, C) =
{\tt true}$, the compiled code $C$ is deemed trustworthy.  Validation can
be performed in several ways, ranging from symbolic interpretation and
static analysis of $S$ and $C$
\cite{Necula-00,Rival-04,Huang-Childers-Soffa-06,Tristan-Leroy-scheduling,Tristan-Leroy-LCM}
to the generation of verification conditions followed by model
checking or automatic theorem proving
\cite{Pnueli-SS-98,Zuck-VOC-02,Zuck-TVOC-05}.  The property $S \approx
C$ being undecidable in general, validators must err on the side
of caution and should reply {\tt false} if they cannot establish
$S \approx C$.\footnote{This conservatism doesn't necessarily
render validators incomplete: a validator can be complete with respect
to a particular code transformation or family of transformations.}

Translation validation generates additional confidence in the
correctness of the compiled code but by itself does not provide
formal guarantees as strong as those provided by a verified
compiler: the validator could itself be unsound.

\begin{definition}[Verified validator] \label{def:validator}
A validator {\it Validate} is said to be verified if it is accompanied
with a formal proof of the following property:
$$
   \forall S,C,~~{\it Validate}(S, C) = {\tt true} \Longrightarrow S \approx C
   \eqno(ii)
$$
\end{definition}

The combination of a verified validator $\it Validate$ with an
unverified compiler $\it Comp$ does provide formal guarantees as
strong as those provided by a verified compiler.  Such a combination
calls the validator after each run of the compiler, reporting a
compile-time error if validation fails:
$$\begin{array}{l}
{\it Comp}'(S) = \\
~~~~ {\tt match\ }{\it Comp}(S) {\tt \ with} \\
~~~~ {\tt {\char124}\ Error} \rightarrow {\tt Error} \\
~~~~ {\tt {\char124}\ OK}(C) \rightarrow {\tt if\ }{\it Validate}(S, C) {\tt \ then\ OK}(C) {\tt \ else\ Error}
\end{array}$$
If the source and target languages are identical, as is often the case
for optimization passes, we also have the option to return the source
code unchanged if validation fails, in effect turning off a
potentially incorrect optimization:
$$\begin{array}{l}
{\it Comp}''(S) = \\
~~~~ {\tt match\ }{\it Comp}(S) {\tt \ with} \\
~~~~ {\tt {\char124}\ Error} \rightarrow {\tt OK}(S) \\
~~~~ {\tt {\char124}\ OK}(C) \rightarrow {\tt if\ }{\it Validate}(S, C) {\tt \ then\ OK}(C) {\tt \ else\ OK}(S)
\end{array}$$

\begin{theorem}
If $\it Validate$ is a verified validator in the sense of
definition~\ref{def:validator}, 
$\it Comp'$ and $\it Comp''$ are verified compilers in the sense of
definition~\ref{def:verified}.
\end{theorem}

Verification of a translation validator is therefore an attractive
alternative to the verification of a compiler, provided the validator
is smaller and simpler than the compiler.

In the presentation above, the validator receives unadorned source and
compiled codes as arguments.  In practice, the validator can also take
advantage of additional information generated by the compiler and
transmitted to the validator as part of~$C$ or separately.  For
instance, the validator of \cite{Rival-04} exploits debugging
information to suggest a correspondence between program points and
between variables of $S$ and $C$.  Credible compilation \cite{Rinard-99}
carries this approach to the extreme: the compiler is supposed to
annotate $C$ with a full proof of $S \approx C$, so that
translation validation reduces to proof checking.

\subsubsection{Proof-carrying code and certifying compilers}

The proof-carrying code (PCC) approach
\cite{Necula-97,Appel-FPCC,Shao-FPCC-07} does not
attempt to establish semantic preservation between a source program
and some compiled code.  Instead, PCC focuses on the generation of
independently-checkable evidence that the compiled code $C$ satisfies
a behavioral specification ${\it Spec}$ such as type and memory
safety.  PCC makes use of a {\em certifying compiler}, which is a
function {\em CComp} that either fails or
returns both a compiled code $C$ and a proof $\pi$ of the property
$C \models {\it Spec}$.
The proof $\pi$, also called a \emph{certificate}, can be checked
independently by the code user; there is no need to trust the code
producer, nor to formally verify the compiler itself.

In a naive view of PCC, the certificate $\pi$ generated by the
compiler is a full proof term and the client-side verifier is a
general-purpose proof checker.  In practice, it is sufficient to
generate enough hints so that such a full proof can be reconstructed
cheaply on the client side by a specialized checker \cite{Necula-01}.
If the property of interest is type safety, PCC can reduce to
type-checking of compiled code, as in Java bytecode verification
\cite{Staerk-01} or typed assembly language \cite{Morrisett-TAL-99}:
the certificate $\pi$ reduces to type annotations, and the client-side
verifier is a type checker.

In the original PCC design, the certifying compiler is specialized for
a fixed property of programs (\eg type and memory safety), and this
property is simple enough to be established by the compiler itself.
For richer properties, it becomes necessary to provide the certifying
compiler with a certificate that the source program $S$ satisfies the
property.  It is also possible to make the compiler generic with
respect to a family of program properties.  This extension of PCC is
called proof-preserving compilation in \cite{Shao-TOPLAS-05} and
certificate translation in \cite{Barthe-Gregoire-06,Barthe-Kunz-08}.

In all cases, it suffices to formally verify the client-side checker
to obtain guarantees as strong as those obtained from compiler
verification.
Symmetrically, a certifying compiler can
be constructed (at least theoretically) from a
verified compiler.  Assume that $\it Comp$ is a verified compiler,
using definition~\ref{d:pres-spec} as our notion of semantic preservation,
and further assume that the verification was conducted with
a proof assistant that produces proof terms, such as
Coq.  Let~$\Pi$ be a proof term for the semantic preservation theorem
of $\it Comp$, namely
$$ \Pi :  \forall S,C,~~{\it Comp}(S) = {\tt OK}(C) \Longrightarrow
                        S \models {\it Spec} \Longrightarrow
                        C \models {\it Spec} $$
Via the Curry-Howard isomorphism, $\Pi$ is a function that takes $S$, $C$,
a proof of ${\it Comp}(S) = {\tt OK}(C)$ and a proof of 
$S \models {\it Spec}$, and returns a proof of $C \models {\it Spec}$.
A certifying compiler of the proof-preserving kind can then be defined
as follows:
$$\begin{array}{l}
{\it CComp}(S: {\it Source}, \pi_s: S \models {\it Spec}) = \\
~~~~ {\tt match\ }{\it Comp}(S) {\tt \ with} \\
~~~~ {\tt {\char124}\ Error} \rightarrow {\tt Error} \\
~~~~ {\tt {\char124}\ OK}(C) \rightarrow {\tt OK}(C, \Pi~S~C~\pi_{eq}~\pi_s)
\end{array}$$
(Here, $\pi_{eq}$ is a proof term for the proposition
${\it Comp}(S) = {\tt OK}(C)$, which trivially holds in the context of
the {\tt match} above.  Actually building this proof term in Coq requires
additional baggage in the definition above that we omitted for
simplicity.)  The accompanying client-side checker is the Coq proof
checker.  While the certificate produced by {\it CComp} is huge (it
contains a proof of soundness for the compilation of all source
programs, not just for $S$), it could perhaps be specialized for $S$
and $C$ using partial evaluation techniques.

\subsection{Composition of compilation passes}

Compilers are naturally decomposed into several passes that
communicate through intermediate languages.  It is fortunate that
verified compilers can also be decomposed in this manner.

Let ${\it Comp}_1$ and ${\it Comp}_2$ be compilers from
languages $L_1$ to $L_2$ and $L_2$ to $L_3$, respectively.
Assume that the semantic preservation property $\approx$ is
transitive.  (This is true for all properties considered
in section~\ref{s:semantic-preservation}.)  Consider the monadic
composition of  ${\it Comp}_1$ and ${\it Comp}_2$:
$$\begin{array}{l}
{\it Comp}(S) = \\
~~~~ {\tt match\ }{\it Comp}_1(S) {\tt \ with} \\
~~~~ {\tt {\char124}\ Error} \rightarrow {\tt Error} \\
~~~~ {\tt {\char124}\ OK}(I) \rightarrow {\it Comp}_2(I)
\end{array}$$

\begin{theorem}
If the compilers  ${\it Comp}_1$ and ${\it Comp}_2$ are verified,
so is their monadic composition $\it Comp$.
\end{theorem}

\subsection{Summary}

The conclusions of this discussion are simple and define the
methodology we have followed to verify the Compcert compiler back-end.

\begin{enumerate}
\item Provided the target language of the compiler has deterministic
semantics, an appropriate specification for the soundness proof of
the compiler is the combination of
definitions~\ref{d:safe-forward-sim}
(forward simulation for safe source programs) and~\ref{def:verified} 
(verified compiler), namely
$$ \forall S, C, B \notin {\tt Wrong}, ~~
  {\it Comp}(S) = {\tt OK}(C) ~\wedge~ S \Downarrow B
  \Longrightarrow
  C \Downarrow B \eqno(i) $$

\item A verified compiler can be structured as a composition of
compilation passes, as is commonly done for conventional compilers.
Each pass can be proved sound independently.  However, all
intermediate languages must be given appropriate formal semantics.

\item For each pass, we have a choice between
proving the code that implements this pass or performing the
transformation via untrusted code, then verifying its results using a
verified validator.  The latter approach can reduce the amount
of code that needs to be proved.  In our experience, the verified
validator approach is particularly effective for advanced
optimizations, but less so for nonoptimizing translation passes and basic
dataflow optimizations.  Therefore, we did not use this approach for
the compilation passes presented in this article, but elected to prove
directly the soundness of these passes.
\footnote{However, a posteriori
validation with a verified validator is used for some auxiliary
heuristics such as graph coloring during register allocation
(section~\ref{s:regalloc-transf}) and node enumeration during CFG
linearization (section~\ref{s:linearize-transf}).}

\item Finally, provided the proof of $(i)$ is
carried out in a prover such as Coq that generates proof terms and
follows the Curry-Howard isomorphism, it is at least theoretically
possible to use the verified compiler in a context of proof-carrying
code.
\end{enumerate}

\section{Infrastructure} \label{s:infrastructure}

This section describes elements of syntax, semantics and proofs that
are used throughout the Compcert development.

\subsection{Programs} \label{s:programs}

The syntax of programs in the source, intermediate and target
languages share the following common shape.

\begin{syntax}
\syntaxclass{Programs:}
\prog & ::=  & \{~{\tt vars} = \id_1 = \data_1^*; \ldots \id_n = \data_n^*;
             & global variables  \\
      &      & ~~~{\tt functs} = \id_1 = \fd_1; \ldots \id_n = \fd_n;
             & functions \\
      &      & ~~~{\tt main} = \id ~\}
             & entry point
\syntaxclass{Function definitions:}
\fd & ::=  & {\tt internal}(\fn) \alt {\tt external}(\ef)
\syntaxclass{Definitions of internal functions:}
\fn & ::=  & \{~ {\tt sig} = \sig; {\tt body} = \ldots; \ldots ~\}
             & (language-dependent)
\syntaxclass{Declarations of external functions:}
\ef & ::=  & \{~ {\tt tag} = \id; {\tt sig} = \sig ~ \}
\syntaxclass{Initialization data for global variables:}
\data & ::=  & {\tt reserve}(n) \alt {\tt int8}(n) \alt {\tt int16}(n) \\
      & \alt & {\tt int32}(n) \alt {\tt float32}(f) \alt {\tt float64}(f)
\syntaxclass{Function signatures:}
\sig  & ::=  & \{~{\tt args} = \vec\tau; ~ {\tt res} = (\tau \alt {\tt void}) \}
\syntaxclass{Types:}
\tau  & ::=  & {\tt int} & integers and pointers \\
      & \alt & {\tt float} & floating-point numbers
\end{syntax}

A program is composed of a list of global variables with their
initialization data, a list of functions, and the name of a
distinguished function that constitutes the program entry point
(like {\tt main} in C).  Initialization data is a sequence of integer or
floating-point constants in various sizes, or ${\tt reserve}(n)$ to denote
$n$ bytes of uninitialized storage.  

Two kinds of function definitions $\fd$ are supported.
Internal functions $\fn$ are defined within the program.  The precise
contents of an internal function depends on the language considered,
but include at least a signature $\sig$ giving the number and types of
parameters and results and a body defining the computation (e.g.~as a
statement in \lang{Cminor} or a list of instructions in \lang{PPC}).
An external function $\ef$ is not defined within the program, but
merely declared with an external name and a signature.  
External functions are
intended to model input/output operations or other kinds of system
calls.  The observable behavior of the program will be defined in
terms of a trace of invocations of external functions (see
section~\ref{s:traces}).

The types $\tau$ used in function signatures and in other parts of
Compcert are extremely coarse: we only distinguish between integers or
pointers on the one hand (type {\tt int}) and floating-point numbers on
the other hand (type {\tt float}).  In particular, we make no attempt to
track the type of data pointed to by a pointer.  These ``types'' are
best thought of as hardware register classes.  Their main purpose is
to guide register allocation and help determine calling conventions
from the signature of the function being called.

Each compilation pass is presented as a total function
${\tt transf}: F_1 \rightarrow ({\tt OK}(F_2) \alt {\tt Error}({\it msg}))$
where $F_1$ and $F_2$ are the types of internal functions for the
source and target languages (respectively) of the compilation pass.
Such transformation functions are generically extended to function
definitions by taking ${\tt transf}(\ef) = {\tt OK}(\ef)$, then to whole
programs as a monadic ``map'' operation over function definitions: 
$$ {\tt transf}(P) = {\tt OK\ } \{ {\tt vars} = P.{\tt vars};
      {\tt functs} = (\ldots~ \id_i = \fd'_i; \ldots);
      {\tt main} = P.{\tt main} \} $$
if and only if $P.{\tt functs} = (\ldots~ \id_i = \fd_i; \ldots)$ and 
${\tt transf}(\fd_i) = {\tt OK}(\fd'_i)$ for all $i$.

\subsection{Values and memory states} \label{s:values-mem}

The dynamic semantics of the Compcert languages manipulate values that
are the discriminated union of 32-bit integers, 64-bit IEEE
double precision floats, pointers, and a special
{\tt undef} value denoting in particular the contents of uninitialized
memory.  Pointers are composed of a block identifier $b$ and a
signed byte offset $\delta$ within this block.

\begin{syntaxleft}
\syntaxclass{Values:}
v & ::=  & {\tt int}(n) & 32-bit machine integer \\
  & \alt & {\tt float}(f) & 64-bit floating-point number \\
  & \alt & {\tt ptr}(b, \ofs) & pointer \\
  & \alt & {\tt undef}
\syntaxclass{Memory blocks:}
b & \in~ & \mathbb{Z} & block identifiers
\syntaxclass{Block offsets:}
\ofs & ::= & n & byte offset within a block (signed)
\end{syntaxleft}

Values are assigned types in the obvious manner:
$$ {\tt int}(n) : {\tt int} \qquad
   {\tt float}(f) : {\tt float} \qquad
   {\tt ptr}(b,\ofs) : {\tt int} \qquad
   {\tt undef} : \tau \mbox{ for all $\tau$} $$

The memory model used in our semantics is detailed in
\cite{Leroy-Blazy-memory-model}.  Memory states $M$ are modeled as
collections of blocks separated by construction and identified by
(mathematical) integers $b$.  Each block has lower and upper bounds
${\cal L}(M, b)$ and ${\cal H}(M, b)$, fixed at allocation time, and
associates values to byte offsets $\ofs \in [{\cal L}(M, b), {\cal H}(M, b))$.
The basic operations over memory states are:
\begin{itemize}
\item ${\tt alloc}(M, l, h) = (b, M')$: 
  allocate a fresh block with bounds $[l, h)$, of size $(h-l)$ bytes;
  return its identifier $b$ and the updated memory state $M'$.
\item ${\tt store}(M, \chunk, b, \ofs, v) = \Some{M'}$:
  store value $v$ in the memory quantity $\chunk$ of block $b$ at
  offset $\ofs$; return update memory state $M'$.
\item ${\tt load}(M, \chunk, b, \ofs) = \Some{v}$:
  read the value $v$ contained in the memory quantity $\chunk$
  of block $b$ at offset $\ofs$.
\item ${\tt free}(M, b) = M'$:
  free (invalidate) the block $b$ and return the updated memory $M'$.
\end{itemize}

The memory quantities $\chunk$ involved in {\tt load} and {\tt store}
operations represent the kind, size and signedness of the datum being
accessed:
\begin{syntaxleft}
\syntaxclass{Memory quantities:}
\chunk & ::=  & {\tt int8signed} \alt {\tt int8unsigned} \alt
                {\tt int16signed} \\
       & \alt & {\tt int16unsigned} \alt {\tt int32} \alt {\tt float32} \alt {\tt float64}
\end{syntaxleft}

The {\tt load} and {\tt store} operations may fail when given an invalid block
$b$ or an out-of-bounds offset $\ofs$.  Therefore, they return
option types, with $\Some{v}$ (read: ``some~$v$'') denoting success
with result $v$, and $\None$ (read: ``none'') denoting failure.
In this particular instance of the memory model of
\cite{Leroy-Blazy-memory-model}, {\tt alloc} and {\tt free} never fail.  In
particular, this means that we assume an infinite memory.  This design
decision is discussed further in section~\ref{s:on-memory}.

The four operations of the memory model satisfy a number of algebraic
properties stated and proved in \cite{Leroy-Blazy-memory-model}.
The following ``load-after-store'' property gives the general flavor
of the memory model.  Assume ${\tt store}(M_1, \chunk, b, \ofs, v) = \Some{M_2}$
and ${\tt load}(M_1, \chunk', b', \ofs') = \Some{v'}$. Then,
$$ {\tt load}(M_2, \chunk', b', \ofs') = \threecases
{\Some{{\tt cast}(v, \chunk')}}
  {if $b' = b$ and $\ofs' = \ofs$ and $\size{\chunk'} = \size{\chunk}$;}
{\Some{v'}}
  {if $b' \not= b$ or $\ofs + \size{\chunk} \le \ofs'$ or
   $\ofs' + \size{\chunk'} \le \ofs$;}
{\Some{{\tt undef}}}
  {otherwise.}
$$
The ${\tt cast}(v, \chunk')$ function performs truncation or
sign-extension of value $v$ as prescribed by the quantity $\chunk'$.
Note that {\tt undef} is returned (instead of a machine-dependent value)
in cases where the quantities $\chunk$ and $\chunk'$ used for writing
and reading disagree, or in cases where the ranges of bytes written
$[\ofs, \ofs + \size{\chunk})$ and read $[\ofs', \ofs' + \size{\chunk'})$
partially overlap.  This way, the memory model hides the endianness
and bit-level representations of integers and floats and makes it
impossible to forge pointers from sequences of bytes 
\cite[section 7]{Leroy-Blazy-memory-model}.

\subsection{Global environments} \label{s:global-envs}

The Compcert languages support function pointers but follow a
``Harvard'' model where functions and data reside in different memory
spaces, and the memory space for functions is read-only (no
self-modifying code).  We use positive block identifiers $b$ to refer
to data blocks and negative $b$ to refer to functions via pointers.
The operational semantics for the
Compcert languages are parameterized by a global environment
$G$ that does not change during execution.  A global environment $G$
maps function blocks $b < 0$ to function definitions.  Moreover, it
maps global identifiers (of functions or global variables) to blocks
$b$.  The basic operations over global environments are:
\begin{itemize}
\item ${\tt funct}(G, b) = \Some{\fd}$: return the function definition $\fd$
  corresponding to the block $b < 0$, if any.
\item ${\tt symbol}(G, \id) = \Some{b}$: return the block $b$
  corresponding to the global variable or function name $\id$, if any.
\item ${\tt globalenv}(\prog) = G$: construct the global environment $G$
  associated with the program $\prog$.
\item ${\tt initmem}(\prog)= M$: construct the initial memory state $M$
  for executing the program $\prog$.
\end{itemize}

The ${\tt globalenv}(\prog)$ and ${\tt initmem}(\prog)$ functions model (at a
high level of abstraction) the operation of a linker and a program
loader. Unique, positive blocks $b$ are allocated and
associated to each global variable $(\id = \data^*)$ of $P$, and the
contents of these blocks are initialized according to $\data^*$.
Likewise, unique, negative blocks~$b$ are associated to each function
definition $(\id = \fd)$ of $\prog$.  In particular, if the functions
of $\prog$ have unique names, the following equivalence holds:
$$ (\id, \fd) \in \prog.{\tt functs} ~\Longleftrightarrow~
   \exists b < 0. ~
   \begin{array}[t]{l}
     {\tt symbol}({\tt globalenv}(\prog), \id) = \Some{b} \\
     {} \wedge  {\tt funct}({\tt globalenv}(\prog), b) = \Some{\fd}
   \end{array}
$$

The allocation of blocks for functions and global variables is
deterministic so that convenient commutation properties hold between
operations on global environments and per-function transformations of
programs as defined in section~\ref{s:programs}.

\begin{lemma}
Assume ${\tt transf}(\prog) = {\tt OK}(\prog')$.
\begin{itemize}
\item ${\tt initmem}(\prog') = {\tt initmem}(\prog)$.
\item If ${\tt symbol}({\tt globalenv}(\prog), \id) = \Some{b}$,
      then ${\tt symbol}({\tt globalenv}(\prog'), \id) = \Some{b}$.
\item If ${\tt funct}({\tt globalenv}(\prog), b) = \Some{\fd}$,
      then there exists a function definition $\fd'$ such that
      ${\tt funct}({\tt globalenv}(\prog'), b) = \Some{\fd'}$
      and ${\tt transf}(\fd) = {\tt OK}(\fd')$.
\end{itemize}
\end{lemma}

\subsection{Traces} \label{s:traces}

We express the observable behaviors of programs in terms of traces
of input-output events, each such event corresponding to an invocation
of an external function.  An event records the external name of the
external function, the values of the arguments provided by the
program, and the return value provided by the environment (\eg the
operating system).  

\begin{syntaxleft}
\syntaxclass{Events:}
\ev & ::= & \event{\id}{\vec{v_\ev}}{v_\ev}
\syntaxclass{Event values:}
v_\ev & ::= & {\tt int}(n) \alt {\tt float}(f)
\syntaxclass{Traces:}
t  & ::= & \epsilon \alt \ev.t & finite traces (inductive) \\
T  & ::= & \epsilon \alt \ev.T & finite or infinite traces (coinductive)
\syntaxclass{Behaviors:}
B  & ::= & {\tt converges}(t, n) & termination with trace $t$ and exit
code $n$ \\
   &\alt & {\tt diverges}(T)     & divergence with trace $T$ \\
   &\alt & {\tt goeswrong}(t)    & going wrong with trace $t$
\end{syntaxleft}

We consider two types of traces: finite traces $t$ for terminating
or ``going wrong'' executions and finite or infinite traces $T$ for diverging
executions.  Note that a diverging program can generate an empty or
finite trace of input-output events (think infinite empty loop).

Concatenation of a finite trace $t$ and a finite trace $t'$ or
infinite trace $T$ is written $t.t'$ or $t.T$.  It is associative and
admits the empty trace $\epsilon$ as neutral element.

The values that are arguments and results of input-output events are
required to be integers or floats.  Since external functions cannot
modify the memory state, passing them pointer values would be useless.
Even with this restriction, events and traces can still model
character-based input-output.  We encapsulate these restrictions
in the following inference rule that defines the effect of
applying an external function $\ef$ to arguments $\vec v$.

\begin{pannel}
\irule{
\mbox{$\vec v$ and $v$ are integers or floats} \\
\mbox{$\vec v$ and $v$ agree in number and types with $\ef.{\tt sig}$} \\
t = \event{\ef.{\tt tag}}{\vec v}{v}
}{
\vdash \ef(\vec v) \evaltrace{t} v
}
\end{pannel}
Note that the result value $v$ and therefore the trace $t$ are not
completely determined by this rule.  We return to this point in
section~\ref{s:ppc-determinism}.

\subsection{Transition semantics} \label{s:trans-sem}
The operational semantics for the source, target and intermediate
languages of the Compcert back-end are defined as labeled transition
systems.  The transition relation for each language is written
$\transtrace{\st}{t}{\st'}$ and denotes one execution step
from state $\st$ to state $\st'$ in global environment $G$.  The
trace $t$ denotes the observable events generated by this execution
step.  Transitions corresponding to an invocation of an external
function record the associated event in $t$.  Other transitions have
$t = \epsilon$.
In addition to the type of states $\st$ and the transition relation 
$\transtrace{\st}{t}{\st'}$, each language defines two predicates:
\begin{itemize}
\item $\initialstate{\prog}{\st}$: the state $\st$ is an initial state
for the program $\prog$.  Typically, $\st$ corresponds to an
invocation of the main function of $\prog$ in the initial memory state
${\tt initmem}(\prog)$.
\item $\finalstate{\st}{n}$: the state $\st$ is a final state with
  exit code $n$.  Typically, this means that the program is returning
  from the initial invocation of its main function, with return value
  ${\tt int}(n)$. 
\end{itemize}
Executions are modeled classically as sequences of transitions from an
initial state to a final state.  We write
$\transp{\st}{t}{\st'}$ to denote one or several transitions
(transitive closure),
$\transm{\st}{t}{\st'}$ to denote zero, one or several transitions
(reflexive transitive closure), and
$\transinf{\st}{T}$ to denote an infinite sequence of
transitions starting with $\st$.  The traces $t$ (finite) and $T$
(finite or infinite) are formed by concatenating the traces of elementary
transitions.  Formally:
\begin{pannel}
\srule{  \transm{\st}{\epsilon}{\st}  }
\irule{
        \transtrace{\st}{t_1}{\st'} \amper \transm{\st'}{t_2}{\st''}
        }{
        \transm{\st}{t_1.t_2}{\st''}
}
\\
\irule{
        \transtrace{\st}{t_1}{\st'} \amper \transm{\st'}{t_2}{\st''}
        }{
        \transp{\st}{t_1.t_2}{\st''}
}
\iruledouble{
        \transtrace{\st}{t}{\st'} \amper \transinf{\st'}{T}
        }{
        \transinf{\st}{t.T}
}
\end{pannel}
As denoted by the double horizontal bar, the inference rule defining
$\transinf{\st}{T}$ is to be interpreted coinductively, as a greatest
fixpoint.
The observable behavior of a program $\prog$ is defined as follows.
Starting from an initial state, if a finite sequence of reductions
with trace $t$ leads to a final state with exit code $n$, the
program has observable behavior ${\tt converges}(t, n)$.  If an infinite
sequence of reductions with trace $T$ is possible, the observable
behavior of the program is ${\tt diverges}(T)$.  Finally, if the
program gets stuck on a non-final state after performing a sequence of
reductions with trace $t$, the behavior is ${\tt goeswrong}(t)$.
\begin{pannel}
\irule{
       \initialstate{\prog}{\st} \amper
       {\tt globalenv}(\prog) \vdash \st \redm{t} \st' \amper
       \finalstate{\st'}{n}
       }{
       \prog \Downarrow {\tt converges}(t, n)
}
\irule{
       \initialstate{\prog}{\st} \amper
       {\tt globalenv}(\prog) \vdash \st \redinf{T}
       }{
       \prog \Downarrow {\tt diverges}(T)
}
\irule{
       \initialstate{\prog}{\st} \amper
       {\tt globalenv}(\prog) \vdash \st \redm{t} \st' \amper
       \st' \not\rightarrow {} \amper \forall n,~\neg\finalstate{\st'}{n}
       }{
       \prog \Downarrow {\tt goeswrong}(t)
}
\end{pannel}
The set of ``going wrong'' behaviors is defined in the obvious manner:
${\tt Wrong} = \{ {\tt goeswrong}(t) \mid t\mbox{ a finite trace} \}$.
\subsection{Program states} \label{s:state-kinds}
The contents of a program state vary from language to language.  For the
assembly language \lang{PPC}, a state is just a pair of a memory state
and a mapping from processor registers to values
(section~\ref{s:ppc-semantics}).  For the other languages of the
Compcert back-end, states come in three kinds written
$\state$, $\callstate$ and $\returnstate$.
\begin{itemize}
\item Regular states $\state$ correspond to an execution point within
an internal function.  They carry the function in question and a
program point within this function, possibly along with additional
language-specific components such as environments giving values to
function-local variables.
\item Call states
$\callstate$ materialize parameter passing from the caller to the
callee.  They carry the function definition $\fd$ being invoked and
either a list of argument values or an environment where the argument
values can be found at conventional locations.  
\item Return states $\returnstate$ correspond to returning from a
function to its caller.  They carry at least the return value or an
environment where this value can be found.
\end{itemize}
All three kinds of states also carry the current memory state as well
as a call stack: a list of frames describing the functions in
the call chain, with the corresponding program points 
where execution should be resumed 
on return, possibly along with function-local environments.

\begin{figure}
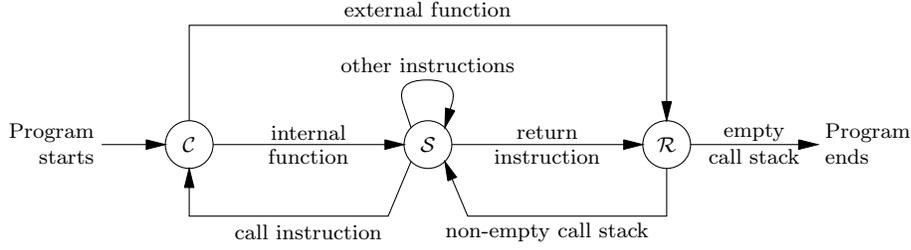


\begin{gpic}

\begin{pgfpicture}
  \pgfsetxvec{\pgfpoint{0.833in}{0in}}
  \pgfsetyvec{\pgfpoint{0in}{0.833in}}
  \begin{pgfscope}
    \pgfpathellipse{\pgfpointxy{2.050}{-0.150}}{\pgfpointxy{0.150}{0}}{\pgfpointxy{0}{0.150}}
    \pgfusepath{stroke}
  \end{pgfscope}
  \pgftext[at=\pgfpointadd{\pgfpointxy{2.050}{-0.150}}{\pgfpoint{0pt}{-0.0 \baselineskip}}]{$\state$}
  \begin{pgfscope}
    \pgfpathellipse{\pgfpointxy{0.550}{-0.150}}{\pgfpointxy{0.150}{0}}{\pgfpointxy{0}{0.150}}
    \pgfusepath{stroke}
  \end{pgfscope}
  \pgftext[at=\pgfpointadd{\pgfpointxy{0.550}{-0.150}}{\pgfpoint{0pt}{-0.0 \baselineskip}}]{$\callstate$}
  \begin{pgfscope}
    \pgfpathellipse{\pgfpointxy{3.550}{-0.150}}{\pgfpointxy{0.150}{0}}{\pgfpointxy{0}{0.150}}
    \pgfusepath{stroke}
  \end{pgfscope}
  \pgftext[at=\pgfpointadd{\pgfpointxy{3.550}{-0.150}}{\pgfpoint{0pt}{-0.0 \baselineskip}}]{$\returnstate$}
  \begin{pgfscope}
    \pgfpathmoveto{\pgfpointxy{2.219}{0.063}}
    \pgfpathlineto{\pgfpointxy{2.156}{-0.044}}
    \pgfpathlineto{\pgfpointxy{2.161}{0.080}}
    \pgfpathlineto{\pgfpointxy{2.219}{0.063}}
    \pgfsetfillcolor{black}
    \pgfsetlinewidth{0.100pt}
    \pgfusepath{fill,stroke}
  \end{pgfscope}
  \begin{pgfscope}
    \pgfpathmoveto{\pgfpointxy{1.944}{-0.044}}
    \pgfpathlineto{\pgfpointxy{1.894}{0.106}}
    \pgfpathcurveto{\pgfpointxy{1.861}{0.206}}{\pgfpointxy{1.911}{0.256}}{\pgfpointxy{2.044}{0.256}}
    \pgfpathcurveto{\pgfpointxy{2.177}{0.256}}{\pgfpointxy{2.230}{0.208}}{\pgfpointxy{2.202}{0.112}}
    \pgfpathlineto{\pgfpointxy{2.159}{-0.032}}
    \pgfusepath{stroke}
  \end{pgfscope}
  \pgftext[at=\pgfpointadd{\pgfpointxy{2.050}{0.350}}{\pgfpoint{0pt}{-0.0 \baselineskip}}]{other instructions}
  \begin{pgfscope}
    \pgfpathmoveto{\pgfpointxy{1.780}{-0.120}}
    \pgfpathlineto{\pgfpointxy{1.900}{-0.150}}
    \pgfpathlineto{\pgfpointxy{1.780}{-0.180}}
    \pgfpathlineto{\pgfpointxy{1.780}{-0.120}}
    \pgfsetfillcolor{black}
    \pgfsetlinewidth{0.100pt}
    \pgfusepath{fill,stroke}
  \end{pgfscope}
  \begin{pgfscope}
    \pgfpathmoveto{\pgfpointxy{0.700}{-0.150}}
    \pgfpathlineto{\pgfpointxy{1.780}{-0.150}}
    \pgfusepath{stroke}
  \end{pgfscope}
  \pgftext[at=\pgfpointadd{\pgfpointxy{1.300}{-0.150}}{\pgfpoint{0pt}{0.5 \baselineskip}}]{internal}
  \pgftext[at=\pgfpointadd{\pgfpointxy{1.300}{-0.150}}{\pgfpoint{0pt}{-0.5 \baselineskip}}]{function}
  \begin{pgfscope}
    \pgfpathmoveto{\pgfpointxy{3.280}{-0.120}}
    \pgfpathlineto{\pgfpointxy{3.400}{-0.150}}
    \pgfpathlineto{\pgfpointxy{3.280}{-0.180}}
    \pgfpathlineto{\pgfpointxy{3.280}{-0.120}}
    \pgfsetfillcolor{black}
    \pgfsetlinewidth{0.100pt}
    \pgfusepath{fill,stroke}
  \end{pgfscope}
  \begin{pgfscope}
    \pgfpathmoveto{\pgfpointxy{2.200}{-0.150}}
    \pgfpathlineto{\pgfpointxy{3.280}{-0.150}}
    \pgfusepath{stroke}
  \end{pgfscope}
  \pgftext[at=\pgfpointadd{\pgfpointxy{2.800}{-0.150}}{\pgfpoint{0pt}{0.5 \baselineskip}}]{return}
  \pgftext[at=\pgfpointadd{\pgfpointxy{2.800}{-0.150}}{\pgfpoint{0pt}{-0.5 \baselineskip}}]{instruction}
  \begin{pgfscope}
    \pgfpathmoveto{\pgfpointxy{0.520}{-0.420}}
    \pgfpathlineto{\pgfpointxy{0.550}{-0.300}}
    \pgfpathlineto{\pgfpointxy{0.580}{-0.420}}
    \pgfpathlineto{\pgfpointxy{0.520}{-0.420}}
    \pgfsetfillcolor{black}
    \pgfsetlinewidth{0.100pt}
    \pgfusepath{fill,stroke}
  \end{pgfscope}
  \begin{pgfscope}
    \pgfpathmoveto{\pgfpointxy{0.550}{-0.420}}
    \pgfpathlineto{\pgfpointxy{0.550}{-0.600}}
    \pgfpathlineto{\pgfpointxy{1.800}{-0.600}}
    \pgfpathlineto{\pgfpointxy{1.944}{-0.256}}
    \pgfusepath{stroke}
  \end{pgfscope}
  \pgftext[at=\pgfpointadd{\pgfpointxy{1.300}{-0.700}}{\pgfpoint{0pt}{-0.0 \baselineskip}}]{call instruction}
  \begin{pgfscope}
    \pgfpathmoveto{\pgfpointxy{2.175}{-0.378}}
    \pgfpathlineto{\pgfpointxy{2.156}{-0.256}}
    \pgfpathlineto{\pgfpointxy{2.230}{-0.355}}
    \pgfpathlineto{\pgfpointxy{2.175}{-0.378}}
    \pgfsetfillcolor{black}
    \pgfsetlinewidth{0.100pt}
    \pgfusepath{fill,stroke}
  \end{pgfscope}
  \begin{pgfscope}
    \pgfpathmoveto{\pgfpointxy{3.550}{-0.300}}
    \pgfpathlineto{\pgfpointxy{3.550}{-0.600}}
    \pgfpathlineto{\pgfpointxy{2.300}{-0.600}}
    \pgfpathlineto{\pgfpointxy{2.202}{-0.367}}
    \pgfusepath{stroke}
  \end{pgfscope}
  \pgftext[at=\pgfpointadd{\pgfpointxy{2.800}{-0.700}}{\pgfpoint{0pt}{-0.0 \baselineskip}}]{non-empty call stack}
  \begin{pgfscope}
    \pgfpathmoveto{\pgfpointxy{3.580}{0.120}}
    \pgfpathlineto{\pgfpointxy{3.550}{0.000}}
    \pgfpathlineto{\pgfpointxy{3.520}{0.120}}
    \pgfpathlineto{\pgfpointxy{3.580}{0.120}}
    \pgfsetfillcolor{black}
    \pgfsetlinewidth{0.100pt}
    \pgfusepath{fill,stroke}
  \end{pgfscope}
  \begin{pgfscope}
    \pgfpathmoveto{\pgfpointxy{0.550}{0.000}}
    \pgfpathlineto{\pgfpointxy{0.550}{0.600}}
    \pgfpathlineto{\pgfpointxy{3.550}{0.600}}
    \pgfpathlineto{\pgfpointxy{3.550}{0.120}}
    \pgfusepath{stroke}
  \end{pgfscope}
  \pgftext[at=\pgfpointadd{\pgfpointxy{2.050}{0.700}}{\pgfpoint{0pt}{-0.0 \baselineskip}}]{external function}
  \begin{pgfscope}
    \pgfpathmoveto{\pgfpointxy{0.280}{-0.120}}
    \pgfpathlineto{\pgfpointxy{0.400}{-0.150}}
    \pgfpathlineto{\pgfpointxy{0.280}{-0.180}}
    \pgfpathlineto{\pgfpointxy{0.280}{-0.120}}
    \pgfsetfillcolor{black}
    \pgfsetlinewidth{0.100pt}
    \pgfusepath{fill,stroke}
  \end{pgfscope}
  \begin{pgfscope}
    \pgfpathmoveto{\pgfpointxy{0.280}{-0.150}}
    \pgfpathlineto{\pgfpointxy{0.000}{-0.150}}
    \pgfusepath{stroke}
  \end{pgfscope}
  \pgftext[at=\pgfpointadd{\pgfpointxy{0.000}{-0.150}}{\pgfpoint{0pt}{0.5 \baselineskip}},right]{Program~}
  \pgftext[at=\pgfpointadd{\pgfpointxy{0.000}{-0.150}}{\pgfpoint{0pt}{-0.5 \baselineskip}},right]{starts~}
  \begin{pgfscope}
    \pgfpathmoveto{\pgfpointxy{4.380}{-0.120}}
    \pgfpathlineto{\pgfpointxy{4.500}{-0.150}}
    \pgfpathlineto{\pgfpointxy{4.380}{-0.180}}
    \pgfpathlineto{\pgfpointxy{4.380}{-0.120}}
    \pgfsetfillcolor{black}
    \pgfsetlinewidth{0.100pt}
    \pgfusepath{fill,stroke}
  \end{pgfscope}
  \begin{pgfscope}
    \pgfpathmoveto{\pgfpointxy{3.700}{-0.150}}
    \pgfpathlineto{\pgfpointxy{4.380}{-0.150}}
    \pgfusepath{stroke}
  \end{pgfscope}
  \pgftext[at=\pgfpointadd{\pgfpointxy{4.100}{-0.150}}{\pgfpoint{0pt}{0.5 \baselineskip}}]{empty}
  \pgftext[at=\pgfpointadd{\pgfpointxy{4.100}{-0.150}}{\pgfpoint{0pt}{-0.5 \baselineskip}}]{call stack}
  \pgftext[at=\pgfpointadd{\pgfpointxy{4.500}{-0.150}}{\pgfpoint{0pt}{0.5 \baselineskip}},left]{~Program}
  \pgftext[at=\pgfpointadd{\pgfpointxy{4.500}{-0.150}}{\pgfpoint{0pt}{-0.5 \baselineskip}},left]{~ends~}
\end{pgfpicture}
\ifx\Setlineno\undefined\else\Setlineno=46\fi
\end{gpic}

\caption{Transitions between the three kinds of program states.}
\label{f:state-auto}

\end{figure}

If we project the transition relation on the three-element set
$\{ \state, \callstate, \returnstate \}$, abstracting away the
components carried by the states, we obtain the finite automaton
depicted in figure~\ref{f:state-auto}.  This automaton is shared by
all languages of the Compcert back-end except \lang{PPC}, and
it illustrates the interplay between the three kinds of states.
Initial states are call states with empty call stacks.  A call state
where the called function is external transitions directly to a return state after
generating the appropriate event in the trace.  A call state where the
called function is internal transitions to a regular state corresponding to the
function entry point, possibly after binding the argument values to
the parameter variables.  Non-call, non-return instructions go from
regular states to regular states.  A non-tail call instruction
resolves the called function, pushes a return frame on the call stack
and transitions to the corresponding call state.  A tail call is
similar but does not push a return frame.  A return instruction
transitions to a return state.  A return state with a 
non-empty call stack pops the top return frame and moves to the
corresponding regular state.  A return state with an empty call stack
is a final state.
\subsection{Generic simulation diagrams} \label{s:simulation-diagrams}
Consider two languages $L_1$ and $L_2$ defined by their transition
semantics as described in section~\ref{s:trans-sem}.  Let $\prog_1$ be
a program in $L_1$ and $\prog_2$ a program in $L_2$ obtained by
applying a transformation to $\prog_1$.  We wish to show that
$\prog_2$ preserves the semantics of $\prog_1$, that is,
$\prog_1 \Downarrow B \Longrightarrow \prog_2 \Downarrow B$ for all
behaviors $B \notin {\tt Wrong}$.
The approach we use throughout this work is to construct a
relation $\st_1 \match \st_2 $ between states of $L_1$ and states of
$L_2$ and show that it is a forward simulation.  First, initial states and
final states should be related by $\match$ in the following sense:
\begin{itemize}
\item Initial states: if $\initialstate{\prog_1}{\st_1}$ and
  $\initialstate{\prog_2}{\st_2}$, then $\st_1 \match \st_2$.
\item Final states: if $\st_1 \match \st_2$ and $\finalstate{\st_1}{n}$,
  then $\finalstate{\st_2}{n}$.
\end{itemize}
Second, assuming $\st_1 \match \st_2$, we need to relate transitions
starting from $\st_1$ in $L_1$ with transitions starting from $S_2$ in
$L_2$.  The simplest property that guarantees semantic preservation is
the following lock-step simulation property:
\begin{definition} Lock-step simulation:
if $\st_1 \match \st_2$ and $G_1 \vdash \st_1
\redtrace{t} \st_1'$, there exists $\st_2'$ such that $G_2 \vdash \st_2
\redtrace{t} \st_2'$ and $\st_1' \match \st_2'$.
\end{definition}
($G_1$ and $G_2$ are the global environments corresponding to
$\prog_1$ and $\prog_2$, respectively.)  Figure~\ref{f:simu-diagrams},
top left, shows the corresponding diagram.

\begin{figure}
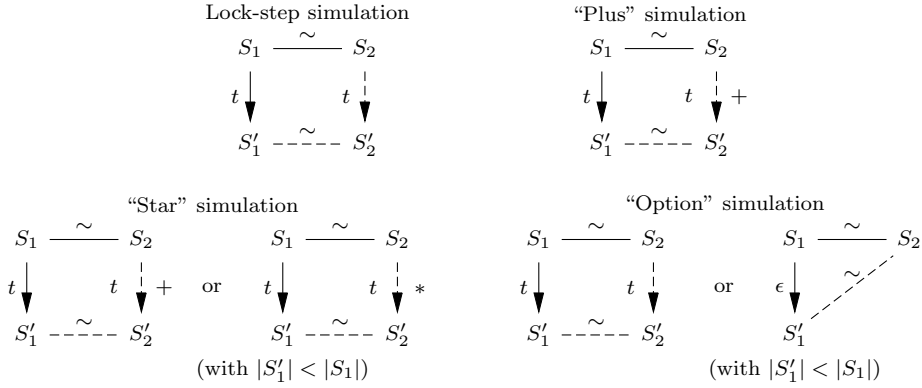


\begin{gpic}

\begin{pgfpicture}
  \pgfsetxvec{\pgfpoint{1.000in}{0in}}
  \pgfsetyvec{\pgfpoint{0in}{1.000in}}
  \pgftext[at=\pgfpointadd{\pgfpointxy{1.290}{0.750}}{\pgfpoint{0pt}{-0.0 \baselineskip}}]{$\st_1$}
  \pgftext[at=\pgfpointadd{\pgfpointxy{1.890}{0.750}}{\pgfpoint{0pt}{-0.0 \baselineskip}}]{$\st_2$}
  \pgftext[at=\pgfpointadd{\pgfpointxy{1.290}{0.250}}{\pgfpoint{0pt}{-0.0 \baselineskip}}]{$\st_1'$}
  \pgftext[at=\pgfpointadd{\pgfpointxy{1.890}{0.250}}{\pgfpoint{0pt}{-0.0 \baselineskip}}]{$\st_2'$}
  \begin{pgfscope}
    \pgfpathmoveto{\pgfpointxy{1.410}{0.750}}
    \pgfpathlineto{\pgfpointxy{1.770}{0.750}}
    \pgfusepath{stroke}
  \end{pgfscope}
  \pgftext[at=\pgfpointadd{\pgfpointxy{1.590}{0.750}}{\pgfpoint{0pt}{0.5 \baselineskip}}]{$\match$}
  \begin{pgfscope}
    \pgfpathmoveto{\pgfpointxy{1.315}{0.470}}
    \pgfpathlineto{\pgfpointxy{1.290}{0.370}}
    \pgfpathlineto{\pgfpointxy{1.265}{0.470}}
    \pgfpathlineto{\pgfpointxy{1.315}{0.470}}
    \pgfsetfillcolor{black}
    \pgfsetlinewidth{0.100pt}
    \pgfusepath{fill,stroke}
  \end{pgfscope}
  \begin{pgfscope}
    \pgfpathmoveto{\pgfpointxy{1.290}{0.630}}
    \pgfpathlineto{\pgfpointxy{1.290}{0.470}}
    \pgfusepath{stroke}
  \end{pgfscope}
  \pgftext[at=\pgfpointadd{\pgfpointxy{1.290}{0.500}}{\pgfpoint{0pt}{-0.0 \baselineskip}}]{$t$~~~~}
  \begin{pgfscope}
    \pgfpathmoveto{\pgfpointxy{1.410}{0.250}}
    \pgfpathlineto{\pgfpointxy{1.770}{0.250}}
    \pgfsetdash{{0.050in}{0.025in}}{0cm}
    \pgfusepath{stroke}
  \end{pgfscope}
  \pgftext[at=\pgfpointadd{\pgfpointxy{1.590}{0.250}}{\pgfpoint{0pt}{0.5 \baselineskip}}]{$\match$}
  \begin{pgfscope}
    \pgfpathmoveto{\pgfpointxy{1.915}{0.470}}
    \pgfpathlineto{\pgfpointxy{1.890}{0.370}}
    \pgfpathlineto{\pgfpointxy{1.865}{0.470}}
    \pgfpathlineto{\pgfpointxy{1.915}{0.470}}
    \pgfsetfillcolor{black}
    \pgfsetlinewidth{0.100pt}
    \pgfusepath{fill,stroke}
  \end{pgfscope}
  \begin{pgfscope}
    \pgfpathmoveto{\pgfpointxy{1.890}{0.630}}
    \pgfpathlineto{\pgfpointxy{1.890}{0.470}}
    \pgfsetdash{{0.050in}{0.025in}}{0cm}
    \pgfusepath{stroke}
  \end{pgfscope}
  \pgftext[at=\pgfpointadd{\pgfpointxy{1.890}{0.500}}{\pgfpoint{0pt}{-0.0 \baselineskip}}]{$t$~~~~~}
  \pgftext[at=\pgfpointadd{\pgfpointxy{3.130}{0.750}}{\pgfpoint{0pt}{-0.0 \baselineskip}}]{$\st_1$}
  \pgftext[at=\pgfpointadd{\pgfpointxy{3.730}{0.750}}{\pgfpoint{0pt}{-0.0 \baselineskip}}]{$\st_2$}
  \pgftext[at=\pgfpointadd{\pgfpointxy{3.130}{0.250}}{\pgfpoint{0pt}{-0.0 \baselineskip}}]{$\st_1'$}
  \pgftext[at=\pgfpointadd{\pgfpointxy{3.730}{0.250}}{\pgfpoint{0pt}{-0.0 \baselineskip}}]{$\st_2'$}
  \begin{pgfscope}
    \pgfpathmoveto{\pgfpointxy{3.250}{0.750}}
    \pgfpathlineto{\pgfpointxy{3.610}{0.750}}
    \pgfusepath{stroke}
  \end{pgfscope}
  \pgftext[at=\pgfpointadd{\pgfpointxy{3.430}{0.750}}{\pgfpoint{0pt}{0.5 \baselineskip}}]{$\match$}
  \begin{pgfscope}
    \pgfpathmoveto{\pgfpointxy{3.155}{0.470}}
    \pgfpathlineto{\pgfpointxy{3.130}{0.370}}
    \pgfpathlineto{\pgfpointxy{3.105}{0.470}}
    \pgfpathlineto{\pgfpointxy{3.155}{0.470}}
    \pgfsetfillcolor{black}
    \pgfsetlinewidth{0.100pt}
    \pgfusepath{fill,stroke}
  \end{pgfscope}
  \begin{pgfscope}
    \pgfpathmoveto{\pgfpointxy{3.130}{0.630}}
    \pgfpathlineto{\pgfpointxy{3.130}{0.470}}
    \pgfusepath{stroke}
  \end{pgfscope}
  \pgftext[at=\pgfpointadd{\pgfpointxy{3.130}{0.500}}{\pgfpoint{0pt}{-0.0 \baselineskip}}]{$t$~~~~}
  \begin{pgfscope}
    \pgfpathmoveto{\pgfpointxy{3.250}{0.250}}
    \pgfpathlineto{\pgfpointxy{3.610}{0.250}}
    \pgfsetdash{{0.050in}{0.025in}}{0cm}
    \pgfusepath{stroke}
  \end{pgfscope}
  \pgftext[at=\pgfpointadd{\pgfpointxy{3.430}{0.250}}{\pgfpoint{0pt}{0.5 \baselineskip}}]{$\match$}
  \begin{pgfscope}
    \pgfpathmoveto{\pgfpointxy{3.755}{0.470}}
    \pgfpathlineto{\pgfpointxy{3.730}{0.370}}
    \pgfpathlineto{\pgfpointxy{3.705}{0.470}}
    \pgfpathlineto{\pgfpointxy{3.755}{0.470}}
    \pgfsetfillcolor{black}
    \pgfsetlinewidth{0.100pt}
    \pgfusepath{fill,stroke}
  \end{pgfscope}
  \begin{pgfscope}
    \pgfpathmoveto{\pgfpointxy{3.730}{0.630}}
    \pgfpathlineto{\pgfpointxy{3.730}{0.470}}
    \pgfsetdash{{0.050in}{0.025in}}{0cm}
    \pgfusepath{stroke}
  \end{pgfscope}
  \pgftext[at=\pgfpointadd{\pgfpointxy{3.730}{0.500}}{\pgfpoint{0pt}{-0.0 \baselineskip}}]{$t$~~~~~$+$}
  \pgftext[at=\pgfpointadd{\pgfpointxy{0.120}{-0.250}}{\pgfpoint{0pt}{-0.0 \baselineskip}}]{$\st_1$}
  \pgftext[at=\pgfpointadd{\pgfpointxy{0.720}{-0.250}}{\pgfpoint{0pt}{-0.0 \baselineskip}}]{$\st_2$}
  \pgftext[at=\pgfpointadd{\pgfpointxy{0.120}{-0.750}}{\pgfpoint{0pt}{-0.0 \baselineskip}}]{$\st_1'$}
  \pgftext[at=\pgfpointadd{\pgfpointxy{0.720}{-0.750}}{\pgfpoint{0pt}{-0.0 \baselineskip}}]{$\st_2'$}
  \begin{pgfscope}
    \pgfpathmoveto{\pgfpointxy{0.240}{-0.250}}
    \pgfpathlineto{\pgfpointxy{0.600}{-0.250}}
    \pgfusepath{stroke}
  \end{pgfscope}
  \pgftext[at=\pgfpointadd{\pgfpointxy{0.420}{-0.250}}{\pgfpoint{0pt}{0.5 \baselineskip}}]{$\match$}
  \begin{pgfscope}
    \pgfpathmoveto{\pgfpointxy{0.145}{-0.530}}
    \pgfpathlineto{\pgfpointxy{0.120}{-0.630}}
    \pgfpathlineto{\pgfpointxy{0.095}{-0.530}}
    \pgfpathlineto{\pgfpointxy{0.145}{-0.530}}
    \pgfsetfillcolor{black}
    \pgfsetlinewidth{0.100pt}
    \pgfusepath{fill,stroke}
  \end{pgfscope}
  \begin{pgfscope}
    \pgfpathmoveto{\pgfpointxy{0.120}{-0.370}}
    \pgfpathlineto{\pgfpointxy{0.120}{-0.530}}
    \pgfusepath{stroke}
  \end{pgfscope}
  \pgftext[at=\pgfpointadd{\pgfpointxy{0.120}{-0.500}}{\pgfpoint{0pt}{-0.0 \baselineskip}}]{$t$~~~~}
  \begin{pgfscope}
    \pgfpathmoveto{\pgfpointxy{0.240}{-0.750}}
    \pgfpathlineto{\pgfpointxy{0.600}{-0.750}}
    \pgfsetdash{{0.050in}{0.025in}}{0cm}
    \pgfusepath{stroke}
  \end{pgfscope}
  \pgftext[at=\pgfpointadd{\pgfpointxy{0.420}{-0.750}}{\pgfpoint{0pt}{0.5 \baselineskip}}]{$\match$}
  \begin{pgfscope}
    \pgfpathmoveto{\pgfpointxy{0.745}{-0.530}}
    \pgfpathlineto{\pgfpointxy{0.720}{-0.630}}
    \pgfpathlineto{\pgfpointxy{0.695}{-0.530}}
    \pgfpathlineto{\pgfpointxy{0.745}{-0.530}}
    \pgfsetfillcolor{black}
    \pgfsetlinewidth{0.100pt}
    \pgfusepath{fill,stroke}
  \end{pgfscope}
  \begin{pgfscope}
    \pgfpathmoveto{\pgfpointxy{0.720}{-0.370}}
    \pgfpathlineto{\pgfpointxy{0.720}{-0.530}}
    \pgfsetdash{{0.050in}{0.025in}}{0cm}
    \pgfusepath{stroke}
  \end{pgfscope}
  \pgftext[at=\pgfpointadd{\pgfpointxy{0.720}{-0.500}}{\pgfpoint{0pt}{-0.0 \baselineskip}}]{$t$~~~~~$+$}
  \pgftext[at=\pgfpointadd{\pgfpointxy{1.090}{-0.500}}{\pgfpoint{0pt}{-0.0 \baselineskip}}]{or}
  \pgftext[at=\pgfpointadd{\pgfpointxy{1.460}{-0.250}}{\pgfpoint{0pt}{-0.0 \baselineskip}}]{$\st_1$}
  \pgftext[at=\pgfpointadd{\pgfpointxy{2.060}{-0.250}}{\pgfpoint{0pt}{-0.0 \baselineskip}}]{$\st_2$}
  \pgftext[at=\pgfpointadd{\pgfpointxy{1.460}{-0.750}}{\pgfpoint{0pt}{-0.0 \baselineskip}}]{$\st_1'$}
  \pgftext[at=\pgfpointadd{\pgfpointxy{2.060}{-0.750}}{\pgfpoint{0pt}{-0.0 \baselineskip}}]{$\st_2'$}
  \begin{pgfscope}
    \pgfpathmoveto{\pgfpointxy{1.580}{-0.250}}
    \pgfpathlineto{\pgfpointxy{1.940}{-0.250}}
    \pgfusepath{stroke}
  \end{pgfscope}
  \pgftext[at=\pgfpointadd{\pgfpointxy{1.760}{-0.250}}{\pgfpoint{0pt}{0.5 \baselineskip}}]{$\match$}
  \begin{pgfscope}
    \pgfpathmoveto{\pgfpointxy{1.485}{-0.530}}
    \pgfpathlineto{\pgfpointxy{1.460}{-0.630}}
    \pgfpathlineto{\pgfpointxy{1.435}{-0.530}}
    \pgfpathlineto{\pgfpointxy{1.485}{-0.530}}
    \pgfsetfillcolor{black}
    \pgfsetlinewidth{0.100pt}
    \pgfusepath{fill,stroke}
  \end{pgfscope}
  \begin{pgfscope}
    \pgfpathmoveto{\pgfpointxy{1.460}{-0.370}}
    \pgfpathlineto{\pgfpointxy{1.460}{-0.530}}
    \pgfusepath{stroke}
  \end{pgfscope}
  \pgftext[at=\pgfpointadd{\pgfpointxy{1.460}{-0.500}}{\pgfpoint{0pt}{-0.0 \baselineskip}}]{$t$~~~~}
  \begin{pgfscope}
    \pgfpathmoveto{\pgfpointxy{1.580}{-0.750}}
    \pgfpathlineto{\pgfpointxy{1.940}{-0.750}}
    \pgfsetdash{{0.050in}{0.025in}}{0cm}
    \pgfusepath{stroke}
  \end{pgfscope}
  \pgftext[at=\pgfpointadd{\pgfpointxy{1.760}{-0.750}}{\pgfpoint{0pt}{0.5 \baselineskip}}]{$\match$}
  \begin{pgfscope}
    \pgfpathmoveto{\pgfpointxy{2.085}{-0.530}}
    \pgfpathlineto{\pgfpointxy{2.060}{-0.630}}
    \pgfpathlineto{\pgfpointxy{2.035}{-0.530}}
    \pgfpathlineto{\pgfpointxy{2.085}{-0.530}}
    \pgfsetfillcolor{black}
    \pgfsetlinewidth{0.100pt}
    \pgfusepath{fill,stroke}
  \end{pgfscope}
  \begin{pgfscope}
    \pgfpathmoveto{\pgfpointxy{2.060}{-0.370}}
    \pgfpathlineto{\pgfpointxy{2.060}{-0.530}}
    \pgfsetdash{{0.050in}{0.025in}}{0cm}
    \pgfusepath{stroke}
  \end{pgfscope}
  \pgftext[at=\pgfpointadd{\pgfpointxy{2.060}{-0.500}}{\pgfpoint{0pt}{-0.0 \baselineskip}}]{$t$~~~~~$*$}
  \pgftext[at=\pgfpointadd{\pgfpointxy{1.460}{-0.870}}{\pgfpoint{0pt}{-0.5 \baselineskip}}]{(with $\size{\st_1'} < \size{\st_1}$)}
  \pgftext[at=\pgfpointadd{\pgfpointxy{2.800}{-0.250}}{\pgfpoint{0pt}{-0.0 \baselineskip}}]{$\st_1$}
  \pgftext[at=\pgfpointadd{\pgfpointxy{3.400}{-0.250}}{\pgfpoint{0pt}{-0.0 \baselineskip}}]{$\st_2$}
  \pgftext[at=\pgfpointadd{\pgfpointxy{2.800}{-0.750}}{\pgfpoint{0pt}{-0.0 \baselineskip}}]{$\st_1'$}
  \pgftext[at=\pgfpointadd{\pgfpointxy{3.400}{-0.750}}{\pgfpoint{0pt}{-0.0 \baselineskip}}]{$\st_2'$}
  \begin{pgfscope}
    \pgfpathmoveto{\pgfpointxy{2.920}{-0.250}}
    \pgfpathlineto{\pgfpointxy{3.280}{-0.250}}
    \pgfusepath{stroke}
  \end{pgfscope}
  \pgftext[at=\pgfpointadd{\pgfpointxy{3.100}{-0.250}}{\pgfpoint{0pt}{0.5 \baselineskip}}]{$\match$}
  \begin{pgfscope}
    \pgfpathmoveto{\pgfpointxy{2.825}{-0.530}}
    \pgfpathlineto{\pgfpointxy{2.800}{-0.630}}
    \pgfpathlineto{\pgfpointxy{2.775}{-0.530}}
    \pgfpathlineto{\pgfpointxy{2.825}{-0.530}}
    \pgfsetfillcolor{black}
    \pgfsetlinewidth{0.100pt}
    \pgfusepath{fill,stroke}
  \end{pgfscope}
  \begin{pgfscope}
    \pgfpathmoveto{\pgfpointxy{2.800}{-0.370}}
    \pgfpathlineto{\pgfpointxy{2.800}{-0.530}}
    \pgfusepath{stroke}
  \end{pgfscope}
  \pgftext[at=\pgfpointadd{\pgfpointxy{2.800}{-0.500}}{\pgfpoint{0pt}{-0.0 \baselineskip}}]{$t$~~~~}
  \begin{pgfscope}
    \pgfpathmoveto{\pgfpointxy{2.920}{-0.750}}
    \pgfpathlineto{\pgfpointxy{3.280}{-0.750}}
    \pgfsetdash{{0.050in}{0.025in}}{0cm}
    \pgfusepath{stroke}
  \end{pgfscope}
  \pgftext[at=\pgfpointadd{\pgfpointxy{3.100}{-0.750}}{\pgfpoint{0pt}{0.5 \baselineskip}}]{$\match$}
  \begin{pgfscope}
    \pgfpathmoveto{\pgfpointxy{3.425}{-0.530}}
    \pgfpathlineto{\pgfpointxy{3.400}{-0.630}}
    \pgfpathlineto{\pgfpointxy{3.375}{-0.530}}
    \pgfpathlineto{\pgfpointxy{3.425}{-0.530}}
    \pgfsetfillcolor{black}
    \pgfsetlinewidth{0.100pt}
    \pgfusepath{fill,stroke}
  \end{pgfscope}
  \begin{pgfscope}
    \pgfpathmoveto{\pgfpointxy{3.400}{-0.370}}
    \pgfpathlineto{\pgfpointxy{3.400}{-0.530}}
    \pgfsetdash{{0.050in}{0.025in}}{0cm}
    \pgfusepath{stroke}
  \end{pgfscope}
  \pgftext[at=\pgfpointadd{\pgfpointxy{3.400}{-0.500}}{\pgfpoint{0pt}{-0.0 \baselineskip}}]{$t$~~~~~~}
  \pgftext[at=\pgfpointadd{\pgfpointxy{3.770}{-0.500}}{\pgfpoint{0pt}{-0.0 \baselineskip}}]{or}
  \pgftext[at=\pgfpointadd{\pgfpointxy{4.140}{-0.250}}{\pgfpoint{0pt}{-0.0 \baselineskip}}]{$\st_1$}
  \pgftext[at=\pgfpointadd{\pgfpointxy{4.740}{-0.250}}{\pgfpoint{0pt}{-0.0 \baselineskip}}]{$\st_2$}
  \pgftext[at=\pgfpointadd{\pgfpointxy{4.140}{-0.750}}{\pgfpoint{0pt}{-0.0 \baselineskip}}]{$\st_1'$}
  \begin{pgfscope}
    \pgfpathmoveto{\pgfpointxy{4.260}{-0.250}}
    \pgfpathlineto{\pgfpointxy{4.620}{-0.250}}
    \pgfusepath{stroke}
  \end{pgfscope}
  \pgftext[at=\pgfpointadd{\pgfpointxy{4.440}{-0.250}}{\pgfpoint{0pt}{0.5 \baselineskip}}]{$\match$}
  \begin{pgfscope}
    \pgfpathmoveto{\pgfpointxy{4.165}{-0.530}}
    \pgfpathlineto{\pgfpointxy{4.140}{-0.630}}
    \pgfpathlineto{\pgfpointxy{4.115}{-0.530}}
    \pgfpathlineto{\pgfpointxy{4.165}{-0.530}}
    \pgfsetfillcolor{black}
    \pgfsetlinewidth{0.100pt}
    \pgfusepath{fill,stroke}
  \end{pgfscope}
  \begin{pgfscope}
    \pgfpathmoveto{\pgfpointxy{4.140}{-0.370}}
    \pgfpathlineto{\pgfpointxy{4.140}{-0.530}}
    \pgfusepath{stroke}
  \end{pgfscope}
  \pgftext[at=\pgfpointadd{\pgfpointxy{4.140}{-0.500}}{\pgfpoint{0pt}{-0.0 \baselineskip}}]{$\epsilon$~~~~}
  \begin{pgfscope}
    \pgfpathmoveto{\pgfpointxy{4.225}{-0.665}}
    \pgfpathlineto{\pgfpointxy{4.655}{-0.335}}
    \pgfsetdash{{0.050in}{0.025in}}{0cm}
    \pgfusepath{stroke}
  \end{pgfscope}
  \pgftext[at=\pgfpointadd{\pgfpointxy{4.440}{-0.500}}{\pgfpoint{0pt}{0.5 \baselineskip}}]{$\match$}
  \pgftext[at=\pgfpointadd{\pgfpointxy{4.140}{-0.870}}{\pgfpoint{0pt}{-0.5 \baselineskip}}]{(with $\size{\st_1'} < \size{\st_1}$)}
  \pgftext[at=\pgfpointadd{\pgfpointxy{1.590}{0.870}}{\pgfpoint{0pt}{0.5 \baselineskip}}]{Lock-step simulation}
  \pgftext[at=\pgfpointadd{\pgfpointxy{3.430}{0.870}}{\pgfpoint{0pt}{0.5 \baselineskip}}]{``Plus'' simulation}
  \pgftext[at=\pgfpointadd{\pgfpointxy{1.090}{-0.130}}{\pgfpoint{0pt}{0.5 \baselineskip}}]{``Star'' simulation}
  \pgftext[at=\pgfpointadd{\pgfpointxy{3.770}{-0.130}}{\pgfpoint{0pt}{0.5 \baselineskip}}]{``Option'' simulation}
\end{pgfpicture}
\ifx\Setlineno\undefined\else\Setlineno=66\fi
\end{gpic}

\caption{Four kinds of simulation diagrams that imply semantic
  preservation.  Solid lines denote hypotheses; dashed lines denote
  conclusions.}
\label{f:simu-diagrams}

\end{figure}

\begin{theorem} \label{th:lockstep-simulation}
Under hypotheses ``initial states'', ``final states'' and ``lock-step
simulation'', $\prog_1 \Downarrow B$ and $B \notin {\tt Wrong}$
imply $\prog_2 \Downarrow B$.
\end{theorem}
\begin{proof}
A trivial induction shows that $\st_1 \match \st_2$ and $G_1 \vdash \st_1
\redm{t} \st_1'$ implies the existence of $\st_2'$ such that $G_2 \vdash \st_2
\redm{t} \st_2'$ and $\st_1' \match \st_2'$.  Likewise, a trivial
coinduction shows that $\st_1 \match \st_2$ and $G_1 \vdash \st_1 \redinf{T}$
implies $G_2 \vdash \st_2 \redinf{T}$.  The result follows from the
definition of $\Downarrow$.
\end{proof}
The lock-step simulation hypothesis is too strong for many program
transformations of interest, however.  Some transformations cause
transitions in $\prog_1$ to disappear in $\prog_2$, \eg removal of
no-operations, elimination of redundant computations, or branch
tunneling.  Likewise, some transformations introduce additional
transitions in $\prog_2$, \eg insertion of spilling and reloading
code.  Naively, we could try to relax the simulation hypothesis as
follows:
\begin{definition} Naive ``star'' simulation:
if $\st_1 \match \st_2$ and $G_1 \vdash \st_1 \redtrace{t} \st_1'$,
there exists $\st_2'$ such that $G_2 \vdash \st_2 \redm{t} \st_2'$
and $\st_1' \match \st_2'$.
\end{definition}
This hypothesis suffices to show the preservation of terminating
behaviors, but does not guarantee that diverging behaviors are
preserved because of the classic ``infinite stuttering'' problem.  The
original program $\prog_1$ could perform infinitely many silent
transitions $\st_1 \red \st_2 \red \ldots \red \st_n \red \ldots$
while the transformed program $\prog_2$ is stuck in a state $\st'$
such that $\st_i \match \st'$ for all $i$.  In this case, $\prog_1$
diverges while $\prog_2$ does not, and semantic preservation does not
hold.
To rule out the infinite stuttering problem, assume we are given a measure
$\size{\st_1}$ over the states of language $L_1$.  This measure ranges
over a type $\cal M$ equipped with a well-founded ordering $<$ (that
is, there are no infinite decreasing chains of elements of $\cal M$).
We require that the measure strictly decreases in cases where
stuttering could occur, making it impossible for stuttering to occur
infinitely. 
\begin{definition} ``Star'' simulation:
if $\st_1 \match \st_2$ and $G_1 \vdash \st_1 \redtrace{t} \st_1'$, either
\begin{enumerate}
\item there exists $\st_2'$ such that $G_2 \vdash \st_2
  \redp{t} \st_2'$ and $\st_1' \match \st_2'$,
\item or $\size{\st_1'} < \size{\st_1}$ and 
  there exists $\st_2'$ such that $G_2 \vdash \st_2 \redm{t} \st_2'$
  and $\st_1' \match \st_2'$.
\end{enumerate}
\end{definition}
Diagrammatically, this hypothesis corresponds to the bottom left part
of figure~\ref{f:simu-diagrams}.
(Equivalently, part 2 of the definition could be replaced by
``or $\size{\st_1'} < \size{\st_1}$ and $t = \epsilon$
and $\st_2 \match \st_1'$'', but the formulation above is more
convenient in practice.)
\begin{theorem} \label{th:star-simulation}
Under hypotheses ``initial states'', ``final states'' and ``star
simulation'', $\prog_1 \Downarrow B$  and $B \notin {\tt Wrong}$
imply $\prog_2 \Downarrow B$.
\end{theorem}
\begin{proof}
A trivial induction shows that $\st_1 \match \st_2$ and $G_1 \vdash \st_1
\redm{t} \st_1'$ implies the existence of $\st_2'$ such that $G_2 \vdash
\st_2 \redm{t} \st_2'$ and $\st_1' \match \st_2'$.  This implies the
desired result if $B$ is a terminating behavior.
For diverging behaviors, we first define (coinductively) the following
``measured'' variant of the $G_2 \vdash \st_2 \redinf{T}$ relation:
\begin{pannel}
\iruledouble{
G_2 \vdash \st_2 \redp{t} \st_2' \amper G_2 \vdash \st_2', \mu' \redinf{T}
}{
G_2 \vdash \st_2, \mu \redinf{t.T}
}
\iruledouble{
G_2 \vdash \st_2 \redm{t} \st_2' \amper \mu' < \mu \amper G_2 \vdash \st_2', \mu' \redinf{T}
}{
G_2 \vdash \st_2, \mu \redinf{t.T}
}
\end{pannel}
The second rule permits a number of potentially stuttering steps to be
taken, provided the measure $\mu$ strictly decreases.  After a finite
number of invocations of this rule, it becomes non applicable and the
first rule must be applied, forcing at least one transition to be taken
and resetting the measure to an arbitrarily-chosen value.  
A straightforward coinduction shows that $G_1 \vdash \st_1 \redinf{T}$
and $\st_1 \match \st_2$ implies $G_2 \vdash \st_2, \size{\st_1} \redinf{T}$.
To conclude, it suffices to prove that
$G_2 \vdash \st_2, \mu \redinf{T}$ implies
$G_2 \vdash \st_2 \redinf{T}$.  This follows by coinduction and the
following inversion lemma, proved by Noetherian induction over $\mu$:
if $G_2 \vdash \st_2, \mu \redinf{T}$, there exists $\st_2'$, $\mu'$, $t$
and $T'$ such that $G_2 \vdash \st_2 \redtrace{t} \st_2'$ and
$G_2 \vdash \st_2', \mu' \redinf{T'}$ and $T = t.T'$.
\end{proof}
Here are two stronger variants of the ``star'' simulation
hypothesis that are convenient in practice. (See
figure~\ref{f:simu-diagrams} for the corresponding diagrams.) 
\begin{definition} ``Plus'' simulation:
if $\st_1 \match \st_2$ and $G_1 \vdash \st_1
  \redtrace{t} \st_1'$, there exists $\st_2'$ such that $G_2 \vdash \st_2
  \redp{t} \st_2'$ and $\st_1' \match \st_2'$.
\end{definition}
\begin{definition} ``Option'' simulation:
if $\st_1 \match \st_2$ and $G_1 \vdash \st_1
  \redtrace{t} \st_1'$, either
\begin{enumerate}
\item there exists $\st_2'$ such that $G_2 \vdash \st_2
  \redtrace{t} \st_2'$ and $\st_1' \match \st_2'$,
\item or $\size{\st_1'} < \size{\st_1}$ and $t = \epsilon$ and
  $\st_1' \match \st_2$.
\end{enumerate}
\end{definition}
Either simulation hypothesis implies the ``star'' simulation property
and therefore semantic preservation per
theorem~\ref{th:star-simulation}.

\section{The source language: Cminor} \label{s:cminor}

The input language of our back-end is called \lang{Cminor}.  It is a
simple, low-level imperative language, comparable to a stripped-down,
typeless variant of C.  Another source of inspiration was the C{\tt --}
intermediate language of Peyton Jones et al. \cite{Peyton-Jones-Ramsey-99}.  
In the CompCert compilation chain, \lang{Cminor} is the
lowest-level language that is still processor independent; it is
therefore an appropriate language to start the back-end part of the
compiler.

\subsection{Syntax}

\lang{Cminor} is, classically, structured in expressions, statements,
functions and whole programs.
\begin{syntax}
\syntaxclass{Expressions:}
a & ::=  & \id  & reading a local variable \\
  & \alt & \cst & constant \\
  & \alt & \unop(a_1) & unary arithmetic operation \\
  & \alt & \binop(a_1,a_2) & binary arithmetic operation \\
  & \alt & \chunk[a_1] & memory read at address $a_1$ \\
  & \alt & a_1 {\tt \ ?\ } a_2 {\tt \ :\ } a_3 & conditional expression
\syntaxclass{Constants:}
\cst & ::=  & n \alt f & integer or float literal \\
     & \alt & {\tt addrsymbol}(\id) & address of a global symbol \\
     & \alt & {\tt addrstack}(\ofs) & address within stack data
\syntaxclass{Unary operators:}
\unop & ::= & {\tt negint} \alt {\tt notint} \alt {\tt notbool} & integer arithmetic \\
      & \alt & {\tt negf} \alt {\tt absf} & float arithmetic \\
      & \alt & {\tt cast8u} \alt {\tt cast8s} \alt {\tt cast16u} \alt {\tt cast16s} &
                           zero and sign extensions \\
      & \alt & {\tt singleoffloat} & float truncation \\
      & \alt & {\tt intoffloat} \alt {\tt intuoffloat} & float-to-int conversions \\
      & \alt & {\tt floatofint} \alt {\tt floatofintu} & int-to-float conversions
\syntaxclass{Binary operators:}
\binop & ::= & {\tt add} \alt {\tt sub} \alt {\tt mul} \alt {\tt div} \alt {\tt divu} \alt {\tt mod} \alt {\tt modu}
         & integer arithmetic \\
       & \alt& {\tt and} \alt {\tt or} \alt {\tt xor} \alt {\tt shl} \alt {\tt shr} \alt {\tt shru}
         & integer bit operation \\
       & \alt& {\tt addf} \alt {\tt subf} \alt {\tt mulf} \alt {\tt divf}
          & float arithmetic \\
       & \alt& {\tt cmp}(c) \alt {\tt cmpu}(c) \alt {\tt cmpf}(c) 
          & comparisons
\syntaxclass{Comparisons:}
c & ::= & {\tt eq} \alt {\tt ne} \alt {\tt gt} \alt {\tt ge} \alt {\tt lt} \alt {\tt le}
\end{syntax}

Expressions are pure: side-effecting operations such as assignment and
function calls are statements, not expressions.  All arithmetic and
logical operators of C are supported.  Unlike in C, there is no
overloading nor implicit conversions between types: distinct arithmetic
operators are provided over integers and over floats; likewise,
explicit conversion operators are provided to convert between floats
and integers and perform zero and sign extensions.  Memory loads are
explicit and annotated with the memory quantity $\chunk$ being
accessed.

\begin{syntax}
\syntaxclass{Statements:}
s & ::=  & {\tt skip}    & no operation \\
  & \alt & \id = a   & assignment to a local variable \\
  & \alt & \chunk[a_1] = a_2 & memory write at address $a_1$ \\
  & \alt & \id^? = a(\vec a) : \sig & function call \\
  & \alt & {\tt tailcall\ }a(\vec a) : \sig & function tail call \\
  & \alt & {\tt return}(a^?) & function return \\
  & \alt & s_1; s_2  & sequence \\
  & \alt & {\tt if}(a) \, \{ s_1 \} {\tt \ else\ } \{ s_2 \} & conditional \\
  & \alt & {\tt loop} \, \{ s_1 \} & infinite loop \\
  & \alt & {\tt block} \, \{ s_1 \} & block delimiting {\tt exit} constructs \\
  & \alt & {\tt exit}(n) & terminate the $(n+1)\th$ enclosing block \\
  & \alt & {\tt switch}(a) \, \{ {\it tbl} \} & multi-way test and exit \\
  & \alt & \lbl: s & labeled statement \\
  & \alt & {\tt goto\ }\lbl & jump to a label
\syntaxclass{Switch tables:}
{\it tbl} & ::= & {\tt default}: {\tt exit}(n) \\
  & \alt & {\tt case\ }i: {\tt exit}(n); {\it tbl}
\end{syntax}

Base statements are {\tt skip}, assignment $\id = a$ to a local variable,
memory store $\chunk[a_1] = a_2$ (of the value of $a_2$ in the
quantity $\chunk$ at address $a_1$), function call (with optional
assignment of the return value to a local variable), function tail
call, and function return (with an optional result).  Function calls
are annotated with the signatures $\sig$ expected for the called function.
A tail call ${\tt tailcall\ }a(\vec a)$ is almost equivalent to
a regular call immediately followed by a {\tt return},
except that the tail call deallocates the current stack data block
before invoking the function.  This enables tail recursion to execute
in constant stack space.

Besides the familiar sequence $s_1;s_2$ and {\tt if}/{\tt then}/{\tt else}
constructs, control flow can be expressed either in an unstructured
way using {\tt goto} and labeled statements or in a structured way
using infinite loops and the {\tt block}/{\tt exit} construct.  ${\tt exit}(n)$
where $n \ge 0$ branches to the end of the $(n+1)\th$ enclosing
{\tt block} construct.  The ${\tt switch}(a) \,\{ {\it tbl} \}$ construct matches
the integer value of $a$ against the cases listed in $\it tbl$ and
performs the corresponding {\tt exit}.  Appropriate nesting of a {\tt switch}
within several {\tt block} constructs suffices to express C-like
structured {\tt switch} statements with fall-through behavior.

\begin{syntaxleft}
\syntaxclass{Internal functions:}
\fn & ::= & \{ ~ {\tt sig} = \sig; & function signature \\
    &  &       ~~{\tt params} = \vec\id; & parameters \\
    &  &       ~~{\tt vars} = \vec\id; & local variables \\
    &  &       ~~{\tt stacksize} = n; & size of stack data in bytes \\
    &  &       ~~{\tt body} = s ~ \} & function body
\end{syntaxleft}

In addition to a parameter list, local variable declarations and a function
body (a statement), a function definition comprises a type
signature $\sig$ and a declaration of how many bytes of
stack-allocated data it needs.  A \lang{Cminor} local variable does
not reside in memory, and its address cannot be taken.  However, the
\lang{Cminor} producer can explicitly stack-allocate some data (such
as, in C, arrays and scalar variables whose addresses are taken).
A fresh memory block of size $\fn.{\tt stacksize}$ is allocated each time
$\fn$ is invoked and automatically freed when it returns.
The ${\tt addrstack}(\delta)$ nullary operator returns a pointer within
this block at byte offset $\delta$.

\begin{figure}

\begin{alltt}
double average(int arr[], int sz)   "average"(arr, sz) : int, int -> float
\{                                   \{
  double s; int i;                    vars s, i;  stacksize 0;
  for (i = 0, s = 0; i < sz; i++)     s = 0.0; i = 0;
    s += arr[i];                      block \{ loop \{
  return s / sz;                        if (i >= sz) exit(0);
\}                                       s = s +\(\sb{f}\) floatofint(int32[arr + i*4]);
                                        i = i + 1;
                                      \} \}
                                      return s /\(\sb{f}\) floatofint(sz);
                                    \} 
\end{alltt}

\caption{An example of a \lang{Cminor} function (right) and the
  corresponding C code (left).}
\label{f:cminor-expl}

\end{figure}

Figure~\ref{f:cminor-expl} shows a simple C function and the
corresponding \lang{Cminor} function, using an ad-hoc concrete syntax
for \lang{Cminor}.  Both functions compute the average value of an array of
integers, using float arithmetic.  Note the explicit address
computation {\tt int32[tbl\ +\ i*4]} to access element {\tt i} of the array, as
well as the explicit {\tt floatofint} conversions.  The {\tt for} loop is
expressed as an infinite loop wrapped in a {\tt block}, so that
${\tt exit}(0)$ in the \lang{Cminor} code behaves like the {\tt break}
statement in C.

\subsection{Dynamic semantics}

The dynamic semantics of \lang{Cminor} is defined using a
combination of natural semantics for the evaluation of expressions
and a labeled transition system in the style of
section~\ref{s:trans-sem} for the execution of statements and
functions.

\subsubsection{Evaluation of expressions}

\begin{figure}

\begin{pannel}
\irule{
E(\id) = \Some v
}{
G, \sp, E, M \vdash \id \eval v
}
\irule{
{\tt eval{\char95}constant}(G, \sp, \cst) = \Some v
}{
G, \sp, E, M \vdash \cst \eval v
}
\irule{
G, \sp, E, M \vdash a_1 \eval v_1 \amper {\tt eval{\char95}unop}(\unop, v_1) = \Some v
}{
G, \sp, E, M \vdash \unop(a_1) \eval v
}
\irule{
G, \sp, E, M \vdash a_1 \eval v_1 \amper
G, \sp, E, M \vdash a_2 \eval v_2 \amper
{\tt eval{\char95}binop}(\binop, v_1, v_2) = \Some v
}{
G, \sp, E, M \vdash \binop(a_1,a_2) \eval v
}
\irule{
G, \sp, E, M \vdash a_1 \eval {\tt ptr}(b, \ofs) \amper
{\tt load}(M, \chunk, b, \ofs) = \Some v
}{
G, \sp, E, M \vdash \chunk[a_1] \eval v
}
\irule{
G, \sp, E, M \vdash a_1 \eval v_1 \amper {\tt istrue}(v_1) \amper
G, \sp, E, M \vdash a_2 \eval v_2
}{
G, \sp, E, M \vdash (a_1 {\tt \ ?\ } a_2 {\tt \ :\ } a_3) \eval v_2
}
\irule{
G, \sp, E, M \vdash a_1 \eval v_1 \amper {\tt isfalse}(v_1) \amper
G, \sp, E, M \vdash a_3 \eval v_3
}{
G, \sp, E, M \vdash (a_1 {\tt \ ?\ } a_2 {\tt \ :\ } a_3) \eval v_3
}
\\
\srule{
G, \sp, E, M \vdash \epsilon: \epsilon
}
\irule{
G, \sp, E, M \vdash a \eval v \amper
G, \sp, E, M \vdash \vec a \eval \vec v
}{
G, \sp, E, M \vdash a.\vec a \eval v.\vec v
}
\end{pannel}
Evaluation of constants:
\begin{eqnarray*}
  {\tt eval{\char95}constant}(G, \sp, i) & = & \Some{{\tt int}(i)} \\
  {\tt eval{\char95}constant}(G, \sp, f) & = & \Some{{\tt float}(f)} \\
  {\tt eval{\char95}constant}(G, \sp, {\tt addrsymbol}(\id)) & = & {\tt symbol}(G, \id) \\
  {\tt eval{\char95}constant}(G, \sp, {\tt addrstack}(\ofs)) & = & \Some{{\tt ptr}(\sp, \ofs)}
\end{eqnarray*}
Evaluation of unary operators (selected cases):
\begin{eqnarray*}
  {\tt eval{\char95}unop}({\tt negf}, {\tt float}(f)) & = & \Some{{\tt float}(- f)} \\
  {\tt eval{\char95}unop}({\tt notbool}, v) & = & \Some{{\tt int}(0)} \mbox{ if ${\tt istrue}(v)$} \\
  {\tt eval{\char95}unop}({\tt notbool}, v) & = & \Some{{\tt int}(1)} \mbox{ if ${\tt isfalse}(v)$}
\end{eqnarray*}
Evaluation of binary operators (selected cases):
\begin{eqnarray*}
  {\tt eval{\char95}binop}({\tt add}, {\tt int}(n_1), {\tt int}(n_2)) & = & \Some{{\tt int}(n_1+n_2)}
                                                     \pmod{2^{32}}
\\
  {\tt eval{\char95}binop}({\tt add}, {\tt ptr}(b,\ofs), {\tt int}(n)) & = & \Some{{\tt ptr}(b, \ofs + n)}
                                                     \pmod{2^{32}}
\\
  {\tt eval{\char95}binop}({\tt addf}, {\tt float}(f_1), {\tt float}(f_2)) & = & \Some{{\tt float}(f_1+f_2)}
\end{eqnarray*}
Truth values:
\begin{eqnarray*}
{\tt istrue}(v) & \defequal & \mbox{$v$ is ${\tt ptr}(b,\ofs)$ or ${\tt int}(n)$ with $n \not= 0$}
\\
{\tt isfalse}(v) & \defequal & \mbox{$v$ is ${\tt int}(0)$}
\end{eqnarray*}
\caption{Natural semantics for Cminor expressions.}
\label{f:cminor-expr}
\end{figure}

Figure~\ref{f:cminor-expr} defines the big-step evaluation of
expressions as the judgment $G, \sp, E, M \vdash a \eval v$, where $a$ is
the expression to evaluate, $v$ its value, $\sp$ the
stack data block, $E$ an environment mapping local variables to
values, and $M$ the current memory state.  The evaluation rules are
straightforward.  Most of the semantics is in the definition of the
auxiliary functions {\tt eval{\char95}constant}, {\tt eval{\char95}unop} and {\tt eval{\char95}binop}, for
which some representative cases are shown.  These functions can return
$\None$, causing the expression to be undefined, if for instance an
argument of an operator is {\tt undef} or of the wrong type.  Some
operators ({\tt add}, {\tt sub} and {\tt cmp}) operate both on integers and on
pointers.  

\subsubsection{Execution of statements and functions}

The labeled transition system that defines the small-step semantics
for statements and function invocations follows the general pattern
shown in sections~\ref{s:trans-sem}~and~\ref{s:state-kinds}.  Program
states have the following shape:
\begin{syntaxleft}
\syntaxclass{Program states:}
\st & ::= & \state(\fn, s, k, \sp, E, M) & regular state \\
    & \alt& \callstate(\fd, \vec v, k, M) & call state \\
    & \alt& \returnstate(v, k, M) & return state
\syntaxclass{Continuations:}
k & ::= & {\tt stop}  & initial continuation \\
  & \alt& s; k    & continue with $s$, then do as $k$ \\
  & \alt& {\tt endblock}(k) & leave a {\tt block}, then do as $k$ \\
  & \alt& {\tt returnto}(\id^?, \fn, \sp, E, k) & return to caller
\end{syntaxleft}
Regular states $\state$ carry the currently-executing function $\fn$,
the statement under consideration $s$, the block identifier for the
current stack data $\sp$, and the values $E$ of local variables.

Following a proposal by Appel and Blazy \cite{Appel-Blazy-07}, we use continuation
terms $k$ to encode both the call stack and the program point within
$\fn$ where the statement $s$ under consideration resides.  A
continuation $k$ records what needs to be done once $s$ reduces to
{\tt skip}, {\tt exit} or {\tt return}.  The {\tt returnto} parts of $k$ represent the
call stack: they record the local states of the calling functions.  The
top part of $k$ up to the first {\tt returnto} corresponds to an execution
context for $s$, represented inside-out in the style of a zipper
\cite{Huet-zipper}.  For example, the continuation $s; {\tt endblock}(\ldots)$
corresponds to the context ${\tt block}\,  \{ \, [\,\,]; s \, \}$.

\begin{figure}

\begin{pannel}
\srule{
\trans{\state(\fn,{\tt skip},(s;k),\sp,E,M)}
      {\state(\fn,s,k,\sp,E,M)}
}
\srule{
\trans{\state(\fn,{\tt skip},{\tt endblock}(k),\sp,E,M)}
      {\state(\fn,{\tt skip},k,\sp,E,M)}
}
\irule{
G, \sp, E, M \vdash a \Rightarrow v
}{
\trans {\state(\fn, (\id = a), k, \sp, E, M)}
       {\state(\fn, {\tt skip}, k, \sp, E\{\id \becomes v\}, M)}
}
\irule{
G, \sp, E, M \vdash a_1 \Rightarrow {\tt ptr}(b,\delta) \amper
G, \sp, E, M \vdash a_2 \Rightarrow v \amper
{\tt store}(M, \chunk, b, \delta, v) = \Some{M'}
}{
\trans {\state(\fn, (\chunk[a_1] = [a_2]), k, \sp, E, M)}
       {\state(\fn, {\tt skip}, k, \sp, E, M')}
}
\srule{
\trans {\state(\fn, (s_1; s_2), k, \sp, E, M)}
       {\state(\fn, s_1, (s_2; k), \sp, E, M)}
}
\irule{
G, \sp, E, M \vdash a \Rightarrow v \amper {\tt istrue}(v)
}{
\trans {\state(\fn, ({\tt if}(a)\{s_1\}{\tt \ else\ }\{s_2\}), k, \sp, E, M)}
       {\state(\fn, s_1, k, \sp, E, M)}
}
\irule{
G, \sp, E, M \vdash a \Rightarrow v \amper {\tt isfalse}(v)
}{
\trans {\state(\fn, ({\tt if}(a)\{s_1\}{\tt \ else\ }\{s_2\}), k, \sp, E, M)}
       {\state(\fn, s_2, k, \sp, E, M)}
}
\srule{
\trans {\state(\fn, {\tt loop}\{s\}, k, \sp, E, M)}
       {\state(\fn, s, ({\tt loop}\{s\}; k), \sp, E, M)}
}
\srule{
\trans {\state(\fn, {\tt block}\{s\}, k, \sp, E, M)}
       {\state(\fn, s, {\tt endblock}(k), \sp, E, M)}
}
\srule{
\trans {\state(\fn, {\tt exit}(n), (s; k), \sp, E, M)}
       {\state(\fn, {\tt exit}(n), k, \sp, E, M)}
}
\srule{
\trans {\state(\fn, {\tt exit}(0), {\tt endblock}(k), \sp, E, M)}
       {\state(\fn, {\tt skip}, k, \sp, E, M)}
}
\srule{
\trans {\state(\fn, {\tt exit}(n+1), {\tt endblock}(k), \sp, E, M)}
       {\state(\fn, {\tt exit}(n), k, \sp, E, M)}
}
\irule{
G, \sp, E, M \vdash a \Rightarrow {\tt int}(n)
}{
\trans {\state(\fn, ({\tt switch}(a)\{ {\it tbl} \}), k, \sp, E, M)}
       {\state(\fn, {\tt exit}({\it tbl}(n)), k, \sp, E, M)}
}
\srule{
\trans {\state(\fn, (\lbl: s), k, \sp, E, M)}
       {\state(\fn, s, k, \sp, E, M)}
}
\irule{
{\tt findlabel}(\lbl, \fn.{\tt body}, {\tt callcont}(k)) = \Some{s', k'}
}{
\trans {\state(\fn, {\tt goto\ }\lbl, k, \sp, E, M)}
       {\state(\fn, s', k', \sp, E, M)}
}
\end{pannel}
$$ \displaylines{{\tt callcont}(s;k) = {\tt callcont}(k) \qquad
   {\tt callcont}({\tt endblock}(k)) = {\tt callcont}(k) \cr
   {\tt callcont}(k) = k \mbox{ otherwise}
} $$
\begin{eqnarray*}
{\tt findlabel}(\lbl, (s_1; s_2), k) & = & 
  \twocases{{\tt findlabel}(\lbl, s_1, (s_2;k))}{if not $\None$;}
           {{\tt findlabel}(\lbl, s_2, k)}{otherwise}
\\
{\tt findlabel}(\lbl, {\tt if}(a)\{s_1\}{\tt \ else\ }\{s_2\}, k) & = &
  \twocases{{\tt findlabel}(\lbl, s_1, k)}{if not $\None$;}
           {{\tt findlabel}(\lbl, s_2, k)}{otherwise}
\\
{\tt findlabel}(\lbl, {\tt loop}\{s\}, k) & = &
  {\tt findlabel}(\lbl, s, ({\tt loop}\{s\}; k))
\\
{\tt findlabel}(\lbl, {\tt block}\{s\}, k) & = &
  {\tt findlabel}(\lbl, s, {\tt endblock}(k))
\\
{\tt findlabel}(\lbl, (\lbl: s), k) & = & \Some{s, k}
\\
{\tt findlabel}(\lbl, (\lbl': s), k) & = &
  {\tt findlabel}(\lbl, s, k) \mbox{ if $\lbl' \not= \lbl$}
\end{eqnarray*}

\caption{Transition semantics for Cminor, part 1: statements.}
\label{f:cminor-sem-1}
\end{figure}
\begin{figure}
\begin{pannel}
\irule{
G, \sp, E, M \vdash a_1 \Rightarrow {\tt ptr}(b,0) \amper
G, \sp, E, M \vdash \vec a \Rightarrow \vec v \amper
{\tt funct}(G, b) = \Some{\fd} \amper \fd.{\tt sig} = \sig
}{
\trans {\state(\fn, (id^? = a_1(\vec a) : \sig), k, \sp, E, M)}
       {\callstate(\fd, \vec v, {\tt returnto}(\id^?, \fn, \sp, E, k), M)}
}
\irule{
G, \sp, E, M \vdash a_1 \Rightarrow {\tt ptr}(b,0) \amper
G, \sp, E, M \vdash \vec a \Rightarrow \vec v \amper
{\tt funct}(G, b) = \Some{\fd} \amper \fd.{\tt sig} = \sig
}{
\trans {\state(\fn, ({\tt tailcall\ }a_1(\vec a) : \sig), k, \sp, E, M)}
       {\callstate(\fd, \vec v, {\tt callcont}(k), M)}
}
\irule{
\fn.{\tt sig}.{\tt res} = {\tt void} \amper k = {\tt returnto}(\ldots) \mbox{ or } k = {\tt stop}
}{
\trans {\state(\fn,{\tt skip},k,\sp,E,M)}
       {\returnstate({\tt undef}, k, {\tt free}(M,\sp))}
}
\irule{
\fn.{\tt sig}.{\tt res} = {\tt void}
}{
\trans {\state(\fn, {\tt return}, k, \sp, E, M)}
       {\returnstate({\tt undef}, {\tt callcont}(k), {\tt free}(M, \sp))}
}
\irule{
\fn.{\tt sig}.{\tt res} \not= {\tt void} \amper G, \sp, E, M \vdash a \Rightarrow v
}{
\trans {\state(\fn, {\tt return}(a), k, \sp, E, M)}
       {\returnstate(v, {\tt callcont}(k), {\tt free}(M, \sp))}
}
\irule{
{\tt alloc}(M, 0, \fn.{\tt stackspace}) = (\sp, M') \amper
E = [\fn.{\tt params} \becomes \vec v; \fn.{\tt vars} \becomes {\tt undef}]
}{
\trans {\callstate({\tt internal}(\fn), \vec v, k, M)}
       {\state(\fn, \fn.{\tt body}, k, \sp, E, M')}
}
\irule{
\vdash \ef(\vec v) \evaltrace{t} v \mbox{ (see section~\ref{s:traces})}
}{
\transtrace{\callstate({\tt external}(\ef), \vec v, k, M)}
       {t} {\returnstate(v, k, M)}
}
\srule{
\trans {\returnstate(v, {\tt returnto}(\id^?, \fn, \sp, E, k), M)}
       {\state(\fn, {\tt skip}, k, \sp, E\{\id^? \becomes v\}, M)}
}

\irule{
{\tt symbol}({\tt globalenv}(\prog), \prog.{\tt main}) = \Some b \amper
{\tt funct}({\tt globalenv}(\prog), b) = \Some \fd
}{
\initialstate{\prog}{\callstate(\fd, \epsilon, {\tt stop}, {\tt initmem}(\prog))}
}
\srule{
\finalstate{\returnstate({\tt int}(n), {\tt stop}, M)}{n}
}
\end{pannel}
\caption{Transition semantics for Cminor, part 2: functions, initial
  states, final states.}
\label{f:cminor-sem-2}
\end{figure}

Figures~\ref{f:cminor-sem-1}~and~\ref{f:cminor-sem-2} list the rules
defining the transition relation $\transtrace{\st}{t}{\st'}$.
The rules in figure~\ref{f:cminor-sem-1} address transitions within
the currently-executing function.  They are roughly of three kinds: 
\begin{itemize}
\item Execution of an atomic computation step.  For example, the rule for
  assignments transitions from $\id = a$ to {\tt skip}.
\item Focusing on the active part of the current statement.  For example,
  the rule for sequences transitions from $(s_1; s_2)$ with
  continuation $k$ to $s_1$ with continuation $s_2; k$.
\item Resuming a continuation that was set apart in the continuation.
  For instance, one of the rules for {\tt skip} transitions
  from {\tt skip} with continuation $s;k$ to $s$ with continuation $k$.
\end{itemize}
Two auxiliary functions over continuations are defined: ${\tt callcont}(k)$
discards the local context part of the continuation $k$, and
${\tt findlabel}(\lbl, s, k)$ returns a pair $(s', k')$ of the
leftmost sub-statement of $s$ labeled $\lbl$ and of a continuation $k'$ that
extends $k$ with the context surrounding $s'$. The combination of these
two functions in the rule for {\tt goto} suffices to implement the
branching behavior of {\tt goto} statements.

Figure~\ref{f:cminor-sem-2} lists the transitions involving call
states and return states, and defines initial states and final
states.  The definitions follow the general pattern depicted in
figure~\ref{f:state-auto}.  In particular, initial states are call
states to the ``main'' function of the program, with no arguments and
the {\tt stop} continuation; symmetrically, final states are return states
with the {\tt stop} continuation.  The rules for function calls require
that the signature of the called function matches exactly the
signature annotating the call; otherwise execution gets stuck.  A
similar requirement exists in the C standard and is essential to
support signature-dependent calling conventions later in the compiler.
Taking again a leaf from the C standard, functions are allowed to
terminate by {\tt return} without an argument or by falling through their
bodies only if their return signatures are {\tt void}.

\subsubsection{Alternate natural semantics for statements and functions}

For some applications,
it is convenient to have an alternate natural (big-step) operational
semantics for \lang{Cminor}.  We have developed such a semantics for
the fragment of \lang{Cminor} that excludes {\tt goto} and labeled
statements.  The big-step judgments for terminating
executions have the following form:
$$\begin{array}{rcll}
G, \sp & \vdash & s, E, M \evaltrace{t} out, E', M' & \mbox{(statements)} \\
G & \vdash & \fd(\vec v), M \evaltrace{t} v, M' & \mbox{(function calls)}
\end{array}$$
$E'$ and $M'$ are the local environment and the memory state at the
end of the execution; $t$ is the trace of events generated during
execution.  Following Huisman and Jacobs \cite{Huisman-Jacobs-00}, the outcome $out$ indicates how the statement $s$
terminated: either normally by running to completion ($out =
{\tt Normal}$); or prematurely by executing an {\tt exit} statement ($out =
{\tt Exit}(n)$), {\tt return} statement ($out = {\tt Return}(v^?)$ where $v^?$
is the value of the optional argument to {\tt return}), or {\tt tailcall}
statement ($out = {\tt Tailreturn}(v)$).
Additionally, we followed the coinductive approach to natural
semantics of Leroy and Grall \cite{Leroy-Grall-coindsem} to define (coinductively)
big-step judgments for diverging executions, of the form
$$\begin{array}{rcll}
G, \sp & \vdash & s, E, M \evalinf{T} & \mbox{(diverging statements)} \\
G & \vdash & \fd(\vec v), M \evalinf{T} & \mbox{(diverging function calls)}
\end{array}$$
The definitions of these judgments can be found in the Coq
development.

\begin{theorem} The natural semantics of \lang{Cminor} is correct with
  respect to its transition semantics:
\begin{enumerate}
\item If $G \vdash \fd(\vec v), M \evaltrace{t} v, M'$, then 
$G \vdash \callstate(\fd, \vec v, k, M) \redm{t} \returnstate(v, k, M')$
for all continuations $k$ such that $k = {\tt callcont}(k)$.
\item If $G \vdash \fd(\vec v), M \evalinf{T}$, then 
$G \vdash \callstate(\fd, \vec v, k, M) \redinf{T}$
for all continuations $k$.  
\end{enumerate}
\end{theorem}

\subsection{Static typing} \label{s:cminor-typing}

\lang{Cminor} is equipped with a trivial type
system having only two types: {\tt int} and {\tt float}.  (Pointers
have static type {\tt int}.)  Function definitions and
function calls are annotated with signatures $\sig$ giving the number
and types of arguments, along with optional result types.  All
operators are monomorphic; therefore, the types of local variables can
be inferred from their uses and are not declared.  

The primary purpose of this trivial type system is to facilitate later
transformations (see sections~\ref{s:regalloc} and~\ref{s:stacking});
for this purpose, all intermediate languages of Compcert are equipped
with similar {\tt int}-or-{\tt float} type systems.  By themselves, these type
systems are too weak to give type soundness properties (absence of
run-time type errors).  For example, performing an integer addition of
two pointers or two {\tt undef} values is statically well-typed but causes
the program to get stuck.  Likewise, calling a function whose
signature differs from that given at the {\tt call} site is a run-time
error, undetected by the type system; its semantics are not defined
and the compiler can (and does) generate incorrect code for this call.
It is the responsibility of the \lang{Cminor} producer to avoid these
situations, \eg\ by using a richer type system.
Nevertheless, the \lang{Cminor} type system enjoys a type preservation
property: values of static type {\tt int} are always integers, pointers or
{\tt undef}, and values of static type {\tt float} are always floating-point
numbers or {\tt undef}.  This weak soundness property plays a role
in the correctness proofs of section~\ref{s:stacking-proof}.

\section{Instruction selection} \label{s:selection}

The first compilation pass of Compcert rewrites expressions to exploit
the combined arithmetic operations and addressing modes of the target
processor.  To take better advantage of the processor's capabilities,
reassociation of integer additions and multiplications is also
performed, as well as a small amount of constant propagation.

\subsection{The target language: CminorSel} \label{s:cminorsel}

The target language for this pass is \lang{CminorSel}, a variant of
\lang{Cminor} that uses a different, processor-specific set of
operators.  Additionally, a syntactic class of condition expressions $ce$
(expressions used only for their truth values) is introduced.
\begin{syntax}
\syntaxclass{Expressions:}
a & ::=  & \id  & reading a local variable \\
  & \alt & \op(\vec a) & operator application \\
  & \alt & {\tt load}(\chunk, \mode, \vec a) & memory read \\
  & \alt & ce {\tt \ ?\ } a_1 {\tt \ :\ } a_2 & conditional expression
\syntaxclass{Condition expressions:}
ce & ::=  & {\tt true} \alt {\tt false} \\
  & \alt & \cond(\vec a) & elementary test \\
  & \alt & ce_1 {\tt \ ?\ } ce_2 {\tt \ :\ } ce_3 & conditional condition
\syntaxclass{Operators (machine-specific):}
\op & ::= & n \alt f \alt {\tt move} \alt \ldots & most of \lang{Cminor} operators \\
    & \alt& {\tt addi}_n \alt {\tt rolm}_{n,m} \alt \ldots & PPC combined operators
\syntaxclass{Addressing modes (machine-specific):}
\mode & ::= & {\tt indexed}(n) & indexed, immediate displacement \\
      & \alt& {\tt indexed2}   & indexed, register displacement \\
      & \alt& {\tt global}(\id,\delta) & address is $\id + \delta$ \\
      & \alt& {\tt based}(\id,\delta) & indexed, displacement is $\id + \delta$ \\
      & \alt& {\tt stack}(\delta) & address is $\mbox{stack pointer} + \delta$
\syntaxclass{Conditions (machine-specific):}
\cond & ::= & {\tt comp}(c) \alt {\tt compimm}(c,n) & signed integer / pointer comparison \\
      & \alt& {\tt compu}(c) \alt {\tt compuimm}(c,n) & unsigned integer
comparison \\
      & \alt& {\tt compf}(c) & float comparison
\syntaxclass{Statements:}
s & ::= & {\tt store}(\chunk, \mode, \vec a, a) & memory write \\
  & \alt & {\tt if}(ce) \, \{ s_1 \} {\tt \ else\ } \{ s_2 \} & conditional statement\\
  & \alt & \ldots & as in \lang{Cminor}
\end{syntax}

For the PowerPC, the machine-specific operators $\op$ include all
\lang{Cminor} nullary, unary and binary operators except {\tt notint},
{\tt mod} and {\tt modu} (these need to be synthesized from other operators)
and adds immediate forms of many integer operators, as well as a
number of combined operators such as not-or, not-and, and
rotate-and-mask. (${\tt rolm}_{n,m}$ is a left rotation by $n$ bits
followed by a logical ``and'' with $m$.)  A memory load or store now
carries an addressing mode $\mode$ and a list of expressions $\vec a$,
from which the address being addressed is computed.  Finally,
conditional expressions and conditional statements now take
condition expressions~$ce$ as arguments instead of normal expressions~$a$.

The dynamic semantics of \lang{CminorSel} resembles that of
\lang{Cminor}, with the addition of a new evaluation judgment for
condition expressions $G, \sp, E, M \vdash ce \Rightarrow ({\tt false} \mid
{\tt true})$.  Figure~\ref{f:cminorsel-sem} shows the main differences
with respect to the \lang{Cminor} semantics.

\begin{figure}

Evaluation of expressions:
\begin{pannel}
\irule{
G, \sp, E, M \vdash \vec a \eval \vec v \amper {\tt eval{\char95}op}(G, \sp, \op, \vec v) = \Some v
}{
G, \sp, E, M \vdash \op(\vec a) \eval v
}
\irule{
G, \sp, E, M \vdash \vec a \eval \vec v \amper
{\tt eval{\char95}mode}(G, \sp, \mode, \vec v) = \Some{{\tt ptr}(b, \ofs)} \amper
{\tt load}(M, \chunk, b, \ofs) = \Some v
}{
G, \sp, E, M \vdash {\tt load}(\chunk, \mode, \vec a) \eval v
}
\irule{
G, \sp, E, M \vdash c \eval {\tt true} \amper G, \sp, E, M \vdash a_1 \eval v_1
}{
G, \sp, E, M \vdash (c {\tt \ ?\ } a_1 {\tt \ :\ } a_2) \eval v_1
}
\hspace*{-4mm}
\irule{
G, \sp, E, M \vdash c \eval {\tt false} \amper G, \sp, E, M \vdash a_2 \eval v_2
}{
G, \sp, E, M \vdash (c {\tt \ ?\ } a_1 {\tt \ :\ } a_2) \eval v_2
}
\end{pannel}
Evaluation of condition expressions:
\begin{pannel}
\srule{
G, \sp, E, M \vdash {\tt true} \eval {\tt true}
}
\srule{
G, \sp, E, M \vdash {\tt false} \eval {\tt false}
}
\irule{
G, \sp, E, M \vdash \vec a \eval \vec v \amper
{\tt eval{\char95}cond}(\cond, \vec v) = \Some{b}
}{
G, \sp, E, M \vdash \cond(\vec a) \eval b
}
\irule{
G, \sp, E, M \vdash ce_1 \eval {\tt true} \amper G, \sp, E, M \vdash ce_2 \eval b
}{
G, \sp, E, M \vdash (ce_1 {\tt \ ?\ } ce_2 {\tt \ :\ } ce_3) \eval b
}
\hspace*{-4mm}
\irule{
G, \sp, E, M \vdash ce_1 \eval {\tt false} \amper G, \sp, E, M \vdash ce_3 \eval b
}{
G, \sp, E, M \vdash (ce_1 {\tt \ ?\ } ce_2 {\tt \ :\ } ce_3) \eval b
}
\end{pannel}
Execution of statements:
\begin{pannel}
\irule{
G, \sp, E, M \vdash \vec a \Rightarrow \vec v \amper
{\tt eval{\char95}mode}(G, \sp, \mode, \vec v) = \Some{{\tt ptr}(b,\delta)} \\
G, \sp, E, M \vdash a \Rightarrow v \amper
{\tt store}(M, \chunk, b, \delta, v) = \Some{M'}
}{
\trans {\state(\fn, {\tt store}(\chunk, \mode, \vec a, a), k, \sp, E, M)}
       {\state(\fn, {\tt skip}, k, \sp, E, M')}
}
\irule{
G, \sp, E, M \vdash ce \Rightarrow {\tt true}
}{
\trans {\state(\fn, ({\tt if}(ce)\{s_1\}{\tt \ else\ }\{s_2\}), k, \sp, E, M)}
       {\state(\fn, s_1, k, \sp, E, M)}
}
\irule{
G, \sp, E, M \vdash ce \Rightarrow {\tt false}
}{
\trans {\state(\fn, ({\tt if}(ce)\{s_1\}{\tt \ else\ }\{s_2\}), k, \sp, E, M)}
       {\state(\fn, s_2, k, \sp, E, M)}
}
\end{pannel}
\caption{Semantics of \lang{CminorSel}.  Only the rules that differ
  from those of \lang{Cminor} are shown.}
\label{f:cminorsel-sem}
\end{figure}

\subsection{The code transformation}

Instruction selection is performed by a bottom-up rewriting of
expressions.  For each \lang{Cminor} operator $\op$, we define a ``smart
constructor'' function written $\overline\op$ that takes
\lang{CminorSel} expressions as arguments, performs shallow
pattern-matching over them to recognize combined or immediate
operations, and returns the corresponding \lang{CminorSel} expression.
For example, here are the smart constructor $\overline{{\tt add}}$ for
integer addition and its helper $\overline{{\tt addi}}_n$ for immediate
integer addition:
\begin{eqnarray*}
\overline{{\tt add}}({\tt addi}_{n_1}(a_1), {\tt addi}_{n_2}(a_2)) & = & 
   \overline{{\tt addi}}_{n_1+n_2}({\tt add}(a_1, a_2)) \\
\overline{{\tt add}}({\tt addi}_{n}(a_1), a_2) & = & 
   \overline{{\tt addi}}_{n}({\tt add}(a_1, a_2)) \\
\overline{{\tt add}}(a_1, {\tt addi}_{n}(a_2)) & = & 
   \overline{{\tt addi}}_{n}({\tt add}(a_1, a_2)) \\
\overline{{\tt add}}(n_1, n_2) & = & n_1+n_2 \\
\overline{{\tt add}}(a_1', a_2') & = & {\tt add}(a_1', a_2') \mbox{ otherwise} \\
\overline{{\tt addi}}_{n_1}(n_2) & = & n_1 + n_2 \\
\overline{{\tt addi}}_{n_1}({\tt addi}_{n_2}(a)) & = & 
   {\tt addi}_{n_1+n_2}(a) \\
\overline{{\tt addi}}_n({\tt addrsymbol}(\id+\delta)) & = &
   {\tt addrsymbol}(\id+(\delta+n)) \\
\overline{{\tt addi}}_n({\tt addrstack}(\delta)) & = &
   {\tt addrstack}(\delta+n) \\
\overline{{\tt addi}}_n(a) & = & {\tt addi}_n(a) \mbox{ otherwise}
\end{eqnarray*}
Here are some cases from other smart constructors that illustrate
reassociation of immediate multiplication and immediate addition, as
well as the recognition of the rotate-and-mask instruction:
\begin{eqnarray*}
\overline{{\tt muli}}_m({\tt addi}_n(a)) & = &
  \overline{{\tt addi}}_{m \times n}({\tt muli}_m(a)) \\
\overline{{\tt shl}}(a, n) & = & \overline{{\tt rolm}}_{n, (-1) \ll n}(a) \\
\overline{{\tt shru}}(a, n) & = & \overline{{\tt rolm}}_{32-n, (-1) \gg n}(a) \\
\overline{{\tt and}}(a, n) & = & \overline{{\tt rolm}}_{0,n}(a) \\
\overline{{\tt rolm}}_{n_1,m_1}({\tt rolm}_{n_2,m_2}(a)) & = &
   {\tt rolm}_{n_1+n_2,m}(a)
 \mbox{with $m = {\tt rol}(m_1, n_2) \wedge m_2$} \\
\overline{{\tt or}}({\tt rolm}_{n,m_1}(a), {\tt rolm}_{n,m_2}(a)) & = & 
   {\tt rolm}_{n, m_1 \vee m_2}(a)
\end{eqnarray*}
While innocuous-looking, these smart constructors are powerful enough
to, for instance, reduce $8 + ({\tt x} + 1) \times 4$ to ${\tt x} \times 4 + 12$,
and to recognize a ${\tt rolm}_{3,-1}({\tt x})$ instruction for
{\tt (x\ <<\ 3)\ {\char124}\ (x\ >>\ 29)}, a C encoding of bit rotation commonly used in
cryptography.

The recognition of conditions and addressing modes is performed by two
functions $\overline{{\tt cond}}(a) = c'$ and
$\overline{{\tt mode}}(a) = (\mode, \vec{a})$.
The translation of expressions is, then, a straightforward bottom-up
traversal, applying the appropriate smart constructors at each step:
\begin{eqnarray*}
\tr{\cst} & = & \overline{\cst} \\
\tr{\unop(a)} & = & \overline{\unop}(\tr{a}) \\
\tr{\binop(a_1,a_2)} & = & \overline{\binop}(\tr{a_1}, \tr{a_2}) \\
\tr{\chunk[a]} & = & {\tt load}(\chunk, \mode, \vec{a}) 
  \mbox{ where $(\mode, \vec{a}) = \overline{{\tt mode}}(\tr{a})$} \\
\tr{a_1{\tt \ ?\ } a_2{\tt \ :\ }a_3} & = & 
  \overline{{\tt cond}}(\tr{a_1}){\tt \ ?\ } \tr{a_2} {\tt \ :\ } \tr{a_3}
\end{eqnarray*}
We omit the translation of statements and functions, which is similar.

\subsection{Semantic preservation}

The first part of the proof that instruction selection preserves
semantics is to show the correctness of the smart constructor
functions.

\begin{lemma} \label{l:sel-smart-constr} ~
\begin{enumerate}
\item If $G, \sp, E, M \vdash a_1 \eval v_1$ 
  and ${\tt eval{\char95}unop}(\unop, v_1) = \Some v$,
  then $G, \sp, E, M \vdash \overline{\unop}(a_1) \eval v$.
\item If $G, \sp, E, M \vdash a_1 \eval v_1$
  and $G, \sp, E, M \vdash a_2 \eval v_2$
  and ${\tt eval{\char95}binop}(\binop, v_1, v_2) = \Some v$,
  then $G, \sp, E, M \vdash \overline{\binop}(a_1,a_2) \eval v$.
\item If $G, \sp, E, M \vdash a \eval v$ and ${\tt istrue}(v)$,
  then $G, \sp, E, M \vdash \overline{{\tt cond}}(a) \eval {\tt true}$.
\item If $G, \sp, E, M \vdash a \eval v$ and ${\tt isfalse}(v)$,
  then $G, \sp, E, M \vdash \overline{{\tt cond}}(a) \eval {\tt false}$.
\item If $G, \sp, E, M \vdash a \eval v$ and
  $\overline{{\tt mode}}(a) = (\mode, \vec{a})$.
  then there exists $\vec v$ such that
  $G, \sp, E, M \vdash \vec{a} \eval \vec v$
  and ${\tt eval{\char95}mode}(\mode, \vec v) = \Some v$.
\end{enumerate}
\end{lemma}

After copious case analysis on the operators and their arguments and
inversion on the evaluations of the arguments, the proof reduces to
showing that the defining equations for the smart constructors are
valid when interpreted over ground machine integers.  For instance, in
the case for $\overline{{\tt rolm}}$ shown above, we have to prove that
$$ {\tt rol}(({\tt rol}(x, n_1) \wedge m_1), n_2) \wedge m_2 ~=~
   {\tt rol}(x, n_1 + n_2) \wedge ({\tt rol}(m_1, n_2) \wedge m_2) $$
which follows from the algebraic properties of rotate-and-left and
   bitwise ``and''.  Completing the
proof of the lemma above required the development of a rather large
and difficult formalization of $N$-bit machine integers and of the
algebraic properties of their arithmetic and logical operations.

Let $\prog$ be the original \lang{Cminor} program and $\prog'$ be the
\lang{CminorSel} program produced by instruction selection.  Let $G,
G'$ be the corresponding global environments.  Semantic preservation
for the evaluation of expressions follows from
lemma~\ref{l:sel-smart-constr} by induction on the \lang{Cminor}
evaluation derivation.

\begin{lemma} \label{l:sel-expr}
If $G, \sp, E, M \vdash a \eval v$, then $G', \sp, E, M \vdash \tr{a} \eval v$.
\end{lemma}

The last part of the semantic preservation proof is a simulation
argument of the form outlined in section~\ref{s:simulation-diagrams}.
Since the structure of statements is preserved by the translation,
transitions in the original and transformed programs match one-to-one,
resulting in a ``lock-step'' simulation diagram.  The relation
$\match$ between \lang{Cminor} and \lang{CminorSel} execution states
is defined as follows:
\begin{eqnarray*}
\state(\fn, s, k, \sp, E, M) & \match &
\state(\tr{\fn}, \tr{s}, \tr{k}, \sp, E, M) \\
\callstate(\fd, \vec v, k, M) & \match &
\callstate(\tr{\fd}, \vec v, \tr{k}, M) \\
\returnstate(v, k, M) & \match &
\returnstate(v, \tr{k}, M)
\end{eqnarray*}
Since the transformed code computes exactly the same values as the
original, environments $E$ and memory states $M$ are identical in
matching states.  Statements and functions appearing in states must be
the translation of one another.  For continuations, we extend
(isomorphically) the translation of statements and functions.

\begin{lemma} \label{l:sel-simu}
If $G \vdash \st_1 \redtrace{t} \st_2$ and $\st_1 \match \st_1'$,
there exists $\st_2'$ such that $G' \vdash \st_1' \redtrace{t} \st_2'$
and $\st_1' \match \st_2'$.
\end{lemma}

The proof is a straightforward case analysis on the transition from
$\st_1$ to $\st_2$.  Semantic preservation for instruction selection
then follows from theorem~\ref{th:lockstep-simulation} and
lemma~\ref{l:sel-simu}.

\section{RTL generation} \label{s:rtlgen}

The second compilation pass translates \lang{CminorSel} to a simple
intermediate language of the RTL kind, with control represented as a
control-flow graph instead of structured statements.  This
intermediate language is convenient for performing optimizations later.

\subsection{The target language: RTL} \label{s:rtl}

The \lang{RTL} language represents functions as a control-flow graph
(CFG) of abstract instructions, corresponding roughly to machine
instructions but operating over pseudo-registers (also called
``temporaries'').  Every function has an unlimited supply of
pseudo-registers, and their values are preserved across function call.
In the following, $r$~ranges over pseudo-registers and~$l$~over labels
of CFG nodes.
\begin{syntaxleft}
\syntaxclass{RTL instructions:}
i & ::=  & {\tt nop}(l) & no operation (go to $l$) \\
  & \alt & {\tt op}(\op, \vec r, r, l) & arithmetic operation \\
  & \alt & {\tt load}(\chunk, \mode, \vec r, r, l) & memory load \\
  & \alt & {\tt store}(\chunk, \mode, \vec r, r, l) & memory store \\
  & \alt & {\tt call}(\sig, (r \alt \id), \vec r, r, l) & function call \\
  & \alt & {\tt tailcall}(\sig, (r \alt \id), \vec r) & function tail call \\
  & \alt & {\tt cond}(\cond, \vec r, \ltrue, \lfalse) &
                  conditional branch \\
  & \alt & {\tt return} \alt {\tt return}(r) & function return
\syntaxclass{RTL control-flow graph:}
g & ::=  & l \mapsto i               & finite map
\syntaxclass{RTL functions:}
\fn & ::= & \{~{\tt sig} = \sig; \\
    &     & ~~~{\tt params} = \vec r;    & parameters\\
    &     & ~~~{\tt stacksize} = n;      & size of stack data block \\
    &     & ~~~{\tt entrypoint} = l;     & label of first instruction\\
    &     & ~~~{\tt code} = g \}         & control-flow graph
\end{syntaxleft}%
Each instruction takes its arguments in a list of pseudo-registers~$\vec r$
and stores its result, if any, in a pseudo-register~$r$.
Additionally, it carries the labels of its possible successors.  We use
instructions rather than basic blocks as nodes of the control-flow
graph because this simplifies semantics and reasoning over static
analyses without significantly slowing compilation \cite{Knoop-dinosaurs}.

The dynamic semantics of \lang{RTL} is defined by the labeled
transition system shown in figure~\ref{f:rtl-sem}.  Program states
have the following form:
\begin{syntaxleft}
\syntaxclass{Program states:}
\st & ::= & \state(\stk, g, \sp, l, R, M) & regular state \\
    & \alt& \callstate(\stk, \fd, \vec v, M) & call state \\
    & \alt& \returnstate(\stk, v, M)         & return state
\syntaxclass{Call stacks:}
\stk & ::= & (\stackframe(r, \fn, \sp, l, R))^* & list of frames
\syntaxclass{Register states:}
R & ::= & r \mapsto v
\end{syntaxleft}%
In regular states, $g$ is the CFG of the function currently executing,
$l$ a program point (CFG node label) within this function, $\sp$ 
its stack data block, and $R$ an assignment of values for the
pseudo-registers of $\fn$.  All three states carry a call stack~$\stk$,
which is a list of frames $\stackframe$ representing pending function
calls and containing the corresponding per-function state
$\fn, \sp, l, R$.

\begin{figure}

\begin{pannel}
\irule{
g(l) = \Some{{\tt nop}(l')}
}{
\trans {\state(\stk, g, \sp, l, R, M)}
       {\state(\stk, g, \sp, l', R, M)}
}
\irule{
g(l) = \Some{{\tt op}(\op, \vec r, r, l')} \amper
{\tt eval{\char95}op}(G, \sp, \op, R(\vec r)) = \Some v
}{
\trans {\state(\stk, g, \sp, l, R, M)}
       {\state(\stk, g, \sp, l', R\{r \becomes v\}, M)}
}
\irule{
g(l) = \Some{{\tt load}(\chunk, \mode, \vec r, r, l')} \\
{\tt eval{\char95}mode}(G, \sp, \mode, R(\vec r)) = \Some{{\tt ptr}(b,\ofs)} \amper
{\tt load}(M, \chunk, b, \ofs) = \Some v
}{
\trans {\state(\stk, g, \sp, l, R, M)}
       {\state(\stk, g, \sp, l', R\{r \becomes v\}, M)}
}
\irule{
g(l) = \Some{{\tt store}(\chunk, \mode, \vec r, r, l')} \\
{\tt eval{\char95}mode}(G, \sp, \mode, R(\vec r)) = \Some{{\tt ptr}(b,\ofs)} \amper
{\tt store}(M, \chunk, b, \ofs, R(r)) = \Some{M'}
}{
\trans {\state(\stk, g, \sp, l, R, M)}
       {\state(\stk, g, \sp, l', R, M')}
}
\irule{
g(l) = \Some{{\tt call}(\sig, r_f, \vec r, r, l')} \amper
R(r_f) = {\tt ptr}(b, 0) \amper {\tt funct}(G, b) = \Some \fd \amper \fd.{\tt sig} = \sig
}{
\trans {\state(\stk, g, \sp, l, R, M)}
       {\callstate(\stackframe(r, g, \sp, l', R).\stk, \fd, R(\vec r), M)}
}
\irule{
g(l) = \Some{{\tt tailcall}(\sig, r_f, \vec r)} \amper
R(r_f) = {\tt ptr}(b, 0) \amper {\tt funct}(G, b) = \Some \fd \amper \fd.{\tt sig} = \sig
}{
\trans {\state(\stk, g, \sp, l, R, M)}
       {\callstate(\stk, \fd, R(\vec r), {\tt free}(M, \sp))}
}
\irule{
g(l) = \Some{{\tt cond}(\cond, \vec r, \ltrue, \lfalse)} \amper
{\tt eval{\char95}cond}(\cond, R(\vec r)) = \Some{{\tt true}}
}{
\trans {\state(\stk, g, \sp, l, R, M)}
       {\state(\stk, g, \sp, \ltrue, R, M)}
}
\irule{
g(l) = \Some{{\tt cond}(\cond, \vec r, \ltrue, \lfalse)} \amper
{\tt eval{\char95}cond}(\cond, R(\vec r)) = \Some{{\tt false}}
}{
\trans {\state(\stk, g, \sp, l, R, M)}
       {\state(\stk, g, \sp, \lfalse, R, M)}
}
\irule{
g(l) = \Some{{\tt return}}
}{
\trans {\state(\stk, g, \sp, l, R, M)}
       {\returnstate(\stk, {\tt undef}, {\tt free}(M, \sp))}
}
\irule{
g(l) = \Some{{\tt return}(r)}
}{
\trans {\state(\stk, g, \sp, l, R, M)}
       {\returnstate(\stk, R(r), {\tt free}(M, \sp))}
}
\irule{
{\tt alloc}(M, 0, \fn.{\tt stacksize}) = (\sp, M')
}{
\trans {\callstate(\stk, {\tt internal}(\fn), \vec v, M)}
       {\state(\stk, \fn.{\tt code}, \sp, \fn.{\tt entrypoint}, [\fn.{\tt params} \mapsto \vec v], M')}
}
\irule{
\vdash \ef(\vec v) \evaltrace{t} v \mbox{ (see section~\ref{s:traces})}
}{
\transtrace{\callstate(\stk, {\tt external}(\ef), \vec v, M)}
       {t} {\returnstate(\stk, v, M)}
}
\srule{
\trans {\returnstate(\stackframe(r, g, \sp, l, R).\stk, v, M)}
       {\state(\stk, g, \sp, l, R[r \becomes v], M)}
}

\irule{
{\tt symbol}({\tt globalenv}(\prog), \prog.{\tt main}) = \Some b \amper
{\tt funct}({\tt globalenv}(\prog), b) = \Some \fd
}{
\initialstate{\prog}{\callstate(\epsilon, \fd, \epsilon, {\tt initmem}(\prog))}
}
\srule{
\finalstate{\returnstate(\epsilon, {\tt int}(n), M)}{n}
}
\end{pannel}
\caption{Semantics of \lang{RTL}.} \label{f:rtl-sem}
\end{figure}

The transition system in figure~\ref{f:rtl-sem} is
unsurprising.  Transitions from a regular state discriminate on the
instruction found at the current program point.  To interpret
arithmetic operations, conditions and addressing modes, we reuse the
functions {\tt eval{\char95}op}, {\tt eval{\char95}cond} and {\tt eval{\char95}mode} of the
\lang{CminorSel} semantics.  Other transitions follow the pattern
described in section~\ref{s:state-kinds}.

\subsection{Relational specification of the translation} \label{s:rtlgen-spec}

The translation from \lang{CminorSel} to \lang{RTL} is conceptually
simple: the structured control is encoded as a CFG\footnote{%
Since \lang{RTL} currently has no instructions performing $N$-way
branches (\ie jump tables), this translation of control includes the
generation of binary decision trees for \lang{Cminor} {\tt switch}
statements.  We do not describe this part of the translation in this
article.}; expressions
are decomposed into sequences of \lang{RTL} instructions;
pseudo-registers are generated to hold the values of \lang{CminorSel}
variables and intermediate results of expression evaluations.
The decomposition of expressions is made trivial by the prior conversion
to \lang{CminorSel}: every operation becomes exactly one {\tt op}
instruction.  However, implementing this translation in Coq is
delicate: since Coq is a pure functional language, we cannot use
imperative updates to build the CFG and generate fresh
pseudo-registers and CFG nodes.  Section~\ref{s:rtlgen-impl} 
describes a solution, based (unsurprisingly) on the use of a monad.

\begin{figure}
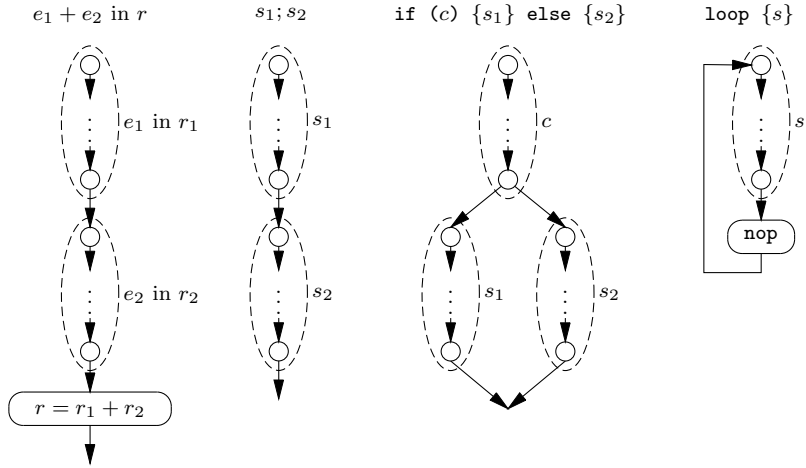


\begin{gpic}

\begin{pgfpicture}
  \pgfsetxvec{\pgfpoint{0.500in}{0in}}
  \pgfsetyvec{\pgfpoint{0in}{0.500in}}
  \begin{pgfscope}
    \pgfpathellipse{\pgfpointxy{0.850}{1.888}}{\pgfpointxy{0.100}{0}}{\pgfpointxy{0}{0.100}}
    \pgfusepath{stroke}
  \end{pgfscope}
  \pgftext[at=\pgfpointadd{\pgfpointxy{0.850}{1.288}}{\pgfpoint{0pt}{-0.0 \baselineskip}}]{$\vdots$}
  \begin{pgfscope}
    \pgfpathellipse{\pgfpointxy{0.850}{0.688}}{\pgfpointxy{0.100}{0}}{\pgfpointxy{0}{0.100}}
    \pgfusepath{stroke}
  \end{pgfscope}
  \begin{pgfscope}
    \pgfpathmoveto{\pgfpointxy{0.900}{1.738}}
    \pgfpathlineto{\pgfpointxy{0.850}{1.538}}
    \pgfpathlineto{\pgfpointxy{0.800}{1.738}}
    \pgfpathlineto{\pgfpointxy{0.900}{1.738}}
    \pgfsetfillcolor{black}
    \pgfsetlinewidth{0.100pt}
    \pgfusepath{fill,stroke}
  \end{pgfscope}
  \begin{pgfscope}
    \pgfpathmoveto{\pgfpointxy{0.850}{1.788}}
    \pgfpathlineto{\pgfpointxy{0.850}{1.738}}
    \pgfusepath{stroke}
  \end{pgfscope}
  \begin{pgfscope}
    \pgfpathmoveto{\pgfpointxy{0.900}{0.988}}
    \pgfpathlineto{\pgfpointxy{0.850}{0.788}}
    \pgfpathlineto{\pgfpointxy{0.800}{0.988}}
    \pgfpathlineto{\pgfpointxy{0.900}{0.988}}
    \pgfsetfillcolor{black}
    \pgfsetlinewidth{0.100pt}
    \pgfusepath{fill,stroke}
  \end{pgfscope}
  \begin{pgfscope}
    \pgfpathmoveto{\pgfpointxy{0.850}{1.038}}
    \pgfpathlineto{\pgfpointxy{0.850}{0.988}}
    \pgfusepath{stroke}
  \end{pgfscope}
  \begin{pgfscope}
    \pgfpathellipse{\pgfpointxy{0.850}{1.288}}{\pgfpointxy{0.300}{0}}{\pgfpointxy{0}{0.800}}
    \pgfsetdash{{0.050in}{0.025in}}{0cm}
    \pgfusepath{stroke}
  \end{pgfscope}
  \pgftext[at=\pgfpointadd{\pgfpointxy{1.200}{1.288}}{\pgfpoint{0pt}{-0.0 \baselineskip}},left]{$e_1$ in $r_1$}
  \begin{pgfscope}
    \pgfpathellipse{\pgfpointxy{0.850}{0.087}}{\pgfpointxy{0.100}{0}}{\pgfpointxy{0}{0.100}}
    \pgfusepath{stroke}
  \end{pgfscope}
  \pgftext[at=\pgfpointadd{\pgfpointxy{0.850}{-0.513}}{\pgfpoint{0pt}{-0.0 \baselineskip}}]{$\vdots$}
  \begin{pgfscope}
    \pgfpathellipse{\pgfpointxy{0.850}{-1.113}}{\pgfpointxy{0.100}{0}}{\pgfpointxy{0}{0.100}}
    \pgfusepath{stroke}
  \end{pgfscope}
  \begin{pgfscope}
    \pgfpathmoveto{\pgfpointxy{0.900}{-0.063}}
    \pgfpathlineto{\pgfpointxy{0.850}{-0.263}}
    \pgfpathlineto{\pgfpointxy{0.800}{-0.063}}
    \pgfpathlineto{\pgfpointxy{0.900}{-0.063}}
    \pgfsetfillcolor{black}
    \pgfsetlinewidth{0.100pt}
    \pgfusepath{fill,stroke}
  \end{pgfscope}
  \begin{pgfscope}
    \pgfpathmoveto{\pgfpointxy{0.850}{-0.013}}
    \pgfpathlineto{\pgfpointxy{0.850}{-0.063}}
    \pgfusepath{stroke}
  \end{pgfscope}
  \begin{pgfscope}
    \pgfpathmoveto{\pgfpointxy{0.900}{-0.813}}
    \pgfpathlineto{\pgfpointxy{0.850}{-1.013}}
    \pgfpathlineto{\pgfpointxy{0.800}{-0.813}}
    \pgfpathlineto{\pgfpointxy{0.900}{-0.813}}
    \pgfsetfillcolor{black}
    \pgfsetlinewidth{0.100pt}
    \pgfusepath{fill,stroke}
  \end{pgfscope}
  \begin{pgfscope}
    \pgfpathmoveto{\pgfpointxy{0.850}{-0.763}}
    \pgfpathlineto{\pgfpointxy{0.850}{-0.813}}
    \pgfusepath{stroke}
  \end{pgfscope}
  \begin{pgfscope}
    \pgfpathellipse{\pgfpointxy{0.850}{-0.513}}{\pgfpointxy{0.300}{0}}{\pgfpointxy{0}{0.800}}
    \pgfsetdash{{0.050in}{0.025in}}{0cm}
    \pgfusepath{stroke}
  \end{pgfscope}
  \pgftext[at=\pgfpointadd{\pgfpointxy{1.200}{-0.513}}{\pgfpoint{0pt}{-0.0 \baselineskip}},left]{$e_2$ in $r_2$}
  \begin{pgfscope}
    \pgfpathmoveto{\pgfpointxy{0.900}{0.388}}
    \pgfpathlineto{\pgfpointxy{0.850}{0.188}}
    \pgfpathlineto{\pgfpointxy{0.800}{0.388}}
    \pgfpathlineto{\pgfpointxy{0.900}{0.388}}
    \pgfsetfillcolor{black}
    \pgfsetlinewidth{0.100pt}
    \pgfusepath{fill,stroke}
  \end{pgfscope}
  \begin{pgfscope}
    \pgfpathmoveto{\pgfpointxy{0.850}{0.587}}
    \pgfpathlineto{\pgfpointxy{0.850}{0.388}}
    \pgfusepath{stroke}
  \end{pgfscope}
  \begin{pgfscope}
    \pgfpathmoveto{\pgfpointxy{0.000}{-1.738}}
    \pgfpatharc{180}{269}{0.075in}
    \pgfusepath{stroke}
  \end{pgfscope}
  \begin{pgfscope}
    \pgfpathmoveto{\pgfpointxy{0.150}{-1.538}}
    \pgfpatharc{90}{180}{0.075in}
    \pgfusepath{stroke}
  \end{pgfscope}
  \begin{pgfscope}
    \pgfpathmoveto{\pgfpointxy{1.700}{-1.688}}
    \pgfpatharc{0}{89}{0.075in}
    \pgfusepath{stroke}
  \end{pgfscope}
  \begin{pgfscope}
    \pgfpathmoveto{\pgfpointxy{1.550}{-1.888}}
    \pgfpatharc{-90}{0}{0.075in}
    \pgfusepath{stroke}
  \end{pgfscope}
  \begin{pgfscope}
    \pgfpathmoveto{\pgfpointxy{0.000}{-1.738}}
    \pgfpathlineto{\pgfpointxy{0.000}{-1.688}}
    \pgfusepath{stroke}
  \end{pgfscope}
  \begin{pgfscope}
    \pgfpathmoveto{\pgfpointxy{0.150}{-1.538}}
    \pgfpathlineto{\pgfpointxy{1.550}{-1.538}}
    \pgfusepath{stroke}
  \end{pgfscope}
  \begin{pgfscope}
    \pgfpathmoveto{\pgfpointxy{1.700}{-1.688}}
    \pgfpathlineto{\pgfpointxy{1.700}{-1.738}}
    \pgfusepath{stroke}
  \end{pgfscope}
  \begin{pgfscope}
    \pgfpathmoveto{\pgfpointxy{1.550}{-1.888}}
    \pgfpathlineto{\pgfpointxy{0.150}{-1.888}}
    \pgfusepath{stroke}
  \end{pgfscope}
  \pgftext[at=\pgfpointadd{\pgfpointxy{0.850}{-1.713}}{\pgfpoint{0pt}{-0.0 \baselineskip}}]{$r = r_1+r_2$}
  \begin{pgfscope}
    \pgfpathmoveto{\pgfpointxy{0.900}{-1.338}}
    \pgfpathlineto{\pgfpointxy{0.850}{-1.538}}
    \pgfpathlineto{\pgfpointxy{0.800}{-1.338}}
    \pgfpathlineto{\pgfpointxy{0.900}{-1.338}}
    \pgfsetfillcolor{black}
    \pgfsetlinewidth{0.100pt}
    \pgfusepath{fill,stroke}
  \end{pgfscope}
  \begin{pgfscope}
    \pgfpathmoveto{\pgfpointxy{0.850}{-1.213}}
    \pgfpathlineto{\pgfpointxy{0.850}{-1.338}}
    \pgfusepath{stroke}
  \end{pgfscope}
  \begin{pgfscope}
    \pgfpathmoveto{\pgfpointxy{0.900}{-2.088}}
    \pgfpathlineto{\pgfpointxy{0.850}{-2.288}}
    \pgfpathlineto{\pgfpointxy{0.800}{-2.088}}
    \pgfpathlineto{\pgfpointxy{0.900}{-2.088}}
    \pgfsetfillcolor{black}
    \pgfsetlinewidth{0.100pt}
    \pgfusepath{fill,stroke}
  \end{pgfscope}
  \begin{pgfscope}
    \pgfpathmoveto{\pgfpointxy{0.850}{-1.888}}
    \pgfpathlineto{\pgfpointxy{0.850}{-2.088}}
    \pgfusepath{stroke}
  \end{pgfscope}
  \pgftext[at=\pgfpointadd{\pgfpointxy{0.850}{2.288}}{\pgfpoint{0pt}{0.5 \baselineskip}}]{$e_1+e_2$ in $r$}
  \begin{pgfscope}
    \pgfpathellipse{\pgfpointxy{2.825}{1.888}}{\pgfpointxy{0.100}{0}}{\pgfpointxy{0}{0.100}}
    \pgfusepath{stroke}
  \end{pgfscope}
  \pgftext[at=\pgfpointadd{\pgfpointxy{2.825}{1.288}}{\pgfpoint{0pt}{-0.0 \baselineskip}}]{$\vdots$}
  \begin{pgfscope}
    \pgfpathellipse{\pgfpointxy{2.825}{0.688}}{\pgfpointxy{0.100}{0}}{\pgfpointxy{0}{0.100}}
    \pgfusepath{stroke}
  \end{pgfscope}
  \begin{pgfscope}
    \pgfpathmoveto{\pgfpointxy{2.875}{1.738}}
    \pgfpathlineto{\pgfpointxy{2.825}{1.538}}
    \pgfpathlineto{\pgfpointxy{2.775}{1.738}}
    \pgfpathlineto{\pgfpointxy{2.875}{1.738}}
    \pgfsetfillcolor{black}
    \pgfsetlinewidth{0.100pt}
    \pgfusepath{fill,stroke}
  \end{pgfscope}
  \begin{pgfscope}
    \pgfpathmoveto{\pgfpointxy{2.825}{1.788}}
    \pgfpathlineto{\pgfpointxy{2.825}{1.738}}
    \pgfusepath{stroke}
  \end{pgfscope}
  \begin{pgfscope}
    \pgfpathmoveto{\pgfpointxy{2.875}{0.988}}
    \pgfpathlineto{\pgfpointxy{2.825}{0.788}}
    \pgfpathlineto{\pgfpointxy{2.775}{0.988}}
    \pgfpathlineto{\pgfpointxy{2.875}{0.988}}
    \pgfsetfillcolor{black}
    \pgfsetlinewidth{0.100pt}
    \pgfusepath{fill,stroke}
  \end{pgfscope}
  \begin{pgfscope}
    \pgfpathmoveto{\pgfpointxy{2.825}{1.038}}
    \pgfpathlineto{\pgfpointxy{2.825}{0.988}}
    \pgfusepath{stroke}
  \end{pgfscope}
  \begin{pgfscope}
    \pgfpathellipse{\pgfpointxy{2.825}{1.288}}{\pgfpointxy{0.300}{0}}{\pgfpointxy{0}{0.800}}
    \pgfsetdash{{0.050in}{0.025in}}{0cm}
    \pgfusepath{stroke}
  \end{pgfscope}
  \pgftext[at=\pgfpointadd{\pgfpointxy{3.175}{1.288}}{\pgfpoint{0pt}{-0.0 \baselineskip}},left]{$s_1$}
  \begin{pgfscope}
    \pgfpathellipse{\pgfpointxy{2.825}{0.088}}{\pgfpointxy{0.100}{0}}{\pgfpointxy{0}{0.100}}
    \pgfusepath{stroke}
  \end{pgfscope}
  \pgftext[at=\pgfpointadd{\pgfpointxy{2.825}{-0.512}}{\pgfpoint{0pt}{-0.0 \baselineskip}}]{$\vdots$}
  \begin{pgfscope}
    \pgfpathellipse{\pgfpointxy{2.825}{-1.112}}{\pgfpointxy{0.100}{0}}{\pgfpointxy{0}{0.100}}
    \pgfusepath{stroke}
  \end{pgfscope}
  \begin{pgfscope}
    \pgfpathmoveto{\pgfpointxy{2.875}{-0.062}}
    \pgfpathlineto{\pgfpointxy{2.825}{-0.262}}
    \pgfpathlineto{\pgfpointxy{2.775}{-0.062}}
    \pgfpathlineto{\pgfpointxy{2.875}{-0.062}}
    \pgfsetfillcolor{black}
    \pgfsetlinewidth{0.100pt}
    \pgfusepath{fill,stroke}
  \end{pgfscope}
  \begin{pgfscope}
    \pgfpathmoveto{\pgfpointxy{2.825}{-0.012}}
    \pgfpathlineto{\pgfpointxy{2.825}{-0.062}}
    \pgfusepath{stroke}
  \end{pgfscope}
  \begin{pgfscope}
    \pgfpathmoveto{\pgfpointxy{2.875}{-0.812}}
    \pgfpathlineto{\pgfpointxy{2.825}{-1.012}}
    \pgfpathlineto{\pgfpointxy{2.775}{-0.812}}
    \pgfpathlineto{\pgfpointxy{2.875}{-0.812}}
    \pgfsetfillcolor{black}
    \pgfsetlinewidth{0.100pt}
    \pgfusepath{fill,stroke}
  \end{pgfscope}
  \begin{pgfscope}
    \pgfpathmoveto{\pgfpointxy{2.825}{-0.762}}
    \pgfpathlineto{\pgfpointxy{2.825}{-0.812}}
    \pgfusepath{stroke}
  \end{pgfscope}
  \begin{pgfscope}
    \pgfpathellipse{\pgfpointxy{2.825}{-0.512}}{\pgfpointxy{0.300}{0}}{\pgfpointxy{0}{0.800}}
    \pgfsetdash{{0.050in}{0.025in}}{0cm}
    \pgfusepath{stroke}
  \end{pgfscope}
  \pgftext[at=\pgfpointadd{\pgfpointxy{3.175}{-0.512}}{\pgfpoint{0pt}{-0.0 \baselineskip}},left]{$s_2$}
  \begin{pgfscope}
    \pgfpathmoveto{\pgfpointxy{2.875}{0.388}}
    \pgfpathlineto{\pgfpointxy{2.825}{0.188}}
    \pgfpathlineto{\pgfpointxy{2.775}{0.388}}
    \pgfpathlineto{\pgfpointxy{2.875}{0.388}}
    \pgfsetfillcolor{black}
    \pgfsetlinewidth{0.100pt}
    \pgfusepath{fill,stroke}
  \end{pgfscope}
  \begin{pgfscope}
    \pgfpathmoveto{\pgfpointxy{2.825}{0.588}}
    \pgfpathlineto{\pgfpointxy{2.825}{0.388}}
    \pgfusepath{stroke}
  \end{pgfscope}
  \begin{pgfscope}
    \pgfpathmoveto{\pgfpointxy{2.875}{-1.412}}
    \pgfpathlineto{\pgfpointxy{2.825}{-1.612}}
    \pgfpathlineto{\pgfpointxy{2.775}{-1.412}}
    \pgfpathlineto{\pgfpointxy{2.875}{-1.412}}
    \pgfsetfillcolor{black}
    \pgfsetlinewidth{0.100pt}
    \pgfusepath{fill,stroke}
  \end{pgfscope}
  \begin{pgfscope}
    \pgfpathmoveto{\pgfpointxy{2.825}{-1.212}}
    \pgfpathlineto{\pgfpointxy{2.825}{-1.412}}
    \pgfusepath{stroke}
  \end{pgfscope}
  \pgftext[at=\pgfpointadd{\pgfpointxy{2.850}{2.288}}{\pgfpoint{0pt}{0.5 \baselineskip}}]{$s_1; s_2$}
  \begin{pgfscope}
    \pgfpathellipse{\pgfpointxy{5.225}{1.888}}{\pgfpointxy{0.100}{0}}{\pgfpointxy{0}{0.100}}
    \pgfusepath{stroke}
  \end{pgfscope}
  \pgftext[at=\pgfpointadd{\pgfpointxy{5.225}{1.288}}{\pgfpoint{0pt}{-0.0 \baselineskip}}]{$\vdots$}
  \begin{pgfscope}
    \pgfpathellipse{\pgfpointxy{5.225}{0.688}}{\pgfpointxy{0.100}{0}}{\pgfpointxy{0}{0.100}}
    \pgfusepath{stroke}
  \end{pgfscope}
  \begin{pgfscope}
    \pgfpathmoveto{\pgfpointxy{5.275}{1.738}}
    \pgfpathlineto{\pgfpointxy{5.225}{1.538}}
    \pgfpathlineto{\pgfpointxy{5.175}{1.738}}
    \pgfpathlineto{\pgfpointxy{5.275}{1.738}}
    \pgfsetfillcolor{black}
    \pgfsetlinewidth{0.100pt}
    \pgfusepath{fill,stroke}
  \end{pgfscope}
  \begin{pgfscope}
    \pgfpathmoveto{\pgfpointxy{5.225}{1.788}}
    \pgfpathlineto{\pgfpointxy{5.225}{1.738}}
    \pgfusepath{stroke}
  \end{pgfscope}
  \begin{pgfscope}
    \pgfpathmoveto{\pgfpointxy{5.275}{0.988}}
    \pgfpathlineto{\pgfpointxy{5.225}{0.788}}
    \pgfpathlineto{\pgfpointxy{5.175}{0.988}}
    \pgfpathlineto{\pgfpointxy{5.275}{0.988}}
    \pgfsetfillcolor{black}
    \pgfsetlinewidth{0.100pt}
    \pgfusepath{fill,stroke}
  \end{pgfscope}
  \begin{pgfscope}
    \pgfpathmoveto{\pgfpointxy{5.225}{1.038}}
    \pgfpathlineto{\pgfpointxy{5.225}{0.988}}
    \pgfusepath{stroke}
  \end{pgfscope}
  \begin{pgfscope}
    \pgfpathellipse{\pgfpointxy{5.225}{1.288}}{\pgfpointxy{0.300}{0}}{\pgfpointxy{0}{0.800}}
    \pgfsetdash{{0.050in}{0.025in}}{0cm}
    \pgfusepath{stroke}
  \end{pgfscope}
  \pgftext[at=\pgfpointadd{\pgfpointxy{5.575}{1.288}}{\pgfpoint{0pt}{-0.0 \baselineskip}},left]{$c$}
  \begin{pgfscope}
    \pgfpathellipse{\pgfpointxy{4.625}{0.087}}{\pgfpointxy{0.100}{0}}{\pgfpointxy{0}{0.100}}
    \pgfusepath{stroke}
  \end{pgfscope}
  \pgftext[at=\pgfpointadd{\pgfpointxy{4.625}{-0.513}}{\pgfpoint{0pt}{-0.0 \baselineskip}}]{$\vdots$}
  \begin{pgfscope}
    \pgfpathellipse{\pgfpointxy{4.625}{-1.113}}{\pgfpointxy{0.100}{0}}{\pgfpointxy{0}{0.100}}
    \pgfusepath{stroke}
  \end{pgfscope}
  \begin{pgfscope}
    \pgfpathmoveto{\pgfpointxy{4.675}{-0.063}}
    \pgfpathlineto{\pgfpointxy{4.625}{-0.263}}
    \pgfpathlineto{\pgfpointxy{4.575}{-0.063}}
    \pgfpathlineto{\pgfpointxy{4.675}{-0.063}}
    \pgfsetfillcolor{black}
    \pgfsetlinewidth{0.100pt}
    \pgfusepath{fill,stroke}
  \end{pgfscope}
  \begin{pgfscope}
    \pgfpathmoveto{\pgfpointxy{4.625}{-0.013}}
    \pgfpathlineto{\pgfpointxy{4.625}{-0.063}}
    \pgfusepath{stroke}
  \end{pgfscope}
  \begin{pgfscope}
    \pgfpathmoveto{\pgfpointxy{4.675}{-0.813}}
    \pgfpathlineto{\pgfpointxy{4.625}{-1.013}}
    \pgfpathlineto{\pgfpointxy{4.575}{-0.813}}
    \pgfpathlineto{\pgfpointxy{4.675}{-0.813}}
    \pgfsetfillcolor{black}
    \pgfsetlinewidth{0.100pt}
    \pgfusepath{fill,stroke}
  \end{pgfscope}
  \begin{pgfscope}
    \pgfpathmoveto{\pgfpointxy{4.625}{-0.763}}
    \pgfpathlineto{\pgfpointxy{4.625}{-0.813}}
    \pgfusepath{stroke}
  \end{pgfscope}
  \begin{pgfscope}
    \pgfpathellipse{\pgfpointxy{4.625}{-0.513}}{\pgfpointxy{0.300}{0}}{\pgfpointxy{0}{0.800}}
    \pgfsetdash{{0.050in}{0.025in}}{0cm}
    \pgfusepath{stroke}
  \end{pgfscope}
  \pgftext[at=\pgfpointadd{\pgfpointxy{4.975}{-0.513}}{\pgfpoint{0pt}{-0.0 \baselineskip}},left]{$s_1$}
  \begin{pgfscope}
    \pgfpathellipse{\pgfpointxy{5.825}{0.087}}{\pgfpointxy{0.100}{0}}{\pgfpointxy{0}{0.100}}
    \pgfusepath{stroke}
  \end{pgfscope}
  \pgftext[at=\pgfpointadd{\pgfpointxy{5.825}{-0.513}}{\pgfpoint{0pt}{-0.0 \baselineskip}}]{$\vdots$}
  \begin{pgfscope}
    \pgfpathellipse{\pgfpointxy{5.825}{-1.113}}{\pgfpointxy{0.100}{0}}{\pgfpointxy{0}{0.100}}
    \pgfusepath{stroke}
  \end{pgfscope}
  \begin{pgfscope}
    \pgfpathmoveto{\pgfpointxy{5.875}{-0.063}}
    \pgfpathlineto{\pgfpointxy{5.825}{-0.263}}
    \pgfpathlineto{\pgfpointxy{5.775}{-0.063}}
    \pgfpathlineto{\pgfpointxy{5.875}{-0.063}}
    \pgfsetfillcolor{black}
    \pgfsetlinewidth{0.100pt}
    \pgfusepath{fill,stroke}
  \end{pgfscope}
  \begin{pgfscope}
    \pgfpathmoveto{\pgfpointxy{5.825}{-0.013}}
    \pgfpathlineto{\pgfpointxy{5.825}{-0.063}}
    \pgfusepath{stroke}
  \end{pgfscope}
  \begin{pgfscope}
    \pgfpathmoveto{\pgfpointxy{5.875}{-0.813}}
    \pgfpathlineto{\pgfpointxy{5.825}{-1.013}}
    \pgfpathlineto{\pgfpointxy{5.775}{-0.813}}
    \pgfpathlineto{\pgfpointxy{5.875}{-0.813}}
    \pgfsetfillcolor{black}
    \pgfsetlinewidth{0.100pt}
    \pgfusepath{fill,stroke}
  \end{pgfscope}
  \begin{pgfscope}
    \pgfpathmoveto{\pgfpointxy{5.825}{-0.763}}
    \pgfpathlineto{\pgfpointxy{5.825}{-0.813}}
    \pgfusepath{stroke}
  \end{pgfscope}
  \begin{pgfscope}
    \pgfpathellipse{\pgfpointxy{5.825}{-0.513}}{\pgfpointxy{0.300}{0}}{\pgfpointxy{0}{0.800}}
    \pgfsetdash{{0.050in}{0.025in}}{0cm}
    \pgfusepath{stroke}
  \end{pgfscope}
  \pgftext[at=\pgfpointadd{\pgfpointxy{6.175}{-0.513}}{\pgfpoint{0pt}{-0.0 \baselineskip}},left]{$s_2$}
  \begin{pgfscope}
    \pgfpathmoveto{\pgfpointxy{4.812}{0.275}}
    \pgfpathlineto{\pgfpointxy{4.625}{0.188}}
    \pgfpathlineto{\pgfpointxy{4.749}{0.352}}
    \pgfpathlineto{\pgfpointxy{4.812}{0.275}}
    \pgfsetfillcolor{black}
    \pgfsetlinewidth{0.100pt}
    \pgfusepath{fill,stroke}
  \end{pgfscope}
  \begin{pgfscope}
    \pgfpathmoveto{\pgfpointxy{5.154}{0.617}}
    \pgfpathlineto{\pgfpointxy{4.780}{0.313}}
    \pgfusepath{stroke}
  \end{pgfscope}
  \begin{pgfscope}
    \pgfpathmoveto{\pgfpointxy{5.701}{0.352}}
    \pgfpathlineto{\pgfpointxy{5.825}{0.188}}
    \pgfpathlineto{\pgfpointxy{5.638}{0.275}}
    \pgfpathlineto{\pgfpointxy{5.701}{0.352}}
    \pgfsetfillcolor{black}
    \pgfsetlinewidth{0.100pt}
    \pgfusepath{fill,stroke}
  \end{pgfscope}
  \begin{pgfscope}
    \pgfpathmoveto{\pgfpointxy{5.296}{0.617}}
    \pgfpathlineto{\pgfpointxy{5.670}{0.313}}
    \pgfusepath{stroke}
  \end{pgfscope}
  \begin{pgfscope}
    \pgfpathmoveto{\pgfpointxy{5.103}{-1.546}}
    \pgfpathlineto{\pgfpointxy{5.225}{-1.713}}
    \pgfpathlineto{\pgfpointxy{5.039}{-1.623}}
    \pgfpathlineto{\pgfpointxy{5.103}{-1.546}}
    \pgfsetfillcolor{black}
    \pgfsetlinewidth{0.100pt}
    \pgfusepath{fill,stroke}
  \end{pgfscope}
  \begin{pgfscope}
    \pgfpathmoveto{\pgfpointxy{4.625}{-1.213}}
    \pgfpathlineto{\pgfpointxy{5.071}{-1.584}}
    \pgfusepath{stroke}
  \end{pgfscope}
  \begin{pgfscope}
    \pgfpathmoveto{\pgfpointxy{5.411}{-1.623}}
    \pgfpathlineto{\pgfpointxy{5.225}{-1.713}}
    \pgfpathlineto{\pgfpointxy{5.347}{-1.546}}
    \pgfpathlineto{\pgfpointxy{5.411}{-1.623}}
    \pgfsetfillcolor{black}
    \pgfsetlinewidth{0.100pt}
    \pgfusepath{fill,stroke}
  \end{pgfscope}
  \begin{pgfscope}
    \pgfpathmoveto{\pgfpointxy{5.825}{-1.213}}
    \pgfpathlineto{\pgfpointxy{5.379}{-1.584}}
    \pgfusepath{stroke}
  \end{pgfscope}
  \pgftext[at=\pgfpointadd{\pgfpointxy{5.250}{2.288}}{\pgfpoint{0pt}{0.5 \baselineskip}}]{\texttt{if ($c$) \{$s_1$\} else \{$s_2$\}}}
  \begin{pgfscope}
    \pgfpathellipse{\pgfpointxy{7.875}{1.888}}{\pgfpointxy{0.100}{0}}{\pgfpointxy{0}{0.100}}
    \pgfusepath{stroke}
  \end{pgfscope}
  \pgftext[at=\pgfpointadd{\pgfpointxy{7.875}{1.288}}{\pgfpoint{0pt}{-0.0 \baselineskip}}]{$\vdots$}
  \begin{pgfscope}
    \pgfpathellipse{\pgfpointxy{7.875}{0.688}}{\pgfpointxy{0.100}{0}}{\pgfpointxy{0}{0.100}}
    \pgfusepath{stroke}
  \end{pgfscope}
  \begin{pgfscope}
    \pgfpathmoveto{\pgfpointxy{7.925}{1.738}}
    \pgfpathlineto{\pgfpointxy{7.875}{1.538}}
    \pgfpathlineto{\pgfpointxy{7.825}{1.738}}
    \pgfpathlineto{\pgfpointxy{7.925}{1.738}}
    \pgfsetfillcolor{black}
    \pgfsetlinewidth{0.100pt}
    \pgfusepath{fill,stroke}
  \end{pgfscope}
  \begin{pgfscope}
    \pgfpathmoveto{\pgfpointxy{7.875}{1.788}}
    \pgfpathlineto{\pgfpointxy{7.875}{1.738}}
    \pgfusepath{stroke}
  \end{pgfscope}
  \begin{pgfscope}
    \pgfpathmoveto{\pgfpointxy{7.925}{0.988}}
    \pgfpathlineto{\pgfpointxy{7.875}{0.788}}
    \pgfpathlineto{\pgfpointxy{7.825}{0.988}}
    \pgfpathlineto{\pgfpointxy{7.925}{0.988}}
    \pgfsetfillcolor{black}
    \pgfsetlinewidth{0.100pt}
    \pgfusepath{fill,stroke}
  \end{pgfscope}
  \begin{pgfscope}
    \pgfpathmoveto{\pgfpointxy{7.875}{1.038}}
    \pgfpathlineto{\pgfpointxy{7.875}{0.988}}
    \pgfusepath{stroke}
  \end{pgfscope}
  \begin{pgfscope}
    \pgfpathellipse{\pgfpointxy{7.875}{1.288}}{\pgfpointxy{0.300}{0}}{\pgfpointxy{0}{0.800}}
    \pgfsetdash{{0.050in}{0.025in}}{0cm}
    \pgfusepath{stroke}
  \end{pgfscope}
  \pgftext[at=\pgfpointadd{\pgfpointxy{8.225}{1.288}}{\pgfpoint{0pt}{-0.0 \baselineskip}},left]{$s$}
  \begin{pgfscope}
    \pgfpathmoveto{\pgfpointxy{7.525}{0.063}}
    \pgfpatharc{180}{270}{0.075in}
    \pgfusepath{stroke}
  \end{pgfscope}
  \begin{pgfscope}
    \pgfpathmoveto{\pgfpointxy{7.675}{0.262}}
    \pgfpatharc{90}{179}{0.075in}
    \pgfusepath{stroke}
  \end{pgfscope}
  \begin{pgfscope}
    \pgfpathmoveto{\pgfpointxy{8.225}{0.113}}
    \pgfpatharc{0}{90}{0.075in}
    \pgfusepath{stroke}
  \end{pgfscope}
  \begin{pgfscope}
    \pgfpathmoveto{\pgfpointxy{8.075}{-0.087}}
    \pgfpatharc{-90}{0}{0.075in}
    \pgfusepath{stroke}
  \end{pgfscope}
  \begin{pgfscope}
    \pgfpathmoveto{\pgfpointxy{7.525}{0.063}}
    \pgfpathlineto{\pgfpointxy{7.525}{0.113}}
    \pgfusepath{stroke}
  \end{pgfscope}
  \begin{pgfscope}
    \pgfpathmoveto{\pgfpointxy{7.675}{0.263}}
    \pgfpathlineto{\pgfpointxy{8.075}{0.263}}
    \pgfusepath{stroke}
  \end{pgfscope}
  \begin{pgfscope}
    \pgfpathmoveto{\pgfpointxy{8.225}{0.113}}
    \pgfpathlineto{\pgfpointxy{8.225}{0.063}}
    \pgfusepath{stroke}
  \end{pgfscope}
  \begin{pgfscope}
    \pgfpathmoveto{\pgfpointxy{8.075}{-0.087}}
    \pgfpathlineto{\pgfpointxy{7.675}{-0.087}}
    \pgfusepath{stroke}
  \end{pgfscope}
  \pgftext[at=\pgfpointadd{\pgfpointxy{7.875}{0.088}}{\pgfpoint{0pt}{-0.0 \baselineskip}}]{\texttt{nop}}
  \begin{pgfscope}
    \pgfpathmoveto{\pgfpointxy{7.925}{0.463}}
    \pgfpathlineto{\pgfpointxy{7.875}{0.263}}
    \pgfpathlineto{\pgfpointxy{7.825}{0.463}}
    \pgfpathlineto{\pgfpointxy{7.925}{0.463}}
    \pgfsetfillcolor{black}
    \pgfsetlinewidth{0.100pt}
    \pgfusepath{fill,stroke}
  \end{pgfscope}
  \begin{pgfscope}
    \pgfpathmoveto{\pgfpointxy{7.875}{0.588}}
    \pgfpathlineto{\pgfpointxy{7.875}{0.463}}
    \pgfusepath{stroke}
  \end{pgfscope}
  \begin{pgfscope}
    \pgfpathmoveto{\pgfpointxy{7.576}{1.939}}
    \pgfpathlineto{\pgfpointxy{7.775}{1.888}}
    \pgfpathlineto{\pgfpointxy{7.575}{1.840}}
    \pgfpathlineto{\pgfpointxy{7.576}{1.939}}
    \pgfsetfillcolor{black}
    \pgfsetlinewidth{0.100pt}
    \pgfusepath{fill,stroke}
  \end{pgfscope}
  \begin{pgfscope}
    \pgfpathmoveto{\pgfpointxy{7.875}{-0.087}}
    \pgfpathlineto{\pgfpointxy{7.875}{-0.287}}
    \pgfpathlineto{\pgfpointxy{7.275}{-0.287}}
    \pgfpathlineto{\pgfpointxy{7.275}{1.893}}
    \pgfpathlineto{\pgfpointxy{7.575}{1.889}}
    \pgfusepath{stroke}
  \end{pgfscope}
  \pgftext[at=\pgfpointadd{\pgfpointxy{7.750}{2.288}}{\pgfpoint{0pt}{0.5 \baselineskip}}]{\texttt{loop \{$s$\}}}
\end{pgfpicture}
\ifx\Setlineno\undefined\else\Setlineno=103\fi

\end{gpic}

\caption{Examples of correspondences between expressions/statements
  and sub-graphs of a CFG.}
\label{f:rtlgen-decomp}

\end{figure}

In the present section, we give a relational, non-executable
specification of the translation: syntactic conditions under which
a \lang{RTL} function is an acceptable translation of a
\lang{CminorSel} function.  The key intuition captured by this
specification is that each subexpression or substatement of the
\lang{CminorSel} function should correspond to a sub-graph of the
\lang{RTL} control-flow graph.  For an expression $a$, this sub-graph
is identified by a start node $l_1$ and an end node $l_e$.  The
instructions on the paths from $l_1$ to $l_e$ should compute the value
of expression $a$, deposit it in a given destination register $r_d$,
and preserve the values of a given set of registers.  For a statement
$s$, the sub-graph has one start node but several end nodes
corresponding to the multiple ways in which a statement can terminate
(normally, by {\tt exit}, or by {\tt goto}).  Figure~\ref{f:rtlgen-decomp}
illustrates this correspondence.  As depicted there, the sub-graph
for a compound expression or compound statement such as $s_1; s_2$
contains sub-graphs for its components $s_1$ and $s_2$, suitably
connected.  In other words, the relational specification of the
translation describes a hierarchical decomposition of the CFG along
the structure of statements and expressions of the original
\lang{CminorSel} code.

The specification for expressions is the predicate
$g, \gamma, \pi \vdash a{\tt \ in\ }r_d \match l_1, l_2$, where $g$ is the CFG,
$\gamma$ an injective mapping from \lang{CminorSel} variables to the registers
holding their values, $\pi$ a set of registers that the instructions
evaluating $a$ must preserve (in addition to those in
$\Range(\gamma)$), $a$ the \lang{CminorSel} expression, $r_d$ the
register where its value must be deposited, and $l_1, l_2$ the start
and end nodes in the CFG.  The following rules give the flavor of the
specification:

\begin{pannel}
\irule{
\gamma(\id) = \Some{r_d}
}{
g, \gamma, \pi \vdash \id{\tt \ in\ }r_d \match l_1, l_1
}
\irule{
\gamma(\id) = \Some{r} \amper
g(l_1) = \Some{{\tt op}({\tt move}, r, r_d, l_2)} \amper
r_d \notin \Range(\gamma) \union \pi
}{
g, \gamma, \pi \vdash \id{\tt \ in\ }r_d \match l_1, l_2
}
\irule{
g, \gamma, \pi \vdash \vec a{\tt \ in\ }\vec r \match l_1, l \amper
g(l) = \Some{{\tt op}(\op, \vec r, r_d, l_2)} \amper
r_d \notin \Range(\gamma) \union \pi
}{
g, \gamma, \pi \vdash \op(\vec a){\tt \ in\ }r_d \match l_1, l_2
}
\srule{
g, \gamma, \pi \vdash \epsilon{\tt \ in\ }\epsilon \match l_1, l_1
}
\irule{
g, \gamma, \pi \vdash a{\tt \ in\ }r \match l_1, l \amper
g, \gamma, \pi \union \{r\} \vdash \vec a{\tt \ in\ }\vec r \match l, l_2 \amper
r \notin \Range(\gamma) \union \pi
}{
g, \gamma, \pi \vdash a.\vec a{\tt \ in\ }r.\vec r \match l_1, l_2
}
\end{pannel}
The freshness side-conditions $r \notin \Range(\gamma) \union \pi$
ensure that the temporary registers used to hold the values of
subexpressions of $a$ do not interfere with registers holding values
of \lang{CminorSel} variables ($\Range(\gamma)$) nor with temporary
registers holding values of previously computed subexpressions
(e.g.~in the expression ${\tt add}(a_1, a_2)$, the value of $a_1$ during
the computation of the value of $a_2$).
The specification for conditional expressions is similar:
$g, \gamma, \pi \vdash c \match l_1, \ltrue, \lfalse$,
but with two exit nodes, $\ltrue$ and $\lfalse$.  The
instruction sequence starting at $l_1$ should terminate at 
$\ltrue$ if $c$ evaluates to {\tt true} and at $\lfalse$ if $c$
evaluates to {\tt false}.  
Finally, the translation of statements is specified by the predicate
$g, \gamma \vdash s \match l_1, l_2, \lexit, \lgoto, \lret, \rret$,
where $\lexit$ is a list of nodes and $\lgoto$ a
mapping from \lang{CminorSel} labels to nodes.
The contract expressed by this complicated predicate is that
the instruction sequence starting at $l_1$ should compute whatever $s$
computes and branch to node $l_2$ if $s$ terminates normally, to
$\lexit(n)$ if $s$ terminates by ${\tt exit}(n)$, to $\lgoto(\lbl)$ if
$s$ performs ${\tt goto\ }\lbl$, and to $\lret$ if $s$ performs a
{\tt return} (after depositing the return value, if any, in register
$\rret$).  For simplicity, figure~\ref{f:rtlgen-decomp} depicts
only the $l_2$ final node, and the following sample rules consider
only the $l_2$ and $\lexit$ final nodes.
\begin{pannel}
\irule{
\gamma(\id) = \Some{r_d} \amper
g, \gamma, \emptyset \vdash a{\tt \ in\ }r_d \match l_1, l_2
}{
g, \gamma \vdash (\id = a) \match l_1, l_2, \lexit
}
\irule{
\gamma(\id) = \Some{r} \amper
g, \gamma, \emptyset \vdash a{\tt \ in\ }r_d \match l_1, l \amper
g(l) = \Some{{\tt op}({\tt move}, r_d, r, l_2)} 
}{
g, \gamma \vdash (\id = a) \match l_1, l_2, \lexit
}
\irule{
g, \gamma \vdash s_1 \match l_1, l, \lexit \amper
g, \gamma \vdash s_2 \match l, l_2, \lexit
}{
g, \gamma \vdash (s_1; s_2) \match l_1, l_2, \lexit
}
\irule{
g, \gamma, \emptyset \vdash c \match l_1, \ltrue, \lfalse \\
g, \gamma \vdash s_1 \match \ltrue, l_2, \lexit \amper
g, \gamma \vdash s_2 \match \lfalse, l_2, \lexit
}{
g, \gamma \vdash {\tt if}(c)\{s_1\}{\tt \ else\ }\{s_2\} \match l_1, l_2, \lexit
}
\irule{
g(l_1) = \Some{{\tt nop}(l)} \amper
g, \gamma \vdash s \match l, l_1, \lexit 
}{
g, \gamma \vdash {\tt loop}\{s\} \match l_1, l_2, \lexit
}
\irule{
g, \gamma \vdash s \match l_1, l_2, l_2.\lexit
}{
g, \gamma \vdash {\tt block}\{s\} \match l_1, l_2, \lexit
}
\irule{
\lexit(n) = \Some{l_1}
}{
g, \gamma \vdash {\tt exit}(n) \match l_1, l_2, \lexit
}
\end{pannel}
The specification for the translation of \lang{CminorSel} functions to
\lang{RTL} functions is, then: if $\tr{\fn} = \Some{\fn'}$, it must be
the case that 
$$\fn'.{\tt code}, \gamma \vdash \fn.{\tt body} \match \fn'.{\tt entrypoint}, l,
\epsilon, \lgoto, l, \rret$$
for some $l$, $\lgoto$ and injective $\gamma$; moreover, the CFG node
at $l$ must contain a ${\tt return}(\rret)$ instruction appropriate
to the signature of $\fn$.
\subsection{Semantic preservation} \label{s:rtlgen-preservation}
We now outline the proof of semantic preservation for RTL generation.
Consider a \lang{CminorSel} program $\prog$ and an \lang{RTL} program
$\prog'$, with $G, G'$ being the corresponding global environments.
Assuming that $\prog'$ is an acceptable translation of $\prog$
according to the relational specification of
section~\ref{s:rtlgen-spec}, we show that executions of $\prog$ are
simulated by executions of $\prog'$.  
For the evaluation of \lang{CminorSel} expressions, the simulation
argument is of the form ``if expression $a$ evaluates to value $v$,
the generated \lang{RTL} code performs a sequence of transitions from
$l_1$ to $l_2$, where $l_1, l_2$ delimit the sub-graph of the CFG
corresponding to $a$, and leaves $v$ in the given destination register
$r_d$''.  Agreement between a \lang{CminorSel} environment $E$
and an \lang{RTL} register state $R$ is written $\gamma \vdash E \match R$
and defined as $R(\gamma(x)) = E(x)$ for all $x \in \Dom(\gamma)$.
\begin{lemma} \label{l:rtlgen-exprs}
Assume $\fn.{\tt code}, \gamma, \pi \vdash a{\tt \ in\ }r_d \match l_1, l_2$
and $\gamma$ is injective.  \\
If $G, \sp, E, M \vdash a \eval v$ and $\gamma \vdash E \match R$,
there exists $R'$ such that
\begin{enumerate}
\item The \lang{RTL} code executes from $l_1$ to $l_2$:
$$G' \vdash \state(\stk, \fn, \sp, l_1, R, M) \redm{\epsilon}
             \state(\stk, \fn, \sp, l_2, R', M)$$
\item Register $r_d$ contains the value of $v$ at the end of this
  execution: $R'(r_d) = v$
\item The values of preserved registers are unchanged:
  $R'(r) = R(r)$ for all $r \in \Range(\gamma) \union \pi$.
This implies $\gamma \vdash E \match R'$ in particular.
\end{enumerate}
\end{lemma}
This lemma, along with a similar lemma for condition expressions, is
proved by induction on the \lang{CminorSel} evaluation derivation.
To relate \lang{CminorSel} and \lang{RTL} execution states, we need to
define a correspondence between \lang{CminorSel} continuations and
\lang{RTL} call stacks.  A continuation $k$ interleaves two aspects
that are handled separately in \lang{RTL}.  The parts of $k$ that lie
between {\tt returnto} markers, namely the continuations $s; k'$ and
${\tt endblock}(k')$, correspond to execution paths within the current
function.  These paths connect the several possible end points for the
current statement with the final {\tt return} of the function.  Just as
a statement $s$ is associated with a ``fan-out'' subgraph of the CFG
(one start point, several end points), this part of $k$ is associated
with a ``fan-in'' subgraph of the CFG (several start points, one end
point corresponding to a {\tt return} instruction).
The other parts of $k$, namely the {\tt returnto} markers, are in
one-to-one correspondence with frames on the \lang{RTL} call stack
$\stk$. 
We formalize these intuitions using two mutually
inductive predicates: $g, \gamma \vdash k \match l_2, \lexit, \lgoto,
\lret, \rret, \stk$ for the local part of the continuation $k$,
and $k \match S$ for call continuations $k$.  The definitions are
omitted for brevity.  The invariant between \lang{CminorSel} states
and \lang{RTL} states is, then, of the following form.
\begin{pannel}
\irule{
g, \gamma \vdash s \match l_1, l_2, \lexit, \lgoto, \lret, \rret \amper
g, \gamma \vdash k \match l_2, \lexit, \lgoto, \lret, \rret, \stk \amper
\gamma \vdash E \match R
}{
\state(\fn, s, k, \sp, E, M) \match \state(\stk, g, \sp, l_1, R, M)
}
\irule{
\tr{\fd} = \Some{\fd'} \amper k \match \stk
}{
\callstate(\fd, \vec v, k, M) \match \callstate(\stk, \fd', \vec v, M)
}
\irule{
k \match \stk
}{
\returnstate(v, k, M) \match \returnstate(\stk, v, M)
}
\end{pannel}
The proof of semantic preservation for statements and functions is
a simulation diagram of the ``star'' kind (see
section~\ref{s:simulation-diagrams}).  Several 
\lang{CminorSel} transitions become no-operations in the translated
\lang{RTL} code, such as self assignments $\id = \id$ and ${\tt exit}(n)$
constructs.  We therefore need to define a measure over
\lang{CminorSel} states that decreases on such potentially stuttering
steps.  After trial and error, an appropriate measure for regular
states is the lexicographically-ordered pair of nonnegative integers
$\size{\state(\fn, s, k, \sp, E, M)} = (\size{s} + \size{k}, \size{s})$
where $\size{s}$ is the number of nodes in the abstract syntax tree
for $s$, and
$$ \size{s;k} = 1 + \size{s} + \size{k}
\qquad
   \size{{\tt endblock}(k)} = 1 + \size{k}
\qquad
   \size{k} = 0 \mbox{ otherwise.} $$
The measure for call states and return states is $(0, 0)$.
\begin{lemma} \label{l:rtlgen-simu}
If $G \vdash \st_1 \redtrace{t} \st_2$ and $\st_1 \match \st_1'$,
either
there exists $\st_2'$ such that $G' \vdash \st_1' \redp{t} \st_2'$
and $\st_1' \match \st_2'$,
or
$\size{\st_2} < \size{\st_1}$ and
there exists $\st_2'$ such that $G' \vdash \st_1' \redm{t} \st_2'$
and $\st_1' \match \st_2'$.
\end{lemma}
The proof is a long case analysis on the transition
$G \vdash \st_1 \redtrace{t} \st_2$.  Semantic preservation for RTL
generation then follows from theorem~\ref{th:star-simulation}.
\subsection{Functional implementation of the translation} \label{s:rtlgen-impl}
We now return to the question left open at the beginning of
section~\ref{s:rtlgen-preservation}: how to define the generation of
\lang{RTL} as a Coq function?  Naturally, the translation proceeds by
a recursive traversal of \lang{CminorSel} expressions and statements,
incrementally adding the corresponding instructions to the CFG and
generating fresh temporary registers to hold intermediate results
within expressions.  Additionally, the translation may fail, e.g.~if
an undeclared local variable is referenced or a statement
label is defined several times.  This would cause no
programming difficulties in a language featuring mutation and
exceptions, but these luxuries are not available in Coq, which is a
pure functional language.
We therefore use a monadic programming style using the state-and-error
monad.  The compile-time state is a triple $(g, l, r)$ where
$g$ is the current state of the CFG, $l$ the next unused CFG node
label, and $r$ the next unused pseudo-register.  Every translation
that computes (imperatively) a result of type $\alpha$ becomes a pure
function with type
${\tt mon}(\alpha) = {\tt state} \fun {\tt Error} \mid {\tt OK}({\tt state}\times\alpha)$.
Besides the familiar {\tt ret} and {\tt bind} monadic combinators, the basic
operations of this monad are:
\begin{itemize}
\item ${\tt newreg} : {\tt mon}({\tt reg})$ generates a fresh temporary register
  (by incrementing the $r$ component of the state);
\item ${\tt add{\char95}instr}(i) : {\tt mon}({\tt node})$ allocates a fresh CFG node $l$,
  stores the instruction $i$ in this node, and returns $l$;
\item ${\tt reserve{\char95}instr} : {\tt mon}({\tt node})$ allocates and returns a fresh
  CFG node, leaving it empty;
\item ${\tt update{\char95}instr}(l, i) : {\tt mon}({\tt unit})$ stores instruction $i$
  in node $l$, raising an error if node $l$ is not empty.
\end{itemize}
(The latter two operations are used when compiling loops and labeled
statements.)  The translation functions, then, are of the following
form:
$$\begin{array}{lll}
{\tt transl{\char95}expr}(\gamma, a, r_d, l_2) & : & {\tt mon}({\tt node}) \\
{\tt transl{\char95}exprlist}(\gamma, \vec a, \vec{r_d}, l_2) & : & {\tt mon}({\tt node}) \\
{\tt transl{\char95}condition}(\gamma, c, \ltrue, \lfalse) & : & {\tt mon}({\tt node}) \\
{\tt transl{\char95}stmt}(\gamma, s, l_2, \lexit, \lgoto, \lret, \rret) & : & {\tt mon}({\tt node})
\end{array}$$
These functions recursively add to the CFG the instructions that
compute the given expression or statement and branch to the given
end nodes $l_2, \lexit, \ldots$.  Each call returns the node of the first
instruction in this sequence (the node written $l_1$ in the relational
specification).  The following Coq excerpts from {\tt transl{\char95}stmt}
should give the flavor of the translation functions:
\begin{alltt}
  match s with
  | Sskip =>
      ret l2
  | Sassign v a =>
      do rd <- new_reg;
      do rv <- find_var \(\gamma\) v;
      do l  <- add_instr (Iop Omove (rd::nil) rv l2);
      transl_expr \(\gamma\) a rd l
  | Sseq s1 s2 =>
      do l <- transl_stmt \(\gamma\) s2 l2 lexit lgoto lret rret;
      transl_stmt \(\gamma\) s1 l lexit lgoto lret rret
  | Sloop s =>
      do l <- reserve_instr;
      do l' <- transl_stmt \(\gamma\) s l lexit lgoto lret rret;
      do x  <- update_instr l (Inop l');
      ret l
\end{alltt}
Inspired by Haskell, {\tt do}~$x$~{\tt <-}~$a; b$ is a user-defined Coq
notation standing for ${\tt bind}~a~(\lambda x.b)$.
Two syntactic invariants of the state play a crucial role in proving
the correctness of the generated CFG against the relational
specification of section~\ref{s:rtlgen-spec}.
First, in a compile-time state $(g, l, r)$, all
CFG nodes above $l$ must be empty:
$\forall l' \ge l,\allowbreak\; g(l') = \None$.
Second, the state evolves in a monotone way: nodes are only added to
the CFG, but an already filled node is never modified; likewise,
temporary registers are never reused.  (If this were not the case,
correct sub-graphs constructed by recursive invocations to the {\tt transl}
functions could become incorrect after later modifications of the CFG.)
We define this monotonicity property as a partial order~$\preceq$ between
states:
$$ (g_1, l_1, r_1) \preceq (g_2, l_2, r_2) \defequal
   l_1 \le l_2 \wedge r_1 \le r_2 \wedge 
   (\forall l,i,~ g_1(l) = \Some{i} \Rightarrow g_2(l) = \Some{i}) $$
It is straightforward but tedious to show that these two invariants
are satisfied by the monadic translation functions, since they hold
for the basic operations of the monad and are preserved by monadic
{\tt bind} composition.  However, we can avoid much proof effort by
taking advantage of Coq's dependent types.  The first invariant (of
one state) can be made a part of the state itself, which
becomes a dependent record type:
\begin{verbatim}
Record state: Set := mkstate {
  st_nextreg: reg;
  st_nextnode: node;
  st_code: graph;
  st_wf: forall (l: node), l >= st_nextnode -> st_code!l = None
}.
\end{verbatim}
Note the $4\th$ field {\tt st{\char95}wf}, which is a proof term that the
first invariant holds.
The second invariant (the partial order $\preceq$) is more difficult, as
it involves two states.  However, it can be expressed in the
definition of the ${\tt mon}(\alpha)$ type by turning the function type
${\tt state} \rightarrow \ldots$ into a dependent function type of the shape
${\Pi (x:{\tt state}).} \allowbreak\; \{ y \mid P~x~y \}$.
Here is the Coq definition of the
dependently-typed monad:
\begin{alltt}
Inductive res (A: Set) (s: state): Set :=
  | Error: res A s
  | OK: A -> forall (s': state), s \(\preceq\) s' -> res A s.
Definition mon (A: Set) : Set := forall (s: state), res A s.
\end{alltt}
The result of a successful monadic computation of type ${\tt mon}(\alpha)$
starting in state $s$ is ${\tt OK}(x, s', \pi)$ where $x: \alpha$ is the
return value, $s'$ the final state, and $\pi$ a proof term for the
proposition $s \preceq s'$.  The two invariants need to be proved when
defining the basic operations of the monad; for instance, in the case
of {\tt ret} and {\tt bind}, the corresponding proofs amount to reflexivity
and transitivity of $\preceq$.  However, they then
automatically hold for all
computations in the dependently-typed monad.
\iffalse
Here is the statement of correctness of the translation functions with
respect to the relational specification.  We say that a register $r$ is
valid in a state $\sigma = (g, l, r')$, and write $\sigma \models r$,
if $r < r'$, i.e.~$r$ was previously allocated, and similarly for a
list of registers.
\begin{lemma}
\begin{itemize}
\item If ${\tt transl{\char95}expr}(\gamma, a, r_d, l_e)(\sigma) =
          {\tt OK}(l_1, \sigma', \_)$ 
  and $\sigma \models \Range(\gamma) \union \pi \union \{r_d\}$
  and either $a = \id \wedge \gamma(\id) = \Some{r_d}$ or 
  $r_d \notin \Range(gamma) \union \pi$,
  then ${\tt st{\char95}code}(\sigma'), \gamma, \pi \vdash a{\tt \ in\ }r \match l_1, l_e$.
\item If ${\tt transl{\char95}stmt}(\gamma, s, l_2, \lexit, \lgoto)(\sigma) =
          {\tt OK}(l_1, \sigma', \_)$ 
  and $\sigma \models \Range(\gamma)$,
  then ${\tt st{\char95}code}(\sigma'), \gamma \vdash s \match l_1, l_2, \lexit, \lgoto$.
\end{itemize}
\end{lemma}
(We omit two similar statements for {\tt transl{\char95}exprlist} and
{\tt transl{\char95}condition}.)  The proof proceeds by structural induction over
$a$ and $s$, exploiting the fact that the relational specifications of
the translation are stable under the addition of fresh nodes to the CFG.
\fi

\section{Optimizations based on dataflow analysis} \label{s:dataflow}

We now describe two optimization passes performed on the \lang{RTL}
form: constant propagation and common subexpression elimination.  Both
passes make use of a generic solver for dataflow inequations, which we
describe first.

\subsection{Generic solvers for dataflow inequations} \label{s:dataflow-solvers}

We formalize forward dataflow analyses as follows.  We are given a
control-flow graph (as a function ${\tt successors}: {\tt node} \fun
{\tt list}({\tt node})$) and a transfer function $T: {\tt node} \times {\cal A}
\fun {\cal A}$, where $\cal A$ is the type of abstract values (the results
of the analysis), equipped with a partial order $\ge$.  Intuitively,
$T(l)$ computes the abstract value ``after'' the instruction at point $l$
as a function of the abstract value ``before'' this instruction.
We are also given a set ${\tt cstrs}$ of pairs of a CFG node
and an abstract value $a : \cal A$, representing e.g.~requirements on the CFG
entry point.  The result of forward dataflow analysis is a solution
$A: {\tt node} \fun {\cal A}$ to the following forward dataflow inequations:
\begin{eqnarray*}
A(s) & \ge & T(l, A(l)) \mbox{ for all $s \in {\tt successors}(l)$} \\
A(l) & \ge & a          \mbox{ for all $(l, a) \in {\tt cstrs}$}
\end{eqnarray*}
We formalize dataflow analysis as inequations instead of the usual
equations $A(l) = \bigsqcup \, \{ T(p, A(p)) \mid l \in {\tt successors}(p) \}$
because we are interested only in the correctness of the solutions,
not in their optimality.

Two solvers for dataflow inequations are provided as Coq functors,
that is, modules parameterized by a module defining the type $\cal A$ and its
operations.  The first solver implements Kildall's
worklist algorithm \cite{Kildall-73}.  It is applicable if the type
$\cal A$ is equipped with a decidable equality, a
least element $\bot$ and an upper bound operation $\sqcup$.  (Again,
since we are not interested in optimality of the results, $\sqcup$ is
not required to compute the least upper bound.)  The second solver
performs propagation over extended basic blocks, setting $A(l) = \top$
in the solution for all points $l$ that have several predecessors.
The only requirement over the type $\cal A$ is that it possesses a
greatest element $\top$.  This propagation-based solver is useful in
cases where upper bounds do not always exist or are too expensive to
compute.

Several mechanical verifications of Kildall's worklist algorithm have
been published already 
\cite{Barthe-02,Coupet-Grimal-Delobel-05,Klein-Nipkow-jinja,%
Bertot-Gregoire-Leroy-05}.
Therefore, we do not detail the correctness proofs for our solvers,
referring the reader to the cited papers and to the Coq development
for more details.

The solvers actually return an option type, with $\None$ denoting
failure and $\Some A$ denoting success with solution $A$.  For
simplicity, we do not require that the CFG is finite, nor in the case
of Kildall's algorithm that the $\ge$ ordering over $A$ is well
founded (no infinite ascending chains).  Consequently, we cannot
guarantee termination of the solvers and must allow them to fail, at
least formally.  The implementations of the two solvers bound the
number of iterations performed to find a fixed point, returning
$\None$ if a solution cannot be found in $N$ iterations, where $N$ is
a very large constant.  Alternatively, unbounded iteration can be
implemented using the approach of Bertot and Komendantsky \cite{Bertot-Komendantsky}, which uses
classical logic and Tarski's theorem to model general recursion in
Coq.  Yet another alternative is to
use the ``verified validator'' approach, where the computation of the
solution is delegated to external, untrusted code, then verified a
posteriori to satisfy the dataflow inequations.
In all these approaches, if the static analysis fails, we can either
abort the compilation process or simply turn off the corresponding
optimization pass, returning the input code unchanged.

Solvers for backward dataflow inequations can be easily derived from
the forward solvers by reversing the edges of the control-flow graph.
In the backward case, the transfer function
$T(l)$ computes the abstract value ``before'' the instruction at point $l$
as a function of the abstract value ``after'' this instruction.  The
solution $A$ returned by the solvers satisfies the backward dataflow
inequations:
\begin{eqnarray*}
A(l) & \ge & T(s, A(s)) \mbox{ for all $s \in {\tt successors}(l)$} \\
A(l) & \ge & a          \mbox{ for all $(l, a) \in {\tt cstrs}$}
\end{eqnarray*}

\subsection{Constant propagation} \label{s:constprop}

\subsubsection{Static analysis}

Constant propagation for a given function starts with a forward
dataflow analysis using the following domain of abstract values:
$$ {\cal A} \defequal r \mapsto (\top \alt \bot \alt {\tt Int}(n) \alt
{\tt Float}(f) \alt {\tt Addrsymbol}(\id + \ofs) $$
That is, at each program point and for each register $r$, we record
whether its value at this point is statically known to be equal to an
integer $n$, a float $f$, or the address of a symbol $\id$ plus an
offset $\delta$, or is unknown ($\top$), or whether this
program point is unreachable ($\bot$).  Kildall's algorithm is used to
solve the dataflow inequations, with the additional constraint that
$A(\fn.{\tt entrypoint}) \ge (r \mapsto \top)$.
If it fails to find a proper solution $A$, we take the trivial
solution $A(l) = (r \mapsto \top)$, effectively turning off the
optimization.  For each function $\fn$ of the program, we write
${\tt analyze}(\fn)$ for the solution (proper or trivial) of the dataflow
equations for $\fn$.

The transfer function $T_\fn$ is the obvious abstract interpretation of
\lang{RTL}'s semantics on this domain:
$$
T_\fn(l, a) = \fourcases
  {a\{r \becomes \overline{{\tt eval{\char95}op}}(\op, a(\vec r))\}}
      {if $\fn.{\tt code}(l) = \Some{{\tt op}(\op, \vec r, r, l')}$}
  {a\{r \becomes \top\}}
      {if $\fn.{\tt code}(l) = \Some{{\tt load}(\chunk, \mode, \vec r, r, l')}$}
  {a\{r \becomes \top\}}
      {if $\fn.{\tt code}(l) = \Some{{\tt call}(\sig, \_, r, l')}$}
  {a}
      {otherwise}
$$
Here, $\overline{{\tt eval{\char95}op}}$ is the abstract interpretation over the
domain $\cal A$ of the {\tt eval{\char95}op} function defining the semantics of
operators.  By lack of an alias analysis, we do not attempt to track
constant values stored in memory; therefore, the abstract value of the
result of a {\tt load} is $\top$.

\subsubsection{Code transformation} 

The code transformation exploiting the results of this analysis is
straightforward: {\tt op} instructions become ``load constant''
instructions if the values of all argument registers are statically
known; {\tt cond} instructions where the condition can be statically
determined to be always {\tt true} or always {\tt false} are turned into {\tt nop}
instructions branching to the appropriate successor; finally,
operators, conditions and addressing modes are specialized to cheaper
immediate forms if the values of some of their arguments are
statically known.  The structure of the control-flow graph is
preserved (no node is inserted nor deleted), making this transformation
easy to express as a morphism over the CFG.  Parts of the CFG can
become unreachable as a consequence of statically resolving some
{\tt cond} instructions.  The corresponding instructions, as well as the
{\tt nop} instructions that replace the statically-resolved {\tt cond}
instructions, will be removed later during branch tunneling and CFG
linearization (sections~\ref{s:tunneling} and~\ref{s:linearize}).

\subsubsection{Semantic preservation} \label{s:constprop-preservation}

The proof of semantic preservation for constant propagation is based
on a ``lock-step'' simulation diagram.  The central invariant of the
diagram is the following: at every program point $l$ within a function
$\fn$, the concrete values $R(r)$ of registers $r$ must agree with the
abstract values ${\tt analyze}(\fn)(l)(r)$ predicted by the dataflow
analysis.  Agreement between a concrete and an abstract value is
written $\models a : v$ and defined as follows.
\begin{pannel}
\srule{ \models v : \top }
\srule{ \models {\tt int}(n) : {\tt Int}(n) }
\srule{ \models {\tt float}(f) : {\tt Float}(f) }
\irule{
{\tt symbol}(G, \id) = \Some b
}{
\models {\tt ptr}(b,\delta) : {\tt Addrsymbol}(\id + \delta) 
}
\end{pannel}
We write $\models R : A$ to mean $\models R(r) : A(r)$ for all registers $r$.
The first part of the proof shows the correctness of the abstract
interpretation with respect to this agreement relation.  For example,
if ${\tt eval{\char95}op}(\op, \vec v) = \Some v$ and
$\models \vec v: \vec a$, we show that $\models v : \overline{{\tt eval{\char95}op}}(\op, \vec a)$.
Likewise, we show that the specialized forms of operations, addressing
modes and conditions produced by the code transformation compute the
same values as the original forms, provided the concrete arguments
agree with the abstract values used to determine the specialized
forms.  These proofs are large but straightforward case analyses.
We then define the relation between pairs of \lang{RTL} states that
is invariant under execution steps.
\begin{pannel}
\irule{
\tr{\fn} = \Some{\fn'} \amper \stk \match \stk' \amper \models R : {\tt analyze}(\fn)(l)
}{
\state(\stk, \fn.{\tt code}, \sp, l, R, M) \match
\state(\stk', \fn'.{\tt code}, \sp, l, R, M)
}
\irule{
\tr{\fd} = \Some{\fd'} \amper \stk \match \stk'
}{
\callstate(\stk, \fd, \vec v, M) \match
\callstate(\stk', \fd', \vec v, M)
}
\irule{
\stk \match \stk'
}{
\callstate(\stk, v, M) \match
\callstate(\stk', v, M)
}
\\
\srule{ \epsilon \match \epsilon }
\irule{
\tr{\fn} = \Some{\fn'} \amper \stk \match \stk' \amper
\forall v,~ \models R\{r \becomes v\} : {\tt analyze}(\fn)(l)
}{
\stackframe(r, \fn.{\tt code}, \sp, l, R).\stk \match
\stackframe(r, \fn'.{\tt code}, \sp, l, R).\stk'
}
\end{pannel}
In the rule for stack frames, $v$ stands for the value
that will eventually be returned to the pending function call.  Since
it is unpredictable, we quantify universally over this return value.
Semantic preservation for constant propagation follows by
theorem~\ref{th:lockstep-simulation} from the simulation lemma below.
\begin{lemma} \label{l:constprop-simu}
If $G \vdash \st_1 \redtrace{t} \st_2$ and $\st_1 \match \st_1'$,
there exists $\st_2'$ such that $G' \vdash \st_1' \redtrace{t} \st_2'$
and $\st_1' \match \st_2'$,
\end{lemma}
\subsection{Common subexpression elimination} \label{s:CSE}
\subsubsection{Static analysis}
Common subexpression elimination is implemented via local value
numbering performed over extended basic blocks.  Value numbers $x$ are
identifiers representing runtime values abstractly.
A forward dataflow analysis associates to each program point $l, \fn$
a pair $(\phi,\eta)$ of a partial mapping $\phi$ from registers to
value numbers and a set $\eta$ of equations between value numbers of
the form $x = \op(\vec x)$ or $x = \chunk,\mode(\vec x)$.
In a sense that will be made semantically precise below, the first
equation form means that the operation $\op$ applied to concrete
values matching the value numbers $\vec x$ returns a concrete value
that matches $x$; likewise, the second equation means that computing
an address using the addressing mode $\mode$ applied to values
matching $\vec x$ and loading a quantity $\chunk$ from this address
in the current memory state returns a value matching $x$.
In addition to $\phi$ and $\eta$, the analysis result at each point
also contains a supply of fresh value numbers and a proof that value
numbers appearing in $\phi$ and $\eta$ are not fresh with respect to
this supply.  We omit these two additional components for simplicity.
The transfer function $T_\fn(l, A)$, where $A = (\phi,\eta)$ is the
analysis result ``before'' point $l$, is defined as follows.
If the instruction at $l$ is a move ${\tt op}({\tt move}, r_s, r_d, l')$, we
record the equality $r_d = r_s$ by returning $(\phi\{r_d \becomes
\phi(r_s)\}, \eta)$.  If the instruction at $l$ is ${\tt op}(\op, \vec r,
r, l')$, we determine the value numbers $\vec x = \phi(\vec r)$,
associating fresh value numbers to elements of $\vec r$ if needed.  We
then check whether $\eta$ contains an equation of the form $x =
\op(\vec x)$ for some $x$.  If so, the computation performed by this
{\tt op} operation has already been performed earlier in the program and
the transfer function returns $(\phi\{r \becomes x\}, \eta)$.  If not,
we allocate a fresh value number $x$ to stand for the result of this
operation and return $(\phi\{r \becomes x\}, \eta \union \{ x =
\op(\vec x) \})$ as the result of the transfer function.
The definition of the transfer function is similar in the case of
{\tt load} instructions, using memory equalities $x = \chunk,\mode(\vec x)$
instead of arithmetic equalities $x = \op(\vec x)$.  For {\tt store}
instructions, in the absence of nonaliasing information we must assume
that the store can invalidate any of the memory equations that
currently hold.  The transfer function therefore removes all such
equations from $\eta$. For {\tt call} instructions, we remove all
equations, since eliminating common subexpressions across function
calls is often detrimental to the quality of register allocation.
Other instructions keep $(\phi, \eta)$ unchanged.
We say that a value numbering $(\phi, \eta)$ is satisfiable in a
register state $R$ and a memory state $M$, and we write $R, M \models \phi,
\eta$, if there exists a valuation $V$ associating concrete values
to value numbers such that
\begin{enumerate}
\item If $\phi(r) = \Some x$, then $R(r) = V(x)$.
\item If $\eta$ contains the equation $x = \op(\vec x)$, then
  ${\tt eval{\char95}op}(\op, V(\vec x)) = \Some{V(x)}$.
\item If $\eta$ contains the equation $x = \chunk,\mode(\vec x)$, then
  there exists $b$ and $\delta$ such that
  ${\tt eval{\char95}addressing}(\mode, V(\vec x)) = \Some{{\tt ptr}(b,\delta)}$
  and ${\tt load}(\chunk, M, b, \delta) = \Some{V(x)}$.
\end{enumerate}
Using this notion of satisfaction, we order value numberings by
entailment: $(\phi', \eta') \ge (\phi, \eta)$ if 
$R, M \models \phi, \eta \Longrightarrow R, M \models \phi', \eta'$ for all $R, M$.
Least upper bounds for this ordering are known to be difficult to
compute efficiently \cite{Kildall-73,Gulwani-Necula-04}.  We sidestep
this issue by solving the dataflow inequation using the approximate
solver based on propagation described in
section~\ref{s:dataflow-solvers} instead of Kildall's algorithm.
\subsubsection{Code transformation}
The actual elimination of common subexpressions is simple.  Consider
an instruction $i$ at point $l$, and let $(\phi,\eta)$ be the results
of the static analysis at point $l$.  If $i$ is
${\tt op}(\op, \vec r, r, l')$ or ${\tt load}(\chunk,\mode, \vec r, r, l')$,
and there exists a register $r'$ such that 
the equation $\phi(r') = \op(\phi(\vec r))$ or respectively
$\phi(r') = \chunk,\mode(\phi(\vec r))$ is in $\eta$,
we rewrite the instruction as a move ${\tt op}({\tt move}, r, r', l')$ from
$r'$ to $r$.  This eliminates the redundant computation, reusing
instead the result of the previous equivalent computation, which is
still available in register $r'$.  In all other cases, the instruction
is unchanged.
\subsubsection{Semantic preservation}
The correctness proof for common subexpression elimination closely
follows the pattern of the proof for constant propagation (see
section~\ref{s:constprop-preservation}), using a lock-step simulation
diagram.  The only difference is the replacement of the hypothesis 
$\models R : A(l)$ in the invariant between states by the hypothesis
$R, M \models A(l)$, meaning that at every program point the current
register and memory states must satisfy the value numbering obtained
by static analysis.

\section{Register allocation} \label{s:regalloc}

The next compilation pass performs register allocation by coloring of
an interference graph.  

\subsection{The target language: LTL} \label{s:ltl}

The target language for register allocation is a variant of \lang{RTL}
called \lang{LTL} (Location Transfer Language).  Syntactically, the
only difference between \lang{LTL} and \lang{RTL} is that the
\lang{RTL} pseudo-registers $r$ appearing as arguments and results of
instructions are replaced by locations $\loc$.  A location is
either a hardware register or an abstract designation of a stack slot.
\begin{syntaxleft}
\syntaxclass{Locations:}
\loc & ::= & \mreg \alt \slot
\syntaxclass{Machine registers:}
\mreg & ::= & {\tt R3} \alt {\tt R4} \alt \ldots & PowerPC integer registers\\
      & \alt& {\tt F1} \alt {\tt F2} \alt \ldots & PowerPC float registers
\syntaxclass{Stack slots:}
\slot & ::= & {\tt local}(\ty, \ofs) & local variables \\
      & \alt& {\tt incoming}(\ty,\ofs) & incoming parameters \\
      & \alt& {\tt outgoing}(\ty,\ofs) & outgoing arguments
\end{syntaxleft}%
In stack slots, $\ty$ is the intended type of the slot ({\tt int} or
{\tt float}) and $\ofs$ an integer representing a word offset in one of
the three areas of the activation record ({\tt local}, {\tt incoming} and
{\tt outgoing}).  

In \lang{LTL}, stack slots are not yet mapped to actual memory
locations.  Their values, along with those of the machine registers,
are recorded in a mapping $L: {\tt loc} \fun {\tt val}$ similar to the mapping
$R: {\tt reg} \fun {\tt val}$ used in the \lang{RTL} semantics and disjoint
from the memory state $M$.  However, the \lang{LTL} semantics treats
locations in a way that anticipates their behaviors once they are later
mapped to actual memory locations and processor registers.
In particular, we account for the possible overlap between distinct stack
slots once they are mapped to actual memory areas.  For instance,
${\tt outgoing}({\tt float},0)$ overlaps with ${\tt outgoing}({\tt int},0)$ and
${\tt outgoing}({\tt int},1)$: assigning to one of these locations invalidates
the values of the other two.  This is reflected in the weak ``good
variable'' property for location maps $L$:
$$ (L\{\loc_1 \becomes v\})(\loc_2) =
   \threecases{v}{if $\loc_1 = \loc_2$;}
              {L(\loc_2)}{if $\loc_1$ and $\loc_2$ do not overlap;}
              {{\tt undef}}{if $\loc_1$ and $\loc_2$ partially overlap.} $$
Contrast with the standard ``good variable'' property for register
maps $R$:
$$ (R\{r_1 \becomes v\})(r_2) =
   \twocases{v}{if $r_1 = r_2$;}
               {R(r_2)}{if $r_1 \not= r_2$.} $$

\begin{figure}

\begin{pannel}
\irule{
\fn.{\tt code}(l) = \Some{{\tt call}(\sig, \loc_f, \vec \loc, \loc, l')} \amper
L(\loc_f) = {\tt ptr}(b, 0) \amper {\tt funct}(G, b) = \Some \fd \amper \fd.{\tt sig} = \sig
}{
\trans {\state(\stk, \fn, \sp, l, L, M)}
       {\callstate(\stackframe(\loc, \fn, \sp, l', {\tt postcall}(L)).\stk,
                   \fd, L(\vec \loc), M)}
}
\end{pannel}
\caption{Semantics of \lang{LTL}.  The transitions not shown are similar
  to those of \lang{RTL}. } \label{f:ltl-sem}
\end{figure}

The dynamic semantics of \lang{LTL} is illustrated in
figure~\ref{f:ltl-sem}.  Apart from the use of location maps $L$ and
overlap-aware update instead of register maps $R$ and normal updates, the only
significant difference with the semantics of \lang{RTL} is the
semantics of {\tt call} instructions.  In preparation for enforcing
calling conventions as described later in section~\ref{s:reload},
processor registers that are temporary or caller-save according to the
calling conventions are set to {\tt undef} in the location state of the
caller, using the function {\tt postcall} defined as
$$ {\tt postcall}(L)(\loc) = \twocases
      {{\tt undef}}{if $\loc$ is a temporary or caller-save register;}
      {L(\loc)}{otherwise.}
$$
This forces the \lang{LTL} producer, namely the register
allocation pass, to ensure that no value live across a function call
is stored in a caller-save register.

\subsection{Code transformation} \label{s:regalloc-transf}

For every function, register allocation is performed in four
steps, which we now outline.

\subsubsection{Type reconstruction for RTL}
The first step performs type reconstruction for the \lang{RTL} source
function in a trivial ``{\tt int}-or-{\tt float}'' type system similar to that
of \lang{Cminor} (see section~\ref{s:cminor-typing}).
To each \lang{RTL} pseudo-register we assign a type {\tt int}
or {\tt float}, based on its uses within the function.  The resulting type
assignment $\Gamma: r \mapsto \tau$ will guide the register allocator,
making sure a machine register or stack location of the correct kind
is assigned to each pseudo-register.

All operators, addressing modes and conditions are monomorphic
(each of their arguments has only one possible type, either {\tt int}
or {\tt float}), except the {\tt move} operation which is polymorphic with
type $\forall \tau. ~ \tau \fun \tau$.  Type reconstruction can
therefore be performed by trivial unification.  We use the verified
validator approach: a candidate type assignment $\Gamma$ is computed by
untrusted Caml code, using in-place unification, then verified for
correctness by a simple type-checker, written and proved correct in
Coq.

\subsubsection{Liveness analysis}
We compute the set $A(l)$ of pseudo-registers live
``after'' every program point $l$.  These sets are a solution of the
following backward dataflow inequations:
\begin{eqnarray*}
A(l) & \supseteq & T(s, A(s)) \mbox{ for all $s$ successor of $l$}
\end{eqnarray*}
The transfer function $T$ computes the set of live pseudo-registers
``before'' an instruction, as a function of this instruction and the
set of live pseudo-registers ``after'' it.  Classically, it removes
the registers defined by the instruction, then adds the registers
used.  A special case is made for {\tt op} or {\tt load} instructions whose
result register is not live ``after'' the instruction, since these
instructions will later be eliminated as dead code.
For instance, if the instruction at $l$ in $F$ is
${\tt op}(\op, \vec r, r, l')$, then
$$T(l, A) = \twocases
   {(A \setminus \{r\}) \union \vec r}
   {if $r \in A$;}
   {A}
   {if $r \notin A$.}
$$
Liveness information is computed by Kildall's algorithm, using the
generic dataflow solver described in section~\ref{s:dataflow-solvers}.

\subsubsection{Construction of the interference graph}
Based on the results of liveness analysis, an interference graph is
built following Chaitin's construction \cite{Chaitin-82}.  Two kinds
of interferences are recorded: between two pseudo-registers $(r, r')$
and between a pseudo-register $r$ and a machine register $\mreg$.  To
enable coalescing during allocation, register affinities arising from
moves (either explicit or implicit through calling conventions) are
also recorded.  (Affinities do not affect the correctness of the
generated code but have a considerable impact on its performance.)

The interference graph is represented by sets of edges: either
unordered pairs $(r, r')$ or ordered pairs $(r, \mreg)$.  The graph is
constructed incrementally by enumerating every \lang{RTL} instruction
and adding interference edges between the defined register $r$ (if
any) and the registers $r' \in A(l) \setminus \{r\}$ live ``across''
the instruction.  For move instructions
$l: {\tt op}({\tt move}, r_s, r_d, l')$, we avoid adding an edge between $r_s$
and $r_d$, as proposed by Chaitin \cite{Chaitin-82}.  Finally, for call
instructions, additional interference edges are introduced pairwise
between the registers live across the call and the caller-save machine
registers.

\subsubsection{Coloring of the interference graph}
To color the interference graph, we use the iterated coalescing
algorithm of George and Appel \cite{Appel-George}.  The result is a mapping
$\Phi: r \mapsto \loc$ from pseudo-registers to locations.  Following
once more the verified validator approach, the actual coloring is
performed by an untrusted implementation of the George-Appel algorithm
written in Caml and using imperative doubly-linked lists for
efficiency.  The candidate coloring returned is then verified by a
simple validator written and proved correct in Coq.  Like many
NP-complete problems, graph coloring is a paradigmatic example of an
algorithm that is significantly easier to validate a posteriori than
to prove correct.  Validation proceeds by enumerating the nodes and
edges of the interference graph, checking the following properties:
\begin{enumerate}
\item Correct colors: $\Phi(r) \not= \Phi(r')$ for all edges $(r, r')$
  of the interference graph; likewise, $\Phi(r) \not= \mreg$ for all
  interference edges $(r, m)$.
\item Register class preservation: the type of the location $\Phi(r)$
  is equal to $\Gamma(r)$ for all pseudo-registers $r$.
\item Validity of locations: for all $r$, the location $\Phi(r)$ is
  either a {\tt local} stack slot or a non-temporary machine register%
\footnote{We reserve 2 integer and 3 float machine registers as
  temporaries to be used later for spilling and reloading.  
  These temporary registers must therefore not be used by the register
  allocator.}.
\end{enumerate}

\subsubsection{LTL generation}
Finally, the actual code transformation from \lang{RTL} to \lang{LTL}
is a trivial per-instruction rewriting of the CFG where each mention
of a pseudo-register $r$ is replaced by the location $\Phi(r)$
allocated to $r$.  For instance, the \lang{RTL} instruction
$l: {\tt op}(\op, \vec r, r, l')$ becomes the \lang{LTL} instruction
$l: {\tt op}(\op, \Phi(\vec r), \Phi(r), l')$.  There are two exceptions
to this rule.   First,
a move instruction $l: {\tt op}({\tt move}, r_s, r_d, l')$ such that $\Phi(r_s) =
\Phi(r_d)$ is turned into a no-operation $l: {\tt nop}(l')$,
therefore performing one step of coalescing.
Second, an {\tt op} or {\tt load} instruction whose result
register is not live after the instruction is similarly turned into a
{\tt nop} instruction, therefore performing dead-code elimination.

\subsection{Semantic preservation}

The proof that register allocation preserves program behaviors is,
once more, based on a lock-step simulation diagram.  The invariant
between states is, however, more complex than those used for constant
propagation or common subexpression elimination.  For these two
optimizations, the register state $R$ was
identical between matching states, because the same values would be
computed (by possibly different instructions) in the original and
transformed program and would be stored in the same registers.  This
assumption no longer holds in the case of register allocation.

A value computed by the original \lang{RTL} program and stored
in pseudo-register $r$ is stored in location $\Phi(r)$ in the
transformed \lang{LTL} program.  Naively, we could relate the
\lang{RTL} register state $R$ and the \lang{LTL} location state $L$ by
$R(r) = L(\Phi(r))$ for all pseudo-registers $r$.  However, this
requirement is too strong, as it essentially precludes sharing a
location between several pseudo-registers.

To progress towards the
correct invariant, consider the semantic interpretation of live and
dead pseudo-registers.  If a pseudo-register $r$ is dead at point $l$
in the original \lang{RTL} code, then its value $R(r)$ has no impact
on the remainder of the program execution: either $r$ will never be
used again, or it will be redefined before being used; in either case,
its value $R(r)$ could be replaced by any other value without any harm.  
A better relation between the values of pseudo-registers $R$ and
locations $L$ at point $l$ is therefore
$$ R(r) = L(\Phi(r)) \mbox{ for all pseudo-registers $r \in T(l, A(l))$} $$
In other words, at each program point $l$, register allocation must
preserve the values of all registers live ``before'' executing the
instruction at this point.  This property that we have never seen
spelled out explicitly in compiler literature captures concisely and
precisely the essence of register allocation.

The invariant between \lang{RTL} and \lang{LTL} states is, then:
\begin{pannel}
\irule{
\stk \match \stk' \amper
\tr{\fn} = \Some{\fn'} \amper
R(r) = L(\Phi(r)) \mbox{ for all $r \in T(l, A(l))$} \\
{\tt typecheck}(\fn) = \Some{\Gamma} \amper {\tt analyze}(\fn) = \Some{A} \amper
{\tt regalloc}(\fn, A, \Gamma) = \Some{\Phi}
}{
\state(\stk, \fn.{\tt code}, \sp, l, R, M) \match
\state(\stk', \fn', \sp, l, L, M)
}
\irule{
\stk \match \stk' \amper
\tr{\fd} = \Some{\fd'} \amper
}{
\callstate(\stk, \fd, \vec v, M) \match
\callstate(\stk', \fd', \vec v, M)
}
\irule{
\stk \match \stk'
}{
\returnstate(\stk, v, M) \match
\returnstate(\stk', v, M)
}
\end{pannel}
The invariant relating frames in the call stacks is similar to that
relating regular states, with a universal quantification on the return
value, as in section~\ref{s:constprop-preservation}.
The proof of the lock-step simulation diagram makes heavy use of the
definition of the transfer function $T$ for liveness analysis,
combined with the following characterization of the register
allocation $\Phi$ with respect to the results $A$ of liveness analysis:
\begin{itemize}
\item For a move instruction $l: {\tt op}({\tt move}, r_s, r_d, l')$, we have
$\Phi(r_d) \not= \Phi(r')$ for all $r' \in A(l) \setminus \{r_s, r_d\}$.
\item For other {\tt op}, {\tt load} or {\tt call} instructions at $l$ with
  destination register $r$, we have
$\Phi(r) \not= \Phi(r')$ for all $r' \in A(l) \setminus \{r\}$.
\item For call instructions $l: {\tt call}(\sig, (r_f \alt \id), \vec r, r_d, l'$),
we additionally have $\Phi(r') \not= \mreg$ for all
$r' \in A(l) \setminus \{r_d\}$ and all callee-save registers $\mreg$.
\end{itemize}

\section{Branch tunneling and no-op elimination} \label{s:tunneling}

Register coalescing and dead-code elimination, performed in the course
of register allocation, generate a number of {\tt nop} instructions.  Now
is a good time to ``short-circuit'' them, rendering them
unreachable and making them candidates for removal during CFG linearization
(section~\ref{s:linearize}).  Since {\tt nop} in a CFG representation
also encodes unconditional branches, this transformation also performs
branch tunneling: the elimination of branches to branches.

\subsection{Code transformation}

Branch tunneling rewrites the \lang{LTL} control-flow graph, replacing
each successor, i.e.~the $l'$ and $l''$ in instructions such as
${\tt op}(\op, \vec \loc, \loc, l')$ or
${\tt cond}(\cond, \vec \loc, l', l'')$, by its effective
destination $D_\fn(l')$.  
Naively, the effective destination is
computed by chasing down sequences of {\tt nop} instructions, stopping at
the first non-{\tt nop} instruction:
\begin{eqnarray*}
D_\fn(l) & = &
    \twocases{D_\fn(l')}{if $\fn.{\tt code}(l) = \Some{{\tt nop}(l')}$;}
             {l}{otherwise.}
\end{eqnarray*}
This is not a proper definition: if the control-flow graph contains
cycles consisting only of {\tt nop} instructions, such as $l : {\tt nop}(l)$,
the computation of $D_\fn(l)$ fails to terminate.

A simple solution is to bound the recursion depth when computing
$D_\fn(l)$, returning~$l$ when the counter reaches~0.  The initial
value~$N$ can be chosen at will; the number of instructions in the
code of function $\fn$ is a good choice.

A more elegant solution, suggested by an anonymous reviewer, uses a
union-find data structure $U_\fn$.  A first scan of the control-flow graph
populates $U_\fn$ by adding an edge from $l$ to $l'$ for each instruction
$l: {\tt nop}(l')$, provided $l$ and $l'$ are not already in the same
equivalence class.  The effective destination $D_\fn(l)$, then,
is defined as the canonical representative of $l$ in the union-find
structure $U_\fn$.

\subsection{Semantic preservation}

Semantic preservation for branch tunneling follows from a simulation
diagram of the ``option'' kind (see figure~\ref{f:simu-diagrams}).
Intuitively, the execution of a non-call, non-return instruction that
causes a transition from point $l_1$ to point $l_2$ in the original
code corresponds to the execution of zero or one instructions in the
tunneled code, from point $D_\fn(l_1)$ to point $D_\fn(l_2)$.  The
``zero'' case can appear for example when the instruction at $l_1$
is ${\tt nop}(l_2)$.  The definition of the invariant between execution
states is therefore
\begin{pannel}
\irule{
\tr{\fn} = \Some{\fn'} \amper \stk \match \stk'
}{
\state(\stk, \fn, \sp, l, L, M) \match
\state(\stk, \fn', \sp, D_\fn(l), L, M)
}
\end{pannel}
for regular states, with the obvious definitions for call states,
return states, and stack frames.
For simulation diagrams of the ``option'' kind, we must provide a
measure to show that it is not possible to take infinitely many ``or
zero'' cases.  For branch tunneling, a suitable measure is the number
of {\tt nop} instructions that are skipped starting with the current
program point:
\begin{eqnarray*}
\size{\state(\stk, \fn, sp, l, L, M)} & = & \#{\tt nop}_\fn(l) \\
\#{\tt nop}_\fn(l) & = &
    \twocases{1 + \#{\tt nop}_\fn(l')}{if $\fn.{\tt code}(l) = \Some{{\tt nop}(l')}$
                                   and $D_\fn(l) \not= l$;}
             {0}{otherwise.}
\end{eqnarray*}
This definition is well founded because of the way $D_\fn$ is
constructed from the union-find structure $U_\fn$.  The crucial
property of the $D$ and $\#{\tt nop}$ functions is that, for all program
points $l$,
$$ D_\fn(l) = l ~\vee~
   \exists l',~ \fn.{\tt code}(l) =  \Some{{\tt nop}(l')}
                \wedge D_\fn(l') = D_\fn(l)
                \wedge \#{\tt nop}_\fn(l') < \#{\tt nop}_\fn(l)
$$
\begin{theorem}
If $\st_1 \match \st_2$ and $G \vdash \st_1 \redtrace{t} \st_1'$, either
there exists $\st_2'$ such that $G' \vdash \st_2
  \redtrace{t} \st_2'$ and $\st_1' \match \st_2'$,
or $\size{\st_1'} < \size{\st_1}$ and $t = \epsilon$ and
  $\st_1' \match \st_2$.
\end{theorem}

\section{Linearization of the control-flow graph} \label{s:linearize}

The next compilation step linearizes the control-flow graphs of
\lang{LTL}, replacing them by lists of instructions with
explicit labels and unconditional and conditional branches to labels,
in the style of assembly code.  While the CFG representation of
control is very convenient for performing dataflow analyses, the
linearized representation makes it easier to insert new
instructions, as needed by some of the subsequent passes.  

Discussions of linearization in the literature focus on trace picking
heuristics that reduce the number of jumps introduced, but consider
the actual production of linearized code trivial.  Our first attempts
at proving directly the correctness of a trace picking algorithm that
builds linearized code and performs branch tunneling on the fly showed
that this is not so trivial.  We therefore perform tunneling in a
separate pass (see section~\ref{s:tunneling}), then implement CFG
linearization in a way that clearly separates heuristics from the
actual production of linearized code.

\subsection{The target language: LTLin} \label{s:ltlin}

The target language for CFG linearization is \lang{LTLin}, a variant
of \lang{LTL} where control-flow graphs are replaced by lists of
instructions.
\begin{syntax}
\syntaxclass{LTLin instructions:}
i & ::=  & {\tt op}(\op, \vec \loc, \loc) & arithmetic operation \\
  & \alt & {\tt load}(\chunk, \mode, \vec \loc, \loc) & memory load \\
  & \alt & {\tt store}(\chunk, \mode, \vec \loc, \loc) & memory store \\
  & \alt & {\tt call}(\sig, (\loc \alt \id), \vec \loc, \loc) & function call \\
  & \alt & {\tt tailcall}(\sig, (\loc \alt \id), \vec \loc) & function tail call \\
  & \alt & {\tt cond}(\cond, \vec \loc, \ltrue) &
                  conditional branch \\
  & \alt & {\tt goto}(l) & unconditional branch \\
  & \alt & {\tt label}(l) & definition of the label $l$ \\
  & \alt & {\tt return} \alt {\tt return}(\loc) & function return
\syntaxclass{LTLin code sequences:}
c & ::=  & i_1 \ldots i_n  & list of instructions
\syntaxclass{LTLin functions:}
\fn & ::= & \{~{\tt sig} = \sig; \\
    &     & ~~~{\tt params} = \vec\loc;    & parameters\\
    &     & ~~~{\tt stacksize} = n;      & size of stack data block \\
    &     & ~~~{\tt code} = c \}         & instructions
\end{syntax}

The dynamic semantics of \lang{LTLin} is similar to that of \lang{LTL}
as far as the handling of data is concerned.  In execution states,
program points within a function $\fn$ are no longer represented by
CFG labels $l$, but instead are represented by code sequences $c$ that
are suffixes of $\fn.{\tt code}$.  The first element of $c$ is the
instruction to execute next.  As shown in figure~\ref{f:ltlin-sem},
most instructions have ``fall-through'' behavior: they transition
from $i.c$ to $c$.  To resolve branches to a label $l$, we use the
auxiliary function ${\tt findlabel}(\fn,l)$ that returns the maximal
suffix of $\fn.{\tt code}$ that starts with ${\tt label}(l)$, if it exists.

\begin{figure}

\begin{pannel}
\irule{
{\tt eval{\char95}op}(G, \sp, \op, L(\vec \loc)) = \Some v
}{
\trans {\state(\stk, \fn, \sp, {\tt op}(\op, \vec \loc, \loc).c, L, M)}
       {\state(\stk, \fn, \sp, c, L\{\loc \becomes v\}, M)}
}
\irule{
{\tt findlabel}(\fn, \ltrue) = \Some{c'}
}{
\trans {\state(\stk, \fn, \sp, {\tt goto}(l).c, L, M)}
       {\state(\stk, \fn, \sp, c', L, M)}
}

\irule{
{\tt eval{\char95}cond}(\cond, L(\vec \loc)) = \Some{{\tt true}} \amper
{\tt findlabel}(\fn, \ltrue) = \Some{c'}
}{
\trans {\state(\stk, \fn, \sp, {\tt cond}(\cond, \vec\loc, \ltrue).c, L, M)}
       {\state(\stk, \fn, \sp, c', L, M)}
}
\irule{
{\tt eval{\char95}cond}(\cond, L(\vec \loc)) = \Some{{\tt false}}
}{
\trans {\state(\stk, \fn, \sp, {\tt cond}(\cond, \vec\loc, \ltrue).c, L, M)}
       {\state(\stk, \fn, \sp, c, L, M)}
}
\end{pannel}
\caption{Semantics of \lang{LTLin} (selected rules).} \label{f:ltlin-sem}
\end{figure}

\subsection{Code transformation} \label{s:linearize-transf}

CFG linearization is performed in two steps that clearly separate the
heuristic, correctness-irrelevant part of linearization from the
actual, correctness-critical code generation part 
\cite[chap.~8]{Appel-tiger}.
We first compute an enumeration $l_1 \ldots l_n$ of the labels of the
CFG nodes reachable from the entry node.  The order of labels in this
enumeration dictates the positions of the corresponding instructions
in the list of \lang{LTLin} instructions that we will generate.
Following the verified validator approach, this enumeration of CFG
nodes is computed by untrusted code written in Caml.  It can implement
any of the textbook heuristics for picking ``hot'' traces that should
execute without branches, including static branch prediction.  This
enumeration is validated by a validator written in Coq that checks the
following two conditions:
\begin{enumerate}
\item No node $l$ appears twice in the enumeration.
\item All nodes $l$ reachable from the function entry point appear in
  the enumeration.
\end{enumerate}
For condition 2, the validator precomputes the set of reachable nodes
using a trivial forward dataflow analysis, where the abstract domain
is $\{{\tt false},{\tt true}\}$ and the transfer function is $T(l, a) = a$.
(By definition, all successors of a reachable node are reachable).

To generate \lang{LTLin} code, we then concatenate the instructions of
the control-flow graph in the order given by the enumeration $l_1,
\ldots, l_n$.  Each instruction is preceded by ${\tt label}(l_i)$ and
followed by a {\tt goto} to its successor, unless this {\tt goto} is
unnecessary because it would branch to an immediately following
{\tt label}.  Formally, the basic code generation function is of the form
$C(i, c)$ and returns a sequence $c'$ of instructions obtained by
prepending the translation of the \lang{LTL} instruction $i$ 
to the initial instruction sequence $c$.  Here are some representative
cases:
\begin{eqnarray*}
C({\tt op}(\op, \vec\loc, \loc, l'), c) & = &
  {\tt op}(\op, \vec \loc, \loc).c
\quad\mbox{if $c$ starts with ${\tt label}(l')$;} \\
C({\tt op}(\op, \vec\loc, \loc, l'), c) & = &
{\tt op}(\op, \vec \loc, \loc).{\tt goto}(l').c
\quad\mbox{otherwise.} \\
C({\tt cond}(\cond, \vec \loc, \ltrue, \lfalse), c) & = & 
  {\tt cond}(\cond, \vec \loc, \ltrue).c \\
&&\mbox{if $c$ starts with ${\tt label}(\lfalse)$;} \\
C({\tt cond}(\cond, \vec \loc, \ltrue, \lfalse), c) & = & 
  {\tt cond}(\neg\cond, \vec \loc, \lfalse).c \\
&&\mbox{if $c$ starts with ${\tt label}(\ltrue)$;} \\
C({\tt cond}(\cond, \vec \loc, \ltrue, \lfalse), c) & = & 
  {\tt cond}(\cond, \vec \loc, \ltrue).{\tt goto}(\lfalse).c \\
&&\mbox{otherwise.}
\end{eqnarray*}
This function is then iterated over the enumeration $\vec l$ of CFG
nodes, inserting the appropriate {\tt label} instructions:
\begin{eqnarray*}
C(\fn, \epsilon) & = & \epsilon \\
C(\fn, l.\vec l) & = & {\tt label}(l).C(i, C(\fn, \vec l))
\mbox{ if $\fn.{\tt code}(l) = \Some i$}
\end{eqnarray*}

\subsection{Semantic preservation}

Each intra-function transition in the original \lang{LTL} code
corresponds to 2 or 3 transitions in the generated \lang{LTLin} code:
one to skip the {\tt label} instruction, one to execute the actual
instruction, and possibly one to perform the {\tt goto} to the successor.
The proof of semantic preservation is therefore based on a simulation
diagram of the ``plus'' kind.  The main invariant is: whenever the
\lang{LTL} program is at program point $l$ in function $\fn$,
the \lang{LTLin} program is at the instruction sequence
${\tt findlabel}(\fn',l)$ in the translation $\fn'$ of $\fn$.
\begin{pannel}
\irule{
\tr{\fn} = \Some{\fn'} \amper \stk \match \stk' \\
{\tt findlabel}(\fn', l) = \Some c \amper
l\mbox{ is reachable from $\fn.{\tt entrypoint}$}
}{
\state(\stk, \fn, \sp, l, L, M) \match
\state(\stk', \fn', \sp, c, L, M)
}
\end{pannel}
A pleasant surprise is that the simulation proof goes through under
very weak assumptions about the enumeration of CFG nodes produced by
the external heuristics: conditions (1) and (2) of
section~\ref{s:linearize-transf} are all it takes to guarantee
semantic preservation.  This shows that our presentation of
linearization is robust: many trace picking heuristics can be
tried without having to redo any of the semantic preservation proofs.

\section{Spilling, reloading, and materialization of calling conventions} \label{s:reload}

The next compilation pass finishes the register allocation process
described in section~\ref{s:regalloc} by inserting explicit ``spill''
and ``reload'' operations around uses of pseudo-registers that have
been allocated stack slots.  Additionally, calling conventions are
materialized in the generated code by inserting moves to and from the
conventional locations used for parameter passing around function
calls.

\subsection{The target language: Linear} \label{s:linear}

\lang{Linear}, the target language for this pass, is a variant of
\lang{LTLin} where the operands of arithmetic operations, memory
accesses and conditional branches are restricted to machine registers
instead of arbitrary locations.  This is consistent with the
RISC instruction set of our target processor.
(Machine registers, written $\mreg$ so far, will now be written $r$
for simplicity.)  Two instructions {\tt getstack} and {\tt setstack} are
provided to move data between machine registers and stack slots
$\slot$.

\begin{syntax}
\syntaxclass{Linear instructions:}
i & ::=  & {\tt getstack}(\slot, r) \alt {\tt setstack}(r, \slot) &
              reading, writing a stack slot \\
  & \alt & {\tt op}(\op, \vec r, r) & arithmetic operation \\
  & \alt & {\tt load}(\chunk, \mode, \vec r, r) & memory load \\
  & \alt & {\tt store}(\chunk, \mode, \vec r, r) & memory store \\
  & \alt & {\tt call}(\sig, (r \alt \id)) & function call \\
  & \alt & {\tt tailcall}(\sig, (r \alt \id)) & function tail call \\
  & \alt & {\tt cond}(\cond, \vec r, \ltrue) &
                  conditional branch \\
  & \alt & {\tt goto}(l) & unconditional branch \\
  & \alt & {\tt label}(l) & definition of the label $l$ \\
  & \alt & {\tt return} & function return
\syntaxclass{Linear code sequences:}
c & ::=  & i_1 \ldots i_n  & list of instructions
\syntaxclass{Linear functions:}
\fn & ::= & \{~{\tt sig} = \sig; \\
    &     & ~~~{\tt stacksize} = n;      & size of stack data block \\
    &     & ~~~{\tt code} = c \}         & instructions
\end{syntax}

Another novelty of \lang{Linear} is that {\tt call}, {\tt tailcall} and
{\tt return} instructions, as well as function definitions, no longer
carry a list of locations for their parameters or results: the
generated \lang{Linear} code contains all the necessary move
instructions to ensure that these parameters and
results reside in the locations determined by the calling
conventions.  Correspondingly, call states and return states no longer
carry lists of values: instead, they carry
full location maps $L$ where the values of arguments and results can be
found at conventional locations, determined as a function
of the signature of the called function.
\begin{syntaxleft}
\syntaxclass{Program states:}
\st & ::= & \state(\stk, \fn, \sp, c, L, M) & regular state \\
    & \alt& \callstate(\stk, \fd, L, M) & call state \\
    & \alt& \returnstate(\stk, L, M)         & return state
\syntaxclass{Call stacks:}
\stk & ::= & (\stackframe(\fn, \sp, c, L))^* & list of frames
\end{syntaxleft}

\begin{figure}

\begin{pannel}
\irule{
{\tt eval{\char95}op}(G, \sp, \op, L(\vec r)) = \Some v
}{
\trans {\state(\stk, \fn, \sp, {\tt op}(\op, \vec r, r).c, L, M)}
       {\state(\stk, \fn, \sp, c, L\{r \becomes v\}, M)}
}
\irule{
L(r) = {\tt ptr}(b, 0) \amper {\tt funct}(G, b) = \Some \fd \amper \fd.{\tt sig} = \sig
}{
\trans {\state(\stk, \fn, \sp, {\tt call}(\sig, r).c, L, M)}
       {\callstate(\stackframe(\fn, \sp, c, L).\stk, \fd, L, M)}
}
\irule{
L(r) = {\tt ptr}(b, 0) \amper {\tt funct}(G, b) = \Some \fd \amper \fd.{\tt sig} = \sig \amper
L' = {\tt exitfun}(\parent{L}{\stk}, L)
}{
\trans {\state(\stk, \fn, \sp, {\tt tailcall}(\sig, r).c, L, M)}
       {\callstate(\stk, \fd, L', M)}
}
\irule{
L' = {\tt exitfun}(\parent{L}{\stk}, L)
}{
\trans {\state(\stk, \fn, \sp, {\tt return}.c, L, M)}
       {\returnstate(\stk, L', {\tt free}(M, \sp))}
}
\irule{
{\tt alloc}(M, 0, \fn.{\tt stacksize}) = (\sp, M') \amper
L' = {\tt entryfun}(L)
}{
\trans {\callstate(\stk, {\tt internal}(\fn), L, M)}
       {\state(\stk, \fn.{\tt code}, \sp, \fn.{\tt code}, L', M')}
}
\irule{
\vdash \ef(\vec v) \evaltrace{t} v \amper
\vec v = L({\tt loc{\char95}arguments}(\ef.{\tt sig})) \amper
L' = L\{{\tt loc{\char95}result}(\ef.{\tt sig}) \becomes v\}
}{
\transtrace {\callstate(\stk, {\tt external}(\ef), L, M)}
         {t}{\returnstate(\stk, L', M)}
}
\srule{
\trans {\returnstate(\stackframe(\fn, \sp, c, L_0).\stk, L, M)}
       {\state(\stk, \fn, \sp, c, L, M)}
}
\end{pannel}
\caption{Semantics of \lang{Linear}.  The transitions not shown are similar
  to those of \lang{LTLin}. } \label{f:linear-sem}
\end{figure}

As shown in the dynamic semantics for \lang{Linear} (see
figure~\ref{f:linear-sem}), the behavior of locations across function
calls is specified by two functions: ${\tt entryfun}(L)$ determines
the locations on entrance to the callee as a function of the locations
$L$ before the {\tt call}, and ${\tt exitfun}(L,L')$ determines the locations
in the callee when the {\tt call} returns as a function of the caller's
locations before the call, $L$, and the callee's locations before the
return, $L'$.  In summary, processor registers are global, but some
are preserved by the callee; {\tt local} and {\tt incoming} slots of the
caller are preserved across the call; and
the {\tt incoming} slots on entrance to the callee are the {\tt outgoing}
slots of the caller.
\begin{center}\normalbar
\begin{tabular}{@{}l|l|l@{}}
Location $l$           & ${\tt entryfun}(L)(l)$ & ${\tt exitfun}(L,L')(l)$ \\
\hline
$r$                    & $L(r)$             & $L(r)$ if $r$ is callee-save \\
                       &                    & $L'(r)$ if $r$ is caller-save \\
${\tt local}(\ty,\ofs)$    & {\tt undef}            & $L({\tt local}(\ty,\ofs))$ \\
${\tt incoming}(\ty,\ofs)$ & $L({\tt outgoing}(\ty,\ofs))$ & $L({\tt incoming}(\ty,\ofs))$  \\
${\tt outgoing}(\ty,\ofs)$ & {\tt undef}            & $L'({\tt incoming}(\ty,\ofs))$ 
\end{tabular}
\end{center}
In other words, the {\tt entryfun} and {\tt exitfun} anticipate, at the level
of the \lang{Linear} semantics, the effect of future transformations
(placement of stack slots in memory and insertion of function
prologues and epilogues to save and restore callee-save registers)
performed in the next compilation pass (section~\ref{s:stacking}).

In the rules for {\tt tailcall} and {\tt return} of figure~\ref{f:linear-sem},
the notation $\parent{L}{\stk}$ stands for the $L$ component of the
top frame in stack $\stk$.  A suitable default is defined for an empty
stack.
$$\parent{L}{(\stackframe(\fn, \sp, c, L).\stk)} = L
\qquad
  \parent{L}{\epsilon} = (\loc \mapsto {\tt undef})
$$

\subsection{Code transformation}

Our strategy for spilling and reloading is simplistic: each use of a
spilled pseudo-register is preceded by a {\tt getstack} instruction to
reload the pseudo-register in a machine register, and each definition
is followed by a {\tt setstack} instruction that spills the result.  No
attempt is made to reuse a reloaded value, nor to delay spilling.  We
reserve 3 integer registers and 3 float registers to hold reloaded
values and results of instructions before spilling.   (This
does not follow compiler textbooks, which prescribe re-running
register allocation to assign registers to reload and spill
temporaries.  However, it is difficult to prove termination for this
practice, and moreover it requires semantic reasoning about
partially-allocated code.)  
The following case should give the flavor of the transformation:
\begin{eqnarray*}
\tr{{\tt op}(\op, \vec \loc, \loc).c} & = &
  {\tt let\ }\vec r = {\tt regs{\char95}for}(\vec \loc) {\tt \ and\ }r = {\tt reg{\char95}for}(\loc){\tt \ in\ }\\
&&{\tt reloads}(\vec \loc, \vec r). {\tt op}(\op, \vec r, r).{\tt spill}(r, \loc). \tr{c}
\end{eqnarray*}
${\tt reg{\char95}for}(\loc)$ returns $r$ if the location $\loc$ is a machine
register $r$, or a temporary register of the appropriate type if
$\loc$ is a stack slot.  ${\tt regs{\char95}for}(\vec\loc)$ does the same for a
list of locations, using different temporary registers for each
location.  ${\tt spill}(r, \loc)$ generates the {\tt move} or
{\tt setstack} operation that sets $\loc$ to the value of register $r$.
Symmetrically, ${\tt reloads}(\vec\loc, \vec r)$ generates the {\tt move} or
{\tt getstack} operations that set $\vec r$ to the values of locations
$\vec \loc$.

For {\tt call} and {\tt tailcall} instructions with signature $\sig$ and
arguments $\vec \loc$, we insert moves from $\vec \loc$ to the
locations dictated by the calling conventions.  These locations
(a mixture of processor register and {\tt outgoing} stack slots)
are determined as a function ${\tt loc{\char95}arguments}(\sig)$ of the signature
$\sig$ of the called function.  Likewise, the result of the call,
which is passed back in the conventional location
${\tt loc{\char95}result}(\sig)$, is moved to the result location of the {\tt call}.
\begin{eqnarray*}
\tr{{\tt call}(\sig, \id, \vec \loc, \loc).c} & = &
  {\tt parallel{\char95}move}(\vec \loc, {\tt loc{\char95}arguments}(\sig)). \\
&& 
  {\tt call}(\sig, \id).
  {\tt spill}({\tt loc{\char95}result}(\sig), \loc).
  \tr{c}
\end{eqnarray*}
Symmetrically, at entry to a function $\fn$, we insert moves from
${\tt loc{\char95}parameters}(\fn.{\tt sig})$ to $\fn.{\tt params}$. (${\tt loc{\char95}parameters}(\sig)$ is 
${\tt loc{\char95}arguments}(\sig)$ where {\tt outgoing} slots are replaced by
{\tt incoming} slots.) 

The moves inserted for function arguments and for function parameters
must implement a parallel assignment semantics: some registers
can appear both as sources and destinations, as in $(r_1,r_2,r_3) :=
(r_2, r_1, r_4)$.  It is folklore that such parallel moves can
be implemented by a sequence of elementary moves using at most one
temporary register of each kind.  Formulating the parallel move
algorithm in Coq and proving its correctness was a particularly
difficult part of this development.  The proof is detailed in a
separate paper \cite{Rideau-Serpette-Leroy-parmov}.

\subsection{Semantic preservation} \label{s:reload-proof}

The correctness proof for the spilling pass is surprisingly involved
because it needs to account for the fact that the location states $L$
and memory states $M$ differ significantly between the original
\lang{LTLin} code and the generated \lang{Linear} code.  An
inessential source of difference is that the \lang{Linear} code makes
use of temporary registers and {\tt outgoing} and {\tt incoming} stack
locations while the \lang{LTLin} code does not.  A deeper difference
comes from the fact that, in \lang{LTLin}, locations other than
function parameters are consistently initialized to the {\tt undef} value,
while in \lang{Linear} some of these locations (\eg processor
registers) just keep whatever values they had in the caller before the
{\tt call} instruction.  Performing arithmetic over this {\tt undef} value
would cause the original \lang{LTLin} program to go wrong, but it 
can still pass this value around, store it in memory locations, and
read it back.  Therefore, some
{\tt undef} values found in \lang{LTLin} register and memory states can
become any other value in the corresponding \lang{Linear} location and
memory states.  To capture this fact, we use the ``less defined than''
ordering $\le$ between values defined by
$$ v \le v' \defequal v = v' \mbox{ or } v = {\tt undef} $$
and extended to memory states as follows:
$$ M \le M' \defequal
  \begin{array}[t]{rl}
  \forall \chunk,b,\ofs,v, & {\tt load}(\chunk,M,b,\ofs) = \Some v \\
     &  {} \Rightarrow \exists v', {\tt load}(\chunk,M',b,\ofs) = \Some{v'} \wedge v \le v'
  \end{array} $$
Leroy and Blazy \cite[section 5.3]{Leroy-Blazy-memory-model} study the properties of
this relation between memory states.  It commutes nicely with the
{\tt store}, {\tt alloc} and {\tt free} operations over memory states.

Putting it all together, we define agreement $L \le L'$ between a
\lang{LTLin} location state $L$ and a \lang{Linear} location state
$L'$ as
$$ L \le L' \defequal L(\loc) \le L'(\loc)
\mbox{ for all non-temporary registers or {\tt local} stack slots $\loc$}
$$
and define the invariant relating \lang{LTLin} and \lang{Linear}
execution states as follows.

\begin{pannel}
\irule{
\stk \match \stk' : \fn.{\tt sig} \amper L \le L' \amper M \le M'
}{
\state(\stk, \fn, \sp, c, L, M) \match
\state(\stk', \tr{\fn}, \sp, \tr{c}, L', M')
}
\irule{
\stk \match \stk' : \fd.{\tt sig} \amper 
\vec v \le L'({\tt loc{\char95}arguments}(\fd.{\tt sig})) \amper
L' \approx \parent{L}{\stk} \amper M \le M'
}{
\callstate(\stk, \fd, \vec v, M) \match
\callstate(\stk', \fd, L', M')
}
\irule{
\stk \match \stk' : \sig \amper 
v \le L'({\tt loc{\char95}result}(\sig)) \amper
L' \approx \parent{L}{\stk} \amper M \le M'
}{
\returnstate(\stk, v, M) \match
\returnstate(\stk', L', M')
}
\end{pannel}
For call states and return states, the second premises capture the
fact that the argument and return values can indeed be found in the
corresponding locations dictated by the calling conventions, modulo
the $\le$ relation between values.  The third premises
$L' \approx \parent{L}{\stk}$ say that the current location state $L'$
and that of the caller $\parent{L}{\stk}$ assign the same values to
any callee-save location.  Finally, in the invariant $\stk \match \stk' :
\sig$ relating call stacks, a call signature $\sig$ is threaded
through the call stack to make sure that the caller and the callee
agree on the result type of the call, and therefore on the location
used to pass the return value.
\begin{pannel}
\irule{
\sig.{\tt res} = \Some{{\tt int}}
}{
\epsilon \match \epsilon : \sig
}
\irule{
\stk \match \stk' : \fn.{\tt sig} \amper
c' = {\tt spill}({\tt loc{\char95}result}(\sig), \loc).\tr{c} \amper
{\tt postcall}(L) \le L'
}{
\stackframe(\loc, \fn, \sp, {\tt postcall}(L), c).\stk \match
\stackframe(\tr{\fn}, \sp, L', c').\stk' : \sig
}
\end{pannel}
Armed with these definitions and the proof of the parallel move
algorithm of \cite{Rideau-Serpette-Leroy-parmov}, we prove
semantic preservation for this pass using a simulation diagram of the
``star'' kind.  The only part of the transformation that could cause
stuttering is the elimination of a redundant move from $\loc$ to
$\loc$.  To prove that stuttering cannot happen, it suffices to note
that the length of the \lang{LTLin} instruction sequence currently
executing decreases in this case.

\section{Construction of the activation record} \label{s:stacking}

The penultimate compilation pass lays out the activation record for
each function, allocating space for stack slots and turning accesses
to slots into memory loads and stores.  Function prologues and
epilogues are added to preserve the values of callee-save registers.

\subsection{The target language: Mach} \label{s:mach}

The last intermediate language in our gentle descent towards assembly
language is called \lang{Mach}.  It is a variant of \lang{Linear}
where the three infinite supplies of stack slots (local, incoming and
outgoing) are mapped to actual memory locations in the stack frames
of the current function (for local and outgoing slots) or the calling
function (for incoming slots).
\begin{syntaxleft}
\syntaxclass{Mach instructions:}
i & ::=  & {\tt setstack}(r, \ty, \ofs) & register to stack move \\
  & \alt & {\tt getstack}(\ty, \ofs, r) & stack to register move \\
  & \alt & {\tt getparent}(\ty, \ofs, r) & caller's stack to register move\\
  & \alt & \ldots & as in \lang{Linear}
\end{syntaxleft}%
In the three new move instructions, $\ty$ is the type of the data
moved and $\ofs$ its word offset in the corresponding activation record.
\begin{syntaxleft}
\syntaxclass{Mach code:}
c & ::=  & i_1 \ldots i_n  & list of instructions
\syntaxclass{Mach functions:}
\fn & ::= & \{~{\tt sig} = \sig; \\
    &     & ~~~{\tt stack{\char95}high} = n;     & upper bound of stack data block \\
    &     & ~~~{\tt stack{\char95}low} = n;      & lower bound of stack data block \\
    &     & ~~~{\tt retaddr} = \delta;   & offset of saved return address \\
    &     & ~~~{\tt link} = \delta;      & offset of back link \\
    &     & ~~~{\tt code} = c \}         & instructions
\end{syntaxleft}%
Functions carry two byte offsets, {\tt retaddr} and {\tt link}, indicating
where in the activation record the function prologue should save the
return address into its caller and the back link to the activation
record of its caller, respectively.

\begin{syntaxleft}
\syntaxclass{Program states:}
\st & ::= & \state(\stk, \fn, \sp, c, R, M) & regular states \\
    & \alt& \callstate(\stk, \fd, R, M) & call states \\
    & \alt& \returnstate(\stk, R, M) & return states
\syntaxclass{Call stacks:}
\stk & ::= & (\stackframe(\fn, \sp, \retaddr, c))^* & list of frames
\end{syntaxleft}

Semantically, the main difference between \lang{Linear} and
\lang{Mach} is that, in \lang{Mach}, the register state $R$ mapping
processor registers to values is global and shared between caller and
callee.  In particular, $R$ is not saved in the call stack, and as
shown in figure~\ref{f:mach-sem}, there is no automatic restoration of
callee-save registers at function return; instead, the \lang{Mach}
code generator produces appropriate {\tt setstack} and {\tt getstack}
instructions to save and restore used callee-save registers at
function prologues and epilogues.

\begin{figure}

\begin{pannel}
\begin{eqnarray*}
{\tt loadstack}(\tau, M, \sp, \ofs) & = &
\twocases{{\tt load}(\chunk(\tau), M, b, \ofs' + \ofs)}{if $\sp = {\tt ptr}(b,\ofs')$}
         {\None}{otherwise} \\
{\tt storestack}(\tau, M, \sp, \ofs, v) & = &
\twocases{{\tt store}(\chunk(\tau), M, b, \ofs' + \ofs, v)}{if $\sp = {\tt ptr}(b,\ofs')$}
         {\None}{otherwise}
\end{eqnarray*}
\irule{
{\tt storestack}(\tau, M, \sp, \ofs, R(r)) = \Some{M'}
}{
\trans {\state(\stk, \fn, \sp, {\tt setstack}(r,\ty,\ofs).c, R, M)}
       {\state(\stk, \fn, \sp, c, R, M')}
}
\irule{
{\tt loadstack}(\tau, M, \sp, \ofs) = \Some v
}{
\trans {\state(\stk, \fn, \sp, {\tt getstack}(\ty,\ofs,r).c, R, M)}
       {\state(\stk, \fn, \sp, c, R\{r \becomes v\}, M)}
}
\irule{
{\tt loadstack}({\tt int}, M, \sp, \fn.{\tt link}) = \Some{\sp'} \amper
{\tt loadstack}(\tau, M, \sp', \ofs) = \Some v
}{
\trans {\state(\stk, \fn, \sp, {\tt getparent}(\ty,\ofs,r).c, R, M)}
       {\state(\stk, \fn, \sp, c, R\{r \becomes v\}, M)}
}
\irule{
R(r) = {\tt ptr}(b, 0) \amper {\tt funct}(G, b) = \Some \fd \amper \fd.{\tt sig} = \sig \amper
{\tt retaddr}(\fn, c, \retaddr)
}{
\trans {\state(\stk, \fn, \sp, {\tt call}(\sig, r).c, R, M)}
       {\callstate(\stackframe(\fn, \sp, \retaddr, c).\stk, \fd, R, M)}
}
\irule{
R(r) = {\tt ptr}(b, 0) \amper {\tt funct}(G, b) = \Some \fd \amper \fd.{\tt sig} = \sig \\
{\tt loadstack}({\tt int}, M, \sp, \fn.{\tt link}) = \Some{\parent{\sp}{\stk}} \amper
{\tt loadstack}({\tt int}, M, \sp, \fn.{\tt retaddr}) = \Some{\parent{\retaddr}{\stk}}
}{
\trans {\state(\stk, \fn, \sp, {\tt tailcall}(\sig, r).c, R, M)}
       {\callstate(\stk, \fd, R, M)}
}
\irule{
{\tt loadstack}({\tt int}, M, \sp, \fn.{\tt link}) = \Some{\parent{\sp}{\stk}} \amper
{\tt loadstack}({\tt int}, M, \sp, \fn.{\tt retaddr}) = \Some{\parent{\retaddr}{\stk}}
}{
\trans {\state(\stk, \fn, \sp, {\tt return}.c, R, M)}
       {\returnstate(\stk, R, M)}
}
\irule{
{\tt alloc}(M, \fn.{\tt stack{\char95}low}, \fn.{\tt stack{\char95}high}) = (b, M_1) \amper
\sp = {\tt ptr}(b, \fn.{\tt stack{\char95}low}) \\
{\tt storestack}({\tt int}, M_1, \sp, \fn.{\tt link}, \parent{\sp}{\stk}) = \Some{M_2} \\
{\tt storestack}({\tt int}, M_2, \sp, \fn.{\tt retaddr}, \parent{\retaddr}{\stk}) = \Some{M_3}
}{
\trans {\callstate(\stk, {\tt internal}(\fn), R, M)}
       {\state(\stk, \fn, \sp, \fn.{\tt code}, R, M_3)}
}
\end{pannel}
\caption{Semantics of \lang{Mach} (selected rules).} \label{f:mach-sem}
\end{figure}

The {\tt setstack} and {\tt getstack} instructions are interpreted as memory
stores and loads relative to the stack pointer.  We write
$\chunk(\tau)$ for the memory quantity appropriate for storing a value of type
$\tau$ without losing information, namely $\chunk({\tt int}) = {\tt int32}$
and $\chunk({\tt float}) = {\tt float64}$.  For the {\tt getparent} instruction,
we recover a pointer to the caller's stack frame by loading from the
{\tt link} location of our stack frame, then load from this pointer.  This
additional indirection is needed since, in our memory model, the
callee's stack frame is not necessarily adjacent to that of the
caller.  This linking of stack frames is implemented by the rule for
function entry. 

The rules for {\tt call} and function entry also make provisions for
saving a return address within the caller's function code in the
{\tt retaddr} location of the activation record.  This return address 
(a pointer to an instruction within a code block) becomes relevant in the next
compilation pass (generation of assembly code), but it is convenient
to reflect it in the semantics of \lang{Mach}.  To this end, the
semantics is parameterized by a predicate ${\tt retaddr}(\fn, c, \retaddr)$ that
relates this return address $\retaddr$ with the caller's function
$\fn$ and the code sequence $c$ that immediately follows the {\tt call}
instruction.  In section~\ref{s:ppcgen-proof}, we will see how to
define this predicate in a way that accurately characterizes the
return address.  The rules for {\tt return} and {\tt tailcall} contain
premises that check that the values contained in the {\tt retaddr} and
{\tt link} locations of the activation record were preserved during the
execution of the function.

\subsection{Code transformation} \label{s:stacking-transf}

\begin{figure}
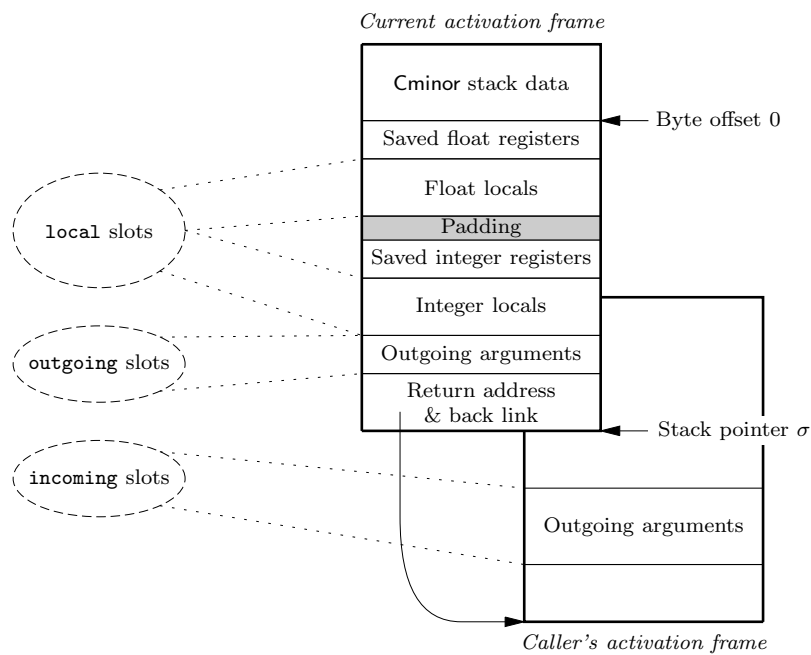


\begin{gpic}

\begin{pgfpicture}
  \pgfsetxvec{\pgfpoint{1.000in}{0in}}
  \pgfsetyvec{\pgfpoint{0in}{1.000in}}
  \begin{pgfscope}
    \pgfpathmoveto{\pgfpointxy{1.825}{-0.512}}
    \pgfpathlineto{\pgfpointxy{3.075}{-0.512}}
    \pgfpathlineto{\pgfpointxy{3.075}{-0.212}}
    \pgfpathlineto{\pgfpointxy{1.825}{-0.212}}
    \pgfpathlineto{\pgfpointxy{1.825}{-0.512}}
    \pgfusepath{stroke}
  \end{pgfscope}
  \pgftext[at=\pgfpointadd{\pgfpointxy{2.450}{-0.362}}{\pgfpoint{0pt}{0.5 \baselineskip}}]{Return address}
  \pgftext[at=\pgfpointadd{\pgfpointxy{2.450}{-0.362}}{\pgfpoint{0pt}{-0.5 \baselineskip}}]{\& back link}
  \begin{pgfscope}
    \pgfpathmoveto{\pgfpointxy{1.825}{-0.212}}
    \pgfpathlineto{\pgfpointxy{3.075}{-0.212}}
    \pgfpathlineto{\pgfpointxy{3.075}{-0.012}}
    \pgfpathlineto{\pgfpointxy{1.825}{-0.012}}
    \pgfpathlineto{\pgfpointxy{1.825}{-0.212}}
    \pgfusepath{stroke}
  \end{pgfscope}
  \pgftext[at=\pgfpointadd{\pgfpointxy{2.450}{-0.112}}{\pgfpoint{0pt}{-0.0 \baselineskip}}]{Outgoing arguments}
  \begin{pgfscope}
    \pgfpathmoveto{\pgfpointxy{1.825}{-0.012}}
    \pgfpathlineto{\pgfpointxy{3.075}{-0.012}}
    \pgfpathlineto{\pgfpointxy{3.075}{0.288}}
    \pgfpathlineto{\pgfpointxy{1.825}{0.288}}
    \pgfpathlineto{\pgfpointxy{1.825}{-0.012}}
    \pgfusepath{stroke}
  \end{pgfscope}
  \pgftext[at=\pgfpointadd{\pgfpointxy{2.450}{0.138}}{\pgfpoint{0pt}{-0.0 \baselineskip}}]{Integer locals}
  \begin{pgfscope}
    \pgfpathmoveto{\pgfpointxy{1.825}{0.288}}
    \pgfpathlineto{\pgfpointxy{3.075}{0.288}}
    \pgfpathlineto{\pgfpointxy{3.075}{0.488}}
    \pgfpathlineto{\pgfpointxy{1.825}{0.488}}
    \pgfpathlineto{\pgfpointxy{1.825}{0.288}}
    \pgfusepath{stroke}
  \end{pgfscope}
  \pgftext[at=\pgfpointadd{\pgfpointxy{2.450}{0.388}}{\pgfpoint{0pt}{-0.0 \baselineskip}}]{Saved integer registers}
  \begin{pgfscope}
    \pgfpathmoveto{\pgfpointxy{1.825}{0.488}}
    \pgfpathlineto{\pgfpointxy{3.075}{0.488}}
    \pgfpathlineto{\pgfpointxy{3.075}{0.613}}
    \pgfpathlineto{\pgfpointxy{1.825}{0.613}}
    \pgfpathlineto{\pgfpointxy{1.825}{0.488}}
    \pgfsetfillcolor{black!20!white}
    \pgfusepath{fill,stroke}
  \end{pgfscope}
  \pgftext[at=\pgfpointadd{\pgfpointxy{2.450}{0.550}}{\pgfpoint{0pt}{-0.0 \baselineskip}}]{Padding}
  \begin{pgfscope}
    \pgfpathmoveto{\pgfpointxy{1.825}{0.613}}
    \pgfpathlineto{\pgfpointxy{3.075}{0.613}}
    \pgfpathlineto{\pgfpointxy{3.075}{0.913}}
    \pgfpathlineto{\pgfpointxy{1.825}{0.913}}
    \pgfpathlineto{\pgfpointxy{1.825}{0.613}}
    \pgfusepath{stroke}
  \end{pgfscope}
  \pgftext[at=\pgfpointadd{\pgfpointxy{2.450}{0.763}}{\pgfpoint{0pt}{-0.0 \baselineskip}}]{Float locals}
  \begin{pgfscope}
    \pgfpathmoveto{\pgfpointxy{1.825}{0.913}}
    \pgfpathlineto{\pgfpointxy{3.075}{0.913}}
    \pgfpathlineto{\pgfpointxy{3.075}{1.113}}
    \pgfpathlineto{\pgfpointxy{1.825}{1.113}}
    \pgfpathlineto{\pgfpointxy{1.825}{0.913}}
    \pgfusepath{stroke}
  \end{pgfscope}
  \pgftext[at=\pgfpointadd{\pgfpointxy{2.450}{1.013}}{\pgfpoint{0pt}{-0.0 \baselineskip}}]{Saved float registers}
  \begin{pgfscope}
    \pgfpathmoveto{\pgfpointxy{1.825}{1.113}}
    \pgfpathlineto{\pgfpointxy{3.075}{1.113}}
    \pgfpathlineto{\pgfpointxy{3.075}{1.512}}
    \pgfpathlineto{\pgfpointxy{1.825}{1.512}}
    \pgfpathlineto{\pgfpointxy{1.825}{1.113}}
    \pgfusepath{stroke}
  \end{pgfscope}
  \pgftext[at=\pgfpointadd{\pgfpointxy{2.450}{1.312}}{\pgfpoint{0pt}{-0.0 \baselineskip}}]{\lang{Cminor} stack data}
  \begin{pgfscope}
    \pgfpathmoveto{\pgfpointxy{3.175}{1.088}}
    \pgfpathlineto{\pgfpointxy{3.075}{1.113}}
    \pgfpathlineto{\pgfpointxy{3.175}{1.138}}
    \pgfpathlineto{\pgfpointxy{3.175}{1.088}}
    \pgfsetfillcolor{black}
    \pgfsetlinewidth{0.100pt}
    \pgfusepath{fill,stroke}
  \end{pgfscope}
  \begin{pgfscope}
    \pgfpathmoveto{\pgfpointxy{3.175}{1.113}}
    \pgfpathlineto{\pgfpointxy{3.325}{1.113}}
    \pgfusepath{stroke}
  \end{pgfscope}
  \pgftext[at=\pgfpointadd{\pgfpointxy{3.325}{1.113}}{\pgfpoint{0pt}{-0.0 \baselineskip}},left]{~Byte offset 0}
  \begin{pgfscope}
    \pgfpathmoveto{\pgfpointxy{2.675}{-1.512}}
    \pgfpathlineto{\pgfpointxy{3.925}{-1.512}}
    \pgfpathlineto{\pgfpointxy{3.925}{-1.212}}
    \pgfpathlineto{\pgfpointxy{2.675}{-1.212}}
    \pgfpathlineto{\pgfpointxy{2.675}{-1.512}}
    \pgfusepath{stroke}
  \end{pgfscope}
  \begin{pgfscope}
    \pgfpathmoveto{\pgfpointxy{2.675}{-1.212}}
    \pgfpathlineto{\pgfpointxy{3.925}{-1.212}}
    \pgfpathlineto{\pgfpointxy{3.925}{-0.812}}
    \pgfpathlineto{\pgfpointxy{2.675}{-0.812}}
    \pgfpathlineto{\pgfpointxy{2.675}{-1.212}}
    \pgfusepath{stroke}
  \end{pgfscope}
  \pgftext[at=\pgfpointadd{\pgfpointxy{3.300}{-1.012}}{\pgfpoint{0pt}{-0.0 \baselineskip}}]{Outgoing arguments}
  \begin{pgfscope}
    \pgfpathmoveto{\pgfpointxy{1.825}{-0.512}}
    \pgfpathlineto{\pgfpointxy{3.075}{-0.512}}
    \pgfpathlineto{\pgfpointxy{3.075}{1.513}}
    \pgfpathlineto{\pgfpointxy{1.825}{1.513}}
    \pgfsetlinewidth{1.000pt}
    \pgfusepath{stroke}
  \end{pgfscope}
  \begin{pgfscope}
    \pgfpathmoveto{\pgfpointxy{1.825}{1.513}}
    \pgfpathlineto{\pgfpointxy{1.825}{-0.512}}
    \pgfsetlinewidth{1.000pt}
    \pgfusepath{stroke}
  \end{pgfscope}
  \begin{pgfscope}
    \pgfpathmoveto{\pgfpointxy{3.925}{-1.512}}
    \pgfpathlineto{\pgfpointxy{2.675}{-1.512}}
    \pgfpathlineto{\pgfpointxy{2.675}{-0.512}}
    \pgfsetlinewidth{1.000pt}
    \pgfusepath{stroke}
  \end{pgfscope}
  \begin{pgfscope}
    \pgfpathmoveto{\pgfpointxy{3.925}{-1.512}}
    \pgfpathlineto{\pgfpointxy{3.925}{0.188}}
    \pgfpathlineto{\pgfpointxy{3.075}{0.188}}
    \pgfsetlinewidth{1.000pt}
    \pgfusepath{stroke}
  \end{pgfscope}
  \begin{pgfscope}
    \pgfpathmoveto{\pgfpointxy{2.575}{-1.487}}
    \pgfpathlineto{\pgfpointxy{2.675}{-1.512}}
    \pgfpathlineto{\pgfpointxy{2.575}{-1.537}}
    \pgfpathlineto{\pgfpointxy{2.575}{-1.487}}
    \pgfsetfillcolor{black}
    \pgfsetlinewidth{0.100pt}
    \pgfusepath{fill,stroke}
  \end{pgfscope}
  \begin{pgfscope}
    \pgfpathmoveto{\pgfpointxy{2.025}{-0.412}}
    \pgfpathlineto{\pgfpointxy{2.025}{-0.962}}
    \pgfpathcurveto{\pgfpointxy{2.025}{-1.329}}{\pgfpointxy{2.132}{-1.512}}{\pgfpointxy{2.345}{-1.512}}
    \pgfpathlineto{\pgfpointxy{2.665}{-1.512}}
    \pgfsetlinewidth{0.000pt}
    \pgfusepath{stroke}
  \end{pgfscope}
  \begin{pgfscope}
    \pgfpathmoveto{\pgfpointxy{3.175}{-0.537}}
    \pgfpathlineto{\pgfpointxy{3.075}{-0.512}}
    \pgfpathlineto{\pgfpointxy{3.175}{-0.487}}
    \pgfpathlineto{\pgfpointxy{3.175}{-0.537}}
    \pgfsetfillcolor{black}
    \pgfsetlinewidth{0.100pt}
    \pgfusepath{fill,stroke}
  \end{pgfscope}
  \begin{pgfscope}
    \pgfpathmoveto{\pgfpointxy{3.175}{-0.512}}
    \pgfpathlineto{\pgfpointxy{3.325}{-0.512}}
    \pgfsetlinewidth{0.000pt}
    \pgfusepath{stroke}
  \end{pgfscope}
  \begin{pgfscope}
    \pgfpathmoveto{\pgfpointxy{3.325}{-0.587}}
    \pgfpathlineto{\pgfpointxy{4.225}{-0.587}}
    \pgfpathlineto{\pgfpointxy{4.225}{-0.437}}
    \pgfpathlineto{\pgfpointxy{3.325}{-0.437}}
    \pgfpathlineto{\pgfpointxy{3.325}{-0.587}}
    \pgfsetfillcolor{white}
    \pgfusepath{fill}
  \end{pgfscope}
  \pgftext[at=\pgfpointadd{\pgfpointxy{3.775}{-0.512}}{\pgfpoint{0pt}{-0.0 \baselineskip}}]{Stack pointer $\sp$}
  \pgftext[at=\pgfpointadd{\pgfpointxy{2.450}{1.562}}{\pgfpoint{0pt}{0.5 \baselineskip}}]{{\it Current activation frame}}
  \pgftext[at=\pgfpointadd{\pgfpointxy{3.300}{-1.562}}{\pgfpoint{0pt}{-0.5 \baselineskip}}]{{\it Caller's activation frame}}
  \begin{pgfscope}
    \pgfpathellipse{\pgfpointxy{0.450}{0.538}}{\pgfpointxy{0.450}{0}}{\pgfpointxy{0}{0.300}}
    \pgfsetlinewidth{0.000pt}
    \pgfsetdash{{0.050in}{0.025in}}{0cm}
    \pgfusepath{stroke}
  \end{pgfscope}
  \pgftext[at=\pgfpointadd{\pgfpointxy{0.450}{0.538}}{\pgfpoint{0pt}{-0.0 \baselineskip}}]{{\tt local} slots}
  \begin{pgfscope}
    \pgfpathellipse{\pgfpointxy{0.450}{-0.162}}{\pgfpointxy{0.450}{0}}{\pgfpointxy{0}{0.200}}
    \pgfsetlinewidth{0.000pt}
    \pgfsetdash{{0.050in}{0.025in}}{0cm}
    \pgfusepath{stroke}
  \end{pgfscope}
  \pgftext[at=\pgfpointadd{\pgfpointxy{0.450}{-0.162}}{\pgfpoint{0pt}{-0.0 \baselineskip}}]{{\tt outgoing} slots}
  \begin{pgfscope}
    \pgfpathellipse{\pgfpointxy{0.450}{-0.762}}{\pgfpointxy{0.450}{0}}{\pgfpointxy{0}{0.200}}
    \pgfsetlinewidth{0.000pt}
    \pgfsetdash{{0.050in}{0.025in}}{0cm}
    \pgfusepath{stroke}
  \end{pgfscope}
  \pgftext[at=\pgfpointadd{\pgfpointxy{0.450}{-0.762}}{\pgfpoint{0pt}{-0.0 \baselineskip}}]{{\tt incoming} slots}
  \begin{pgfscope}
    \pgfpathmoveto{\pgfpointxy{0.768}{0.750}}
    \pgfpathlineto{\pgfpointxy{1.825}{0.912}}
    \pgfsetlinewidth{0.000pt}
    \pgfsetdash{{1pt}{0.050in}}{0cm}
    \pgfusepath{stroke}
  \end{pgfscope}
  \begin{pgfscope}
    \pgfpathmoveto{\pgfpointxy{0.900}{0.538}}
    \pgfpathlineto{\pgfpointxy{1.825}{0.613}}
    \pgfsetlinewidth{0.000pt}
    \pgfsetdash{{1pt}{0.050in}}{0cm}
    \pgfusepath{stroke}
  \end{pgfscope}
  \begin{pgfscope}
    \pgfpathmoveto{\pgfpointxy{0.900}{0.538}}
    \pgfpathlineto{\pgfpointxy{1.825}{0.288}}
    \pgfsetlinewidth{0.000pt}
    \pgfsetdash{{1pt}{0.050in}}{0cm}
    \pgfusepath{stroke}
  \end{pgfscope}
  \begin{pgfscope}
    \pgfpathmoveto{\pgfpointxy{0.768}{0.325}}
    \pgfpathlineto{\pgfpointxy{1.825}{-0.012}}
    \pgfsetlinewidth{0.000pt}
    \pgfsetdash{{1pt}{0.050in}}{0cm}
    \pgfusepath{stroke}
  \end{pgfscope}
  \begin{pgfscope}
    \pgfpathmoveto{\pgfpointxy{0.768}{-0.021}}
    \pgfpathlineto{\pgfpointxy{1.825}{-0.012}}
    \pgfsetlinewidth{0.000pt}
    \pgfsetdash{{1pt}{0.050in}}{0cm}
    \pgfusepath{stroke}
  \end{pgfscope}
  \begin{pgfscope}
    \pgfpathmoveto{\pgfpointxy{0.768}{-0.304}}
    \pgfpathlineto{\pgfpointxy{1.825}{-0.212}}
    \pgfsetlinewidth{0.000pt}
    \pgfsetdash{{1pt}{0.050in}}{0cm}
    \pgfusepath{stroke}
  \end{pgfscope}
  \begin{pgfscope}
    \pgfpathmoveto{\pgfpointxy{0.768}{-0.621}}
    \pgfpathlineto{\pgfpointxy{2.675}{-0.812}}
    \pgfsetlinewidth{0.000pt}
    \pgfsetdash{{1pt}{0.050in}}{0cm}
    \pgfusepath{stroke}
  \end{pgfscope}
  \begin{pgfscope}
    \pgfpathmoveto{\pgfpointxy{0.768}{-0.904}}
    \pgfpathlineto{\pgfpointxy{2.675}{-1.212}}
    \pgfsetlinewidth{0.000pt}
    \pgfsetdash{{1pt}{0.050in}}{0cm}
    \pgfusepath{stroke}
  \end{pgfscope}
\end{pgfpicture}
\ifx\Setlineno\undefined\else\Setlineno=69\fi
\end{gpic}

\caption{Layout of \lang{Mach} activation frames; mapping from
  \lang{Linear} stack slots to frame locations.}
\label{f:stack-layout}

\end{figure}

The translation from \lang{Linear} to \lang{Mach} proceeds in two
steps.  First, the \lang{Linear} code is scanned to determine which
stack slots and callee-save registers it uses.  Based on this
information, the activation record is laid out following the general
shape pictured in figure~\ref{f:stack-layout}.  The total size of the
record and the byte offsets for each of its areas are
determined.  From these offsets, we can define a function $\Delta$
mapping callee-save registers, {\tt local} and {\tt outgoing} stack slots to
byte offsets, as suggested in figure~\ref{f:stack-layout}.  
(Note that these offsets are negative, while positive offsets
within the activation frame correspond to \lang{Cminor} stack data.
This choice is compatible with our pointer and memory model, where
offsets in pointers are signed integers, and it simplifies the
soundness proof.)

The generation of \lang{Mach} code is straightforward.  {\tt getslot} and
{\tt setslot} \lang{Linear} instructions are rewritten as follows:
\begin{eqnarray*}
\tr{{\tt getslot}({\tt local}(\tau,\ofs), r)} & = &
{\tt getstack}(\tau, \Delta({\tt local}(\tau,\ofs)), r) \\
\tr{{\tt getslot}({\tt outgoing}(\tau,\ofs), r)} & = &
{\tt getstack}(\tau, \Delta({\tt outgoing}(\tau,\ofs)), r) \\
\tr{{\tt getslot}({\tt incoming}(\tau,\ofs), r)} & = &
{\tt getparent}(\tau, \Delta({\tt outgoing}(\tau,\ofs)), r) \\
\tr{{\tt setslot}(r, {\tt local}(\tau,\ofs))} & = &
{\tt setstack}(r, \tau, \Delta({\tt local}(\tau,\ofs))) \\
\tr{{\tt setslot}(r, {\tt outgoing}(\tau,\ofs))} & = &
{\tt setstack}(r, \tau, \Delta({\tt outgoing}(\tau,\ofs)))
\end{eqnarray*}
Moreover, instructions to save and restore the callee-save registers
$r_1, \ldots, r_n$ used in the function are inserted as function
epilogue and prologue:
\begin{eqnarray*}
\tr{{\tt return}} & = & \ldots {\tt getstack}(\tau(r_i), \Delta(r_i), r_i)
\ldots {\tt return} \\
\tr{\fn} & = & \{ {\tt code} = \ldots {\tt setstack}(r_i, \tau(r_i),
\Delta(r_i)) \ldots \tr{\fn.{\tt code}}; \ldots \}
\end{eqnarray*}
Moreover, the translation of a function $\tr{\fn}$ must fail if the
``frame'' part of the activation record (everything except the
\lang{Cminor} stack data) is bigger than $2^{31}$ bytes.  Indeed, in
this case the signed integer offsets used to access locations within the
activation record would overflow, making it impossible to access some
frame components.

\subsection{Semantic preservation} \label{s:stacking-proof}

While the code transformation outlined above is simple, its proof of
correctness is surprisingly difficult, mostly because it entails
reasoning about memory separation properties (between areas of the
activation record and between different activation records).  To
manage this complexity, we broke the proof in two sub-proofs, using an
alternate semantics for \lang{Mach} to connect them. 

In this alternate semantics, the ``frame'' part of activation records
does not reside in memory; instead, its contents are recorded
separately in a component $\Phi$ of regular states.  This environment
$\Phi$ maps (type, byte offset) pairs to values, taking overlap into
account in the style of the maps $L$ from locations to values
introduced in section~\ref{s:ltl}.  Each function activation has its
own frame environment $\Phi$, saved and restored from the call stack
$\Sigma$.  In the alternate semantics for \lang{Mach}, the {\tt getstack}
and {\tt setstack} instructions are reinterpreted as accesses and updates
to $\Phi$, and {\tt getparent} instructions as accesses to
$\parent{\Phi}{\stk}$.

The first part of the proof shows a simulation diagram of the ``plus''
kind between executions of the original \lang{Linear} code and
executions of the generated \lang{Mach} code using the alternate
semantics outlined above.  Thanks to the alternate semantics, memory
states are identical between pairs of matching \lang{Linear} and
\lang{Mach} states.  The main invariant to be enforced is agreement
between, on the \lang{Linear} side,  the location states $L$ and
$\parent{L}{\stk}$ of the current and calling functions and, on the
\lang{Mach} side, the register state $R$ and the frame states $\Phi$
and $\parent{\Phi}{\stk'}$.
This agreement captures the following conditions:
\begin{itemize}
\item $L(r) = R(r)$ for all registers $r$;
\item $L(s) = \Phi(\tau, \Delta(s))$
  for all {\tt local} and {\tt outgoing} stack slots $s$ of type $\tau$
  used in the current function;
\item $L(s) = \parent{\Phi}{\stk'}(\tau, \Delta(s))$
  for all {\tt incoming} slots $s$ of type $\tau$ used in the
  current function;
\item $\parent{L}{\stk}(r) = \Phi(\tau, \Delta(r))$ for all
  callee-save registers $r$ of type $\tau$ used in the current function;
\item $\parent{L}{\stk}(r) = R(r)$ for all callee-save registers $r$ not
  used in the current function.
\end{itemize}
The preservation of agreement during execution steps follows mainly
from separation properties between the various areas of
activation records.

The second part of the proof shows a lock-step simulation result
between executions of \lang{Mach} programs that use the alternate and
standard semantics, respectively.  Here, the memory states $M$ (of the
alternate semantics) and $M'$ (of the standard semantics) differ:
for each activation record $b$, the block $M'(b)$ is larger than the
block $M(b)$ (because it contains additional ``frame'' data), but the
two blocks have the same contents on byte offsets valid for $M(b)$
(these offsets correspond to the \lang{Cminor} stack data).  We
capture this connection between memory states by the ``memory
extension'' ordering $M \le M'$:
$$ M \le M' \defequal
  \forall \chunk,b,\ofs,v,~ {\tt load}(\chunk,M,b,\ofs) = \Some v
     \Longrightarrow {\tt load}(\chunk,M',b,\ofs) = \Some{v}
$$
Leroy and Blazy \cite[section 5.2]{Leroy-Blazy-memory-model} study the properties of
this relation between memory states and shows that it commutes nicely
with the {\tt store}, {\tt alloc} and {\tt free} operations over memory states.

Another invariant that we must maintain is that the contents of the
block $M'(b)$ agree with the frame state $\Phi$ of the alternate
\lang{Mach} semantics:
\begin{eqnarray*}
\fn, M, M' \models \Phi \match b, \delta \defequal
\begin{array}[t]{ll}
        & b \mbox{ is valid in } M \wedge {\cal L}(M, b) = 0 \\
\wedge & {\cal L}(M', b) = \fn.{\tt stack{\char95}low} \wedge  {\cal H}(M', b) \ge 0 \\
\wedge & \delta = \fn.{\tt stack{\char95}low} \mbox{ mod } 2^{32} \\
\wedge & \forall \tau, n, ~~
          \fn.{\tt stack{\char95}low} \le n \wedge n + \size{\tau} \le 0
          \Longrightarrow {} \\
       & ~~~~~~~~~~~{\tt load}(\chunk(\tau), M', b, n) =
                     \Some{\Phi(\ty,n)}
\end{array}
\end{eqnarray*}

The invariant between alternate \lang{Mach} states and standard
\lang{Mach} states is, then, of the following form:
\begin{pannel}
\irule{
\stk \match \stk' \amper \fn, M, M' \models \Phi \match \sp, \ofs \amper M \le M' \amper
\stk' \prec \sp
}{
\state(\stk, \fn, {\tt ptr}(\sp,\ofs), c, R, \Phi, M) \match
\state(\stk', \fn, {\tt ptr}(\sp,\ofs), c, R, M')
}
\end{pannel}
The notation $\stk' \prec \sp$ means that the stack blocks $\sp_1,
\ldots, \sp_n$ appearing in the call stack $\stk'$ are such that
$\sp > \sp_1 > \ldots > \sp_n$ and are therefore pairwise distinct.
To complete the proof of simulation between the alternate and standard
semantics for \lang{Mach}, we need to exploit well-typedness
properties.  The frame environments $\Phi$ used in the alternate
semantics of \lang{Mach} satisfy the classic ``good variable'' property
$\Phi\{(\tau,\delta) \becomes v\}(\tau,\delta) = v$ regardless of
whether the value~$v$ matches the claimed type~$\tau$.  However, once
frames are mapped to memory locations, writing \eg a float value with
memory quantity {\tt int32} and reading it back with the same memory
quantity results in {\tt undef}, not the stored value.  More precisely:
\begin{lemma}
Assume ${\tt storestack}(\tau, M, \sp, \ofs, v) = \Some{M'}$.
Then, \\ ${\tt loadstack}(\tau, M', \sp, \ofs) = \Some v$ if and only if 
$v : \tau$ (as defined in section~\ref{s:values-mem}).
\end{lemma}
Therefore, we need to exploit the well-typedness of the \lang{Mach}
code in a trivial {\tt int}-or-{\tt float} type system in the style of
section~\ref{s:cminor-typing} to ensure that the values $v$ stored in
a stack location of type $\tau$ always semantically belong to type
$\tau$.  We say that a \lang{Mach} alternate execution state $\st$ is
well typed if
\begin{itemize}
\item all functions and code sequences appearing in $\st$ are well typed;
\item all abstract frames $\Phi$ appearing in $\st$ are such
that $\forall \tau,\ofs, ~ \Phi(\tau,\ofs) : \tau$;
\item all register states $R$ appearing in $\st$ are such that
$\forall r, ~ R(r) : \tau(r)$.
\end{itemize}
We can then prove a ``subject reduction'' lemma showing that
well-typedness is preserved by transitions.
\begin{lemma}
If $\transtrace{\st}{t}{\st'}$ in the \lang{Mach} abstract semantics
and $\st$ is well typed, then $\st'$ is well typed.
\end{lemma}
Combining this well-typedness property with the invariant between
alternate and standard \lang{Mach} states, we prove the following
lock-step simulation result between the two \lang{Mach} semantics.
\begin{lemma}
If $\transtrace{\st_1}{t}{\st_1'}$ in the \lang{Mach} abstract semantics
and $\st_1$ is well typed and $\st_1 \match \st_2$, there exists $\st_2'$
such that $\transtrace{\st_2}{t}{\st_2'}$ in the standard semantics
and $\st_1' \match \st_2'$.
\end{lemma}
Semantic preservation for the compiler pass that constructs activation
records then follows from the two sub-proofs of simulation outlined above.

\section{The output language: PowerPC assembly language} \label{s:ppc}

\begin{figure}

\noindent Machine instructions:

\smallskip

{
\def\column{\begin{minipage}[t]{1.9cm}\begin{alltt}}
\def\endcolumn{\end{alltt}\end{minipage}}

~~~\begin{column}
add
addi
addis
addze
and.
andc
andi.
andis.
b
bctr
bctrl
bf
bl
bs
\end{column}
\begin{column}
blr
bt
cmplw
cmplwi
cmpw
cmpwi
cror
divw
divwu
eqv
extsb
extsh
fabs
fadd
\end{column}
\begin{column}
fcmpu
fdiv
fmadd
fmr
fmsub
fmul
fneg
frsp
fsub
lbz
lbzx
lfd
lfdx
lfs
\end{column}
\begin{column}
lfsx
lha
lhax
lhz
lhzx
lwz
lwzx
mfcr
mflr
mr
mtctr
mtlr
mulli
mullw
\end{column}
\begin{column}
nand
nor
or
orc
ori
oris
rlwinm
slw
sraw
srawi
srw
stb
stbx
stfd
\end{column}
\begin{column}
stfdx
stfs
stfsx
sth
sthx
stw
stwx
subfc
subfic
xor
xori
xoris
\end{column}
}

\bigskip

Macro-instructions:
\begin{center}
\begin{tabular}[t]{l}
{\tt allocframe}, {\tt freeframe}: allocation and deallocation of a
  stack frame \\
{\tt lfi}: load a floating-point constant in a float register \\
{\tt fcti}, {\tt fctiu}: conversion from floats to integers \\
{\tt ictf}, {\tt iuctf}: conversion from integers to floats
\end{tabular}
\end{center}

\caption{The subset of PowerPC instructions used in Compcert.}
\label{f:ppc-instrs}

\end{figure}

The target language for our compiler is \lang{PPC}, an abstract syntax
for a subset of the PowerPC assembly language \cite{PowerPC-94},
comprising 82 of the 200+ instructions of this processor, plus 7
macro-instructions.  The supported instructions are listed in
figure~\ref{f:ppc-instrs}.

\subsection{Syntax} \label{s:ppc-syntax}

The syntax of \lang{PPC} has the following shape:

\begin{syntaxleft}
\syntaxclass{Integer registers:}
\ireg & ::= & {\tt R0} \alt {\tt R1} \alt \ldots \alt {\tt R31} 
\syntaxclass{Float registers:}
\freg & ::= & {\tt F0} \alt {\tt F1} \alt \ldots \alt {\tt F31} 
\syntaxclass{Condition bits:}
\crbit & ::= & {\tt CR0} \alt {\tt CR1} \alt {\tt CR2} \alt {\tt CR3}
\syntaxclass{Constants:}
\cst & ::= & n \alt {\tt lo16}(\id + \ofs) \alt {\tt hi16}(\id + \ofs)
\syntaxclass{Instructions:}
i & ::= & {\tt label}(\lbl) \alt {\tt bt}(\crbit, \lbl) \\
  & \alt & {\tt add}(\ireg, \ireg', \ireg'') \alt {\tt addi}(\ireg, \ireg', \cst) \\
  & \alt & {\tt fadd}(\freg, \freg', \freg'') \alt \ldots
\syntaxclass{Internal functions:}
\fn & ::= & i^*
\end{syntaxleft}

\lang{PPC} is an assembly language, not a machine language.  This is
apparent in the use of symbolic labels in branch instructions such as
{\tt bt}, and in the use of symbolic constants ${\tt lo16}(\id + \ofs)$ and
${\tt hi16}(\id + \ofs)$ as immediate operands of some instructions.
(These constants, resolved by the linker, denote the low-order and
high-order 16 bits of the memory address of symbol $\id$ plus offset
$\ofs$.)

Moreover, \lang{PPC} features a handful of macro-instructions that
expand to canned sequences of actual instructions during
pretty-printing of the abstract syntax to concrete assembly syntax.
These macro-instructions include allocation and deallocation of the
stack frame (mapped to arithmetic on the stack pointer register),
conversions between integers and floats (mapped to memory transfers
and complicated bit-level manipulations of IEEE floats), and loading
of a floating-point literal (mapped to a load from a memory-allocated
constant).  The reason for treating these operations as basic
instructions is that the memory model and the axiomatization of IEEE
float arithmetic that we use are too abstract to verify the
correctness of the corresponding canned instruction sequences.  (For
example, our memory model cannot express that two abstract stack
frames are adjacent in memory.)  We leave this verification to
future work, but note that these canned sequences of instructions are
identical to those used by GCC and other PowerPC compilers and
therefore have been tested extensively.

\subsection{Semantics} \label{s:ppc-semantics}

Program states in \lang{PPC} are pairs $\st ::= (R, M)$ of a memory
state $M$ and a register state $R$ associating values to the processor
registers that we model, namely integer registers $\ireg$,
floating-point registers $\freg$, bits 0 to 3 of the condition
register {\tt CR}, the program counter {\tt PC}, and the special ``link'' and
``counter'' registers {\tt LR} and {\tt CTR}.  

The core of \lang{PPC}'s operational semantics is a transition
function $T(i, \st) = \Some{\st'}$ that determines the state $\st'$
after executing instruction $i$ in initial state $\st$.  In
particular, the program counter {\tt PC} is incremented (for fall-through
instructions) or set to the branch target (for branching instructions).
We omit the definition of $T$ in this article, as it is a very large
case analysis, but refer the reader to the Coq development for more
details.  The semantics of \lang{PPC}, then, is defined by just two
transition rules:

\begin{pannel}
\irule{
R({\tt PC}) = {\tt ptr}(b,n) \amper G(b) = \Some{{\tt internal}(c)} \amper c\#n = \Some{i} \\
T(i, (R, M)) = \Some{(R', M')}
}{
\trans {(R, M)} {(R', M')}
}
\irule{
R({\tt PC}) = {\tt ptr}(b,0) \amper G(b) = \Some{{\tt external}(\ef)} \\
{\tt extcall{\char95}arguments}(R, M, \ef.{\tt sig}) = \Some{\vec v} \amper
\ef(\vec v) \evaltrace{t} v \\
R' = R\{{\tt PC} \becomes R({\tt LR}),~ {\tt extcall{\char95}result}(\ef.{\tt sig}) \becomes v\}
}{
\transtrace {(R, M)} {t} {(R', M)}
}
\end{pannel}
The first rule describes the execution of one \lang{PPC} instruction
within an internal function.  The notation $c\#n$ denotes the $n\th$
instruction in the list $c$.  The first three premises model
abstractly the act of reading and decoding the instruction pointed to
by the program counter {\tt PC}.  For simplicity, we consider that all
instructions occupy one byte in memory, so that incrementing the
program counter corresponds to branching to the next instruction.  It
would be easy to account for variable-length instructions (\eg 4 bytes for
regular instructions, 0 for labels, and $4n$ bytes for macro-instructions).
The second rule describes the big-step execution of an invocation of
an external function.  The {\tt extcall{\char95}arguments} function extracts the
arguments to the call from the registers and stack memory locations
prescribed by the signature of the external call; likewise,
{\tt extcall{\char95}result} denotes the register where the result must be stored.
Conventionally, the address of the instruction following the call is
found in register {\tt LR}; setting {\tt PC} to $R({\tt LR})$ therefore returns to
the caller.
\subsection{Determinism} \label{s:ppc-determinism}
Determinism of the target language plays an important role in the
general framework described in section~\ref{s:framework}.  It is
therefore time to see whether the semantics of \lang{PPC} is
deterministic.  In the general case, the answer is definitely ``no'':
in the rule for external function calls, the result value $v$ of the
call is unconstrained and can take any value, resulting in different
executions for the same \lang{PPC} program.  However, this is a form
of external nondeterminism: the choice of the result value $v$ is not
left to the program, but is performed by the ``world'' (operating
system, user interaction, \ldots) in which the program executes.  As
we now formalize, if the world is deterministic, so is the semantics
of \lang{PPC}.
A deterministic world is modeled as a partial function $W$ taking the
identifier of an external call and its argument values and returning
the result value of the call as well as an updated world $W'$.  A
finite trace $t$ or infinite trace $T$ is legal in a world $W$,
written $W \models t$ or $W \models T$, if the result values of external calls
recorded in the trace agree with what the world $W$ imposes:
\begin{pannel}
\srule{ W \models \epsilon }
\irule{
W(\id,\vec v) = \Some{v, W'} \amper W' \models t
}{
W \models \event{\id}{\vec v}{v}. t
}
\iruledouble{
W(\id,\vec v) = \Some{v, W'} \amper W' \models T
}{
W \models \event{\id}{\vec v}{v}. T
}
\end{pannel}
We extend this definition to program behaviors in the obvious way:
\begin{pannel}
\irule{
W \models t
}{
W \models {\tt converges}(t, n)
}
\irule{
W \models T
}{
W \models {\tt diverges}(T)
}
\irule{
W \models t
}{
W \models {\tt goeswrong}(t)
}
\end{pannel}
We could expect that a \lang{PPC} program has at most one behavior $B$
that is legal in a deterministic initial world $W$.  This is true for
terminating behaviors, but for diverging behaviors a second source of
apparent nondeterminism appears, caused by the coinductive definition
of the infinite closure relation $\transinf{\st}{T}$ in
section~\ref{s:trans-sem}.  Consider a program $P$ that diverges
silently, such as an infinite empty loop.  According to the
definitions in section~\ref{s:trans-sem}, this program has behaviors
${\tt diverges}(T)$ for any finite or infinite trace $T$, not just $T =
\epsilon$ as expected.  Indeed, no finite observation of the execution
of $P$ can disprove the claim that it executes with a trace $T \not=
\epsilon$.  However, using classical logic (the axiom of excluded
middle), it is easy to show that the set of possible traces admits a
minimal element for the prefix ordering between traces $T$.  This
minimal trace is infinite if the program is {\em reactive}
(performs infinitely many I/O operations separated by finite numbers
of internal computation steps) and finite otherwise (if the program
eventually loops without performing any I/O).  
By restricting observations to legal behaviors with minimal traces, we
finally obtain the desired determinism property for \lang{PPC}.
\begin{theorem}
Let $P$ be a \lang{PPC} program, $W$ be a deterministic initial world,
and $P \Downarrow B$ and $P \Downarrow B'$ be two executions of $P$.
If $W \models B$ and $W \models B'$ and moreover the traces of $B$ and $B'$ are
minimal, then $B' = B$ up to bisimilarity of infinite traces.
\end{theorem}

\section{Generation of PowerPC assembly language} \label{s:ppcgen}

\subsection{Code generation}

The final compilation pass of Compcert translates from \lang{Mach} to
\lang{PPC} by expanding \lang{Mach} instructions into canned sequences
of \lang{PPC} instructions.  For example, a \lang{Mach} conditional
branch ${\tt cond}(\cond, \vec r, \ltrue)$ becomes a {\tt cmplw}, {\tt cmplwi},
{\tt cmpw}, {\tt cmpwi} or {\tt fcmp} instruction that sets condition bits,
followed in some cases by a {\tt cror} instruction to merge two condition
bits, followed by a {\tt bt} or {\tt bf} conditional branch.  Moreover,
\lang{Mach} registers are injected into \lang{PPC} integer or float
registers. 

The translation deals with various idiosyncrasies of the PowerPC
instruction set, such as the limited range for immediate arguments to
integer arithmetic operations, and the fact that register {\tt R0} reads
as 0 when used as argument to some instructions.  Two registers ({\tt R2}
and {\tt F13}) are reserved for these purposes.  The translation is
presented as a number of ``smart constructor'' functions that
construct and combine sequences of \lang{PPC} operations.  To give the
flavor of the translation, here are the smart constructors for ``load
integer immediate'' and ``add integer immediate''.  The functions
${\tt low}(n)$ and ${\tt high}(n)$ compute the signed 16-bit integers such that
$n = {\tt high}(n) \times 2^{16} + {\tt low}(n)$.
\begin{eqnarray*}
{\tt loadimm}(r,n) & = & \threecases
    {{\tt addi}(r, {\tt R0}, n)}{if ${\tt high}(n) = 0$;}
    {{\tt addis}(r, {\tt R0}, {\tt high}(n))}{if ${\tt low}(n) = 0$;}
    {{\tt addis}(r, {\tt R0}, {\tt high}(n)); {\tt ori}(r, r, {\tt low}(n))}{otherwise} \\
{\tt addimm}(r_d,r_s,n) & = & \fourcases
    {{\tt loadimm}({\tt R2}, n); {\tt add}(r_d,r_s,{\tt R2})
       \hbox to 0pt{~~~~~~~~~~if $r_d = {\tt R0}$ or $r_s = {\tt R0}$;\hss}}{}
    {{\tt addi}(r_d, r_s, n)}{if ${\tt high}(n) = 0$;}
    {{\tt addis}(r_d, r_s, {\tt high}(n))}{if ${\tt low}(n) = 0$;}
    {{\tt addis}(r_d, r_s, {\tt high}(n)); {\tt addi}(r_d, r_d, {\tt low}(n))}{otherwise}
\end{eqnarray*}

Just as the generation of \lang{Mach} code must fail if the activation
record is too large to be addressed by machine integers
(section~\ref{s:stacking-transf}), the \lang{PPC} generator must fail
if the translation of a function contains $2^{31}$ or more
instructions, since it would then be impossible to address some of the
instructions via a signed 32-bit offset from the beginning of the
function.

\subsection{Semantic preservation} \label{s:ppcgen-proof}

Semantic preservation for \lang{PPC} generation is proved using a
simulation diagram of the ``option'' type.  The two main invariants to
be preserved are:
\begin{enumerate}
\item The {\tt PC} register contains a pointer ${\tt ptr}(b,\ofs)$ that
  corresponds to the \lang{Mach} function $\fn$ and code sequence $c$
  currently executing:
$$ \fn, c \match {\tt ptr}(b,\ofs) \defequal
   G(b) = \Some{{\tt internal}(\fn)} ~\wedge~ \tr{c} = {\tt suffix}(\tr{\fn}, \ofs) $$
\item The \lang{Mach} register state $R$ and stack pointer $\sp$
  agree with the \lang{PPC} register state $R'$:
$$ R, \sp \match R' \defequal R'({\tt R1}) = \sp ~\wedge~
                              \forall r, R(r) = R'(\overline r) $$
where $\overline r$ denotes the \lang{PPC} register associated with
the \lang{Mach} register $r$.  (Conventionally, the {\tt R1} register is
used as the stack pointer.)
\end{enumerate}
More precisely, matching between a \lang{Mach} execution state and a
\lang{PPC} execution state is defined as follows:

\begin{pannel}
\irule{
\stk {\tt \ wf} \amper \fn, c \match R'({\tt PC}) \amper R, \sp \match R'
}{
\state(\stk, \fn, \sp, c, R, M) \match (R', M)
}
\irule{
\stk {\tt \ wf} \amper R'({\tt PC}) = {\tt ptr}(b, 0) \amper G(b) = \Some{\fd} \amper
R'({\tt LR}) = \parent{\retaddr}{\stk} \amper R, \parent{\sp}{\stk} \match R'
}{
\callstate(\stk, \fd, R, M) \match (R', M)
}
\irule{
\stk {\tt \ wf} \amper R'({\tt PC}) = \parent{\retaddr}{\stk} \amper
R, \parent{\sp}{\stk} \match R'
}{
\returnstate(\stk, R, M) \match (R', M)
}
\end{pannel}
The invariant $\stk {\tt \ wf}$ over \lang{Mach} call stacks is defined as
$\fn, c \match \retaddr$ for all $\stackframe(\fn,\sp,\retaddr,c) \in \stk$.
The case for call states reflects the convention that the caller saves
its return address in register {\tt LR} before jumping to the first
instruction of the callee.
As mentioned in section~\ref{s:mach}, the \lang{Mach} semantics is
parameterized by an oracle ${\tt retaddr}(\fn, c, \retaddr)$ that guesses
the code pointer $\retaddr$ pointing to the \lang{PPC} code
corresponding to the translation of \lang{Mach} code $c$ within
function $\fn$.  We construct a suitable oracle by noticing that the
translation of a \lang{Mach} code $i_1 \ldots i_n$ is simply the
concatenation $\tr{i_1} \ldots \tr{i_n}$ of the translations of the
instructions.  Therefore, the offset of the return address $\retaddr$
is simply the position of the suffix $\tr{c}$ of $\tr{\fn}$.  It is
always uniquely defined if $c$ is a suffix of $\fn.{\tt code}$, which the
\lang{Mach} semantics guarantees.
The simulation diagram is of the ``star'' kind because \lang{Mach}
transitions from return states $\returnstate$ to regular states
$\state$ correspond to zero transitions in the generated \lang{PPC}
code.  The absence of infinite stuttering is trivially shown by
associating measure 1 to return states and measure 0 to regular and
call states.
The proof of the simulation diagram is long because of the high number
of cases to consider but presents no particular difficulties.  
To reason more easily about the straight-line execution of a sequence of
non-branching \lang{PPC} instructions, we make heavy use of the
following derived execution rule:
\begin{pannel}
\irule{
R_0({\tt PC}) = {\tt ptr}(b,\ofs) \amper G(b) = \Some{{\tt internal}(c)} \amper
{\tt suffix}(c, \ofs) = i_1 \ldots i_n.c' \amper n > 0 \\
\mbox{for all $k \in \{1, \ldots, n\}$,~}
T(i_k, (R_{k-1}, M_{k-1})) = \Some{(R_k, M_k)} \mbox{ and }
R_k({\tt PC}) = R_{k-1}({\tt PC}) + 1
}{
\transp {(R_0, M_0)}{\epsilon}{(R_n, M_n)}
}
\end{pannel}

\section{The Coq development} \label{s:coq-devel}

We now discuss the Coq implementation of the algorithms and
proofs described in this paper.  

\subsection{Specifying in Coq}

The logic implemented by Coq is the Calculus of Inductive
and Coinductive Constructions (CICC), a very powerful constructive
logic.  It supports equally well three familiar styles of writing
specifications: by functions and pattern-matching, by inductive or
coinductive predicates representing inference rules, and by ordinary
predicates in first-order logic.  All three styles are used in the
Compcert development, resulting in specifications and statements of
theorems that remain quite close to what can be found in programming
language research papers. 

CICC also features higher-order logic, dependent types, a hierarchy of
universes to enforce predicativity, and an ML-style module system.
These advanced logical features are used sparingly in our development:
higher-order functions and predicates are used for iterators over data
structures and to define generic closure operations; simple forms of
dependent types are used to attach logical invariants to some data
structures and some monads (see section~\ref{s:rtlgen-impl});
coinduction is used to define and reason about infinite transition
sequences; parameterized modules are used to implement the generic
dataflow solvers of section~\ref{s:dataflow-solvers}.  Most of our
development is conducted in first-order logic, however.

Two logical axioms (not provable in Coq) are used: function
extensionality (if $\forall x,~f(x) = g(x)$, then the functions $f$
and $g$ are equal) and proof irrelevance (any two proof terms of the
same proposition $P$ are equal).  It might be possible to avoid using
these two axioms, but they make some proofs significantly easier.  The
proofs of the semantic preservation theorems are entirely
constructive.  The only place where classical logic is needed, under
the form of the axiom of excluded middle, is to show the existence of
minimal traces for infinite executions (see
section~\ref{s:ppc-determinism}).  This set of three axioms is
consistent with the predicative version of Coq's logic
that we use.

\subsection{Proving in Coq}

While Coq proofs can be checked a posteriori in batch mode, they are
developed interactively using a number of tactics as elementary proof
steps.  The sequence of tactics used constitutes the proof script.
Building such scripts is surprisingly addictive, in a videogame kind
of way, but reading them is notoriously difficult.  We were initially
concerned that adapting and reusing proof scripts when specifications
change could be difficult, forcing many proofs to be rewritten from
scratch.  In practice, careful decomposition of the proofs in separate
lemmas enabled us to reuse large parts of the development even in the
face of major changes in the semantics, such as switching from the
``mixed-step'' semantics described in \cite{Leroy-compcert-06} to the
small-step transition semantics described in this paper.  

Our proofs frequently use the basic proof automation facilities
provided by Coq, mostly {\tt eauto} (Prolog-style resolution), {\tt omega}
(Presburger arithmetic) and {\tt congruence}
(a decision procedure for ground equalities with uninterpreted symbols).
However, these tactics do not combine automatically, and significant
manual massaging of the goals is necessary before they apply.  The
functional induction and functional inversion mechanisms of 
Coq~8.1 \cite{Barthe-func-ind-06} helped reason about functions
defined by complex pattern-matching.

Coq also provides a dedicated  \lang{Ltac} language for users to
define their own tactics.  We used this facility occasionally, for
instance to define a ``monadic inversion'' tactic that recursively
simplifies hypotheses of the form $({\tt do}~x$~{\tt <-}~$a; b)~s =
{\tt OK}(s',r)$ into 
$\exists s_1.~ \exists ~x.~ ~a~s = {\tt OK}(s_1, x) ~\wedge~ b~s_1 = {\tt OK}(s', r)$.
There is no doubt that a Coq expert could have found more
opportunities for domain-specific tactics and could have improved
significantly our proof scripts.

As mentioned earlier, most of the development is conducted in
first-order logic, suggesting the possibility of using automated
theorem provers such as SMT solvers.  Preliminary experiments with
using SMT solvers to prove properties of the Compcert memory model
\cite{Leroy-Blazy-memory-model} indicate that many but not all of the
lemmas can be proved automatically.  While fully automated
verification of a program like Compcert appears infeasible with
today's technology, we expect that our interactive proof scripts would
shrink significantly if Coq provided a modern SMT solver as one of its
tactics.

\subsection{Size of the development}

\bgroup

\def\rotatetext#1{
  \begin{tikzpicture}
  \draw (0,0) node[rotate=90]{#1};
  \end{tikzpicture}
}

\begin{figure}

\normalbar

\begin{tabular}{@{}|l|*{7}{@{~}r@{~}|}}
\hline
& \rotatetext{Code (in Coq)}
& \rotatetext{Code (in Caml)}
& \rotatetext{Specifications}
& \rotatetext{Statements}
& \rotatetext{Proof scripts}
& \rotatetext{Directives}
& \rotatetext{Total} \\
\hline
Data structures, proof auxiliaries
   &   413 &    83 &     - &   471 &  1127 &   336 &  2430 \\
Integers, floats, values, memory model
   &   850 &    21 &     - &  1196 &  2819 &   119 &  5005 \\
Common syntactic and semantic definitions
   &   122 &     - &   111 &   439 &   775 &   517 &  1964 \\
Cminor semantics
   &     - &     - &   563 &    39 &   186 &    85 &   873 \\
CminorSel semantics, operators
   &     - &     - &   639 &    88 &   221 &    85 &  1033 \\
Instruction selection, reassociation
   &   838 &     - &     - &   399 &   755 &    82 &  2074 \\
RTL semantics
   &     - &     - &   403 &    62 &   198 &    69 &   732 \\
RTL generation
   &   413 &    55 &     - &   931 &  1237 &   161 &  2797 \\
Dataflow optimizations over RTL
   &  1162 &     - &     - &   596 &  1586 &   282 &  3626 \\
Calling conventions
   &   117 &     - &     - &   127 &   315 &    37 &   596 \\
LTL semantics
   &     0 &     0 &   414 &    48 &   129 &    40 &   631 \\
Register allocation
   &   303 &   544 &     0 &   564 &   975 &   190 &  2576 \\
LTLin semantics
   &     0 &     0 &   243 &     5 &    10 &    27 &   285 \\
Branch tunneling and code linearization
   &   137 &    52 &     0 &   349 &   674 &    97 &  1309 \\
Linear semantics
   &     0 &     0 &   280 &     3 &     4 &    27 &   314 \\
Spilling and reloading
   &   333 &     0 &     0 &   763 &  1669 &   171 &  2936 \\
Mach semantics
   &     - &     - &   545 &    38 &   106 &    73 &   762 \\
Layout of the activation record
   &   236 &     - &     - &   801 &  1510 &   203 &  2750 \\
PPC semantics
   &     - &     - &   520 &     6 &     7 &    29 &   562 \\
PPC generation
   &   424 &     - &     - &   832 &  1904 &   171 &  3331 \\
Compiler driver, PPC printer
   &    34 &   290 &     - &   223 &   336 &   101 &   984 \\
\hline
Total
   &  5382 &  1045 &  3718 &  7980 & 16543 &  2902 & 37570 \\
   & 14\% & 3\% & 10\% & 21\% & 44\% & 8\% & \\
\hline

\end{tabular}

\caption{Size of the development (in non-blank, non-comment lines of code).}
\label{f:size}

\end{figure}

\egroup

The size of the Coq development can be estimated from the line counts
given in figure~\ref{f:size}.  The whole development represents
approximately 37000 lines of Coq (excluding comments and blank lines)
plus 1000 lines of code directly written in Caml.  The overall effort
represents approximately 2~person-years of work.

The Coq datatype and function definitions that implement the compiler
itself (column ``Code in Coq'' in figure~\ref{f:size}) account for
14\% of the source.  In other terms, the Coq verification is about 6
times larger than the program being verified.  The remaining 86\%
comprise 10\% of specifications (mostly, the operational semantics for
the source, target and intermediate languages), 21\% of statements of
lemmas, theorems and supporting definitions, 44\% of proof scripts and
8\% of directives, module declarations, and custom proof tactics.

The sizes of individual passes are relatively consistent: the most
difficult passes (RTL generation, register allocation, spilling and
reloading, and layout of the activation record) take between 2000 and
3000 lines each, while simpler passes (constant propagation, CSE,
tunneling, linearization) take less than 1500 lines each.  One outlier
is the PPC code generation pass which, while conceptually simple,
involves large definitions and proofs by case analysis, totaling more than
3300 lines.  Among the supporting libraries, the formalizations of
machine integer arithmetic and of the memory model are the largest and
most difficult, requiring 1900 and 2300 lines respectively.

Checking all the proofs in the development takes about 7.5 minutes of
CPU time on a 2.4 GHz Intel Core 2 processor equipped with 4 Gb of RAM
and running 64-bit Linux.  The version of Coq used is 8.1pl3.
Parallel {\tt make} with 2 cores results in a wall-clock time of 4.5
minutes.

\section{Experimental results} \label{s:experimental}

\subsection{Extracting an executable compiler}

As mentioned in the Introduction, the verified parts of the Compcert
compiler are programmed directly in Coq, then automatically translated
to executable Caml code using Coq's extraction facility
\cite{Letouzey-02,Letouzey-08}.  To obtain an executable compiler,
this extracted Caml code is combined with:
\begin{itemize}
\item A compiler front-end, translating the \lang{Clight} subset of
  C to \lang{Cminor}.  This front-end is itself extracted from a Coq
  development.  An earlier version of this front-end compiler is
  described in \cite{Blazy-Dargaye-Leroy-06}.
\item A parser for C that generates \lang{Clight} abstract syntax.
  This parser is built on top of the CIL library \cite{CIL2}.
\item Hand-written Caml implementations of the heuristics that we
  validate a posteriori (graph coloring, \lang{RTL} type
  reconstruction, etc), of the pretty-printer for PowerPC assembly
  code, and of a {\tt cc}-style compiler driver.
\end{itemize}
The resulting compiler runs on any platform supported by Caml and
generates PowerPC code that runs under MacOS~X.  While the
soundness proof for Compcert does not account for separate
compilation and assumes that whole programs are compiled at once, the
compiler can be used to separately compile C source files and link
them with precompiled libraries, which is convenient for testing.
(The calling conventions implemented by Compcert are compatible enough
with the standard PowerPC ABI to support this.)

Program extraction performs two main tasks.  First, it eliminates
the parts of Coq terms that have no computational content, by a
process similar to program slicing.  For instance, if a data structure
carries a logical invariant, every instance of this structure contains
a proof term showing that the invariant is satisfied.  This subterm does
not contribute to the final result of the program, only to its
correctness, and is therefore eliminated by extraction.  The second
task is to bridge the gap between Coq's rich type system and Caml's
simpler Hindley-Milner type system.  Uses of first-class polymorphism
or general dependent types in Coq can lead to programs that are not
typeable in Caml; extraction works around this issue by inserting
unsafe Caml coercions, locally turning off Caml's type checking.  
This never happens in Compcert, since the source Coq code is written
in pedestrian ML style, using only Hindley-Milner types.

Generally speaking, the Caml code extracted from the Compcert
development looks like what a Caml programmer would write if confined
to the purely functional subset of the language.  There are two
exceptions.  The first is related to the handling of default cases in
Coq pattern-matching.  Consider a data type with 5 constructors {\tt A} to
{\tt E}, and a pattern-matching {\tt (A,\ A)\ ->\ x\ {\char124}\ ({\char95},\ {\char95})\ ->\ y}.  Internally,
Coq represents this definition by a complete matching having $5 \times
5$ cases, 24 of which are {\tt y}.  Extraction does not yet re-factor the
default case, resulting in 24 copies of the code {\tt y}.  On large data
types, this can lead to significant code explosion, which we limited
on a case by case basis, often by introducing auxiliary functions.

The other problem with extraction we encountered is the
$\eta$-expansions that extraction sometimes performs in the hope of
exposing more opportunities for program slicing.  These expansions can
introduce inefficiencies by un-sharing computations.  Consider for
example a curried function of two arguments $x$ and $y$ that takes
$x$, performs an expensive computation, then returns a function
$\lambda y \ldots$.  After $\eta$-expansion, this expensive
computation is performed every time the second argument is passed.  We
ran into this problem twice in Compcert.  Manual patching of the
extracted Caml code was necessary to undo this ``optimization''.

\subsection{Benchmarks}

\bgroup

\def\doline#1#2#3#4#5#6{
  #1 & #2 & #3\% & #4\% & #5\% &
\begin{tikzpicture}[x=0.044cm,y=0.5cm]
\filldraw[fill=black!60] (0,0) rectangle (#5,0.2);
\filldraw[fill=black!25] (0,0.2) rectangle (#4,0.4);
\filldraw[fill=black] (0,0.4) rectangle (#3,0.6);
\filldraw[fill=white] (0,0.6) rectangle (100,0.8);
\end{tikzpicture}
}

\def\legend#1{
\begin{tikzpicture}
\filldraw[fill=#1] (0,0) rectangle (0.3,0.15);
\end{tikzpicture}
}

\begin{figure}
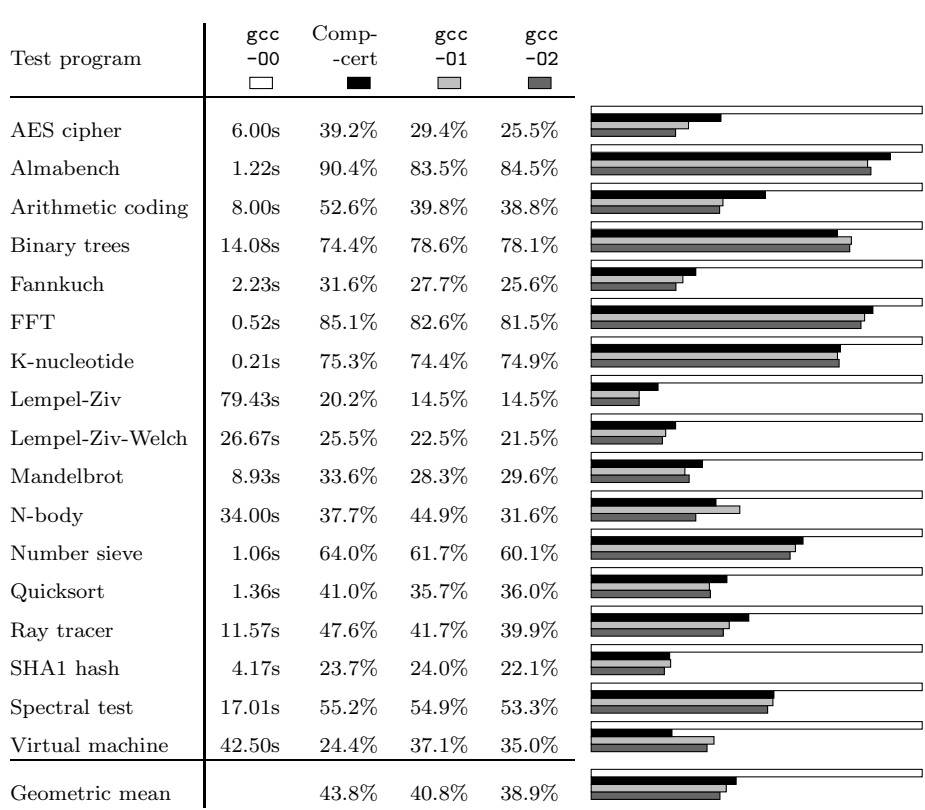


\normalbar
\begin{tabular}{@{}l|rrrrl}
             & {\tt gcc} & Comp- & {\tt gcc} & {\tt gcc} \\
Test program & {\tt -O0} & -cert & {\tt -O1} & {\tt -O2} \\ 
             & \legend{white} & \legend{black} & \legend{black!25} & \legend{black!60} \\
\cline{1-5}
\\[-2mm]
\doline{AES cipher      }{6.00s}{39.2}{29.4}{25.5}{25.8} \\
\doline{Almabench       }{1.22s}{90.4}{83.5}{84.5}{78.9} \\
\doline{Arithmetic coding}{8.00s}{52.6}{39.8}{38.8}{37.3} \\
\doline{Binary trees     }{14.08s}{74.4}{78.6}{78.1}{75.4} \\
\doline{Fannkuch        }{2.23s}{31.6}{27.7}{25.6}{25.6} \\
\doline{FFT             }{0.52s}{85.1}{82.6}{81.5}{80.9} \\
\doline{K-nucleotide     }{0.21s}{75.3}{74.4}{74.9}{73.5} \\
\doline{Lempel-Ziv      }{79.43s}{20.2}{14.5}{14.5}{13.8} \\
\doline{Lempel-Ziv-Welch}{26.67s}{25.5}{22.5}{21.5}{21.4} \\
\doline{Mandelbrot      }{8.93s}{33.6}{28.3}{29.6}{28.6} \\
\doline{N-body           }{34.00s}{37.7}{44.9}{31.6}{31.7} \\
\doline{Number sieve    }{1.06s}{64.0}{61.7}{60.1}{60.0} \\
\doline{Quicksort       }{1.36s}{41.0}{35.7}{36.0}{36.0} \\
\doline{Ray tracer      }{11.57s}{47.6}{41.7}{39.9}{37.7} \\
\doline{SHA1 hash       }{4.17s}{23.7}{24.0}{22.1}{22.0} \\
\doline{Spectral test   }{17.01s}{55.2}{54.9}{53.3}{52.8} \\
\doline{Virtual machine }{42.50s}{24.4}{37.1}{35.0}{35.0} \\
\cline{1-5}
\\[-2mm]
\doline{Geometric mean  }{}{43.8}{40.8}{38.9}{38.2}
\end{tabular}

\caption{Execution times of compiled code.  Times are given relative
  to those obtained with {\tt gcc\ -O0}.  Lower percentages and shorter
  bars mean faster.}
\label{f:perfs}

\end{figure}

\egroup

Performance of the generated code can be estimated from the timings
given in figure~\ref{f:perfs}.  Since Compcert accepts only a subset
of the C language (excluding variadic functions and {\tt long\ long}
arithmetic types, for instance), standard benchmark suites cannot be
used, and we reverted to a small home-grown test suite.  The test
programs range from 50 to 3000 lines of C code, and include
computational kernels (FFT, N-body, etc.), cryptographic primitives
(AES, SHA1), text compression algorithms, a virtual machine
interpreter, and a ray tracer derived from the ICFP 2000 programming
contest.  The PowerPC code generated by Compcert was benchmarked
against the code generated by GCC version 4.0.1 at optimization levels
0, 1 and 2.  (Higher GCC optimization levels make no significant
differences on this test suite.) For the purpose of this benchmark,
Compcert was allowed to recognize the fused multiply-add and
multiply-sub PowerPC instructions.  (These instructions are normally
not used in Compcert because they produce results different from a
multiply followed by an add or sub, but since GCC uses them
nonetheless, it is fair to allow Compcert to do so as well.)
Measurements were performed on an Apple PowerMac workstation with two
2.0 GHz PowerPC 970 (G5) processors and 6 Gb of RAM, running MacOS~10.4.11.

As the timings in figure~\ref{f:perfs} show, Compcert generates code
that is more than twice as fast as that generated by GCC without
optimizations, and competitive with GCC at optimization levels 1 and 2.
On average, Compcert code is only 7\% slower than {\tt gcc\ -O1} and 12\%
slower than {\tt gcc\ -O2}.  The test suite is too small to draw definitive
conclusions, but these results strongly suggest that 
while Compcert is not going to win any prize in high performance
computing, the performance of generated programs is
adequate for critical embedded code.

Compilation times are higher for Compcert than for GCC but remain
acceptable: to compile the 3000-line ray tracer, Compcert takes 4.6\,s
while {\tt gcc\ -O1} takes 2.7\,s.  There are several possible reasons for
this slowdown.  One is that Compcert proceeds in a relatively high
number of passes, each of which reconstructs entirely the
representation of the function being compiled.  Another is the use of
purely functional data structures (balanced trees) instead of
imperative data structures (bitvectors, hash tables).  We were
careful, however, to use functional data structures with good
asymptotic complexity (mainly AVL trees and radix-2 trees), which
cost only a logarithmic factor compared with imperative data structures.

\section{Discussion and perspectives} \label{s:perspectives}

We now discuss some of the design choices and limitations of our work
and outline directions for future work.

\subsection{On Cminor as a target language}

The \lang{Cminor} intermediate language was designed to allow
relatively direct translation of a large subset of the C language.
In particular, the memory model closely matches that of C, and the
{\tt block}/{\tt exit} mechanism makes it easy to translate C loops (including
{\tt break} and {\tt continue}) in a compositional manner.  The C feature that
appears most difficult to support is variadic functions; this is not
surprising given that variadic functions account for most of the
complexity of function calling conventions in C compilers.  

Other features of C could be supported with small extensions to
\lang{Cminor} and the Compcert back-end.  For instance, \lang{Cminor}
currently performs all floating-point arithmetic in double precision,
but it is planned to add single-precision float operators to better
support ISO~C's floating-point semantics.  Also, large {\tt switch}
statements could be compiled more efficiently if multi-way branches
(jump tables) were added to \lang{RTL} and later intermediate
languages.

Less obviously, \lang{Cminor} can also be used as a target language
when compiling higher-level source languages.  Function pointers and
tail-call optimization are supported, enabling the compilation of
object-oriented and functional languages.  Exceptions in the style of
ML, C{\tt ++} or Java are not supported natively in \lang{Cminor}
but could be encoded (at some run-time cost) either as special
return values or using continuation-passing style.
For the former approach, it could be worthwhile to add functions with
multiple return values to \lang{Cminor}.

For languages with automatic memory management, \lang{Cminor} provides
no native support for accurate garbage collection: mechanisms for
tracking GC roots through register and stack allocation in the style
of C{\tt --} \cite{Peyton-Jones-Ramsey-99} are not provided and appear
difficult to specify and prove correct.  However, the \lang{Cminor}
producer can explicitly register GC roots in \lang{Cminor} stack
blocks, in the style of Henderson \cite{Henderson-uncooperative}.  Zaynah
Dargaye has prototyped a verified front-end
compiler from the mini-ML functional language to \lang{Cminor} that
follows this approach \cite{Dargaye-these}.

\subsection{On retargeting}

While some parts of the Compcert back-end are obviously specific to
the PowerPC (\eg generation of assembly language,
section~\ref{s:ppcgen}), most parts are relatively independent of the
target processor and could conceivably be reused for a different
target.  To make this claim more precise, we experimented with
retargeting the Compcert back-end for the popular ARM processor.
The three aspects of this port that required significant changes in
the back-end and in its proof are:
\begin{itemize}
\item Reflecting the differences in instruction sets, the types and
  semantics of machine-specific operators, addressing modes and
  conditions (section~\ref{s:cminorsel}) change.  This impacts the
  instruction selection pass (section~\ref{s:selection}) but also the
  abstract interpretation of these operators performed by constant
  propagation (section~\ref{s:constprop}).
\item Calling conventions and stack layout differ.  Most differences
  are easy to abstract over, but the standard ARM calling convention
  could not be supported: it requires that floats are passed in pairs of
  integer registers, which our value and memory model cannot express yet.  
  Nonstandard calling conventions, using float registers to pass
  floats, had to be used instead.
\item ARM has fewer registers than PowerPC (16 integer registers and 8
  float registers instead of 32 and 32).  Consequently, we had to
  reduce the number of registers reserved to act as temporary
  registers.  This required some changes in the spilling and reloading
  pass (section~\ref{s:reload}).
\end{itemize}
Overall, the port of Compcert and its proof to the ARM processor took
about 3~weeks.  Three more weeks were needed to revise the
modular structure of the Coq development, separating the
processor-specific parts from the rest, adapting the PowerPC and
ARM-specific parts so that they have exactly the same interface,
and making provisions for supporting other target processors in the
future.  Among the 37500 lines of the initial development,
28000 (76\%) were found to be processor-independent and
8900 (24\%) ended up in PowerPC-specific modules; an additional
8800 lines were added to support the ARM processor.

\subsection{On optimizations}

As mentioned in the Introduction, the main objective for the Compcert
project was to prove end-to-end semantic preservation.  This led us to
concentrate on non-optimizing transformations that are required in a
compiler and to spend less time on optimizations that are optional.
Many interesting optimizations \cite{Muchnick,Appel-tiger} remain to be proved
correct and integrated in Compcert.

In separate work not yet part of Compcert,
Jean-Baptiste Tristan
verified two additional optimizations: instruction scheduling by list
scheduling and trace scheduling \cite{Ellis-Bulldog} and lazy code
motion (LCM) \cite{Knoop-LCM}.  These optimizations are more advanced
than those described in section~\ref{s:dataflow}, since they move
instructions across basic blocks and even across loop boundaries in
the case of LCM.  In both cases, Tristan used a translation validation
approach (see section~\ref{s:translation-validation}) where the code
transformation is performed by untrusted Caml code, then verified a
posteriori using a verifier that is formally proved correct in Coq.
In the case of instruction scheduling, validation is performed by
symbolic execution of extended basic blocks
\cite{Tristan-Leroy-scheduling}.  For LCM, validation exploits
equations between program variables obtained by an available expressions
analysis, combined with an anticipability analysis \cite{Tristan-Leroy-LCM}.
On these two examples, the verified translation validation approach
was effective, resulting in relatively simple semantic correctness
proofs that are insensitive to the many heuristics decisions taken by
these two optimizations.  We conjecture that for several other
optimizations, the verified translation validation approach is
simpler than proving directly the correctness of the optimization.

Many advanced optimizations are formulated in terms of static single
assignment (SSA) representation rather than over classic
intermediate representations like those currently used in Compcert.
SSA enables more efficient static analyses and sometimes simpler code
transformations.  A typical SSA-based optimization that interests us
is global value numbering \cite{Rosen-Wegman-Zadeck}.  Since the
beginning of Compcert we have been considering using SSA-based
intermediate languages, but were held off by two difficulties.  First,
the dynamic semantics for SSA is not obvious to formalize (see
\cite{Blech-Glesner-05,Beringer-07,Pop-SSA-phd} for various approaches).
Second, the SSA property is global to the code of a
whole function and not straightforward to exploit locally within
proofs.  Functional representations such as A-normal forms could offer
some of the benefits of SSA with clearer semantics.  In a translation
validation setting, it might not be necessary to reason directly over
the SSA form: the untrusted optimizations could convert to SSA, use
efficient SSA-based algorithm and convert out of SSA; the validator,
which is the only part that needs proving, can still use conventional
\lang{RTL}-like representations.

\subsection{On memory} \label{s:on-memory}

Whether to give formal semantics to imperative languages or to reason
over pointer-based programs and transformations over these programs,
the memory model is a crucial ingredient.  The formalization of memory
used in Compcert can be extended and refined in several directions.

The first is to add memory allocation and deallocation primitives to
\lang{Cminor}, in the style of C's {\tt malloc} and {\tt free}.  Both can be
implemented in \lang{Cminor} by carving sub-blocks out of a large
global array, but primitive support for these operations could
facilitate reasoning over \lang{Cminor} programs.  Supporting dynamic
allocation is easy since it maps directly to the {\tt alloc} function of
our memory model.  (This model does not assume that allocations and
deallocations follow a stack-like discipline.)  Dynamic,
programmer-controlled deallocation requires more care: mapping it to
the {\tt free} operation of our memory model opens up the possibility that
a \lang{Cminor} function explicitly deallocates its own stack block.
This could invalidate semantic preservation: if the \lang{Cminor}
function does not use its stack block and does not terminate, nothing
wrong happens, but if some of its variables are spilled to memory, the
corresponding \lang{Mach} code could crash when accessing a spilled
variable.  Explicit deallocation of stack frames must therefore be
prevented at the \lang{Cminor} level, typically by tagging memory
blocks as belonging to the stack or to the heap.

Another limitation of our memory model is that it completely hides the
byte-level representation of integers and floats in memory.  This
makes it impossible to express in \lang{Cminor} some C programming
idioms of dubious legality but high practical usefulness such as
copying byte-per-byte an arbitrary data structure to another of the
same layout (in the style of {\tt memmove}).  Doing so in \lang{Cminor}
would fill the destination structure with {\tt undef} values.  As
mentioned in section~\ref{s:ppc-syntax}, this feature of our memory
model also prevents us from reasoning about machine code that manipulates
the IEEE representation of floats at the bit level.  As discussed in 
\cite[section 7]{Leroy-Blazy-memory-model}, a strong reason for hiding
byte-level in-memory representations is to ensure that pointer values
cannot be forged from integers or floats; this guarantee plays a
crucial role in proving semantic preservation for certain memory
transformations.  A topic for future work is to refine the memory
model to obtain the best of both worlds: unforgeable pointers and
byte-level access to the representations of integers and floats.

Compared with ``real'' memory implementations or even with the version
of our memory model presented in \cite{Leroy-Blazy-memory-model}, the
memory model currently used in Compcert makes two simplifying
assumptions: (1) {\tt free} never fails, therefore allowing repeated
deallocation of a given block; (2) {\tt alloc} never fails, therefore
modeling an infinite memory.  Assumption (1) is not essential, and
our proofs extend straightforwardly to showing that the compiler never
inserts a double {\tt free} operation.  
Assumption (2) on infinite memory is more difficult to
remove, because in general a compiler does not preserve the stack
memory consumption of the program it compiles.  It is easy to show
that the generated code performs exactly the same dynamic allocations
and deallocations as the source program; therefore, heap memory usage
is preserved.  However, in Compcert, the sizes of stack blocks can
increase arbitrarily between \lang{Cminor} and \lang{PPC}, owing to
the spilling of \lang{Cminor} variables to the stack described in
section~\ref{s:stacking}.  A \lang{Cminor} program that executes
correctly within $N$ bytes of stack space can therefore be translated
to a \lang{PPC} program that runs out of stack, significantly
weakening the semantic preservation theorems that we proved.

In other words, while heap memory usage is clearly preserved, it seems
difficult to prove a bound on stack usage on the source program and
expect this resource certification to carry over to compiled code:
stack consumption, like execution time, is a program property that is
not naturally preserved by compilation.  A simpler alternative is to
establish the memory bound directly on the compiled code.  If
recursion and function pointers are not used, which is generally the
case for critical embedded software, a simple, verified static
analysis over \lang{Mach} code that approximates the call graph can
provide the required bound $N$ on stack usage.  We could, then, prove
a strong semantic preservation theorem of the form ``if the \lang{PPC}
stack is of size at least $N$, the generated \lang{PPC} code behaves
like the source \lang{Cminor} program''.  For programs that use
recursion or function pointers, the issue remains open, however.

\subsection{On the multiplicity of passes and intermediate languages}

The number of compilation passes in Compcert is relatively high by
compiler standards, but not shockingly so.  The main motivation here
was to have passes that do exactly one thing each, but do it well and
in a complete manner.  Combining several passes together tends to
complicate their proofs super-linearly.  For example, the
first published version of Compcert \cite{Leroy-compcert-06} performed
register allocation, reloading, spilling and enforcement of calling
conventions all in one pass.  Splitting this pass in two (register
allocation in section~\ref{s:regalloc}, spilling in
section~\ref{s:reload}) resulted in a net simplification of the
proofs.

What is more surprising is the high number of intermediate languages
involved in Compcert: with the exception of dataflow-based
optimizations (constant propagation, CSE, and tunneling), each pass
introduces a new intermediate language to use as its target language.
Many of these intermediate languages are small variations on one
another.  Yet we found it necessary to define each intermediate
language separately, rather than identifying them with subsets of
a small number of more general intermediate languages (like for
instance GCC does with its Tree and RTL representations).  

The general problem we face is that of transmitting information
between compiler passes: what are the properties of its output that a
compiler pass guarantees and that later passes can rely on?  
These guarantees can be positive properties (\eg ``the \lang{Mach} code
generated by the stack layout pass is compatible with treating
registers as global variables'') but also negative, ``don't care''
properties (\eg ``the \lang{LTL} code generated by register allocation
does not use temporary registers and is insensitive to modifications
of caller-save registers across function calls'').  

Some of these guarantees can be captured syntactically.  For instance,
code produced by register allocation never mentions temporary
registers.  In Compcert, such syntactic guarantees for a compiler pass
are enforced either by the abstract syntax of its target
language or via an additional inductive predicate on this abstract
syntax.  As mentioned in section~\ref{s:cminor-typing}, all
intermediate languages from \lang{Cminor} to \lang{Mach} are weakly
typed in a {\tt int}-or-{\tt float} type system.  The corresponding typing
rules are a good place to carry additional restrictions on the syntax
of intermediate languages.  For example, \lang{LTL}'s typing rules
enforce the restriction that all locations used are either
non-temporary registers or stack slots of the {\tt local} kind.

Other guarantees are semantic in nature and cannot be expressed by
syntactic restriction.  One example is the fact that \lang{LTL} code
generated by register allocation does not expect caller-save registers
to be preserved across function calls.  Such properties need to be
reflected in the dynamic semantics for the target language of the
considered pass.  Continuing the previous example, the \lang{LTL}
semantics captures the property by explicitly setting caller-save
registers to the {\tt undef} value when executing a function call
(section~\ref{s:ltl}).  Later passes can then refine these {\tt undef}
values to whatever value is convenient for them
(section~\ref{s:reload-proof}). 

In other cases, we must guarantee that an intermediate representation
not only doesn't care about the values of some locations but actually
preserves whatever values they hold.  Just setting these locations to
{\tt undef} is not sufficient to capture this guarantee.  The approach we
followed is to anticipate, in the semantics of an intermediate
language, the actual values that these locations will take after later
transformations.  For example, the \lang{Linear} semantics
(section~\ref{s:linear}) anticipates the saving and reloading of
callee-save registers performed by \lang{Mach} code generated by the
stack layout pass (section~\ref{s:stacking}).  Another instance of
this technique is the semantics of \lang{Mach} that anticipates the
return address that will be stored by \lang{PPC} code in the slot of
the activation record reserved for this usage (sections~\ref{s:mach}
and~\ref{s:ppcgen-proof}).  While this technique of semantic
anticipation was effective in Compcert, it is clearly not as modular
as one would like: the semantics of an intermediate language becomes
uncomfortably dependent on the effect of later compilation passes.
Finding better techniques to capture behaviors of the form ``this
generated code does not depend on the value of $X$ and guarantees to
preserve the value of $X$'' is an open problem.

\subsection{Toward machine language and beyond} \label{s:hardware-verif}

The Compcert compilation chain currently stops at the level of an
assembly language following a Harvard architecture and equipped
with a C-like memory model.  The next step ``down'' would be machine
language following a von Neumann architecture and representing memory
as a finite array of bytes.  The main issue in such a refinement is to
bound the amount of memory needed for the call stack, as discussed in
section~\ref{s:on-memory}.  Once this is done, we believe that the
refinement of the memory model can be proved correct as an instance of
the memory embeddings studied by
Leroy and Blazy \cite[section~5.1]{Leroy-Blazy-memory-model}.  Additionally, symbolic
labels must be resolved into absolute addresses and relative
displacements, and instructions must be encoded and stored in memory.
Based on earlier work such as \cite{Moore-96}, such a refinement
is likely to be tedious but should not raise major
difficulties.

The main interest in going all the way to machine language is to
connect our work with existing and future hardware verification
efforts at the level of instruction set architectures (\ie machine
language) and micro-architectures.  Examples of such hardware
verifications include the Piton project \cite{Moore-89,Moore-96}
(from a high-level assembly language to an NDL netlist for a
custom microprocessor), Fox's verification of
the ARM6 micro-architecture \cite{Fox-ARM6-03},
and the VAMP project \cite{Beyer-VAMP}
(from the DLX instruction set to a gate-level description of a
processor).
Sharing a specification of an instruction set architecture between the
verification of a compiler and the verification of a hardware
architecture strengthens the confidence we can have in both
verifications.

\subsection{Toward shared-memory concurrency}

Shared-memory concurrency is back into fashion these days and is
infamous for raising serious difficulties both with the verification of
concurrent source programs and with the reuse, in a concurrent
setting, of languages and compilers designed for sequential
execution \cite{Boehm-threads-library}.  An obvious question,
therefore, is whether the Compcert back-end and its soundness proof
could be extended to account for shared-memory concurrency.

It is relatively easy to give naive interleaving semantics for
concurrency in \lang{Cminor} and the other languages of Compcert, but
semantic preservation during compilation obviously fails: if arbitrary
data races are allowed in the source \lang{Cminor} program, the
transformations performed by the compiler introduce additional
interleavings and therefore additional behaviors not present in the
source program.  For instance, the evaluation of an expression is an
atomic step in the \lang{Cminor} semantics but gets decomposed into
several instructions during compilation.  The weakly consistent
hardware memory models implemented by today's processors add even more
behaviors that cannot be predicted easily by the \lang{Cminor}
semantics.

Our best hope to show a semantic preservation result in a
shared-memory concurrent setting is, therefore, to restrict ourselves
to race-free source programs that implement a proper mutual exclusion
discipline on their memory accesses.  A powerful way to characterize
such programs is concurrent separation logic \cite{OHearn-07}.  Using
this approach, Hobor et al. \cite{Hobor-Appel-Zappa-08} develop an operational
semantics for Concurrent Cminor, an extension of \lang{Cminor} with
threads and locks.  This semantics is pseudo-sequential in that
threads run sequentially between two operations over locks, and their
interleaving is determined by an external {\em oracle} that appears as
a parameter to the semantics.  It is conceivable that, for a fixed but
arbitrary oracle, the Compcert proofs of semantic preservation would
still hold.  The guarantees offered by concurrent separation logic
would then imply that the pseudo-sequential semantics for the
generated \lang{PPC} code captures all possible actual executions of
this code, even in the presence of arbitrary interleavings and
weakly-consistent hardware memory.  This approach is very promising,
but much work remains to be done.

\section{Related work} \label{s:related}

We have already discussed the relations between
compiler verification and other approaches to trusted compilation in
section~\ref{s:framework}.
Proving the correctness of compilers has been an active research
topic for more than 40 years, starting with the seminal work of
McCarthy and Painter \cite{McCarthy-Painter-67}.  Since then, a great many on-paper proofs
for program analyses and compiler transformations have been published
--- too many to survey here.  Representative examples include the works
of Clemmensen and Oest \cite{Clemmensen-84},
Chirica and Martin \cite{Chirica-86},
Guttman et al. \cite{Guttman-VLISP},
M\"{u}ller-Olm \cite{Muller-phd-97} and
Lacey et al. \cite{Lacey-Jones-02}.  We refer the reader to Dave's annotated
bibliography \cite{Dave-compiler-verif-03} for further references.
In the following, we restrict ourselves to discussing correctness
proofs of compilers that involve mechanized verification.

Milner and Weyhrauch \cite{Milner-compiler-correctness} were arguably the first to
mechanically verify semantic preservation for a compilation algorithm
(the translation of arithmetic expressions to stack machine code),
using the Stanford LCF prover.  The first mechanized verification of a
full compiler is probably that of Moore \cite{Moore-89,Moore-96},
although for a rather low-level, assembly-style source language.

The Verifix project \cite{Goos-Zimmermann-99} had goals broadly similar to
ours: the construction of mathematically correct compilers.  It
produced several methodological approaches, but few parts of this
project led to machine-checked proofs.  The parts closest to the
present work are the formal verification (in PVS)
of a compiler from a subset of Common Lisp to Transputer code
\cite{Dold-Vialard-01} and of instruction selection by a bottom-up
rewrite system \cite{Dold-Gaul-Zimmerman-98}.  

Strecker \cite{Strecker-02} and Klein and Nipkow \cite{Klein-Nipkow-jinja} verified
non-optimizing byte-code compilers from a subset of Java to a subset
of the Java Virtual Machine using Isabelle/HOL.  They did not address
compiler optimizations nor generation of actual machine code.  Another
verification of a byte-code compiler is that of Gr\'{e}goire \cite{Gregoire-phd},
for a functional language.

The Verisoft project \cite{Verisoft-project} is an ambitious attempt at
end-to-end formal verification that covers the whole spectrum from
application to hardware.  The compiler part of Verisoft is the formal
verification of a compiler for a Pascal-like language called C0
down to DLX machine code using the Isabelle/HOL proof assistant
\cite{Leinenbach-Paul-05,Strecker-C0-05,Leinenbach-Petrova-08}.
The publications on this verification lack details but suggest
a compiler that is much simpler than Compcert and generates
unoptimized code.

Li et al. \cite{Li-Owens-Slind-07,Li-Slind-07} describe an original
approach to the generation of trusted ARM code where the input
language is a subset of the HOL specification language.  Compilation
is not entirely automatic: the user chooses interactively which
transformations to apply, but the system produces formal evidence (a
HOL proof term) that the generated ARM code conforms to the HOL
specification.

Chlipala \cite{Chlipala-07} developed and proved correct a compiler for
simply-typed $\lambda$-calculus down to an idealized assembly
language.  This Coq development cleverly uses dependent types and
type-indexed denotational semantics, resulting in remarkably compact
proofs.  Another Coq verification of a compiler from a simply-typed
functional language to an idealized assembly language is that of
Benton and Hur \cite{Benton-Hur}.  Like Chlipala's, their proof has a strong
denotational semantics flavor; it builds upon the concepts of
step-indexed logical relations and biorthogonality.  It is unclear yet
whether such advanced semantic techniques can be profitably applied to
low-level, untyped languages such as those considered in this paper,
but this is an interesting question.

The formal verification of static analyses, usable both to support
compiler optimizations or to establish safety properties of programs,
has received much attention.  Considerable efforts have been expended
on formally verifying Java's dataflow-based bytecode verification;
see \cite{Hartel-01,Leroy-bytecode-verification-03} for a survey.
Cachera et al. \cite{Cachera-Jensen-05} and Pichardie \cite{Pichardie-phd} develop a
framework for abstract interpretation and dataflow analysis, formally
verified using Coq.  A related project is 
Rhodium \cite{Lerner-Chambers-05}, a domain-specific language to
describe program analyses and transformations.  From a Rhodium
specification, both executable code and an automatically-verified
proof of semantic preservation are generated.  Rhodium achieves
a high degree of automation, but applies only to the
optimization phases of a compiler and not to the non-optimizing
translations from one language to another.

\section{Conclusions} \label{s:conclusions}

The formal verification of a compiler back-end presented in this
article provides strong evidence that the initial goal of formally
verifying a realistic compiler can be achieved, within the limitations
of today's proof assistants, and using only elementary semantic and
algorithmic approaches.  It is, however, just one exploration within a
wide research area: the certification, using formal
methods, of the verification tools, code generators, compilers and
run-time systems that participate in the development, validation 
and execution of critical software.  In addition, we hope that this work
also contributes to renewing scientific interest in the semantic
understanding of compiler technology, in mechanized operational
semantics, and in integrated environments for programming and proving.

Looking back at what was achieved, we did not completely rule out all
uncertainties concerning the soundness of the compiler, but reduced
the problem of trusting the whole compiler down to trusting
(1) the formal semantics for the source (\lang{Cminor}) and target
(\lang{PPC}) languages; (2) the compilation chain used to produce the
executable for the compiler (Coq's extraction facility and the OCaml
compiler); and (3) the Coq proof assistant itself.  Concerning (3), it
is true that an inconsistency in Coq's logic or a bug in Coq's
implementation could theoretically invalidate all the guarantees we
obtained about Compcert.  As Hales \cite{Hales-formal-proof} argues, this
is extremely unlikely, and proofs mechanically checked by a proof
assistant that generates proof terms are orders of magnitude
more trustworthy than even carefully hand-checked mathematical proofs.

To address concern (2), ongoing work within the Compcert project
studies the feasibility of formally verifying Coq's extraction
mechanism and a compiler from \lang{Mini-ML} (the target language for
this extraction) to \lang{Cminor}.  Composed with the Compcert back-end,
these efforts could eventually result in a trusted execution path for
programs written and verified in Coq, like Compcert itself, therefore
increasing confidence further through a form of bootstrapping.

The main source of uncertainty is concern (1): do the formal
semantics of \lang{Cminor} and \lang{PPC}, along with the underlying
memory model, capture the intended behaviors?  One
could argue that they are small enough (about 2500 lines of Coq) to
permit manual review.  Another effective way to increase confidence in these
semantics is to use them in other formal verifications, such as the
\lang{Clight} to \lang{Cminor} and \lang{Mini-ML} to \lang{Cminor}
front-ends developed within the Compcert project, and the axiomatic
semantics for \lang{Cminor} of \cite{Appel-Blazy-07}.  Future work in
this direction could include connections with architectural-level
hardware verification (as outlined in section~\ref{s:hardware-verif})
or with verifications of program provers, model checkers and static
analyzers for C-like languages.  Drawing and formalizing such
connections would not only strengthen even further the confidence we
can have in each component involved, but also progress towards the
availability of high-assurance environments for the development and
validation of critical software.

\begin{acknowledgements}
This work was supported in part by Agence Nationale de la Recherche,
grant number ANR-05-SSIA-0019.  We would like to thank the
following colleagues for their input.  
E.~Ledinot made us aware of the need for a formally verified
C~compiler in the aerospace industry.  Several of the techniques
presented here were discussed and prototyped with members of the
Concert INRIA coordinated research action and of the Compcert
ANR advanced research project, especially Y.~Bertot, S.~Blazy,
B.~Gr\'egoire, P.~Letouzey and L.~Rideau.
Our formalization of dataflow analyses
and constant propagation builds on that of B.~Gr\'egoire.  L.~Rideau
and B.~Serpette contributed the hairy correctness proof for parallel
moves.  D.~Doligez proved properties of finite maps and type
reconstruction for RTL.  A.~Appel and S.~Blazy invented the
continuation-based transition semantics for \lang{Cminor}.  G.~Necula
argued convincingly in favor of a posteriori verification of static
analyses and optimizations.
Last but not least, the author is grateful to the JAR anonymous
reviewers for their meticulous reading of this long article and for
their helpful suggestions for improvements.  One of them went so far
as to suggest the union-find based algorithm for branch tunneling
presented in section~\ref{s:tunneling}, which is original to the best
of our knowledge and more elegant than our previous bounded-depth
algorithm.

\end{acknowledgements}

\bibliographystyle{spmpsci}
\ifx\french\undefined \def\biling#1#2{#1} \else \def\biling#1#2{#2}
  \fi\ifx\abbrevbib\undefined \def\abbrev#1#2{#1} \else \def\abbrev#1#2{#2} \fi

\end{document}